\documentclass[11pt]{article}
\usepackage{amsmath}
\usepackage{amsthm}
\usepackage{amssymb}
\usepackage{algorithm}
\usepackage{subfig}
\usepackage{color}
\usepackage[english]{babel}
\usepackage{graphicx}
\usepackage{wrapfig,epsfig}
\usepackage{epstopdf}
\usepackage{url}
\usepackage{graphicx}
\usepackage{color}
\usepackage{epstopdf}
\usepackage{algpseudocode}
\usepackage[T1]{fontenc}
\usepackage{bbm}

\usepackage{tikz}
\usepackage{hyperref}
\hypersetup{colorlinks=true,citecolor=blue,linkcolor=blue} %%% Zhao : maybe we should comment this in submission
\usetikzlibrary{arrows}
\usepackage[margin=1in]{geometry}
\author{
  Zhao Song\thanks{Work done while visiting IBM Almaden, and supported in part by UTCS TAship (CS429 Fall 16 Computer Organization and Architecture and CS395T Fall 16 Sublinear Algorithms).}\\
  \texttt{zhaos@utexas.edu}\\
  UT-Austin
  \and
  David P. Woodruff\thanks{Supported in part by the XDATA program of the Defense Advanced Research Projects Agency (DARPA), administered through Air Force Research Laboratory contract FA8750-12-C-0323.}\\
  \texttt{dpwoodru@us.ibm.com}\\
  IBM Almaden
  \and
  Peilin Zhong\thanks{Supported in part by Simons Foundation, and NSF CCF-1617955.}\\
  \texttt{peilin.zhong@columbia.edu}\\
  Columbia University
}

\date{}
\title{Low Rank Approximation with Entrywise $\ell_1$-Norm Error\thanks{A preliminary version of this paper appears in Proceedings of the 49th Annual ACM SIGACT Symposium on the Theory of Computing (STOC 2017).}}
\newtheorem{theorem}{Theorem}[section]
\newtheorem{lemma}[theorem]{Lemma}
\newtheorem{definition}[theorem]{Definition}

\newtheorem{corollary}[theorem]{Corollary}

\newtheorem{fact}[theorem]{Fact}
\newtheorem{remark}[theorem]{Remark}
\newtheorem{claim}[theorem]{Claim}

\newtheorem{hypothesis}[theorem]{Hypothesis}

\newcommand{\wh}{\widehat}
\newcommand{\wt}{\widetilde}

\newcommand{\eps}{\epsilon}
\newcommand{\N}{\mathcal{N}}

\newcommand{\RHS}{\mathrm{RHS}}

\DeclareMathOperator*{\E}{{\bf {E}}}

\DeclareMathOperator{\OPT}{OPT}
\DeclareMathOperator{\supp}{supp}
\DeclareMathOperator{\poly}{poly}

\DeclareMathOperator{\nnz}{nnz}

\DeclareMathOperator{\rank}{rank}

\DeclareMathOperator{\EMD}{EMD}
\DeclareMathOperator{\EEMD}{EEMD}
\DeclareMathOperator{\constraints}{constraints}
\DeclareMathOperator{\degree}{degree}
\DeclareMathOperator{\variables}{variables}
\DeclareMathOperator{\TV}{TV}
\DeclareMathOperator{\KL}{KL}
\DeclareMathOperator{\cost}{cost}
\newcommand{\SAT}{$\mathsf{3}$-$\mathsf{SAT}$~}
\newcommand{\MAXCUT}{$\mathsf{MAX}$-$\mathsf{CUT}$~}
\newcommand{\CNF}{$\mathsf{CNF}$~}
\newcommand{\ETH}{$\mathsf{ETH}$} %%% Not sure.

\makeatletter
\newcommand*{\RN}[1]{\expandafter\@slowromancap\romannumeral #1@}
\makeatother

\usepackage{lineno}

\usepackage{etoolbox}

\makeatletter

\newcommand{\define}[4][ignore]{%
  \ifstrequal{#1}{ignore}{}{
  \@namedef{thmtitle@#2}{#1}}%
  \@namedef{thm@#2}{#4}%
  \@namedef{thmtypen@#2}{lemma}%
  \newtheorem{thmtype@#2}[theorem]{#3}%
  \newtheorem*{thmtypealt@#2}{#3~\ref{#2}}%
}

\newcommand{\state}[1]{%
  \@namedef{curthm}{#1}
  \@ifundefined{thmtitle@#1}{
  \begin{thmtype@#1}
    }{
  \begin{thmtype@#1}[\@nameuse{thmtitle@#1}]
  }
    \label{#1}
    \@nameuse{thm@#1}
  \end{thmtype@#1}
  \@ifundefined{thmdone@#1}{
  \@namedef{thmdone@#1}{stated}%
  }{}
}

\newcommand{\restate}[1]{%
  \@namedef{curthm}{#1}
  \@ifundefined{thmtitle@#1}{
    \begin{thmtypealt@#1}
    }{
  \begin{thmtypealt@#1}[\@nameuse{thmtitle@#1}]
  }
    \@nameuse{thm@#1}
  \end{thmtypealt@#1}
  \@ifundefined{thmdone@#1}{
  \@namedef{thmdone@#1}{stated}%
  }{}
}

\newcommand{\thmlabel}[1]{
  \@ifundefined{thmdone@\@nameuse{curthm}}{\label{#1}
    }{\tag*{\eqref{#1}}}
}
\makeatother
\begin{document}
%\linenumbers

\begin{titlepage}
  \maketitle
  \begin{abstract}
We study the $\ell_1$-low rank approximation problem, where for a given $n \times d$ matrix $A$ and approximation factor $\alpha \geq 1$, the goal is to output a rank-$k$ matrix $\widehat{A}$ for which
$$\|A-\widehat{A}\|_1 \leq \alpha \cdot \min_{\textrm{rank-}k\textrm{ matrices} ~A'}\|A-A'\|_1,$$
where for an $n \times d$ matrix $C$, we let $\|C\|_1 = \sum_{i=1}^n \sum_{j=1}^d |C_{i,j}|$. This error measure is
known to be more robust than
the Frobenius norm in the presence of outliers
and is indicated in models where Gaussian assumptions
on the noise may not apply.  
The problem was shown to be NP-hard by Gillis and Vavasis and a number
of heuristics have been proposed. It was
asked in multiple places if there are any approximation algorithms.

We give
the first provable approximation algorithms for $\ell_1$-low rank approximation, showing that it is possible
to achieve approximation factor $\alpha = (\log d) \cdot \poly(k)$ in
$\nnz(A) + (n+d) \poly(k)$ time, where $\nnz(A)$ denotes the number of
non-zero entries of $A$. If $k$ is constant, we further improve
the approximation ratio to $O(1)$ with a $\poly(nd)$-time algorithm. Under
the Exponential Time Hypothesis, we show there is no $\poly(nd)$-time
algorithm achieving a $(1+\frac{1}{\log^{1+\gamma}(nd)})$-approximation,
for $\gamma > 0$ an arbitrarily small constant, even when $k = 1$.

We give a number of additional 
results for $\ell_1$-low rank approximation: nearly tight upper and lower
bounds for column subset selection, CUR decompositions, extensions to low rank
approximation with respect to $\ell_p$-norms for $1 \leq p < 2$
and earthmover distance, low-communication distributed
protocols and low-memory streaming algorithms, algorithms with
limited randomness, and bicriteria algorithms.
We also give a preliminary empirical evaluation.

  \end{abstract}
  \thispagestyle{empty}
\end{titlepage}

%\thispagestyle{empty}
%\clearpage
%\setcounter{page}{1}

%\newpage
%\clearpage
%\setcounter{page}{1}
{
\hypersetup{linkcolor=black}
\tableofcontents
}
\newpage

\section{Introduction}\label{sec:intro}
Two well-studied problems in numerical linear algebra are regression and low rank approximation. In regression, one is given an $n \times d$ matrix $A$, and an $n \times 1$ vector $b$, and one seeks an $x \in \mathbb{R}^d$ which minimizes $\|Ax-b\|$ under some norm. For example, for least squares regression one minimizes $\|Ax-b\|_2$. In low rank approximation, one is given an $n \times d$ matrix $A$, and one seeks a rank-$k$ matrix $\widehat{A}$ which minimizes $\|A-\widehat{A}\|$ under some norm. For example, in Frobenius norm low rank approximation, one minimizes $\|A-\widehat{A}\|_F = \left (\sum_{i,j} (A_{i,j} - \widehat{A}_{i,j})^2 \right )^{1/2}$. Algorithms for regression are often used as subroutines for low rank approximation. Indeed, one of the main insights of \cite{dmm06,dmm06b,s06,dmm08,cw09} was to use results for generalized least squares regression for Frobenius norm low rank approximation. Algorithms for $\ell_1$-regression, in which one minimizes $\|Ax-b\|_1 = \sum_i |(Ax)_i - b_i|$, were also used \cite{bd13,sw11} to fit a set of points to a hyperplane, which is a special case of entrywise $\ell_1$-low rank approximation, the more general problem being to find a rank-$k$ matrix $\widehat{A}$ minimizing $\sum_{i,j} |A_{i,j} - \widehat{A}_{i,j}|$. 

Randomization and approximation were introduced to significantly speed up algorithms for these problems, resulting in algorithms achieving relative error approximation with high probability. Such algorithms are based on sketching and sampling techniques; we refer to \cite{w14} for a survey. For least squares regression, a sequence of work \cite{s06,cw13,mm13,nelson2013osnap,lmp13,bdn15,c16} shows how to achieve algorithms running in $\nnz(A) + \poly(d)$ time. For Frobenius norm low rank approximation, using the advances for regression this resulted in $\nnz(A) + (n+d)\poly(k)$ time algorithms. 
For $\ell_1$-regression, sketching and sampling-based methods \cite{c05,sw11,cdmmmw13,cw13,mm13,lmp13,wz13,cw15a,cp15} led to an $\nnz(A) + \poly(d)$ time algorithm.

Just like Frobenius norm low rank approximation is the analogue of least squares regression, entrywise $\ell_1$-low rank approximation is the analogue of $\ell_1$-regression. Despite this analogy, {\it no non-trivial upper bounds with provable guarantees} are known for $\ell_1$-low rank approximation. Unlike Frobenius norm low rank approximation, which can be solved exactly using the singular value decomposition, no such algorithm or closed-form solution is known for $\ell_1$-low rank approximation. Moreover, the problem was recently shown to be NP-hard \cite{gv15}. A major open question is whether there exist approximation algorithms, sketching-based or otherwise, for $\ell_1$-low rank approximation. Indeed, the question of obtaining betters algorithms was posed in section 6 of \cite{gv15}, in \cite{stack:pose}, and as the second part of open question 2 in \cite{w14}, among other places. The earlier question of NP-hardness was posed in Section 1.4 of \cite{kv09}, for which the question of obtaining approximation algorithms is a natural followup. The goal of our work is to answer this question. 

We now formally define the $\ell_1$-low rank approximation problem:
we are given an $n \times d$ matrix $A$ and approximation factor $\alpha \geq 1$, and we would like, with large constant probability, to output a rank-$k$ matrix $\widehat{A}$ for which
\begin{eqnarray*}%\label{eqn:guarantee}
  \|A-\widehat{A}\|_1 \leq \alpha \cdot \min_{\textrm{rank-}k\textrm{ matrices~}A'}\|A-A'\|_1,
  \end{eqnarray*}
where for an $n \times d$ matrix $C$, we let $\|C\|_1 = \sum_{i=1}^n \sum_{j=1}^d |C_{i,j}|$. This notion of low rank approximation has been proposed as a more robust alternative
to Frobenius norm low rank approximation \cite{kk03,kk05,kimefficient,kwak08,zlsyO12,bj12,bd13,bdb13,meng2013cyclic,mkp13,mkp14,mkcp16,park2016iteratively}, and is sometimes referred to as $\ell_1$-matrix factorization or robust PCA. $\ell_1$-low rank approximation gives improved results over Frobenius norm low rank approximation since outliers are less exaggerated, as one does not square their contribution in the objective. The outlier values are often erroneous values that are far away from the nominal data, appear only a few times in the data matrix, and would not appear again under normal system operation. These works also argue $\ell_1$-low rank approximation can better handle missing data, is appropriate in noise models for which the noise is not Gaussian, e.g., it produces the maximum likelihood estimator for Laplacian noise \cite{g08,kahl2008practical,vt01}, and can be used in image processing to prevent image occlusion \cite{y12}.
%The increased interest in $\ell_1$-low rank approximation may also be partly attributed to the popularity of sparse recovery methods that rely on $\ell_1$-based calculations for signal reconstruction (see, e.g., \cite{d06}). 

To see that $\ell_1$-low rank approximation and Frobenius norm low rank approximation can 
give very different results, consider the $n \times n$ matrix $A = \left [ \begin{smallmatrix} n &0\\0&B \end{smallmatrix} \right ]$, where $B$ is {\it any} $(n-1) \times (n-1)$ matrix with $\|B\|_F < n$. The
best rank-$1$ approximation with Frobenius norm error is given by $\widehat{A} = n \cdot e_1 e_1^\top$, where $e_1$ is the first standard unit vector. Here $\widehat{A}$ ignores all but the first row and column of $A$, which may be undesirable in the case that this row and column represent an outlier. Note $\|A-\widehat{A}\|_1 = \|B\|_1$. If, for example, $B$ is the all $1$s matrix, then $\widehat{A} = [0, 0; 0, B]$ is a rank-$1$ approximation for which $\|A-\widehat{A}\|_1 = n$, and therefore this solution is a much better solution to the $\ell_1$-low rank approximation problem than $n \cdot e_1 e_1^\top$, for which $\|A-n\cdot e_1 e_1^\top\|_1 = (n-1)^2.$

Despite the advantages of $\ell_1$-low rank approximation, its main disadvantage is its computationally intractability. It is not rotationally invariant and most tools for Frobenius low rank approximation do not apply. To the best of our knowledge, all previous works only provide heuristics. We provide hard instances for previous work in Section \ref{sec:exp}, showing these algorithms at best give a $\poly(nd)$-approximation (though even this is not shown in these works). We also mention why a related objective function, robust PCA \cite{wgrpm09,clmw11,nnasj14,nyh14,chd16,zzl15}, does not give a provable approximation factor for $\ell_1$-low rank approximation. Using that for an $n \times d$ matrix $C$, $\|C\|_F \leq \|C\|_1 \leq \sqrt{nd} \|C\|_F$, a Frobenius norm low rank approximation gives a $\sqrt{nd}$ approximation for $\ell_1$-low rank approximation. A bit better is to use algorithms for low rank approximation with respect to the sum of distances, i.e., to find a rank-$k$ matrix $\widehat{A}$ minimizing $\|A-\widehat{A}\|_{1,2}$, where for an $n \times d$ matrix $C$, $\|C\|_{1,2} = \sum_{i=1}^n \|C_i\|_2$, where $C_i$ is the $i$-th row of $C$. A sequence of work \cite{dv07,fmsw10,fl11,sv12,cw15b} shows how to obtain an $O(1)$-approximation to this problem in $\nnz(A) + (n+d)\poly(k) + \exp(k)$ time, and using that $\|C\|_{1,2} \leq \|C\|_1 \leq \sqrt{d}\|C\|_{1,2}$ results in an $O(\sqrt{d})$-approximation.

There are also many variants of Frobenius norm low rank approximation for which nothing is known for $\ell_1$-low rank approximation, such as column subset selection and CUR decompositions, distributed and streaming algorithms, algorithms with limited randomness, and bicriteria algorithms. Other interesting questions include low rank approximation for related norms, such as $\ell_p$-low rank approximation in which one seeks a rank-$k$ matrix $\widehat{A}$ minimizing $\sum_{i=1}^n \sum_{j=1}^d (A_{i,j}-\widehat{A}_{i,j})^p$. Note for $1 \leq p < 2$ these are also more robust than the SVD. 

\subsection{Our Results}
We give the first efficient algorithms for $\ell_1$-low rank approximation with provable approximation guarantees. By symmetry of the problem, we can assume $d \leq n$. We first give an algorithm which runs in $O(\nnz(A)) + n \cdot \poly(k)$ time and solves the $\ell_1$-low rank approximation problem with approximation factor $(\log d) \cdot \poly(k)$. This is an exponential improvement over the previous approximation factor of $O(\sqrt{d})$, provided $k$ is not too large, and is polynomial time for every $k$. Moreover, provided $\nnz(A) \geq n \cdot \poly(k)$, our time is optimal up to a constant factor as any relative error algorithm must spend $\nnz(A)$ time. We also give
a hard instance for our algorithm ruling out $\frac{\log d}{k \log k} + k^{1/2-\gamma}$ approximation for arbitrarily small constant $\gamma > 0$,
and hard instances for a general class of algorithms based on linear sketches, ruling out $k^{1/2-\gamma}$ approximation. 
%We also show for $t \leq \poly(k)$, such sketches
%cannot achieve better than a $k^{1/2-\gamma}$ approximation for $\gamma > 0$ an arbitrarily small constant. 

Via a different algorithm, we show how to achieve an $\wt{O}(k)$-approximation factor in \\$\poly(n)d^{\wt{O}(k)}2^{\wt{O}(k^2)}$ time. This is useful for constant $k$, for which it gives an $O(1)$-approximation in $\poly(n)$ time, improving the $O(\log d)$-approximation for constant $k$ of our earlier algorithm.
%Here we make the mild assumption that the entries of $A$ can each be specified with $\poly(nd)$ bits (that is, the entries can be as large as $\exp(\poly(nd))$), whereas the algorithm in the previous paragraph makes no assumptions on the bit complexity.
The approximation ratio of this algorithm, although $O(1)$ for constant $k$, depends on $k$.
%We also show if the entries of $A$ are in the range $\{-b, -b+1, \ldots, b\}$ for an
%integer 
%$b \leq \poly(n)$, then
We also show one can find a rank-$2k$ matrix $\widehat{A}$ in $\poly(n)$ time for constant $k$ for which 
$\|A-\widehat{A}\|_1 \leq C \min_{\textrm{rank-}k\textrm{ matrices~}A'}\|A-A'\|_1$, where $C > 1$ is an absolute constant independent of $k$. We refer to this as a
bicriteria algorithm.
Finally, one can output a rank-$k$ matrix $\widehat{A}$, instead of a rank-$2k$ matrix $\widehat{A}$, in $\poly(n)$ time with the same absolute constant
$C$ approximation factor, under an additional assumption that the entries of $\widehat{A}$ are integers in the range $\{-b, -b+1, \ldots, b\}$ for an
integer $b \leq \poly(n)$. Unlike our previous algorithms, this very last algorithm has a bit complexity assumption, and runs in $\poly(b)$ time instead of $\poly(\log(b))$ time.

Under the Exponential Time Hypothesis ($\mathsf{ETH}$), we show there is no $\poly(n)$-time
algorithm achieving a $(1+\frac{1}{\log^{1+\gamma}(n)})$-approximation,
for $\gamma > 0$ an arbitrarily small constant, even when $k = 1$. The latter
strengthens the NP-hardness result of \cite{gv15}. 

We also give a number of results for variants of $\ell_1$-low rank approximation which are
studied for Frobenius norm low rank approximation; prior to our work nothing was known about these problems. 

{\bf Column Subset Selection and CUR Decomposition:} In the column subset selection problem, one seeks a
small subset $C$ of columns of $A$ for which there is a matrix $X$ for which $\|CX-A\|$ is small, under some
norm. The matrix $CX$ provides a low rank approximation to $A$ which is often more interpretable, since it
stores actual columns of $A$, preserves sparsity, etc. These have been extensively studied when
the norm is the Frobenius or operator norm (see, e.g., \cite{bmd09,dr10,bdm11} and the references therein).
We initiate the study of this problem with respect to
the $\ell_1$-norm. We first prove an existence result, namely, that there exist matrices $A$ for which any
subset $C$ of $\poly(k)$ columns satisfies
$\min_X \|CX-A\|_1 \geq k^{1/2-\gamma} \cdot \min_{\textrm{rank-}k\textrm{ matrices~}A'}\|A-A'\|_1$, where $\gamma > 0$ is an arbitrarily small constant. 
%Note that if $C$ has fewer than $k$ columns, the error can
%be infinite if $A$ has rank $k$.
This result 
is in stark contrast to the Frobenius norm for which for every matrix there exist $O(\frac{k}{\epsilon})$ columns
for which the approximation factor is $1+\epsilon$.
We also show that our bound is nearly optimal in this regime,
by showing for every matrix there exists a subset of $O(k \log k)$
columns providing an $O(\sqrt{k \log k})$-approximation.
One can find such columns in $\poly(n) d^{O(k \log k)}$ time by enumerating
and evaluating the cost of each subset. Although this is exponential in $k$, 
we show it is possible to find $O(k \log k)$ columns providing a larger 
$O(k \log k \log d)$-approximation in polynomial time for every $k$.

We extend these results to the CUR decomposition problem (see, e.g., \cite{dmm08,bw14}), in which one seeks a factorization
$CUR$ for which $C$ is a subset of columns of $A$, $R$ is a subset of rows of $A$, and
$\|CUR-A\|$ is as small as possible. In the case of Frobenius norm, one can choose $O(k/\epsilon)$ columns and
rows, have rank$(U) = k$, have $\|CUR-A\|_F$ be at most $(1+\epsilon)$ times the optimal cost,
and find the factorization in $\nnz(A)\log n + n \cdot \poly( (\log n) k/\epsilon)$ time \cite{bw14}.
Using our column subset selection results, we give an $\nnz(A) + n \cdot \poly(k)$ time algorithm choosing
$O(k \log k)$ columns and rows, for which rank$(U) = k$, and for which $\|CUR-A\|_1$ is $\poly(k) \log d$
times the cost of any rank-$k$ approximation to $A$.

{\bf $\ell_p$-Low Rank Approximation and EMD-Low Rank Approximation:} We also give
the first algorithms with provable approximation guarantees for the $\ell_p$-low rank approximation problem,
$1 \leq p < 2$, 
in which we are given an $n \times d$ matrix $A$ and approximation factor $\alpha \geq 1$, and would like,
with large constant probability, to output a rank-$k$ matrix $\widehat{A}$ for which
\begin{eqnarray*}%\label{eqn:guarantee}
  \|A-\widehat{A}\|_p^p \leq \alpha \cdot \min_{\textrm{rank-}k\textrm{ matrices~} A'}\|A-A'\|_p^p,
  \end{eqnarray*}
where for an $n \times d$ matrix $C$, $\|C\|_p^p = \sum_{i=1}^n \sum_{j=1}^d |C_{i,j}|^p$. We obtain
similar algorithms for this problem as for $\ell_1$-low rank approximation. For instance, we obtain an
$\nnz(A) + n \cdot \poly(k)$ time algorithm with approximation ratio $(\log d) \cdot \poly(k)$.
%By embedding
%the earthmover distance (EMD) into $\ell_1$ with $O(\log d)$ distortion \cite{c02,it03},
We also provide the first low rank approximation with respect to sum of earthmover distances (of the $n$
rows of $A$ and $\widehat{A}$) with a $(\log^2 d) \poly(k)$ approximation factor. This low rank error measure was used, e.g., in \cite{sl09}.
Sometimes such applications also require a non-negative factorization, which we do not provide. 

{\bf Distributed/Streaming Algorithms, and Algorithms with Limited Randomness:}
There is a growing body of work on low rank approximation in the distributed (see, e.g., \cite{tisseur1999parallel,qu2002principal,bai2005principal,sensors2008,macua2010consensus,fegk13,poulson2013elemental,kvw14,bklw14,kdd16,bwz16,wz16})
and streaming models (see, e.g., \cite{mcgregor2006open,cw09,kl11,gp13,lib13,klmms14,woo14}), though almost exclusively for the Frobenius norm. One distributed
model is the {\it arbitrary partition model} \cite{kvw14} in which there are $s$ servers, each holding an $n \times d$ matrix
$A^i$, and they would like to 
output a $k \times d$ matrix $V^\top$ for which $\min_U \|UV^\top-A\|$ is as small as possible (or, a centralized coordinator
may want to output this).
%There is also a weaker {\it row partition model} in which the
%rows are partitioned across servers.
We give $\wt{O}(snk)$-communication algorithms achieving a $\poly(k, \log(n))$-approximation for $\ell_1$-low rank approximation in the arbitrary partition model, which is optimal for this approximation factor (see \cite{bw14} where lower bounds for Frobenius norm approximation with $\poly(n)$ multiplicative approximation were shown - such lower bounds apply to $\ell_1$ low rank approximation). We also consider the {\it turnstile streaming model} \cite{muthu} in which we receive positive or negative updates to its entries and wish to output a rank-$k$ factorization at the end of the stream. We give an algorithm using $\wt{O}(nk) + \poly(k)$ space to achieve a $\poly(k, \log(n))$-approximation, which is space-optimal for this approximation factor, up to the degree of the $\poly(k)$ factor.
To obtain these results, we show our algorithms can be implemented using $\wt{O}(dk)$ random bits. 

We stress for all of our results, we do not make assumptions on $A$ such as low coherence or condition number; our results hold for any $n \times d$ input matrix $A$. 

We report a promising preliminary empirical evaluation of our algorithms in Section \ref{sec:exp}. 

\begin{remark}
  We were just informed of the concurrent and independent work \cite{cgklp16},
  which also obtains approximation algorithms for $\ell_1$-low rank approximation. That paper obtains 
  a $2^{O(k)} \log d$-approximation in $(\log d)^k \poly(nd)$ time. Their algorithm is not 
  polynomial time once $k = \wt{\Omega}(\log d)$, whereas we obtain a polynomial time algorithm for every $k$ (in fact $\nnz(A) + (n+d)\poly(k)$
  time). Our approximation factor is also $\poly(k) \log d$, which is an exponential improvement over theirs in terms of $k$. In \cite{cgklp16}
  they also obtain a $2^k$-approximation in $\poly(nd) d^{O(k)}$ time. In contrast, we obtain an $\wt{O}(k)$-approximation
  in $\poly(nd) d^{\wt{O}(k)} 2^{\wt{O}(k^2)}$ time. The dependence in \cite{cgklp16} on $k$ in the approximation ratio is exponential, whereas ours is polynomial. 
  \end{remark}

\subsection{Technical Overview}
{\bf Initial Algorithm and Optimizations:} 
    Let $A^*$ be a rank-$k$ matrix
    for which 
    $\|A-A^*\|_1 = \min_{\textrm{rank-}k\textrm{ matrices~}A'}\|A-A'\|_1$. Let $A^* = U^* V^*$ be a factorization
    for which $U^*$ is $n \times k$ and $V^*$ is $k \times d$. Suppose we somehow
    knew $U^*$ and consider 
    the multi-response $\ell_1$-regression problem
    $\min_V \|U^*V - A\|_1 = \min_V \sum_{i=1}^d \|U^* V_i - A_i\|_1$, where $V_i, A_i$ denote the
    $i$-th columns of $V$ and $A$, respectively. We could solve this with linear programming
    though this is not helpful for our argument here.

    Instead, inspired by recent advances in sketching
    for linear algebra (see, e.g., \cite{w14} for a survey), we could choose a random matrix $S$
    and solve $\min_V \|SU^*V-SA\|_1 = \min_V \sum_{i=1}^d \|(SU^*)V_i - SA_i\|_1$.
    If $V$ is an approximate minimizer of the
    latter problem, we could hope $V$ is an approximate minimizer of the former problem.
    If also $S$ has a small number $t$ of rows, then we could instead solve
    $\min_V \sum_{i=1}^d \|(SU^*)V_i - SA_i\|_2$, that is, minimize the sum of Euclidean norms rather
    than the sum of $\ell_1$-norms. Since
    $t^{-1/2}\|(SU^*)V_i - SA_i\|_1 \leq \|(SU^*)V_i-SA_i\|_2 \leq \|(SU^*)V_i-SA_i\|_1$,
    we would obtain a $\sqrt{t}$-approximation to the problem $\min_V \|SU^*V-SA\|_1$. A crucial
    observation is that the solution to $\min_V \sum_{i=1}^d \|(SU^*)V_i - SA_i\|_2$ is
    given by $V = (SU^*)^\dagger SA$, which implies that $V$ is in the row span of $SA$. If also $S$
    were {\it oblivious} to $U^*$, then we could compute $SA$ without ever knowing $U^*$. Having a
    low-dimensional space containing a good solution in its span is our starting point. 

    For this to work, we need a distribution on oblivious
    matrices $S$ with a small number of rows, for which an approximate minimizer $V$ to
    $\min_V \|SU^*V-SA\|_1$ is also an approximate minimizer to $\min_V \|U^*V - A\|_1$.
    It is {\it unknown} if there exists a distribution on $S$
    with this property. What is known is that if $S$ has $O(d \log d)$ rows,
    then the Lewis weights (see, e.g., \cite{cp15} and references therein) 
    of the concatenated matrix $[U^*, A]$ give a distribution for which the optimal
    $V$ for the latter problem is a $(1+\epsilon)$-approximation to the former problem; see also
    earlier work on $\ell_1$-leverage scores \cite{c05,ddhkm09} which have $\poly(d)$ rows and the same
    $(1+\epsilon)$-approximation guarantee. Such distributions are 
    not helpful here as (1) they are not oblivious, and (2) the number $O(d \log d)$ of rows gives 
    an $O(\sqrt{d \log d})$ approximation factor, which is much larger than what we want. 

    There are a few oblivious distributions $S$ which are useful for {\it single-response}
    $\ell_1$-regression $\min\|U^*v - a\|_1$ for column vectors $v, a \in \mathbb{R}^k$
    \cite{sw11,cdmmmw13,wz13}.
    In particular, if $S$ is an $O(k \log k) \times n$
    matrix of i.i.d. Cauchy random variables, then the solution $v$ to $\min \|SU^*v-Sa\|_1$
    is an $O(k \log k)$-approximation to $\min\|U^*v - a\|_1$ \cite{sw11}. The important property of
    Cauchy random variables is that if $X$ and $Y$ are independent Cauchy random variables,
    then $\alpha X + \beta Y$ is distributed as a Cauchy random variable times 
    $|\alpha| + |\beta|$, for any scalars $\alpha, \beta \in \mathbb{R}$. The $O(k \log k)$
    approximation arises because 
    all possible regression solutions are in the column span of $[U^*, a]$
    which is $(k+1)$-dimensional, and the sketch $S$ gives an approximation
    factor of $O(k \log k)$ to preserve every vector norm in this subspace.
    If we instead had a multi-response regression problem
    $\min \|SU^*V^*-SA\|_1$ the dimension of the column span of $[U^*, A]$ would be $d+k$, and this approach
    would give an $O(d \log d)$-approximation. Unlike Frobenius norm
    multi-response regression $\min \|SU^*V^*-SA\|_F$, which can be bounded if  
    $S$ is a subspace embedding for $U^*$ and satisfies an approximate matrix product theorem 
    \cite{s06}, there is no convenient linear-algebraic analogue for the $\ell_1$-norm. 

    We first note that since regression is a minimization problem,
    to obtain an $O(\alpha)$-approximation by solving the
    sketched version of the problem, it suffices that (1) for the optimal $V^*$, we have
    $\|SU^*V^*-SA\|_1 \leq O(\alpha) \|U^*V^*-A\|_1$, and (2) for all $V$, we have
    $\|SU^*V-SA\|_1 \geq \Omega(1) \cdot \|U^*V-A\|_1$.

    We show (1) holds for $\alpha = O(\log d)$ and any number of rows of $S$. 
    Our analysis follows by truncating the Cauchy random variables $(SU^*V_j^* - SA_j)_i$ for $i \in [O(k \log k)]$ and
    $j \in [d]$, so that their expectation exists, and applying
    linearity of expectation across the $d$ columns. This is inspired from an argument of Indyk \cite{i06}
    for embedding a vector into a lower-dimensional vector while preserving its $\ell_1$-norm; for {\it single-response}
    regression this is the statement that $\|SU^*v^*-Sa\|_1 = \Theta(1)\|U^*v-a\|_1$, implied
    by \cite{i06}. However, for {\it multi-response} regression we have to work entirely with expectations, rather than the tail
    bounds in \cite{i06}, since the
    Cauchy random variables $(SU^*V_j - SA_j)_i$, while independent across $i$, are dependent
    across $j$. Moreover, our $O(\log d)$-approximation factor is not an artifact of our analysis
    - we show in Section \ref{sec:hardinstance} that there
    is an $n \times d$ input matrix $A$ for which with probability $1-1/\poly(k)$,
    {\it there is no $k$-dimensional space in the span of $SA$ achieving a $\left (\frac{\log d}{t \log t} + k^{1/2-\gamma} \right)$-approximation},
    for $S$ a Cauchy matrix with $t$ rows, where $\gamma > 0$ is an arbitrarily small constant. This shows $(k \log d)^{\Omega(1)}$-inapproximability.
    Thus, the fact that we
    achieve $O(\log d)$-approximation instead of $O(1)$ is fundamental for a matrix $S$ of Cauchy random variables
    or any scaling of it.
%%%To David, you forgot to type a } after {\it there is no, 
%%% I'm not sure if that's right place.

    While we cannot show (2), we instead show for all $V$, 
    $\|SU^*V-SA\|_1 \geq \|U^*V-A\|_1/2 - O(\log d)\|U^*V^*-A\|_1$ if $S$ has $O(k \log k)$ rows.
    This suffices for regression, since the only matrices $V$ for which the cost
    is much smaller in the sketch space are those providing an $O(\log d)$ approximation in the original space.
    The guarantee follows from the triangle inequality: $\|SU^*V-SA\|_1 \geq \|SU^*V-SU^*V^*\|_1 - \|SU^*V^*-SA\|_1$ and the fact that $S$ is known
    to not contract any vector in the column span of $U^*$ if $S$ has $O(k \log k)$ rows \cite{sw11}. Because of this, we have
    $\|SU^*V-SU^*V^*\|_1 = \Omega(1) \|U^*V-U^*V^*\|_1 = \Omega(1) (\|U^*V-A\|_1 - \|U^*V^*-A\|_1)$, where we again use the triangle inequality.
    We also bound the additive term $\|SU^*V^*-SA\|_1$ by $O(\log d)\|U^*V^*-A\|_1$ using (1) above. 

    Given that $SA$ contains a good rank-$k$ approximation in its row span, our algorithm with a slightly worse $\poly(n)$ time
    and $\poly(k\log(n))$-approximation can be completely described here.
    Let $S$ and $T_1$ be independent $O(k \log k) \times n$ matrices of i.i.d. Cauchy random variables, and let $R$ and $T_2$ be independent
    $d \times O(k \log k)$
    matrices of i.i.d. Cauchy random variables. Let $$X =(T_1 A R)^\dagger ((T_1 A R)(T_1 A R)^\dagger(T_1 A T_2) (SAT_2) (SAT_2)^\dagger)_k (SAT_2)^\dagger,$$
    which is the rank-$k$ matrix minimizing $\|T_1 A R X S A T_2 - T_1 A T_2\|_F$,
    where for a matrix $C$, $C_k$ is its best
    rank-$k$ approximation in Frobenius norm. Output $\widehat{A} = AR X SA$ as the solution to $\ell_1$-low rank
    approximation of $A$. We show with constant probability that $\widehat{A}$ is a $\poly(k\log(n))$-approximation.

    To improve the approximation factor, after computing
    $SA$, we $\ell_1$-project each of the rows of $A$ onto $SA$ using linear programming or fast algorithms for
    $\ell_1$-regression \cite{cw13,mm13}, obtaining an $n \times d$ matrix $B$ of rank $O(k \log k)$. We then apply the algorithm
    in the previous paragraph with $A$ replaced by $B$. This ultimately leads to a $\log d \cdot \poly(k)$-approximation. 

    To improve the running time from $\poly(n)$ to $\nnz(A) + n \cdot \poly(k)$,
    we show a similar analysis holds for the sparse Cauchy matrices of \cite{mm13}; see also the matrices in \cite{wz13}.
\\\\
{\bf CUR Decompositions:} 
    To obtain a CUR decomposition, we first find a $\log d \cdot \poly(k)$-approximate rank-$k$ approximation $\widehat{A}$ as above.
    Let $B_1$ be an $n \times k$ matrix whose columns span those of $\widehat{A}$, and consider the regression $\min_V \|B_1 V - A\|_1$.
    Unlike the problem $\min_V \|U^*V-A\|_1$ where $U^*$ was unknown, we {\it know} $B_1$ so can compute
    its Lewis weights efficiently, sample by them, and obtain a regression problem $\min_V \|D_1(B_1V-A)\|_1$ where $D_1$ is a sampling
    and rescaling matrix. Since $$\|D_1(B_1V-A)\|_1 \leq \|D_1(B_1V-B_1V^*)\|_1 + \|D_1(B_1V^*-A)\|_1,$$
    where $V^* = \textrm{argmin}_V \|B_1 V - A\|_1$, we can bound the first term by
    $O(\|B_1V-B_1V^*\|_1)$ using that $D_1$ is a subspace embedding if it has $O(k \log k)$ rows, 
    while the second term is $O(1) \|B_1V^*-A\|_1$ by a Markov bound. Note that
    $\|B_1 V^* - A\|_1 \leq (\log d) \cdot \poly(k) \min_{\textrm{rank-}k\textrm{ matrices~}A'}\|A-A'\|_1$. 
    %By a sonclude that 
    %$\|D_1(B_1V-A)\|_1 = \Theta(1) \|B_1V-A\|_1$.
    By switching to $\ell_2$ as before, we see that $\widehat{V} = (D_1 B_1)^{\dagger} D_1A$ contains a
    $(\log d) \poly(k)$-approximation in its span. Here $D_1A$ is an actual subset of rows of $A$, as required in a CUR decomposition. Moreover
    the subset size is $O(k \log k)$. 
    We can sample by the Lewis weights of $\widehat{V}$ to obtain a subset $C$ of $O(k \log k)$ rescaled columns of $A$,
    together with a rank-$k$ matrix $U$ for which $\|CUR-A\|_1 \leq (\log d) \poly(k) \min_{\textrm{rank-}k\textrm{ matrices~}A'}\|A-A'\|_1$.
\\\\
{\bf Algorithm for Small $k$:}
    Our CUR decomposition shows how we might obtain an $O(1)$-approximation for constant $k$ in $\poly(n)$ time.
    If we knew the Lewis weights of $U^*$, an $\alpha$-approximate solution to the problem $\min_V \|D_1(U^*V-A)\|_1$ would
    be an $O(\alpha)$-approximate solution to the problem $\min_V \|U^*V-A\|_1$, where $D_1$
    is a sampling and rescaling matrix of $O(k \log k)$ rows of $A$. Moreover, an $O(\sqrt{k \log k})$-approximate solution to
    $\min_V \|D_1(U^*V-A)\|_1$ is given by $V = (D_1U^*)^{\dagger}D_1A$, which implies the $O(k \log k)$ rows of $D_1 A$ contain 
    an $O(\sqrt{k \log k})$-approximation. For small $k$, we can guess every subset
    of $O(k \log k)$ rows of $A$ in $n^{O(k \log k)}$ time (if $d \ll n$, by taking transposes at the beginning one can replace this with $d^{O(k \log k)}$ time). For each guess, we set up the problem $\min_{\textrm{rank-}k\ U} \|U(D_1A) - A\|_1$.
    If $D_2$ is a sampling and rescaling matrix according to the Lewis weights of $D_1A$, then by a similar triangle inequality
    argument as for our CUR decomposition,
    minimizing $\|U(D_1A)D_2 - AD_2\|_1$ gives an $O(\sqrt{k \log k})$ approximation. By switching to $\ell_2$, this implies there is an
    $O(k \log k)$-approximation of the form $AD_2WD_1A$, where $W$ is an $O(k \log^2 k) \times O(k \log k)$ matrix of rank $k$. By setting
    up the problem $\min_{\textrm{rank-}k\ W}\|AD_2WD_1A - A\|_1$, one can sample from Lewis weights on the left and right to reduce this to a
    problem independent of $n$ and $d$, after which one can use polynomial optimization to solve it in $\exp(\poly(k))$ time.
    One of our guesses $D_1 A$ will be correct, and for this guess we obtain an $\wt{O}(k)$-approximation. For each guess we can compute
    its cost and take the best one found. This gives 
    an $O(1)$-approximation for constant $k$, removing the $O(\log d)$-factor from the approximation of our earlier algorithm. 
\\\\
{\bf Existential Results for Subset Selection:}
    In our algorithm for small $k$, the first step was to show there exist $O(k \log k)$ rows of $A$ which
    contain a rank-$k$ space which is an $O(\sqrt{k \log k})$-approximation (though algorithmically we instead found $O(k \log k)$ rows giving an $O(k \log k \log d)$ approximation in polynomial time). 
    %To find such a set, we can
    %try all subsets $R$ of $O(\sqrt{k \log k})$ rows, and for each solve the linear programming problem $\min_U \|UR-A\|_1$ to find
    %the best $R$. This takes $n^{O(k \log k)}$ time.
    %
    %If we want a rank-$k$ matrix $U$, we can first $\ell_1$-project
    %the rows of $A$ onto the row span of $R$, obtaining an $n \times d$ matrix $B$ of rank $O(k \log k)$. Then we solve the problem
    %$\min_{U \in \mathbb{R}^{n \times k}, X \in \mathbb{R}^{k \times O(k \log k)}} \|UXR-B\|_1$. Since $B$ has small rank, we can guess $O(k \log^2 k)$
    %rows of $U$ and weights corresponding to the Lewis weights of $[U, B]$, form a sampling and rescaling matrix $D_2$, and consider
    %the problem $\min_{U \in \mathbb{R}^{n \times k}, X \in \mathbb{R}^{k \times O(k \log k)}} \|D_2UXC-D_2B\|_1$. We also compute the Lewis weights of $[R, B]$,
    %form a sampling and rescaling matrix $D_3$, and reduce the problem to
    %$\min_{U \in \mathbb{R}^{n \times k}, X \in \mathbb{R}^{k \times O(k \log k)}} \|D_2UXRD_3 -D_2B D_3\|_1$. We now have a small optimization problem, and
    %polynomial system solvers can solve this in $\poly(n) \exp(\poly(k))$ time \cite{bpr96}. The row span of $XR$ is our
    %desired $k$-dimensional subspace. 
    
    While for Frobenius norm one can find $O(k)$ rows with an $O(1)$-approximation in their span, one of our main negative results for $\ell_1$-low rank
    approximation is that this is impossible,
    showing that the best approximation one can obtain with $\poly(k)$ rows is $k^{1/2-\gamma}$ for an arbitrarily small constant
    $\gamma > 0$. Our hard
    instance is an $r \times (r+k)$ matrix $A$ in which the first $k$ columns are i.i.d. Gaussian, and the remaining
    $r$ columns are an identity matrix. Here, $r$ can be twice the number of rows one
    is choosing. The optimal $\ell_1$-low rank approximation has cost at most $r$, obtained by choosing the first $k$ columns.
    
    Let $R \in \mathbb{R}^{r/2 \times k}$ denote the first $k$ entries of the $r/2$ chosen rows, and let $y$ denote the first
    $k$ entries of an unchosen row. For $r/2 > k$, there exist many solutions $x \in \mathbb{R}^{r/2}$ for which $x^\top R = y$. However,
    we can show the following tradeoff:
    \begin{center}
      {\it whenever $\|x^\top R - y\|_1 < \frac{\sqrt{k}}{\poly(\log k)}$, then $\|x\|_1 > \frac{\sqrt{k}}{\poly(\log k)}$}.
      \end{center}
    Then no matter which linear combination $x^\top$ of the rows of $R$ one chooses to approximate $y$ by, either one incurs
    a $\frac{\sqrt{k}}{\poly(\log k)}$ cost on the first $k$ coordinates, or since $A$ contains an identity matrix, one
    incurs cost $\|x\|_1 > \frac{\sqrt{k}}{\poly(\log k)}$ on the last $r$ coordinates of $x^\top R$. 

    To show the tradeoff, consider an $x \in \mathbb{R}^{r/2}$. We
    decompose $x = x^0 + \sum_{j \geq 1} x^j$, where $x^j$ agrees with $x$ on coordinates which have absolute value in the range
    $\frac{1}{\sqrt{k} \log^{c} k} \cdot [2^{-j}, 2^{-j+1}]$, and is zero otherwise. Here, $c > 0$ is a constant, and $x^0$ denotes
    the restriction of $x$ to all coordinates of absolute value at least $\frac{1}{\sqrt{k} \log^c k}$. Then $\|x\|_1 < \frac{\sqrt{k}}{\log^c k}$,
    as otherwise we are done. Hence, $x^0$ {\it has small support}. Thus, one can build a small net for all $x^0$ vectors by
    choosing the support, then placing a net on it. For $x^j$ for $j > 0$, the support sizes are increasing so the net size needed
    for all $x^j$ vectors is larger. However, since $x^j$ has all coordinates of roughly the same magnitude on its support, its
    {\it $\ell_2$-norm is decreasing in $j$}. Since $(x^j)^\top R \sim N(0, \|x^j\|_2^2 I_k)$,
    this makes it much less likely that individual coordinates of $(x^j)^\top R$ can be large. Since this probability goes down rapidly, we can
    afford to union bound over the larger net size.
    What we show is that for any sum of the form $\sum_{j \geq 1} x^j$, at most $\frac{k}{10}$
    of its coordinates are at least $\frac{1}{\log k}$ in magnitude.

%    can have a much larger number of non-zero
%    coordinates than $x^0$, its $\ell_1$-norm is still constrained to be at most $\frac{\sqrt{k}}{\log^c k}$ and so it has at most
%    $2^{j+1} \frac{k}{\log^{c-c'} k}$ non-zero coordinates and so $\|x^j\|_2^2$ is at most
%    $\Theta \left (\frac{2^{-j}}{\log^{c+c'} k} \right )$. We can
%    put a net on the $x^j$ for each $j \geq 1$, and use these nets to argue that for any $\sum_{j \geq 1} x^j$,
%    at most $k/10$ of its coordinates are at least $\frac{1}{\log k}$ in magnitude.
%    
%    We can assume $\|x\|_1 < \frac{\sqrt{k}}{\log^c k}$, for a constant $c > 0$,
%    otherwise the algorithm already incurs this cost. Let $x^0$ be the vector which agrees with $x$
%    on those coordinates which are at least
%    $\frac{1}{\sqrt{k} \log^{c'} k}$ for a constant $0 < c' < c$,
%    and on other coordinates is $0$. Then $x^0$ has at most $\frac{k}{\log^{c-c'} k}$ non-zero coordinates. There
%    are ${r/2 \choose \frac{k}{\log^{c-c'} k}}$ choices of such coordinates and for each choice we put a $\frac{1}{\poly(k)}$-sized
%    net on the coordinates. The number of net points is at most $r^{k/(\log^{c-c'} k)} 2^{k/(\log^{c-c'-1} k)}$. For $r \leq k \poly(\log k)$, by
    %    setting $c$ to be a large enough constant, this is at most $2^{k/\log k}$, say.
    
    For $\|x^\top R-y\|_1$ to be at most $\frac{\sqrt{k}}{\log^c k}$, for at least $\frac{k}{2}$ coordinates $i$, we must have
    $|(x^\top R-y)_i| < \frac{2}{\sqrt{k} \log^c k}$. With probability $1-2^{-\Omega(k)}$, $|y_i| \geq \frac{1}{100}$ on at least
    $\frac{2k}{3}$ coordinates. From the previous paragraph,
    it follows there are at least $\frac{k}{2}-\frac{k}{10}-\frac{k}{3} = \Omega(k)$
    coordinates $i$ of $x$ for which (1) $|(x^\top R-y)_i| < \frac{2}{\sqrt{k} \log^c k}$, (2) $|\sum_{j \geq 1} x^j_i| < \frac{1}{\log k},$
    and (3) $|y_i| \geq \frac{1}{100}$. On these $i$, $(x^0)^\top R_i$ must be in an interval of width $\frac{1}{\log k}$ at distance
    at least $\frac{1}{100}$ from the origin. Since $(x^0)^\top R \sim N(0, \|x^0\|_2^2 I_k)$, for any value of $\|x^0\|_2^2$
    the probability this happens on $\Omega(k)$ coordinates is at most $2^{-\Theta(k)}$. Since the net size for $x^0$ is small,
    we can union bound over every sequence $x^0, x^1, \ldots, $ coming from our nets.

    Some care is needed to union bound over all possible subsets $R$ of rows which can be chosen. We handle this by conditioning on a few
    events of $A$ itself, which imply corresponding events {\it for every subset of rows}. These events are such that if $R$ is the chosen
    set of half the rows, and $S$ the remaining set of rows of $A$, then the event that a constant fraction of rows in $S$ are close to the row span
    of $R$ is $2^{-\Theta(kr)}$, which is small enough to union bound over all choices of $R$.

    Curiously, we also show there are some matrices
    $A \in \mathbb{R}^{n \times d}$ for which any $\ell_1$ rank-$k$
    approximation in the entire row span of $A$ cannot achieve better than a
    $(2-\Theta(1/d))$-approximation.
    \\\\    
{\bf Bicriteria Algorithm:}
Our algorithm for small $k$ gives an $O(1)$-approximation in $\poly(n)$ time for constant $k$, but the approximation factor depends on $k$.
We show how one can find a rank-$2k$ matrix $\widehat{A}$ for which $\|A-\widehat{A}\|_1 \leq C \cdot \OPT$,
where $C$ is an absolute constant, and $\OPT =  \min_{\textrm{rank-}k\textrm{ matrices~}A'}\|A-A'\|_1$.
%Unlike our earlier algorithms, our output $\widehat{A}$ is not int the span of a small subset of rows and columns of $A$,
%thereby bypassing our hardness for column subset selection. 
%While this is impossible
    %for column or row subset selection, %, even if allowing one to output a rank-$k \poly(\log k)$ approximation.
%   By leaving the span of a small subset of rows and columns of the input,
%   we show it is possible to output a rank-$2k$ $O(1)$-approximation to the best rank-$k$ matrix, where the constant in the $O(1)$
%   is independent of $k$.
We first find a rank-$k$ matrix $B^1$ for which $\|A-B^1\|_1 \leq p \cdot \OPT$ for a factor $1 \leq p \leq \poly(n)$. We can use any of our
algorithms above for this.

    Next consider the problem $\min_{V \in \mathbb{R}^{k \times d}}\|U^*V-(A-B^1)\|_1$, and let $U^*V^*$ be a best $\ell_1$-low rank approximation
    to $A-B^1$; we later explain why we look at this problem.
    We can assume $V^*$ is an {\it $\ell_1$ well-conditioned basis} \cite{c05,ddhkm09}, since we can replace $U^*$
    with $U^*R^{-1}$ and $V^*$ with $RV^*$ for any invertible linear transformation $R$. For any vector $x$ we then have
    $\frac{\|x\|_1}{f} \leq \|x^\top V^*\|_1 \leq e \|x\|_1$, where $1 \leq e, f \leq \poly(k)$. This implies all entries of $U^*$ are at most $2f\|A-B\|_1$,
    as otherwise one could replace $U^*$ with $0^{n \times k}$ and reduce the cost. Also, any entry of $U^*$ smaller than $\frac{\|A-B\|_1}{100 e n k p}$ can
    be replaced with $0$ as this incurs additive error $\frac{\OPT}{100}$. If we round the entries of $U^*$ to integer multiples of $\frac{\|A-B\|_1}{100 e n k p}$, then we only have $O(enkpf)$ possibilities for each entry of $U^*$, and still obtain an $O(1)$-approximation. We refer to the rounded $U^*$ as $U^*$,
    abusing notation.

    Let $D$ be a sampling and rescaling matrix with $O(k \log k)$ non-zero diagonal entries, corresponding to sampling by the Lewis weights of $U^*$. We do not
    know $D$, but handle this below. By the triangle inequality, for any $V$, 
    \begin{eqnarray*}
      \|D(U^*V-(A-B^1))\|_1 & = & \|D(U^*V-U^*V^*)\|_1 \pm \|D(U^*V^*-(A-B^1))\|_1\\
      &= & \Theta(1)\|U^*V-U^*V^*\|_1 \pm O(1)\|U^*V^*-(A-B^1)\|_1,
      \end{eqnarray*}
    where the Lewis weights give 
    $\|D(U^*V-U^*V^*)\|_1 = \Theta(1)\|U^*V-U^*V^*\|_1$ and a Markov bound gives $\|D(U^*V^*-(A-B^1))\|_1 = O(1)\|U^*V^*-(A-B^1)\|_1$. Thus, minimizing
    $\|DU^*V-D(A-B^1)\|_1$ gives a fixed constant factor approximation to the problem  $\min_{V \in \mathbb{R}^{k \times d}}\|U^*V-(A-B^1)\|_1$.
    The non-zero diagonal entries of $D$ can be assumed to be integers between $1$ and $n^2$. %(since they correspond to inverses of the probabilities,
%    and rounding the probabilities to powers of $2$ just affects the total number of samples by a factor of $2$).

    We guess the entries of $DU^*$ and
    note for each entry there are only $O(enkpf \log (n^2))$ possibilities. One of our guesses corresponds to Lewis weight sampling by $U^*$. We solve
    for $V$ and by the guarantees of Lewis weights, the row span of this $V$ provides an $O(1)$-approximation. We can find the
    corresponding $U$ via linear programming. As mentioned above, we do not know $D$, but can enumerate over all $D$ and all possible $DU^*$.
    The total time is $n^{\poly(k)}$.

    After finding $U$, which has $k$ columns, we output the rank-$2k$ space formed by the column span of $[U, B^1]$.
    By including the column span of $B^1$, we ensure
    our original transformation of the problem $\min_{V \in \mathbb{R}^{k \times d}}\|U^* \cdot V-A\|_1$ to the problem
    $\min_{V \in \mathbb{R}^{k \times d}} \|U^* \cdot V-(A-B^1)\|_1$ is valid,
    since we can first use the column span of $B^1$ to replace $A$ with $A-B^1$.
    Replacing $A$ with $A-B^1$ ultimately results in a rank-$2k$
    output. Had we used $A$ instead of $A-B^1$ our output would have been
    rank $k$ but would have additive error $\frac{\|A\|_1}{\poly(k/\epsilon)}$.
    If we assume the entries of $A$ are in $\{-b,-b+1, \ldots, b\}$, then we can lower bound the cost $\|U^*V-A\|_1$, given that it is non-zero, 
    by $(ndb)^{-O(k)}$ (if it is zero then we output $A$) using Lemma 4.1 in \cite{cw09} and relating entrywise $\ell_1$-norm to Frobenius norm.
    We can go through the same arguments above with $A-B$ replaced by $A$ and our running time will now be $(ndb)^{\poly(k)}$.
    \\\\
{\bf Hard Instances for Cauchy Matrices and More General Sketches:}
We consider a $d \times d$ matrix $A = I_d + (\log d) e_1^\top e$, where $e_1 = (1, 0, \ldots, 0)$
and $e = (1, 1, \ldots, 1)$ and $I_d$ is the $d \times d$ identity. For an
$O(k \log k) \times d$ matrix $S$ of i.i.d. Cauchy random variables,
$SA = S + (\log d) S_1^\top e$, where $S_1$ is the first column of $S$. For a typical column
of $SA$, all entries are at most $\poly(k) \log d$ in magnitude.
Thus, in order to approximate
the first row of $A$, which is $(\log d) e$, by $x^\top SA$ for an 
$x \in \mathbb{R}^{k \log k}$, we need $\|x\|_1 \geq \frac{1}{\poly(k)}$.
Also $\|x^\top S\|_1 = \Omega(\|x\|_1 d \log d)$ with $1-\exp(-k \log k)$
probability, for $d$ large enough, so by a net argument 
$\|x\|_1 \leq \poly(k)$ for all $x$.

However, {\it there are} entries of $SA$ that are very
large, i.e., about one which is $r = \Theta(d k \log k)$ in magnitude, and in general
about $2^i$ entries about $r2^{-i}$ in magnitude. These entries typically
occur in columns $C_j$ of $SA$ for which all other entries in the column are bounded by
$\poly(k)$ in magnitude. Thus, $|x^\top C_j| \approx r2^{-i}$ for
about $2^i$ columns $j$. For each such column, if $r2^{-i} \gg \log d$, then
we incur cost $\frac{r2^{-i}}{\poly(k)}$ in approximating the first row of $A$.
In total the cost is $\frac{r \log r}{\poly(k)} = \frac{d \log d}{\poly(k)}$, but the optimal cost is at most $d$,
giving a $\frac{\log d}{\poly(k)}$ lower bound. We optimize this to a $\frac{\log d}{k \log^2 k}$ lower bound. 

When $k$ is large this bound deteriorates, but we also show a $k^{1/2-\gamma}$
lower bound for arbitrarily small constant $\gamma > 0$.
This bound applies to any oblivious sketching matrix. The idea is similar to our row
subset selection lower bound. Let $A$ be as in our row subset selection lower bound, consider $SA$,
and write $S = U\Sigma V^\top$ in its full SVD. Then $SA$ is in the row span of the top $O(k \log k)$
rows of $V^\top A$, since $\Sigma$ only has $O(k \log k)$ non-zero singular values. Since the first $k$ columns of
$A$ are rotationally invariant, $V^\top A$ has first $k$ columns i.i.d. Gaussian and remaining columns equal to $V^\top$.
Call the first $O(k \log k)$ rows of $V^\top A$ the matrix $B$. We now try to approximate a row of $A$ by a vector in the
row span of $B$. There are two issues that make this setting different from row subset selection:
(1) $B$ no longer contains an identity submatrix, and (2) the rows of $B$ depend on the rows of $A$. We handle the first
issue by building nets for subsets of coordinates of $x^\top V^\top$ rather than $x$ as before; since $\|x^\top V^\top\|_2 = \|x\|_2$
similar arguments can be applied. We handle the second issue by observing that if the number of rows of $B$ is
considerably smaller than that of $A$, then the distribution of $B$ had we replaced a random row of
$A$ with zeros would be statistically close to i.i.d. Gaussian. Hence, typical rows of $A$ can be regarded as being
independent of $B$. 
\\\\
{\bf Limited Independence, Distributed, and Streaming Algorithms:}
    We show for an $n \times d$ matrix $A$, if we left-multiply by an
    $O(k \log k) \times n$ matrix $S$ in which each row
    is an independent vector of $\wt{O}(d)$-wise independent Cauchy random
    variables, $SA$ contains
    a $\poly(k) \log d$-approximation in its span.
    This allows players in a distributed model
    to share a common $S$ by exchanging $\wt{O}(kd)$ bits,
    independent of $n$.  
    We use Lemma 2.2 of \cite{knw10} which shows for a reasonably smooth approximation
    $f$ to an indicator function,
    ${\bf E}[f(X)] = \E[f(Y)] + O(\epsilon)$, where $X = \sum_i a_i X_i$, 
    $Y = \sum_i a_i Y_i$, $a \in \mathbb{R}^n$ is fixed,
    $X$ is a vector of i.i.d. Cauchy random variables, and $Y$ is a vector
    of $\wt{O}(1/\epsilon)$-wise independent random variables.

    To show the row span
    of $SA$ contains a good rank-$k$ approximation, we argue $\|Sy\|_1
    = \Omega(\|y\|_1)$ for a fixed $y \in \mathbb{R}^n$ with $1-\exp(-k \log k)$
    probability. We apply the above lemma with $\epsilon = \Theta(1)$.
    We also need for an $n \times d$ matrix $A$ with unit-$\ell_1$ columns,
    that $\|SA\|_1 = \wt{O}(kd)$. We fool the expectation of a truncated Cauchy
    by taking a weighted sum of $O(\log(dk))$ indicator functions
    and applying the above lemma with $\epsilon = \Theta(1/d)$. An issue is there are
    $\wt{\Theta}(k d)$ Cauchy random variables corresponding to the entries of $SA$,
    some of which can be as large as $\wt{\Theta}(kd)$, so to fool their expectation
    (after truncation)
    we need $\epsilon = \wt{\Theta}(1/(dk))$, resulting in $\wt{O}(dk^2)$
    seed length
    and ruining our optimal $\wt{O}(dk)$ communication. We show we can instead pay a
    factor of $k$ in our approximation and maintain $\wt{O}(dk)$-wise
    independence. The distributed and streaming algorithms, given this,
    follow algorithms for Frobenius norm low rank approximation in 
    \cite{kvw14,bwz16}.
\\\\
{\bf Hardness Assuming Exponential Time Hypothesis:}
By inspecting the proof of NP-hardness of \cite{gv15}, it at best
gives a $(1+\frac{1}{n^{\gamma}})$-inapproximability for an arbitrarily small constant $\gamma > 0$.
We considerably strengthen this to $(1+\frac{1}{\log^{1+\gamma}n})$-inapproximability by taking a
modified version of the
$n \times n$ hard instance of \cite{gv15} and planting it in a $2^{o(n)} \times 2^{o(n)}$
matrix padded with tiny values. Under the $\mathsf{ETH}$, the maximum cut problem that \cite{gv15} and that
we rely on cannot be solved in $2^{o(n)}$ time, so our transformation is efficient.
Although we use the maximum cut problem as in \cite{gv15} for our $n \times n$ hard instance,
in order to achieve our inapproximability we need to use that under the $\mathsf{ETH}$
this problem is hard to approximate even if the input graph is sparse and even
up to a constant factor; such additional conditions were not needed in \cite{gv15}. 
\\\\    
{\bf $\ell_p$-Low Rank Approximation and EMD-Low Rank Approximation:}
    Our algorithms for entrywise $\ell_p$-Norm Error are similar to our
    algorithms for $\ell_1$. We use $p$-stable random variables in place of Cauchy
    random variables, and note that the $p$-th power of a $p$-stable random variable
    has similar tails to that of a Cauchy, so many of the same arguments apply. Our algorithm
    for EMD low rank approximation immediately follows by embedding EMD into $\ell_1$. 
    \\\\
        {\bf Counterexamples to Heuristics:}
        %For any $\eps \in (0,0.5)$ and $\gamma >0$, let $A\in \mathbb{R}^{(2n+2)\times (2n+2)}$ with
%\begin{align*}
Let $A=\text{diag}(n^{2+\gamma},n^{1.5+\epsilon},B,B)\in \mathbb{R}^{ (2n+2) \times (2n+2)}$ %$A=
%\begin{bmatrix}
%n^{2+\gamma} & 0 & 0 & 0 \\
%0 & n^{1.5+\epsilon} & 0 & 0 \\
%0 & 0 & B & 0 \\
%0 & 0 & 0 & B
%\end{bmatrix},$
%\end{align*}
where $\eps \in (0, .5), \gamma > 0$, and 
$B$ is the $n \times n$ all $1$s matrix. %Note $A\in \mathbb{R}^{(2n+2)\times (2n+2)}$. 
%Suppose we want to find the rank $k=3$ solution for $A$.
For this $A$ we show the four heuristic algorithms \cite{kk05,dzhz06,kwak08,bdb13}
cannot achieve an $n^{\min ( \gamma, 0.5-\eps ) } $ approximation ratio when the rank parameter $k = 3$.% We also mention why a related objective function, robust PCA \cite{wgrpm09,clmw11,nnasj14,nyh14,chd16,zzl15}, does not give a provable approximation factor for $\ell_1$-low rank approximation.

%We target on the following problem in this paper,
%
%
%\begin{problem}
%Given $A\in \mathbb{R}^{n\times d}$, for any $k$ we want to output a rank-$k$ matrix $\widehat{A}\in \mathbb{R}^{n \times d}$ s.t.
%\begin{equation*}
%\| \widehat{A} - A \|_1 \leq C \underset{\rank-k~A'}{\min} \| A' - A\|_1
%\end{equation*}
%where $n\gg d\gg k$.
%\end{problem}

\subsection{Several Theorem Statements, an Algorithm, and a Roadmap}
%This section will contain the following things, 1. Algorithm for $\poly(k) \cdot \log d$ approximation, 2. Algorithm for $\poly(k) \cdot \log d$ CUR decomposition 3. Exponential running time for CUR decomposition 4. Hard instance for column subset selection. 5. Experiments results, if the code of our simplest algorithm is short, we should copy it here.
%This section will provide several main theorems and algorithms.
%
%We first provide two upper bounds.
\vspace{-4mm}
\begin{algorithm}[h]\caption{Main Meta-Algorithm}
\begin{algorithmic}[1]
\Procedure{\textsc{L1LowRankApprox}}{$A,n,d,k$} \Comment{Theorem \ref{thm:main1}}
\State Choose sketching matrix $S$ (a Cauchy matrix or a sparse Cauchy matrix.)
\State Compute $SA$, form $C$ by $C^i \leftarrow \arg\min_{x} \| x S A - A^i \|_1$. Form $B=C \cdot SA $.
\State Choose sketching matrices $T_1, R, D, T_2$ (Cauchy matrices or sparse Cauchy matrices.)
\State Solve $\min_{X,Y} \| T_1 BRX YDB T_2 - T_1 B T_2\|_F$.
\State \Return $BRX$,$YDB$.
\EndProcedure
\end{algorithmic}
\end{algorithm}
\vspace{-0.5cm}
\begin{theorem}[Informal Version of Theorem~\ref{thm:polyklogd_approx_algorithm}]\label{thm:main1}
Given $A\in \mathbb{R}^{n\times d}$, there is an algorithm which in $\nnz(A) + (n+d)\cdot\poly(k)$ time, outputs a (factorization of a) rank-k matrix $A'$ such that with probability $9/10$,
%\begin{align*}
$\| A' - A \|_1 \leq (\log d)\poly(k) \underset{U\in \mathbb{R}^{n\times k}, V\in \mathbb{R}^{k\times d}}{\min} \| UV - A\|_1.$
%\end{align*}
\end{theorem}

\begin{theorem}[Informal Version of Theorem~\ref{thm:k_approx_algorithm}]\label{thm:main2}
Given $A\in \mathbb{R}^{n\times d}$, there is an algorithm that takes $\poly(n)d^{\wt{O}(k)} 2^{\wt{O}(k^2)}$ time and outputs a rank-k matrix $A'$ such that, with probability $9/10$,
$\| A' - A \|_1 \leq \wt{O}(k) \underset{U\in \mathbb{R}^{n\times k}, V\in \mathbb{R}^{k\times d}}{\min} \| UV - A\|_1.$ In addition, $A'$ is a CUR decomposition.
\end{theorem}
%\begin{theorem}[Informal of Theorem \ref{thm:use_yaominmax}]
%Given $A\in\mathbb{R}^{n\times d},S\in\mathbb{R}^{r\times n},k,\gamma$, we say that an algorithm ${\cal M}(A,S,k,\gamma)$ which outputs a rank-$k$ matrix $B\in\mathbb{R}^{n\times r}$ ``succeeds'', if~$\|BSA-A\|_1\leq k^{0.5-\gamma} \cdot \underset{\rank-k~A'}{\min} \| A' - A \|_1$ holds.
%Let $\Pi$ denote a distribution over matrices $S\in \mathbb{R}^{r\times n}$. For any $k\geq 1$, for arbitrarily small constant $\gamma>0$, arbitrary fixed constants $c_1,c_2>0$ and $\min(n,d) \geq k^{c_2}$, if for all $ A\in\mathbb{R}^{n\times d}$, it holds that
%\begin{align*}
%$\underset{S\sim \Pi}{\Pr}[{\cal M}(A,S,k,\gamma)\mathrm{~succeeds~}]\geq \Omega(1/k^{c_1}).$
%\end{align*}
%Then $r = \Omega(k^{c_2-2c_1-2}).$
%\end{theorem}
%
%\begin{theorem}[Lower bound for Cauchy embeddings, informal of Theorem \ref{sec:hardinstance_cauchy} and \ref{thm:use_yaominmax}]
%Let $\Pi$ denote the distribution over $m\times d$ random Cauchy matrices with $m < o(d)$. For arbitrarily small constant $\gamma>0$. There exists some matrix $A\in \mathbb{R}^{d\times d}$, draw $S\sim \Pi$, with arbitrarily large constant probability, there is no better than $O(k^{0.5-\gamma} + \log d)$-approximation solution to $\min_{\rank-k~A'}\|A'-A\|_1$ in row span of $SA$.
%\end{theorem}
\begin{theorem}[Informal Version of Theorem~\ref{thm:hard_for_row_subset}]\label{thm:main3}
For any $k\geq 1$, and any constant $c\geq 1$, let $n = k^c$. There exists a matrix $A$ such that for any matrix $A'$ in the span of $n/2$ rows of $A$,
$
\| A' - A \|_1 = \Omega(k^{0.5-\alpha}) \underset{U\in \mathbb{R}^{n\times k}, V\in \mathbb{R}^{k\times d}}{\min} \| UV - A\|_1,
$
where $\alpha>0$ is an arbitrarily small constant.
\end{theorem}

%We state our main lower bound for column subset selection here.
\vspace{-0.5cm}
\paragraph{Road map}
Section~\ref{sec:notation} introduces some notation and definitions. Section~\ref{sec:preli} includes several useful tools. We provide several $\ell_1$-low rank approximation algorithms in Section~\ref{sec:alg}. Section~\ref{sec:l1} contains the no contraction and no dilation analysis for our main algorithm. The results for $\ell_p$ and earth mover distance are presented in Section~\ref{sec:lp} and \ref{sec:emd}.
We provide our existential hardness results for Cauchy matrices, row subset selection and oblivious subspace embeddings in Section~\ref{sec:hardinstance}.
We provide our computational hardness results in Section~\ref{sec:hardness}. We analyze limited independent random Cauchy variables in Section~\ref{sec:limind}.
Section~\ref{sec:distributed} presents the results for the distributed setting.  Section~\ref{sec:streaming} presents the results for the streaming setting. Section~\ref{sec:exp} contains the experimental results of our algorithm and several heuristic algorithms, as well as counterexamples to heuristic algorithms.

%%%% Cut-line between first 10 pages and appendix

%We will not include this in the final submssion.
%\newpage
%\input{trash.tex}

\appendix

\newpage
\section{Notation}\label{sec:notation}
Let $\mathbb{N}_+$ denote the set of positive integers. For any $n\in \mathbb{N}_{+}$, let $[n]$ denote the set $\{1,2,\cdots,n\}$.
For any $p\in [1,2]$, the $\ell_p$-norm of a vector $x\in \mathbb{R}^d$ is defined as
\begin{equation*}
\| x \|_p = \bigl( \sum_{i=1}^d | x_i |^p \bigr)^{1/p}.
\end{equation*}
For any $p\in [1,2)$, the $\ell_p$-norm of a matrix $A\in \mathbb{R}^{n\times d}$ is defined as
\begin{equation*}
\| A \|_p = \bigl( \sum_{i=1}^n \sum_{j=1}^d |A_{ij}|^p \bigr)^{1/p}.
\end{equation*}
Let $\| A\|_F$ denote the Frobenius norm of matrix $A$. Let $\nnz(A)$ denote the number of nonzero entries of $A$. Let $\det(A)$ denote the determinant of a square matrix $A$. Let $A^\top$ denote the transpose of $A$. Let $A^\dagger$ denote the Moore-Penrose pseudoinverse of $A$. Let $A^{-1}$ denote the inverse of a full rank square matrix. We use $A_j$ to denote the $j^{\text{th}}$ column of $A$, and $A^i$ to denote the $i^{\text{th}}$ row of $A$. For an $n\times d$ matrix $A$, for $S$ a subset of $[n]$ and $T$ a subset of $[d]$, we let $A^S$ denote the $|S|\times d$ submatrix of $A$ with rows indexed by $S$, while $A_T$ denotes the $n\times |T|$ submatrix of $A$ with columns indexed by $T$, and $A_T^S$ denote the $|S|\times |T|$ submatrix $A$
with rows in $S$ and columns in $T$.

For any function $f$, we define $\wt{O}(f)$ to be $f\cdot \log^{O(1)}(f)$. In addition to $O(\cdot)$ notation, for two functions $f,g$, we use the shorthand $f\lesssim g$ (resp. $\gtrsim$) to indicate that $f\leq C g$ (resp. $\geq$) for an absolute constant $C$. We use $f\eqsim g$ to mean $cf\leq g\leq Cf$ for constants $c,C$. %For clairity, we have made no attempt to optimize the values of the constraints in our analyses.
We use $\OPT$ to denote $\min_{\rank-k~A_k}\| A_k -A\|_1$, unless otherwise
specified.

\section{Preliminaries}\label{sec:preli}

\subsection{Polynomial system verifier}\label{sec:polynomial_system_verifier}
Renegar \cite{r92a,r92b} and Basu $et~al.$ \cite{bpr96} independently provided an algorithm for the decision problem for the existential theory of the reals, which is to decide the truth or falsity of a sentence $(x_1, \cdots, x_v) F(f_1, \cdots, f_m)$ where $F$ is a quantifier-free Boolean formula with atoms of the form $\text{sign}(f_i) =\sigma$ with $\sigma \in \{0,1,-1\}$. Note that this problem is equivalent to deciding if a given semi-algebraic set is empty or not. Here we formally state that theorem. For a full discussion of algorithms in real algebraic geometry, we refer the reader to \cite{basu2005algorithms} and \cite{basu2014algorithms}.
\begin{theorem}[Decision Problem \cite{r92a,r92b,bpr96}]
  \label{decision_solver_thm}

Given a real polynomial system $P(x_1, x_2, \cdots, x_v)$ having $v$ variables and $m$ polynomial constraints $f_i (x_1, x_2, \cdots, x_v) \Delta_i 0, \forall i \in [m]$, where $\Delta_i$ is any of the ``standard relations'': $\{ >, \geq, =, \neq, \leq, < \}$, let $d$ denote the maximum degree of all the polynomial constraints and let $H$ denote the maximum bitsize of the coefficients of all the polynomial constraints. Then in
\begin{equation*}
(m d)^{O(v)} \poly(H),
\end{equation*}
time one can determine if there exists a solution to the polynomial system $P$.
\end{theorem}
Recently, this technique has been used to solve a number of low-rank approximation and matrix factorization problems \cite{agkm12,m13,cw15b,bdl16,rsw16}.

\subsection{Cauchy and $p$-stable transform}

\begin{definition}[Dense Cauchy transform]
Let $S = \sigma \cdot C \in \mathbb{R}^{m\times n}$ where $\sigma$ is a scalar, and each entry of $C\in\mathbb{R}^{m\times n}$ is chosen independently from the standard Cauchy distribution. For any matrix $A\in \mathbb{R}^{n\times d}$, $SA$ can be computed in $O(m\cdot \nnz(A))$ time.
\end{definition}
\begin{definition}[Sparse Cauchy transform]
Let $\Pi = \sigma \cdot S C\in \mathbb{R}^{m \times n}$, where $\sigma$ is a scalar, $S\in \mathbb{R}^{m\times n}$ has each column chosen independently and uniformly from the $m$ standard basis vectors of $\mathbb{R}^{m}$, and $C\in \mathbb{R}^{n\times n}$ is a diagonal matrix with diagonals chosen independently from the standard Cauchy distribution.  For any matrix $A\in \mathbb{R}^{n\times d}$, $\Pi A$ can be computed in $O(\nnz(A))$ time.
\end{definition}

\begin{definition}[Dense $p$-stable transform]
Let $p\in (1,2)$. Let $S = \sigma \cdot C \in \mathbb{R}^{m\times n}$ where $\sigma$ is a scalar, and each entry of $C\in\mathbb{R}^{m\times n}$ is chosen independently from the standard $p$-stable distribution. For any matrix $A\in \mathbb{R}^{n\times d}$, $SA$ can be computed in $O(m \nnz(A) )$ time.
\end{definition}

\begin{definition}[Sparse $p$-stable transform]
Let $p\in (1,2)$. Let $\Pi = \sigma \cdot S C\in \mathbb{R}^{m\times n}$, where $\sigma$ is a scalar, $S\in \mathbb{R}^{m\times n}$ has each column chosen independently and uniformly from the $m$ standard basis vectors of $\mathbb{R}^{m}$, and $C\in \mathbb{R}^{n\times n}$ is a diagonal matrix with diagonals chosen independently from the standard $p$-stable distribution. For any matrix $A\in \mathbb{R}^{n\times d}$, $\Pi A$ can be computed in $O(\nnz(A))$ time.
\end{definition}

\subsection{Lewis weights}\label{sec:lewis_weights}
We follow the exposition of Lewis weights from \cite{cp15}.
\begin{definition}
For a matrix $A$, let $a_i$ denote $i^{\text{th}}$ row of $A$, and $a_i(=(A^i)^\top )$ is a column vector. The statistical leverage score of a row $a_i$ is
\begin{align*}
\tau_i(A) \overset{\mathrm{def}}{=} a_i^\top (A^\top A)^{-1} a_i = \| (A^\top A )^{-1/2} a_i \|_2^2.
\end{align*}
For a matrix $A$ and norm $p$, the $\ell_p$ Lewis weights $w$ are the unique weights such that for each row $i$ we have
\begin{align*}
w_i = \tau_i ( W^{1/2-1/p} A ).
\end{align*}
or equivalently
\begin{align*}
a_i^\top (A^\top W^{1-2/p} A)^{-1} a_i = w_i^{2/p}.
\end{align*}
\end{definition}

%\begin{definition}
%For a matrix $A\in\mathbb{R}^{n\times d}$, suppose row $i$ is sampled with probability $p_i\geq c\cdot w_i$ where $w_i$ is the $\ell_p$ Lewis weights of the $i^{\text{th}}$ row of $A$. A sampling/rescaling diagonal matrix $D\in\mathbb{R}^{n\times n}$ with at most $m$ non-zero entries according to $\ell_p$ Lewis weights of $A$ is defined as: $D_{i,i}=$
%\end{definition}

\begin{lemma}[Lemma~2.4 of \cite{cp15} and Lemma~7 of \cite{clmmps15}]\label{lem:compute_lewis_weight}
Given a matrix $A\in\mathbb{R}^{n\times d}$, $n\geq d$, for any constant $C>0,4>p\geq 1$, there is an algorithm which can compute $C$-approximate $\ell_p$ Lewis weights for every row $i$ of $A$ in $O((nnz(A)+d^\omega\log d)\log n) $ time, where $\omega < 2.373$ is the matrix multiplication exponent\cite{s69,cw87,w12}.
\end{lemma}

\begin{lemma}[Theorem~7.1 of \cite{cp15}]\label{lem:num_samples}
Given matrix $A\in\mathbb{R}^{n\times d}$ ($n\geq d$) with $\ell_p$ ($4>p\geq 1$) Lewis weights $w$, for any set of sampling probabilities $p_i$, $\sum_i p_i=N$,
\begin{align*}
p_i\geq f(d,p)w_i,
\end{align*}
if $S\in\mathbb{R}^{N\times n}$ has each row chosen independently as the $i^{\text{th}}$ standard basis vector, times $1/p_i^{1/p}$, with probability $p_i/N$, then with probability at least $0.999$,
\begin{align*}
\forall x\in\mathbb{R}^d, \frac12\|Ax\|_p^p\leq\|SAx\|_p^p\leq 2\|Ax\|_p^p
\end{align*}
Furthermore, if $p=1$, $N=O(d\log d)$. If $1<p<2$, $N=O(d\log d\log \log d)$. If $2\leq p<4$, $N=O(d^{p/2}\log d)$.
\end{lemma}

%%% Peilin wrote the following new paragraph, David please take a look.
{%\color{red}
Given a matrix $A\in\mathbb{R}^{n\times d}$ ($n\geq d$), by Lemma~\ref{lem:num_samples} and Lemma~\ref{lem:compute_lewis_weight}, we are able to compute a sampling/rescaling matrix $S$ in $O((nnz(A)+d^\omega\log d)\log n) $ with $\wt{O}(d)$ nonzero entries such that
\begin{align*}
\forall x\in\mathbb{R}^d, \frac12\|Ax\|_p^p\leq\|SAx\|_p^p\leq 2\|Ax\|_p^p.
\end{align*}
Sometimes, $\poly(d)$ is much smaller than $\log n$. In this case, we are able to compute the such sampling/rescaling matrix $S$ in $n\poly(d)$ time in the following way: basically we can run one of the input sparsity $\ell_p$ embedding algorithm (see e.g. \cite{mm13}) to compute a well conditioned basis $U$ of column span of $A$ in $n\poly(d)$ time. By sampling according to the well conditioned basis (see e.g. \cite{c05,ddhkm09,w14}), we can compute a sampling/rescaling matrix $S_1$ such that $(1-\varepsilon)\|Ax\|_p^p\leq\|S_1Ax\|_p^p\leq (1+\varepsilon)\|Ax\|_p^p$ where $\varepsilon\in(0,1)$ is an arbitrary constant. Notice that $S_1$ has $\poly(d)$ nonzero entries, thus $S_1A$ has size $\poly(d)$. Now, we apply Lewis weights sampling according to $S_1A$, we can get a sampling/rescaling matrix $S$ such that 
\begin{align*}
\forall x\in\mathbb{R}^d, (1-\frac{1}{3})\|S_1Ax\|_p^p\leq\|SS_1Ax\|_p^p\leq (1+\frac{1}{3})\|S_1Ax\|_p^p.
\end{align*}
It means that
\begin{align*}
\forall x\in\mathbb{R}^d, \frac{1}{2}\|Ax\|_p^p\leq\|SS_1Ax\|_p^p\leq 2\|Ax\|_p^p.
\end{align*}
Note that $SS_1$ is still a sampling/rescaling matrix according to $A$, and the number of non-zero entries is $\wt{O}(d)$. And the total running time is thus $n\poly(d)$.
}

%%% Below is David's email on April 11, 2017
%input sparsity $\ell_1$ algorithms to get the $\elll_1$ leverage scores - say Meng Mahoney, and they spend $\nnz(A)$ time to get a change of basis. They multiply by Gaussians but if you don't care about $\poly(d)$ then in $n \poly(d)$ time they get a sampling distribution - then they have to sample $\poly(d/\epsilon)$ rows. Now, they have a $\poly(d/\epsilon) \times d$ matrix - now you can apply Lewis weights to this matrix which and so the $\log n$ factor will turn into a $\log k$ factor for you, so it should be fine.

\subsection{Frobenius norm and $\ell_2$ relaxation }

\begin{theorem}[Generalized rank-constrained matrix approximations, Theorem 2 in \cite{ft07}]\label{thm:reduce_to_frobenius}
%For any matrix $A$ with SVD $A=\sum_{i=1}^{k} \sigma_i u_i v_i^\top$, define SVD,
%\begin{equation*}
%P_{A,L} = \sum_{i=1}^{\rank(A)} u_i u_i^\top \in \mathbb{R}^{n\times n}, P_{A,R} = \sum_{i=1}^{\rank(A)} v_i v_i^\top \in \mathbb{R}^{d\times d}.
%\end{equation*}
Given matrices $A\in \mathbb{R}^{n\times d}$, $B\in \mathbb{R}^{n\times p}$, and $C\in \mathbb{R}^{q\times d}$, let the SVD of $B$ be $B=U_B\Sigma_B V_B^\top$ and the SVD of $C$ be $C=U_C\Sigma_C V_C^\top$. Then,
\begin{equation*}
B^\dagger ( U_B U_B^\top A V_C C_C^\top )_k C^\dagger = \underset{ \rank-k ~X\in \mathbb{R}^{p\times q} }{\arg \min} \| A - B X C \|_F,
\end{equation*}
where $(U_B U_B^\top A V_C V_C^\top)_k \in \mathbb{R}^{p\times q}$ is of rank at most $k$ and denotes the best rank-$k$ approximation to $U_B U_B^\top A V_C V_C^\top \in \mathbb{R}^{p\times d}$ in Frobenius norm.
\end{theorem}

\begin{claim}[$\ell_2$ relaxation of $\ell_p$-regression]\label{cla:ell2_relax_ell1_regression}
Let $p\in [1,2)$. For any $A\in \mathbb{R}^{n\times d}$ and $b\in \mathbb{R}^n$, define $x^* =\underset{x\in \mathbb{R}^d}{\arg\min} \| A x - b\|_p $ and $x'=\underset{x \in \mathbb{R}^d}{\arg\min} \| A x - b\|_2$. Then,
\begin{equation*}
\| Ax^* -b\|_p \leq \| A x' - b\|_p \leq n^{1/p-1/2} \cdot \| A x^* - b\|_p.
\end{equation*}
\end{claim}

\begin{proof}
The lower bound trivially holds by definition; we will focus on proving the upper bound. Because $Ax-b$ is an $n$-dimensional vector, $\forall x$,
\begin{equation}\label{eq:ell1_less_than_ell2}
\frac{1}{ n^{1/p-1/2} } \| A x - b\|_p \leq \| A x - b\|_2 \leq \| A x -b \|_p.
\end{equation}
Then,
\begin{align*}
 ~& \| A x'-b\|_p \\
\leq ~& \sqrt{n} \| A x' -b \|_2 & \text{~by~LHS~of~Equation~(\ref{eq:ell1_less_than_ell2})} \\
\leq ~& \sqrt{n} \| A x^* -b \|_2 & \text{~by~$x' = \underset{x}{\arg\min} \| A x-b\|_2$} \\
\leq ~& \sqrt{n} \| A x^* -b \|_p & \text{~by~RHS~of~Equation~(\ref{eq:ell1_less_than_ell2})}
\end{align*}
This completes the proof.
\end{proof}

\begin{claim}[Frobenius norm relaxation of $\ell_p$-low rank approximation]\label{cla:frobenius_relax_ell1_lowrank}
Let $p\in [1,2)$ and for any matrix $A\in \mathbb{R}^{n\times d}$, define $A^* = \underset{\rank-k~B\in \mathbb{R}^{n\times d}}{\arg \min} \| B - A\|_p$ and $A' =\underset{\rank-k~B \in \mathbb{R}^{n\times d} }{\arg \min} \| B - A \|_F$. Then
\begin{equation}
\| A^* - A \|_p \leq \| A'- A \|_p \leq (nd)^{1/p-1/2} \| A^* - A \|_p.
\end{equation}
\end{claim}

\begin{proof}
The lower bound of $\| A '- A\|_p $ trivially holds by definition. We show an upper bound of $\| A'- A\|_p$ in the rest of the proof. For any $A'- A \in \mathbb{R}^{n\times d}$, we have
\begin{equation}\label{eq:ell1_less_than_frobenius}
\frac{1}{ (nd)^{1/p-1/2}}\| A'-A\|_p \leq \|A'-A\|_F \leq \| A' - A \|_p.
\end{equation}
Then,
\begin{align*}
&\| A' - A \|_p \\
\leq ~& (nd)^{1/p-1/2} \| A'-A \|_F & \text{~by~LHS~of~Equation~(\ref{eq:ell1_less_than_frobenius})}\\
\leq ~& (nd)^{1/p-1/2} \| A^*-A \|_F & \text{~by~}A' = \underset{\rank-k~B}{\arg\min} \| B - A\|_p \\
\leq ~& (nd)^{1/p-1/2} \| A^*-A \|_p. & \text{~by~RHS~of~Equation~(\ref{eq:ell1_less_than_frobenius})}
\end{align*}
\end{proof}

\subsection{Converting entry-wise $\ell_1$ and $\ell_p$ objective functions into polynomials}\label{sec:converting_l1_lp_into_polynomial}

\begin{claim}[Converting absolute value constraints into variables]\label{cla:convert_absolute_constraints_into_variables}
Given $m$ polynomials $f_1(x), f_2(x),$ $\cdots,$ $f_m(x)$ where $x\in \mathbb{R}^{v}$, solving the problem
\begin{equation}
\min_{x \in \mathbb{R}^v }\sum_{i=1}^m |f_i(x)|,
\end{equation}
is equivalent to solving another minimization problem with $O(m)$ extra constraints and $m$ extra variables,
\begin{align*}
& \min_{x \in \mathbb{R}^v ,\sigma \in \mathbb{R}^m } \sum_{i=1}^m \sigma_i f_i(x) \\
 \mathrm{~s.t.} ~& \sigma_i^2 = 1, \forall i\in [m] \\
~& f_i(x) \sigma_i \geq 0, \forall i\in[m] .\\
\end{align*}
\end{claim}

\begin{claim}{(Handling $\ell_p$)}
Given $m$ polynomials $f_1(x), f_2(x),$ $\cdots,$ $f_m(x)$ where $x\in \mathbb{R}^{v}$ and $p=a/b$ for positive integers $a$ and $b$, solving the problem
\begin{equation}
\min_{x \in \mathbb{R}^v }\sum_{i=1}^m |f_i(x)|^p,
\end{equation}
is equivalent to solving another minimization problem with $O(m)$ extra constraints and $O(m)$ extra variables,
\begin{align*}
& \min_{x \in \mathbb{R}^v ,\sigma \in \mathbb{R}^m } \sum_{i=1}^m y_i \\
 \mathrm{~s.t.} ~& \sigma_i^2 = 1, \forall i\in [m] \\
~& f_i(x) \sigma_i \geq 0, \forall i\in[m] \\
~& (\sigma_i f_i(x))^a = y_i^b,\forall i\in [m] \\
~& y_i \geq 0, \forall i\in [m].
\end{align*}
\end{claim}

\subsection{Converting entry-wise $\ell_1$ objective function into a linear program}
\begin{claim}
Given any matrix $A\in \mathbb{R}^{n\times d}$ and matrix $B\in \mathbb{R}^{k\times d}$, the problem $\min_{U \in \mathbb{R}^{n\times k} } \| U B - A \|_1$ can be solved by solving the following linear program,
\begin{align*}
& \min_{U\in \mathbb{R}^{n\times k}, x\in \mathbb{R}^{n\times d}} \sum_{i=1}^n \sum_{j=1}^m x_{i,j}\\
~& U_i B^j - A_{i,j} \leq x_{i,j}, \forall i\in [n], j\in [d]\\
~& U_i B^j - A_{i,j} \geq -x_{i,j}, \forall i\in[n], j\in [d]\\
~& x_{i,j} \geq 0, \forall i\in [n], j\in [d],
\end{align*}
where the number of constraints is $O(nd)$ and the number of variables is $O(nd)$.
\end{claim}

\section{$\ell_1$-Low Rank Approximation}\label{sec:alg}
This section presents our main $\ell_1$-low rank approximation algorithms. Section~\ref{sec:three_existence_results} provides our three existence results. Section~\ref{sec:input_sparsity_algorithm} shows an input sparsity algorithm with $\poly(k) \log^2 d \log n$-approximation ratio. Section~\ref{sec:polyklogd_approx_algorithm} improves the approximation ratio to $\poly(k)\log d$. Section~\ref{sec:k_approx_algorithm} explains how to obtain $\wt{O}(k)$ approximation ratio. Section~\ref{sec:constant_approx_algorithm} improves the approximation ratio to $O(1)$ by outputting a rank-$2k$ solution. Section~\ref{sec:cur_decomposition_algorithm} presents our algorithm for CUR decomposition. Section~\ref{sec:rankr_B_algorithm} includes some useful properties. Our $\ell_1$-low rank approximation algorithm for a $\rank$-$r$ (where $k \leq r\leq (n,d)$ ) matrix is used as a black box (by setting $r=\poly(k)$) in several other algorithms.

\subsection{Existence results via dense Cauchy transforms, sparse Cauchy transforms, Lewis weights}\label{sec:three_existence_results}

The goal of this section is to present the existence results in Corollary~\ref{cor:three_existence_results}. We first provide some bicriteria algorithms in Theorem~\ref{thm:three_bicriteria_algorithms} which can be viewed as a ``warmup''. Then the proof of our bicriteria algorithm actually implies the existence results.
%\iffalse
%\begin{algorithm}[h]\caption{Bicriteria Algorithm}
%\begin{algorithmic}[1]
%\Procedure{\textsc{L1LowRankApproxBicriteria}}{$A,n,d,k$} \Comment{Theorem \ref{thm:three_bicriteria_algorithms}}
%\State Choose a sketch matrix $S\in \mathbb{R}^{m\times n}$ and sketch dimension $m$ according to one of
%\State $\{$dense Cauchy transform, sparse Cauchy transform, or Lewis weights matrices$\}$.
%\State  Solve $\min_{U\in \mathbb{R}^{n\times m}} \| U SA - A\|_1$
%\State \Return $U,SA$
%\EndProcedure
%\end{algorithmic}
%\end{algorithm}
%\fi
\begin{theorem}\label{thm:three_bicriteria_algorithms}%\label{thm:rank_klogk_approx_sqrtklogklogd}
Given matrix $A\in \mathbb{R}^{n\times d}$, for any $k\geq 1$, there exist bicriteria algorithms with running time $T$ (specified below), which output two matrices $U\in \mathbb{R}^{n\times m}$, $V\in \mathbb{R}^{m \times d}$ such that, with probability $9/10$,
\begin{equation*}
\| UV - A\|_1 \leq \alpha \underset{\rank-k~A_k}{ \min } \| A_k - A \|_1.
\end{equation*}

\rm{(\RN{1})}. Using a dense Cauchy transform, \\ $T=\poly(n,d,k)$, $m=O(k\log k)$, $\alpha =O(\sqrt{k\log k} \log d)$.

\rm{(\RN{2})}. Using a sparse Cauchy transform,\\
 $T= \poly(n,d,k) $,$m= O(k^5\log^5 k)$, $\alpha= O({k^{4.5} \log^{4.5} k} \log d )$.

\rm{(\RN{3})}. Sampling by Lewis weights,\\
$T=(nd)^{\wt{O}(k)}$, $m=O(k\log k)$, $\alpha = O(\sqrt{k\log k})$.

\end{theorem}

The matrices in \rm{(\RN{1})}, \rm{(\RN{2})}, \rm{(\RN{3})} here, are the same as those in \rm{(\RN{1})}, \rm{(\RN{2})}, \rm{(\RN{3})}, \rm{(\RN{4})} of Lemma~\ref{lem:con_dil_summary}. Thus, they have the properties shown in Section~\ref{sec:properties_of_con_dil}.

\begin{proof}
We define
\begin{align*}
\OPT := \min_{\rank-k~A_k} \| A_k - A \|_1.
\end{align*}
We define $U^*\in\mathbb{R}^{n\times k},V^*\in\mathbb{R}^{k\times d}$ to be the optimal solution such that
$\|U^*V^*-A\|_1=\OPT.$
%For fixed optimal $U^*$, for all $V$, we can rewrite the original problem into multiple regression problems,
%\begin{equation*}
%\underset{V \in \mathbb{R}^{k\times d} }{\min } \|  U^* V - A \|_1 = \underset{V\in \mathbb{R}^{k\times d} }{\min} \sum_{i=1}^d \| U^* V_i -  A_i\|_1 = \sum_{i=1}^d  \underset{V_i\in \mathbb{R}^{ k} }{\min} \| U^* V_i -  A_i\|_1
%\end{equation*}

Part (\RN{1}). Apply the dense Cauchy transform $S\in \mathbb{R}^{m\times n}$ with $m=O(k\log k)$ rows, and $\beta=O(\log d)$.

Part (\RN{2}). Apply the sparse Cauchy transform $S$($=\Pi\in \mathbb{R}^{m\times n}$) with $m=O(k^5\log^5 k)$ rows, and $\beta = O(\sigma \log d) = O(k^2\log^2 k \log d)$.

Part (\RN{3}). Use $S$ ($=D\in \mathbb{R}^{n\times k}$) to denote an $n\times n$ matrix which is a sampling and rescaling diagonal matrix according to the Lewis weights of matrix $U^*$. It has $m=O(k\log k)$ rows, and $\beta=O(1)$. Sometimes we abuse notation, and should regard $D$ as a matrix which has size $m\times n$, where $m=O(k\log k)$.

We can just replace $M$ in Lemma~\ref{lem:con_dil_summary} with $U^*V^*-A$, replace $U$ in Lemma~\ref{lem:con_dil_summary} with $U^*$, and replace $c_1c_2$ with $O(\beta)$. So, we can apply Lemma~\ref{lem:con_dil_summary} for $S$. Then we can plug it in Lemma~\ref{lem:general_sketch_SUV}, we have: with constant probability, for any $c\geq 1$, for any $V'\in\mathbb{R}^{k\times d}$ which satisfies
\begin{align}\label{eq:c1_1}
\|SU^*V'-SA\|_1\leq c\cdot \min_{V\in\mathbb{R}^{k\times d}}\|SU^*V-SA\|_1,
\end{align}
it has
\begin{align}\label{eq:c1_2}
\|U^*V'-A\|_1\leq c\cdot O(\beta)\|U^*V^*-A\|_1.
\end{align}

%Then for each $i\in [d]$, we define $\wt{V}_i$,
%\begin{align*}
%\wt{V}_i = \underset{V_i\in \mathbb{R}^{ k} }{\arg \min} \| S U^* V_i -  S A_i\|_1
%\end{align*}

%Due to Lemma~\ref{lem:general_sketch_SUV}, with constant probability, we have
%\begin{align}\label{eq:U*wtV_is_O1_opt_A}
%\| U^* \wt{V} - A\|_1 \leq  \beta \OPT.
%\end{align}

Define $\widehat{V}_i = \underset{V_i \in \mathbb{R}^k }{\arg\min} \| SU^* V_i - SA_i \|_2$ for each $i\in [d]$. By
using Claim \ref{cla:ell2_relax_ell1_regression} with $n= m$ and $d= k$, it shows
\begin{align*}
 \|SU^*\wh{V}-SA\|_1=\sum_{i=1}^d \| SU^* \widehat{V}_i - S A_i\|_1 \leq ~& \sum_{i=1}^d \sqrt{m}    \| SU^* \wt{V}_i - S A_i\|_1=\sqrt{m}\min_{V\in\mathbb{R}^{k\times d}}\|SU^*V-SA\|_1.
\end{align*}
%It implies
%\begin{align}\label{eq:DU*whV_is_sqrtklogk_SA_opt}
%\| SU^* \wh{V} - SA \|_1 \leq \sqrt{m} \min_{V\in \mathbb{R}^{k\times d}} \| S U^* V - SA \|_1
%\end{align}
which means $\wh{V}$ is a $\sqrt{m}$-approximation solution to problem, $\underset{V\in \mathbb{R}^{k\times d}}{\min}\| SU^*V - SA\|_1$.

%Using Lemma~\ref{lem:general_sketch_SUV} with Equation~(\ref{eq:U*wtV_is_O1_opt_A}) and (\ref{eq:DU*whV_is_sqrtklogk_SA_opt}),
Now, let us look into Equation~(\ref{eq:c1_1}) and Equation~(\ref{eq:c1_2}), we can obtain that
\begin{align*}
\| U^* \wh{V} - A\|_1 \leq \sqrt{m} O(\beta) \OPT.
\end{align*}

Because $\widehat{V}_i$ is the optimal solution of the $\ell_2$ regression problem, we have
\begin{equation*}
\widehat{V}_i = (SU^*)^\dagger S A_i \in \mathbb{R}^k, \forall i\in [d] \text{, which means~} \widehat{V} = (SU^*)^\dagger SA \in \mathbb{R}^{k\times d}.
\end{equation*}

Plugging $\wh{V}$ into original problem,
we obtain
\begin{align*}
\| U^* (SU^*)^\dagger \cdot SA - A \|_1 \leq \sqrt{m}O(\beta) \OPT.
\end{align*}

It means
\begin{align}\label{eq:existance_result}
\min_{\rank-k\ X\in\mathbb{R}^{n\times m}}\| X SA - A \|_1 \leq \sqrt{m}O(\beta) \OPT.
\end{align}

If we ignore the constraint on the rank of $X$, we can get a bicriteria solution:

For part (\RN{1}), notice that $X$ is an $n\times m$ matrix which can be found by using a linear program, because matrices $SA \in \mathbb{R}^{m \times d}$ and $A\in \mathbb{R}^{n\times d}$ are known.

For part (\RN{2}), notice that $X$ is an $n\times m$ matrix which can be found by using a linear program, because matrices $SA \in \mathbb{R}^{m \times d}$ and $A\in \mathbb{R}^{n\times d}$ are known.

For part (\RN{3}), notice that $X$ is an $n\times m$ matrix which can be found by using a linear program, when the span of rows of $DA \in \mathbb{R}^{m \times d}$ is known.
We assume that $D$ is known in all the above discussions. But $D$ is actually unknown. So we need to try all the possible choices of the row span of $DA$. Since $D$ samples at most $m=O(k\log k)$ rows of $A$, then the total number of choices of selecting $m$ rows from $n$ rows is ${n \choose m} = n^{O(k\log k)}$. This completes the proof.
%%%%%%

\end{proof}
Equation (\ref{eq:existance_result}) in the proof of our bicriteria solution implies the following result,
\begin{corollary}\label{cor:three_existence_results}
Given $A\in \mathbb{R}^{n\times d}$, there exists a $\rank$-$k$ matrix $A'\in \mathbb{R}^{n\times d}$ such that $A'\in\mathrm{rowspan}(S'A)\subseteq \mathrm{rowspan}(A)$ and $\| A'-A\|_1 \leq \alpha \cdot \underset{\rank-k~A_k}{\min} \| A-A_k\|_1$, where $S'\in\mathbb{R}^{m\times n}$ is a sketching matrix.
If $S'$

\rm{(\RN{1})}. indicates the dense Cauchy transform, then $\alpha=O(\sqrt{k\log k} \log d)$.

\rm{(\RN{2})}. indicates the sparse Cauchy transform, then $\alpha= O(k^{4.5}\log^{4.5}k \log d)$.

\rm{(\RN{3})}. indicates sampling by Lewis weights, then $\alpha = O(\sqrt{k\log k})$.
\end{corollary}

\begin{proof}
Define $\OPT= \underset{U \in \mathbb{R}^{n\times k}, V\in \mathbb{R}^{k\times d}}{\min} \| UV - A\|_1$.

{\bf Proof of (\RN{1}).}
Choose $S$ to be a dense Cauchy transform matrix with $m$ rows, then
\begin{align*}
\underset{U \in \mathbb{R}^{n\times k}, Z \in \mathbb{R}^{k\times m} }{\min} \| U Z S A - A\|_1  \leq O(\sqrt{m} \log d) \OPT,
\end{align*}
where $m=O(k\log k)$. Choosing $A'= UZSA$ completes the proof.

{\bf Proof of (\RN{2}).}
Choose $\Pi = SD\in \mathbb{R}^{m\times n}$ where $S\in \mathbb{R}^{m\times n}$ has each column chosen independently and uniformly from the $m$ standard basis vectors of $\mathbb{R}^m$, and where $D$ is a diagonal matrix with diagonals chosen independently from the standard Cauchy distribution, then
\begin{align*}
\underset{U \in \mathbb{R}^{n\times k}, Z \in \mathbb{R}^{k\times m} }{\min} \| U Z \Pi A - A\|_1  \leq O(\sqrt{m}  \sigma \log d) \OPT,
\end{align*}
where $m=O(k^5\log^5 k)$ and $\sigma=O(k^2 \log^2 k)$. Choosing $A'= UZ \Pi A$ completes the proof.

{\bf Proof of (\RN{3}).}

Choose $D$ to be the sampling and rescaling matrix corresponding to the Lewis weights of $U^*$, and let it have $m=O(k\log k)$ nonzero entries on the diagonal, then
\begin{align*}
\underset{U \in \mathbb{R}^{n\times k}, Z \in \mathbb{R}^{k\times m} }{\min} \| U Z D A - A\|_1  \leq O(\sqrt{m} ) \OPT.
\end{align*}
 Choosing $A'= UZDA$ completes the proof.
\end{proof}

%%%%%%%%%%%%%%%%%%%%%%%%%%%%%%%%%%%%%%%%%%%%%%%%%%%%%%%%%%%%%%%%%%%%%%%%%%%%%%%%%%%%%%%%%%%%%%%%%%%%%%%%%%%%%%%%%%%%%%%%%%%%%%%%%%%%%%%%%%%%%%%%%%%%%%%%%%%%%%%%%%%%%%%%%%%%%%%%%%%%%%%%%%%%%%%%%%%%%%%%%%%%%%%%%%%%%%%%%%%%%%%%%%%%%%%%%%%%%%%%%%%%%%%%%%%%%%%%%%%%%%%%%%%%%%%%%%%%%%%%%%%%%%%%%%%%%%%%%%%%%%%%%%%%%%%%%%%%%%%%%%%%%%%%%%%%%%%%%%%%%%%%%%%%%%%%%%%%%%%%%%%%%%%%%%%%%%%%%%%%%%%%%%%%%%%%%%%%%%%%%%%%%%%%%%%%%%%%%%%%%%%%%%%%%%%%%%%%%%%%%%%%%%%%%%%%%%%%%%%%%%%%%%%%%%%%%%%%%%%%%%%%%%%%%%%%%%%%%%%%%%%%%%%%%%%%%%%%%%%%%%%%%%%%%%%%%%%%%%%%%%%%%%%%%%%%%%%%%%%%%%%%%%%%%%%%%%%%%%%%%%%%%%%%%%%%%%%%%%%%%%%%%%%%%%%%%%%%%%%%%%%%%%%%%%%%%%%%%%%%%%%%%%%%%%%%%%%%%%%%%%%%%%%

\subsection{Input sparsity time, $\poly(k,\log n, \log d)$-approximation for an arbitrary matrix $A$}\label{sec:input_sparsity_algorithm}
The algorithm described in this section is actually worse than the algorithm described in the next section. But this algorithm is easy to extend to the distributed and streaming settings (See Section~\ref{sec:distributed} and Section~\ref{sec:streaming}).

\begin{algorithm}[h]\caption{Input Sparsity Time Algorithm}
\begin{algorithmic}[1]
\Procedure{\textsc{L1LowRankApproxInputSparsity}}{$A,n,d,k$} \Comment{Theorem \ref{thm:input_sparsity_algorithm}}
\State Set $s\leftarrow r\leftarrow t_1\leftarrow \wt{O}(k^5)$, $t_2\leftarrow \wt{O}(k)$.
\State Choose sparse Cauchy matrices $S\in \mathbb{R}^{s\times n}$, $R\in \mathbb{R}^{d\times r}$, $T_1\in \mathbb{R}^{t_1\times n}$.
\State Choose dense Cauchy matrices $T_2\in \mathbb{R}^{d\times t_2}$.
\State Compute $S\cdot A$, $A \cdot R$ and $T_1\cdot A \cdot T_2$.
\State Compute $XY = \arg\min_{X,Y}\| T_1 A R XY S A T_2 - T_1 A T_2 \|_F$.
\State \Return $ARX,YSA$.
\EndProcedure
\end{algorithmic}
\end{algorithm}

\begin{theorem}\label{thm:input_sparsity_algorithm}
Given matrix $A\in \mathbb{R}^{n\times d}$, for any $k\geq 1$, there exists an algorithm which takes $O(\nnz(A)) + (n+d) \cdot \poly(k)$ time and outputs two matrices $U\in \mathbb{R}^{n\times k}$, $V\in \mathbb{R}^{k\times d}$ such that
\begin{equation*}
\| UV - A \|_1 \leq O(\poly(k) \log n \log^2 d) \underset{\rank-k~ A_k}{\min}\| A_k - A\|_1
\end{equation*}
holds with probability $9/10$.
\end{theorem}
\begin{proof}
Choose a Cauchy matrix $S\in \mathbb{R}^{s\times n}$ (notice that $S$ can be either a dense Cauchy transform matrix or a sparse Cauchy transform matrix). Using Corollary~\ref{cor:three_existence_results}, we have
\begin{align*}
 \underset{U \in \mathbb{R}^{n\times k}, Z\in \mathbb{R}^{k\times s}}{\min}\| UZSA - A \|_1 \leq \alpha_s \OPT,
\end{align*}
where $\alpha_s$ is the approximation by using matrix $S$. If $S$ is a dense Cauchy transform matrix, then due to Part (\RN{1}) of Corollary~\ref{cor:three_existence_results}, $\alpha_s=O(\sqrt{k\log k} \log d)$, $s=O(k\log k)$, and computing $SA$ takes $O(s \nnz(A))$ time. If $S$ is a sparse Cauchy transform matrix, then due to Part \RN{2} of Corollary~\ref{cor:three_existence_results},  $\alpha_s=\wt{O}(k^{4.5}\log d)$, $s = \wt{O}(k^5)$, and computing $SA$ takes $\nnz(A)$ time.

We define $U^*,Z^*=\arg\min_{U,Z}\| UZSA -A \|_1$. For the fixed $Z^*\in \mathbb{R}^{k\times s}$, choose a Cauchy matrix $R \in \mathbb{R}^{d\times r}$ (note that $R$ can be either a dense Cauchy transform matrix or a sparse Cauchy transform matrix) and sketch on the right of $( U Z SA - A )$.
%  we define $U'\in \mathbb{R}^{n\times k}$,
%\begin{align}
%U' = \underset{U \in \mathbb{R}^{n\times k}}{\min}\| UZ^*SAR - AR \|_1.
%\end{align}
%using Lemma~\ref{lem:general_sketch_SUV}, we have
%\begin{align*}
%\| U'Z^*SA - A \|_1 \leq \alpha_r \underset{U \in \mathbb{R}^{n\times k}}{\min}\| UZ^*SA - A \|_1 = \alpha_r \underset{U \in \mathbb{R}^{n\times k}, Z\in \mathbb{R}^{k\times s}}{\min}\| UZSA - A \|_1 \leq \alpha_s \alpha_r \OPT .
%\end{align*}
If $R$ is a dense Cauchy transform matrix, then $\alpha_r = O(\log n)$, $r=O(k\log k)$, computing $AR$ takes $O(r\cdot \nnz(A))$ time. If $R$ is a sparse Cauchy transform matrix, then $\alpha_r= \wt{O}(k^2) \log n$, $r=\wt{O}(k^5)$, computing $AR$ takes $O(\nnz(A))$ time.

 Define a row vector $\widehat{U}^j =  A^j R ( (Z^*SA) R)^\dagger \in \mathbb{R}^k$. Then 
\begin{align*}
\forall j\in[n],\|\widehat{U}^jZ^*SAR-A^jR\|_2=\min_{x\in\mathbb{R}^k} \|x^\top Z^*SAR-A^jR\|_2.
\end{align*}
Recall that $r$ is the number of columns of $R$. Due to Claim \ref{cla:ell2_relax_ell1_regression},
\begin{equation*}
\sum_{j=1}^n  \| A^j R ( (Z^*SA) R )^\dagger Z^*SAR - A^j R\|_1 \leq O(\sqrt{r}) \sum_{j=1}^n \min_{ U^j \in \mathbb{R}^k } \| U^j Z^*SA R - A^jR\|_1,
\end{equation*}
which is equivalent to
\begin{align*}
 \| A  R ( (Z^*SA) R )^\dagger Z^*SA R - A R \|_1 &\leq O(\sqrt{r})  \min_{ U\in \mathbb{R}^{n\times k} }\| U Z^*SA R - A R\|_1,
\end{align*}
where $AR$ is an $n \times r$ matrix and $SA$ is an $s \times d$ matrix.

Using Lemma \ref{lem:general_sketch_SUV}, we obtain,
\begin{align*}
 \| A  R ( (Z^*SA) R )^\dagger Z^*SA  - A \|_1 &\leq O(\sqrt{r} \alpha_r)  \min_{ U\in \mathbb{R}^{n\times k} }\| U Z^*SA  - A \|_1.
\end{align*}

We define $X^*\in \mathbb{R}^{r\times k}$, $Y^*\in \mathbb{R}^{k\times s}$,
\begin{equation*}
X^*, Y^* = \underset{X \in \mathbb{R}^{r \times k}, Y \in \mathbb{R}^{k\times s} }{\arg\min} \| AR X Y SA -A\|_1.
\end{equation*}
Then,
\begin{align*}
\| ARX^* Y^*SA - A \|_1 & \leq \| A  R ( (Z^*SA) R )^\dagger Z^*SA  - A \|_1 \\
& \leq O(\sqrt{r} \alpha_r)  \min_{ U\in \mathbb{R}^{n\times k} }\| U Z^*SA  - A \|_1 \\
& =  O(\sqrt{r} \alpha_r)  \min_{ U\in \mathbb{R}^{n\times k}, Z\in \mathbb{R}^{k\times s} }\| U ZSA  - A \|_1 \\
& \leq O(\sqrt{r}\alpha_r \alpha_s)\OPT.
\end{align*}
It means that $ARX^*$, $Y^*SA$ gives an $O(\alpha_r \alpha_s \sqrt{r})$-approximation to the original problem.

Thus it suffices to use Lemma \ref{lem:solve_ARXYSA} to solve
\begin{equation*}
 \underset{X \in \mathbb{R}^{r \times k}, Y \in \mathbb{R}^{k\times s} }{\min} \| AR X Y SA -A\|_1,
\end{equation*}
by losing an extra $\poly(k) \log d$ factor in the approximation ratio.

 By using a sparse Cauchy transform (for the place discussing the two options), combining the approximation ratios and running times all together, we can get $\poly(k)\log(n) \log^2(d)$-approximation ratio with $O(\nnz(A))+(n+d)\poly(k)$ running time. This completes the proof.
\end{proof}

\begin{lemma}\label{lem:solve_ARXYSA}
Given matrices $A\in \mathbb{R}^{n\times d}$, $SA\in\mathbb{R}^{s\times n}$,$RA\in \mathbb{R}^{r\times d}$ where $S\in \mathbb{R}^{s \times n}$, $R\in \mathbb{R}^{d\times r}$ with $\min(n,d) \geq \max(r,s)$. For any $1 \leq k\leq \min(r,s)$, there exists an algorithm that takes $O(\nnz(A))+ (n+d)\poly(s,r,k) $ time to output two matrices $X' \in  \mathbb{R}^{r \times k}, Y' \in \mathbb{R}^{k\times s}$ such that
\begin{align*}
\| AR X' \cdot Y' SA - A\|_1 \leq \poly(r,s) \log (d)   \min_{X \in \mathbb{R}^{r\times k}, Y \in \mathbb{R}^{ k \times s} } \| AR X Y SA  -  A \|_1 %\wt{O}(r^2(t_1+s))  \sqrt{t_1 t_2}
\end{align*}
holds with probability at least $.999$.
\end{lemma}
\begin{proof}
Choose sketching matrices $T_1 \in \mathbb{R}^{t_1 \times n}$ to sketch on the left of $(ARXYSA-A)$ (note that $S$ can be either a dense Cauchy transform matrix or a sparse Cauchy transform matrix).
If $T_1$ is a dense Cauchy transform matrix, then $t_1=O(r\log r)$, $\alpha_{t_1}=O(\log d)$, and computing $T_1A$ takes $O(t_1 \cdot \nnz(A) )$ time. If $T_1$ is a sparse Cauchy transform matrix, then $t_1=\wt{O}(r^5)$, $\alpha_{t_1} = \wt{O}(r^2) \log d$, and computing $T_1 A$ takes $\nnz(A)$ time.

Choose dense Cauchy matrices $T_2^\top \in \mathbb{R}^{t_2 \times d}$ to sketch on the right of $T_1(ARXYSA - A)$ with $t_2 = O( (t_1+s)\log (t_1+s))$. We get the following minimization problem,
\begin{equation}\label{eq:TARXYSAT_minus_TAT_ell1}
\min_{X\in \mathbb{R}^{r\times k}, Y \in \mathbb{R}^{k \times s}} \| T_1 AR X Y SA T_2 - T_1 A T_2\|_1.
\end{equation}
%Let $\wh{X},\wh{Y}$ denote the optimal solution of Equation~(\ref{eq:TARXYSAT_minus_TAT_ell1}).

Define $X',Y'$ to be the optimal solution of
\begin{equation*}
\min_{X \in \mathbb{R}^{r \times k}, Y \in \mathbb{R}^{k\times s}} \| T_1 AR X Y SA T_2 - T_1 A T_2\|_F.
\end{equation*}
%Define $\wt{X}, \wt{Y}$ to the be optimal solution of
%\begin{equation*}
%\min_{X\in \mathbb{R}^{r\times k}, Y \in \mathbb{R}^{k \times s}} \| AR X Y SA -A\|_1
%\end{equation*}

Due to Claim~\ref{cla:frobenius_relax_ell1_lowrank},
\begin{align*}
\| T_1 AR X' Y' SA T_2 - T_1 A T_2 \|_1 \leq \sqrt{t_1 t_2 } \min_{X\in \mathbb{R}^{r \times k}, Y \in \mathbb{R}^{k \times s}} \| T_1 AR X Y SA T_2 - T_1 A T_2\|_1.
\end{align*}

Due to Lemma~\ref{lem:general_sketch_T1BXYCT2}
\begin{align*}
\|  AR X' Y' SA  - A \|_1 \leq \sqrt{t_1 t_2 } \alpha_{t_1}\log t_1 \min_{X\in \mathbb{R}^{r \times k}, Y \in \mathbb{R}^{k \times s}} \|   AR X Y SA   -  A \|_1.
\end{align*}
%\iffalse
%Then, we have
%\begin{align*}
% & ~~\quad\| AR X' \cdot Y' S A - A \|_1  \\
%& \lesssim ~ \| T_1 AR X' \cdot Y' S A - T_1 A  \|_1 &\text{~by~Lemma~\ref{lem:sparse_cauchy_l1_no_contraction}, regard~} U^*=AR\\
%& \lesssim ~ \| T_1 AR X' \cdot Y' S A T_2 - T_1 A T_2 \|_1 & \text{~by~Lemma~\ref{lem:dense_cauchy_l1_k_subspace}} \\
%& \leq ~ O(\sqrt{rs}) \| T_1 AR X' \cdot Y' S A T_2 - T_1 A T_2 \|_F \\
%& \leq ~ O(\sqrt{rs}) \| T_1 AR \wt{X} \cdot \wt{Y} S A T_2 - T_1 A T_2 \|_F \\
%& \leq ~ O(\sqrt{rs}) \| T_1 AR \wt{X} \cdot \wt{Y} S A T_2 - T_1 A T_2 \|_1 \\
%& \lesssim ~ O(\sqrt{rs}\log t_1) \| T_1 AR\wt{X} \cdot \wt{Y} S A - T_1 A \|_1 & \text{~by~Lemma~\ref{lem:dense_cauchy_l1_no_dilation},\#rows}(T_1)=t_1\\
%& \lesssim ~ O(\sqrt{rs} \alpha_{t_1} \log t_1   ) \|  AR\wt{X} \cdot \wt{Y} S A - A \|_1 & \text{~by~Lemma~\ref{lem:sparse_cauchy_l1_no_dilation},\#cols}(A)=d\\
%& = ~ O(\sqrt{rs} \alpha_{t_1}\log t_1   ) \min_{X, Y^\top \in \mathbb{R}^{O(k\log k) \times k}} \| AR X Y SA -A\|_1
%\end{align*}
%\fi
%Notice that, the reason for the last two inequality steps follow by no dilation lemma is, it not only holds for the optimal solution, but also holds for a fixed one.

It remains to solve
\begin{align*}
\min_{X\in \mathbb{R}^{r\times k}, Y \in \mathbb{R}^{k \times s} }\| T_1 A R X Y SA T_2- T_1 A T_2\|_F.
\end{align*}

 By using Theorem \ref{thm:reduce_to_frobenius} and choosing $T_1$ to be a sparse Cauchy transform matrix, we have that the optimal $\rank$-$k$ solution $X'Y'$ is $(T_1 A R)^\dagger (U_BU_B^\top (T_1 A T_2) V_C V_C^\top )_k (SAT_2)$ which can be computed in $O(\nnz(A)) + (n+d)\poly(s,r,k)$ time. Here, $U_B$ are the left singular vectors of $T_1AR$. $V_C$ are the right singular vectors of $SAT_2$.
\end{proof}

An alternative way of solving Equation (\ref{eq:TARXYSAT_minus_TAT_ell1}) is using a polynomial system verifier. Note that a  polynomial system verifier does not allow absolute value constraints. Using Claim \ref{cla:convert_absolute_constraints_into_variables}, we are able to remove these absolute value constraints by introducing new constraints and variables. Thus, we can get a better approximation ratio but by spending exponential running time in $k$. In the previous step, we should always use a dense Cauchy transform to optimize the approximation ratio.

\begin{corollary}
Given $A\in \mathbb{R}^{n\times d}$, there exists an algorithm which takes $nd \cdot \poly(k)+ (n+d) \cdot 2^{\wt{O}(k^2)}$ time and outputs two matrices $U\in \mathbb{R}^{n\times k}$, $V\in \mathbb{R}^{k\times d}$ such that
\begin{equation*}
\|UV - A \|_1 \leq O(\poly(k) \log n \log^2 d) \underset{\rank-k~ A_k}{\min}\| A_k - A\|_1
\end{equation*}
holds with probability $9/10$.
\end{corollary}

The $\poly(k)$ factor in the above corollary is much smaller than that in Theorem~\ref{thm:input_sparsity_algorithm}.

%%%%%%%%%%%%%%%%%%%%%%%%%%%%%%%%%%%%%%%%%%%%%%%%%%%%%%%%%%%%%%%%%%%%%%%%%%%%%%%%%%%%%%%%%%%%%%%%%%%%%%%%%%%%%%%%%%%%%%%%%%%%%%%%%%%%%%%%%%%%%%%%%%%%%%%%%%%%%%%%%%%%%%%%%%%%%%%%%%%%%%%%%%%%%%%%%%%%%%%%%%%%%%%%%%%%%%%%%%%%%%%%%%%%%%%%%%%%%%%%%%%%%%%%%%%%%%%%%%%%%%%%%%%%%%%%%%%%%%%%%%%%%%%%%%%%%%%%%%%%%%%%%%%%%%%%%%%%%%%%%%%%%%%%%%%%%%%%%%%%%%%%%%%%%%%%%%%%%%%%%%%%%%%%%%%%%%%%%%%%%%%%%%%%%%%%%%%%%%%%%%%%%%%%%%%%%%%%%%%%%%%%%%%%%%%%%%%%%%%%%%%%%%%%%%%%%%%%%%%%%%%%%%%%%%%%%%%%%%%%%%%%%%%%%%%%%%%%%%%%%%%%%%%%%%%%%%%%%%%%%%%%%%%%%%%%%%%%%%%%%%%%%%%%%%%%%%%%%%%%%%%%%%%%%%%%%%%%%%%%%%%%%%%%%%%%%%%%%%%%%%%%%%%%%%%%%%%%%%%%%%%%%%%%%%%%%%%%%%%%%%%%%%%%%%%%%%%%%%%%%%%%%%%

\subsection{$\poly(k, \log d)$-approximation for an arbitrary matrix $A$}\label{sec:polyklogd_approx_algorithm}
In this section, we explain how to get an $O(\log d ) \cdot \poly(k)$ approximation.

\begin{algorithm}[h]\caption{$\poly(k) \log d$-approximation Algorithm}
\begin{algorithmic}[1]
\Procedure{\textsc{L1LowRankApproxPolykLogd}}{$A,n,d,k$} \Comment{Theorem \ref{thm:polyklogd_approx_algorithm}}
\State Set $s\leftarrow \wt{O}(k^5)$.%$\leftarrow t_1\leftarrow \wt{O}(k^5)$, $t_2\leftarrow \wt{O}(k)$.
\State Choose sparse Cauchy matrices $S\in \mathbb{R}^{s\times n}$ and compute $S\cdot A$.
\State Implicitly obtain $B=U_B V_B$ by finding $V_B=SA\in \mathbb{R}^{s\times d}$ and
  $U_B\in \mathbb{R}^{n\times s}$ where $\forall i\in [n]$, row vector $(U_B)^i$ gives an $O(1)$ approximation to $\min_{x\in \mathrm{R}^{1\times s} } \| x S A - A^i \|_1$.
\State $U,V\leftarrow$\textsc{L1LowRankApproxB}($U_B,V_B,n,d,k,s$). \Comment{Theorem~\ref{thm:rank_r_approx_polyr_B}}
\State \Return $U,V$.
\EndProcedure
\end{algorithmic}
\end{algorithm}

Intuitively, our algorithm has two stages. In the first stage, we just want to find a low rank matrix $B$ which is a good approximation to $A$. Then, we can try to find a rank-$k$ approximation to $B$. Since now $B$ is a low rank matrix, it is much easier to find a rank-$k$ approximation to $B$. The procedure \textsc{L1LowRankApproxB}($U_B,V_B,n,d,k,s$) corresponds to Theorem~\ref{thm:rank_r_approx_polyr_B}.

%remove $\log n$ from approximation factor by paying slightly bigger $\poly(k)$ factors.
\begin{theorem}\label{thm:polyklogd_approx_algorithm}
Given matrix $A\in \mathbb{R}^{n\times d}$, for any $k\geq 1$, there exists an algorithm which takes $\nnz(A)+ (n+d) \cdot \poly(k)$ time to output two matrices $U\in \mathbb{R}^{n\times k}$, $V\in \mathbb{R}^{k\times d}$ such that
\begin{equation*}
\| U V - A \|_1 \leq \poly(k) \log d \underset{\rank-k~ A_k}{\min}\| A_k - A\|_1
\end{equation*}
holds with probability $9/10$.
\end{theorem}

\begin{proof}
We define
\begin{align*}
\OPT := \underset{\rank-k~ A_k}{\min}\| A_k - A\|_1.
\end{align*}

The main idea is to replace the given $n\times d$ matrix $A$ with another low rank matrix $B$ which also has size $n \times d$. Choose $S\in\mathbb{R}^{s\times n}$ to be a Cauchy matrix, where $s\leq\poly(k)$ (note that, if $S$ is a dense Cauchy transform matrix, computing $SA$ takes $O(s\nnz(A))$ time, while if $S$ is a sparse Cauchy transform matrix, computing $SA$ takes $O(\nnz(A))$ time). Then $B$ is obtained by taking each row of $A$ and replacing it with its closest point (in $\ell_1$-distance) in the row span of $SA$. By using Part \RN{2} of Corollary~\ref{cor:three_existence_results}, we have,
\begin{align*}
\min_{U\in \mathbb{R}^{n\times k},Z\in \mathbb{R}^{k\times s}} \| U Z S A - A \|_1 \leq O(\sqrt{s}\poly(k) \log d) \OPT.
\end{align*}
 We define $B$ to be the product of two matrices $U_B\in \mathbb{R}^{n\times s}$ and $V_B \in \mathbb{R}^{s\times d}$. We define $V_B$ to be $SA$ and $U_B$ to be such that for any $i\in [n], (U_B)^i $ gives an $O(1)$-approximation to problem  $ \underset{x\in \mathbb{R}^{1\times s}}{\min} \| x SA - A^i \|_1$, i.e.,
\begin{align*}
\| (U_B)^i S A - A^i \|_1 \leq O(1) \min_{x\in \mathbb{R}^{1\times s}} \| x SA - A^i\|_1, \forall i\in [n],
\end{align*}
which means
\begin{align*}
\| U_B SA - A \|_1 \leq O(1) \min_{X\in \mathbb{R}^{n\times s}} \| X SA -  A\|_1.
\end{align*}

For a fixed $SA \in \mathbb{R}^{s\times d}$, we can compute $D\in \mathbb{R}^{d\times d}$, which is a sampling and rescaling matrix corresponding to Lewis weights of $(SA)^\top$, and let $m=O(s\log s)$ be the number of nonzero entries on the diagonal of $D$.

Define $\wh{X}=\underset{X \in \mathbb{R}^{n\times s}}{\arg\min} \| X SA D - A D \|_1$, thus by Lemma~\ref{lem:con_dil_summary} and Lemma~\ref{lem:general_sketch_SUV}, we have
\begin{align*}
\| \widehat{X} SA - A \|_1 \leq O(1) \min_{X\in \mathbb{R}^{n\times s}} \|X SA - A\|_1.
\end{align*}
Notice that computing Lewis weights takes $d\poly(s)$ time. 
We can use $\ell_1$-regression solver and linear programming to find $\widehat{X} \in \mathbb{R}^{n\times s}$ in $(n+d)\poly(s)$ time.
 Thus $U_B$ can be found in $O(\nnz(A))+(n+d)\poly(s)$ time.

%%% To Peilin, Please double check the above running time.

 By the definition of $B$, it is an $n\times d$ matrix. Na\"ively we can write down $B$ after finding $U_B$ and $V_B$. The time for writing down $B$ is $O(nd)$. To avoid this, we can just keep a factorization $U_B$ and $V_B$. We are still able to run algorithm \textsc{L1LowRankApproxB}. Because $s=\poly(k)$, the running time of algorithm  \textsc{L1LowRankApproxB} is still $O(\nnz(A))+(n+d)\poly(k)$.

 By the definition of $B$, we have that $B$ has rank at most $s$. Suppose we then solve $\ell_1$-low rank approximation problem for rank-$s$ matrix $B$, finding a rank-$k$ $g$-approximation matrix $U V$. Due to Lemma \ref{lem:solution_to_B_is_solution_to_A}, we have that if $B$ is an $f$-approximation solution to $A$, then $U V$ is also an $O(fg)$-approximation solution to $A$,
\begin{equation*}
\| U V - A\|_1 \leq O(\log d) \cdot \poly(k) \cdot g \OPT.
\end{equation*}
Using Theorem \ref{thm:rank_r_approx_polyr_B} we have that $g=\poly(s)$, which completes the proof.

%So $| U^* V-A|_1 \leq 3 (\log d) \poly(k) \cdot \OPT$.
%\Zhao{The following paragraph needed to be written in a formal way}
%It follows we can just run our whole algorithm on $B$ instead of $A$. But the nice thing about $B$, even though it is $n \times d$, is the rank is only $O(k \log k)$, since all rows are in the row span of $SA$. This means when we do our analysis for B, looking at $\underset{V \in \mathbb{R}^{k\times d} }{\min} \|SU^* V-SB \|_1$, we can just say that $S$ is a subspace embedding for the columns of $U^*$ together with the columns of $B$, since this only has rank $O(k \log k)$, so $S$ can have $O(k \log^2 k)$ rows. Then the optimum here is only distorted by $\poly(k)$ using Mahoney and Meng \cite{mm13}, which may be less than $O(\log d)$; that is, we don't use a separate bound for the optimal solution, we just use the dilation bound for all vectors. Similarly, when we sketch on the right hand side, and similarly when we follow this up by sketching on both sides.

%Thus overall we'll get $O(\log d) \cdot \poly(k)$, with a worse $\poly(k)$ factor but only a single $\log d$ factor, as well as no dependence on n.
\end{proof}

%%%%%%%%%%%%%%%%%%%%%%%%%%%%%%%%%%%%%%%%%%%%%%%%%%%%%%%%%%%%%%%%%%%%%%%%%%%%%%%%%%%%%%%%%%%%%%%%%%%%%%%%%%%%%%%%%%%%%%%%%%%%%%%%%%%%%%%%%%%%%%%%%%%%%%%%%%%%%%%%%%%%%%%%%%%%%%%%%%%%%%%%%%%%%%%%%%%%%%%%%%%%%%%%%%%%%%%%%%%%%%%%%%%%%%%%%%%%%%%%%%%%%%%%%%%%%%%%%%%%%%%%%%%%%%%%%%%%%%%%%%%%%%%%%%%%%%%%%%%%%%%%%%%%%%%%%%%%%%%%%%%%%%%%%%%%%%%%%%%%%%%%%%%%%%%%%%%%%%%%%%%%%%%%%%%%%%%%%%%%%%%%%%%%%%%%%%%%%%%%%%%%%%%%%%%%%%%%%%%%%%%%%%%%%%%%%%%%%%%%%%%%%%%%%%%%%%%%%%%%%%%%%%%%%%%%%%%%%%%%%%%%%%%%%%%%%%%%%%%%%%%%%%%%%%%%%%%%%%%%%%%%%%%%%%%%%%%%%%%%%%%%%%%%%%%%%%%%%%%%%%%%%%%%%%%%%%%%%%%%%%%%%%%%%%%%%%%%%%%%%%%%%%%%%%%%%%%%%%%%%%%%%%%%%%%%%%%%%%%%%%%%%%%%%%%%%%%%%%%%%%%%%%%
\subsection{$\wt{O}(k)$-approximation for an arbitrary matrix $A$}\label{sec:k_approx_algorithm}

\begin{algorithm}[h]\caption{$\wt{O}(k)$-approximation Algorithm}
\begin{algorithmic}[1]
\Procedure{\textsc{L1LowRankApproxK}}{$A,n,d,k$} \Comment{Theorem \ref{thm:k_approx_algorithm}}
\State $r\leftarrow O(k\log k), m\leftarrow t_1 \leftarrow O(r\log r), t_2 \leftarrow O(m\log m)$.
\State Guess a diagonal matrix $R\in \mathbb{R}^{d\times d}$ with only $r$ $1$s. \Comment{ $R$ selects $r$ columns of $A\in \mathbb{R}^{n\times d}$.}
\State Compute a sampling and rescaling matrix $D\in \mathbb{R}^{n\times n},T_1\in \mathbb{R}^{n \times n}$ corresponding to the Lewis weights of $AR$, and let them have $m,t_1$ nonzero entries on the diagonals, respectively.
\State Compute a sampling and rescaling matrix $T_2^\top\in \mathbb{R}^{d \times d}$ according to the Lewis weights of $(DA)^\top$, and let it have $t_2$ nonzero entries on the diagonal.
\State Solve $\min_{X,Y} \| T_1 AR XY DA  T_2 - T_1 A T_2\|_1$.
\State Take the best solution $X,Y$ over all guesses of $R$.
\State \Return $ARX$, $YDA$.
\EndProcedure
\end{algorithmic}
\end{algorithm}
\begin{theorem}\label{thm:k_approx_algorithm}
Given matrix $A\in \mathbb{R}^{n\times d}$, there exists an algorithm that takes $\poly(n) \cdot d^{\wt{O}(k)} \cdot 2^{\wt{O}(k^2)}$ time and outputs two matrices $U\in \mathbb{R}^{n\times k}$, $V\in \mathbb{R}^{k\times d}$ such that
\begin{align*}
\| UV - A \|_1 \leq \wt{O}(k) \underset{\rank-k~A_k}{\min} \| A_k -A \|_1
\end{align*}
holds with probability $9/10$.
\end{theorem}
\begin{proof}
We define
\begin{align*}
\OPT := \underset{\rank-k~A_k}{\min} \| A_k - A \|_1.
\end{align*}
Let $U^*\in\mathbb{R}^{n\times k},V^*\in\mathbb{R}^{k\times d}$ satisfy
\begin{align*}
\|U^*V^*-A\|_1=\OPT.
\end{align*}

Let $S^\top\in\mathbb{R}^{d\times d}$ denote the sampling and rescaling matrix corresponding to the Lewis weights of $(V^*)^\top$, where the number of nonzero entries on the diagonal of $S$ is $s=r=O(k\log k)$. Let $R^\top \in \mathbb{R}^{d\times d}$ denote a diagonal matrix such that $\forall i\in[d]$, if $S_{i,i}\not =0$, then $R_{i,i}=1$, and if $S_{i,i}=0$, then $R_{i,i}=0$.
Since $\mathrm{rowspan}(R^\top A^\top)=\mathrm{rowspan}(S^\top A^\top)$,
\begin{align*}
\min_{Z\in \mathbb{R}^{m\times k}, V\in \mathbb{R}^{k\times d}}\| ASZV -A \|_1 = \min_{Z\in \mathbb{R}^{m\times k}, V\in \mathbb{R}^{k\times d}}\| ARZV -A \|_1.
\end{align*}
 %Let $R^\top \in \mathbb{R}^{d\times d}$ denote the sampling-rescaling matrix of $A^\top\in \mathbb{R}^{d\times n}$ according to its Lewis weights, and the number of nonzeros on diagonal of $R$ is $r=O(k\log k)$.
 Combining with Part \RN{3} of Corollary~\ref{cor:three_existence_results}, there exists a $\rank$-$k$ solution in the column span of $AR$, which means,
\begin{align}\label{eq:ARZV_minus_A_is_sqrtklogk_opt}
\min_{Z\in \mathbb{R}^{m\times k}, V\in \mathbb{R}^{k\times d}}\| ARZV -A \|_1 \leq O(\sqrt{r}) \OPT .
\end{align}
Because the number of $1$s of $R$ is $r$, and the size of the matrix is $d\times d$, there are ${d \choose r} = d^{\wt{O}(k)}$ different choices for locations of $1$s on the diagonal of $R$. We cannot compute $R$ directly, but we can guess all the choices of locations of $1$s. Regarding $R$ as selecting $r$ columns of $A$, then there are $d^{\wt{O}(k)}$ choices. There must exist a ``correct'' way of selecting a subset of columns over all all choices. After trying all of them, we will have chosen the right one.

For a fixed guess $R$, we can compute $D\in \mathbb{R}^{n\times n}$, which is a sampling and rescaling matrix corresponding to the Lewis weights of $AR$, and let $m=O(k\log^2 k)$ be the number of nonzero entries on the diagonal of $D$.

%Define $\wh{V}$
By Equation~(\ref{eq:ARZV_minus_A_is_sqrtklogk_opt}), there exists a $W\in \mathbb{R}^{r\times k}$ such that,
\begin{align}\label{eq:ARWV_minus_A_is_sqrtklogk_opt}
\min_{ V\in \mathbb{R}^{k\times d}}\| ARWV -A \|_1 \leq O(\sqrt{r}) \OPT.
\end{align}
%We fix such $W$.
We define $\wh{V}_i=\underset{V_i \in \mathbb{R}^{k\times d}}{\arg\min} \| DARW V_i -  D A_i \|_2$, $\forall i \in [d]$, which means $\wh{V}_i = (DARW)^\dagger DA_i \in \mathbb{R}^k $. Then $\wh{V}=(DARW)^\dagger DA \in \mathbb{R}^{k\times d} $. %We define $V' = \underset{V\in \mathbb{R}^{k\times d} }{\arg\min} \| DARWV - DA \|_1$ and
We define $V^*=\underset{V\in \mathbb{R}^{k\times d} }{\arg\min} \| ARW V - A \|_1$. Then, by Claim~\ref{cla:ell2_relax_ell1_regression}, it has
\begin{align*}
\|DARW\wh{V}-DA\|_1\leq O(\sqrt{m}) \min_{V\in\mathbb{R}^{k\times d}}\|DARWV-DA\|_1. %\leq O(\sqrt{m}) \|DARWV^*-DA\|_1.
\end{align*}
By applying Lemma~\ref{lem:con_dil_summary}, Lemma~\ref{lem:general_sketch_SUV} and Equation~(\ref{eq:ARWV_minus_A_is_sqrtklogk_opt}), we can show
\begin{align*}
\| ARW \wh{V} - A\|_1 \leq O(\sqrt{m})\|ARWV^*-A\|_1 \leq O(\sqrt{mr}) \OPT \leq  \wt{O}(k)\OPT.
\end{align*}

 %using Lewis weights with Lemma~\ref{lem:lewis_weights_l1_no_contraction}, Lemma~\ref{lem:lewis_weights_l1_no_dilation}, Lemma~\ref{lem:general_sketch_SUV} and Claim~\ref{cla:ell2_relax_ell1_regression}, we can show
%\begin{align*}
%\| ARW \wh{V} - A\|_1 \leq O(\sqrt{m}) \| ARW V^* - A\|_1 \leq O(\sqrt{mr}) \OPT \leq = \wt{O}(k)\OPT.
%\end{align*}

%\iffalse
%\begin{align*}
%& ~ \quad ~\| ARW \wh{V} - A\|_1 \\
%& \lesssim ~ \| DARW \wh{V} - DA\|_1 &\text{~by~Lemma~\ref{lem:lewis_weights_l1_no_contraction}} \\
%& \leq ~ \sqrt{k} \| DARW V' - DA\|_1 &\text{~by~$\ell_2$~relaxation~of~$\ell_1$} \\
%& \leq ~ \sqrt{k} \| DARW V^* - DA\|_1 &\text{~by~definition~of~} V' \\
%& \lesssim ~ \sqrt{k} \| ARW V^* - A\|_1 &\text{~by~Lemma~\ref{lem:lewis_weights_l1_no_dilation}}  \\
%& \leq ~ \wt{O}(k) \OPT & \text{~by~Equation~(\ref{eq:ARWV_minus_A_is_sqrtklogk_opt})}
%\end{align*}
%\fi
Plugging $\wh{V} = (DARW )^\dagger DA$ into $\| ARW \wh{V}- A\|_1$, we obtain that
\begin{align*}
\| ARW (DARW)^\dagger DA - A\|_1 \leq \wt{O}(k) \OPT.
\end{align*}
and it is clear that,
\begin{align*}
\min_{X\in\mathbb{R}^{r\times k},Y\in\mathbb{R}^{k\times m}}\| AR X Y DA - A \|_1 \leq \| ARW (DARW)^\dagger DA - A\|_1\leq  \wt{O}(k) \OPT.
\end{align*}

Recall that we guessed $R$, so it is known. We can compute $T_1\in \mathbb{R}^{n\times n}$, which is a sampling and rescaling diagonal matrix corresponding to the Lewis weights of $AR$, and $t_1=O(r\log r)$ is the number of nonzero entries on the diagonal of $T_1$.

Also, $DA$ is known, and the number of nonzero entries in $D$ is $m=O(k\log^2 k)$. We can compute $T_2^\top \in \mathbb{R}^{d\times d}$, which is a sampling and rescaling matrix corresponding to the Lewis weights of $(DA)^\top$, and $t_2=O(m\log m)$ is the number of nonzero entries on the diagonal of $T_2$.

Define $X^*,Y^* = \underset{X\in \mathbb{R}^{r\times k}, Y\in \mathbb{R}^{k \times m} }{\arg\min} \| T_1 (AR) XY (DA) T_2 - T_1 A T_2 \|_1$.
Thus, using Lemma~\ref{lem:general_sketch_T1BXYCT2}, we have
\begin{align*}
 \| (AR) X^* Y^* (DA) - A  \|_1 \leq \wt{O}(k) \OPT.
\end{align*}
To find $X^*,Y^*$, we need to solve this minimization problem
\begin{align*}
 \underset{X\in \mathbb{R}^{r\times k}, Y\in \mathbb{R}^{k \times m} }{\min} \| T_1 (AR) XY (DA) T_2 - T_1 A T_2 \|_1,
\end{align*}
which can be solved by a polynomial system verifier (see more discussion in Section~\ref{sec:polynomial_system_verifier} and \ref{sec:converting_l1_lp_into_polynomial}).

In the next paragraphs, we explain how to solve the above problem by using a polynomial system verifier. Notice that $T_1 (AR)$ is known and $(DA)T_2$ is also known. First, we create $r\times k$ variables for the matrix $X$, i.e., one variable for each entry of $X$. Second, we create $k\times m$ variables for matrix $Y$, i.e., one variable for each entry of $Y$. Putting it all together and creating $t_1\times t_2$ variables $\sigma_{i,j},\forall i\in[t_1], j\in [t_2]$ for handling the unknown signs, we write down the following optimization problem
\begin{align*}
\min_{X,Y} & ~\sum_{i=1}^{t_1} \sum_{j=1}^{t_2} \sigma_{i,j} (T_1 (AR) XY (DA) T_2)_{i,j}\\
\text{s.t}& ~\sigma_{i,j}^2 =1, \forall i\in [t_1], j\in [t_2] \\
& ~\sigma_{i,j} \cdot (T_1 (AR) XY (DA) T_2)_{i,j} \geq 0, \forall i\in[t_1], j\in [t_2].
\end{align*}
Notice that the number of constraints is $O(t_1t_2) = \wt{O}(k^2)$, the maximum degree is $O(1)$, and the number of variables $O(t_1t_2+ km+rk) = \wt{O}(k^2)$. Thus the running time is,
\begin{align*}
(\#\constraints\cdot \#\degree)^{O(\#\variables)} = 2^{\wt{O}(k^2)}.
\end{align*}
To use a polynomial system verifier, we need to discuss the bit complexity. Suppose that all entries are multiples of $\delta$, and the maximum is $\Delta$, i.e., each entry $\in\{-\Delta, \cdots, -2\delta,-\delta,0,\delta, 2\delta, \cdots, \Delta\}$, and $\Delta/\delta=2^{\poly(nd)}$. Then the running time is $O(\poly(\Delta/\delta)) \cdot 2^{\wt{O}(k^2)} = \poly(nd) 2^{\wt{O}(k^2)}$.

Also, a polynomial system verifier is able to tell us whether there exists a solution in a semi-algebraic set. In order to find the solution, we need to do a binary search over the cost $C$. In each step of the binary search we use a polynomial system verifier to determine if there exists a solution in,
\begin{align*}
 & ~\sum_{i=1}^{t_1} \sum_{j=1}^{t_2} \sigma_{i,j} (T_1 (AR) XY (DA) T_2)_{i,j} \leq C\\
& ~\sigma_{i,j}^2 =1, \forall i\in [t_1], j\in [t_2] \\
& ~\sigma_{i,j} \cdot (T_1 (AR) XY (DA) T_2)_{i,j} \geq 0, \forall i\in[t_1], j\in [t_2].
\end{align*}

In order to do binary search over the cost, we need to know an upper bound on the cost and also a lower bound on the minimum nonzero cost. The upper bound on the cost is $C_{\max}=O(nd \Delta)$, and the minimum nonzero cost is $C_{\min} = 2^{-\Omega(\poly(nd))}$. Thus, the total number of steps for binary search is $O(\log(C_{\max}/C_{\min}))$. Overall, the running time is
\begin{align*}
nd^{\wt{O}(k)} \cdot 2^{\wt{O}(k^2)} \cdot \log(\Delta/\delta) \cdot \log(C_{\max}/C_{\min}) = \poly(n) d^{\wt{O}(k)} 2^{\wt{O}(k^2)}.
\end{align*}
This completes the proof.
\end{proof}

Instead of solving an $\ell_1$ problem at the last step of \textsc{L1LowRankApproxK} by using a polynomial system verifier, we can just solve a Frobenius norm minimization problem. This slightly improves the running time and pays an extra $\poly(k)$ factor in the approximation ratio. Thus, we obtain the following corollary,
\begin{corollary}\label{cor:polyk_approx_algorithm}
Given matrix $A\in \mathbb{R}^{n\times d}$, there exists an algorithm that takes $\poly(n) \cdot d^{\wt{O}(k)} $ time which outputs two matrices $U\in \mathbb{R}^{n\times k}$, $V\in \mathbb{R}^{k\times d}$ such that
\begin{align*}
| UV - A \|_1 \leq \poly(k) \underset{\rank-k~A_k}{\min} \| A_k -A \|_1
\end{align*}
holds with probability $9/10$.
\end{corollary}
%%%%%%%%%%%%%%%%%%%%%%%%%%%%%%%%%%%%%%%%%%%%%%%%%%%%%%%%%%%%%%%%%%%%%%%%%%%%%%%%%%%%%%%%%%%%%%%%%%%%%%%%%%%%%%%%%%%%%%%%%%%%%%%%%%%%%%%%%%%%%%%%%%%%%%%%%%%%%%%%%%%%%%%%%%%%%%%%%%%%%%%%%%%%%%%%%%%%%%%%%%%%%%%%%%%%%%%%%%%%%%%%%%%%%%%%%%%%%%%%%%%%%%%%%%%%%%%%%%%%%%%%%%%%%%%%%%%%%%%%%%%%%%%%%%%%%%%%%%%%%%%%%%%%%%%%%%%%%%%%%%%%%%%%%%%%%%%%%%%%%%%%%%%%%%%%%%%%%%%%%%%%%%%%%%%%%%%%%%%%%%%%%%%%%%%%%%%%%%%%%%%%%%%%%%%%%%%%%%%%%%%%%%%%%%%%%%%%%%%%%%%%%%%%%%%%%%%%%%%%%%%%%%%%%%%%%%%%%%%%%%%%%%%%%%%%%%%%%%%%%%%%%%%%%%%%%%%%%%%%%%%%%%%%%%%%%%%%%%%%%%%%%%%%%%%%%%%%%%%%%%%%%%%%%%%%%%%%%%%%%%%%%%%%%%%%%%%%%%%%%%%%%%%%%%%%%%%%%%%%%%%%%%%%%%%%%%%%%%%%%%%%%%%%%%%%%%%%%%%%%%%%%%%

%%%%%%%%%%%%%%%%%%%%%%%%%%%%%%%%%%%%%%%%%%%%%%%%%%%%%%%%%%%%%%%%%%%%%%%%%%%%%%%%%%%%%%%%%%%%%%%%%%%%%%%%%%%%%%%%%%%%%%%%%%%%%%%%%%%%%%%%%%%%%%%%%%%%%%%%%%%%%%%%%%%%%%%%%%%%%%%%%%%%%%%%%%%%%%%%%%%%%%%%%%%%%%%%%%%%%%%%%%%%%%%%%%%%%%%%%%%%%%%%%%%%%%%%%%%%%%%%%%%%%%%%%%%%%%%%%%%%%%%%%%%%%%%%%%%%%%%%%%%%%%%%%%%%%%%%%%%%%%%%%%%%%%%%%%%%%%%%%%%%%%%%%%%%%%%%%%%%%%%%%%%%%%%%%%%%%%%%%%%%%%%%%%%%%%%%%%%%%%%%%%%%%%%%%%%%%%%%%%%%%%%%%%%%%%%%%%%%%%%%%%%%%%%%%%%%%%%%%%%%%%%%%%%%%%%%%%%%%%%%%%%%%%%%%%%%%%%%%%%%%%%%%%%%%%%%%%%%%%%%%%%%%%%%%%%%%%%%%%%%%%%%%%%%%%%%%%%%%%%%%%%%%%%%%%%%%%%%%%%%%%%%%%%%%%%%%%%%%%%%%%%%%%%%%%%%%%%%%%%%%%%%%%%%%%%%%%%%%%%%%%%%%%%%%%%%%%%%%%%%%%%%%%%
\subsection{Rank-$2k$ and $O(1)$-approximation algorithm for an arbitrary matrix $A$}\label{sec:constant_approx_algorithm}
In this section, we show how to output a rank-$2k$ solution that is able to achieve an $O(1)$-approximation.

\begin{algorithm}[h]\caption{Bicriteria $O(1)$-approximation Algorithm}
\begin{algorithmic}[1]
\Procedure{\textsc{L1LowRankApproxBicriteria}}{$A,n,d,k$} \Comment{Theorem \ref{thm:constant_approx_algorithm}}
%\State $U_B, V_B \leftarrow$ \textsc{L1LowRankApproxK}($A,n,d,k$) \Comment{Corollary~\ref{cor:polyk_approx_algorithm} or Theorem~\ref{thm:k_approx_algorithm}}
\State $U_B,V_B \leftarrow \underset{U\in \mathbb{R}^{n\times k}, V\in \mathbb{R}^{k\times d} }{\min} \| UV - A \|_F$.
\State $r\leftarrow O(k\log k)$.%, m\leftarrow t_1 \leftarrow O(r\log r), t_2 \leftarrow O(m\log m)$.
\State Guess a diagonal matrix $D\in \mathbb{R}^{n\times n}$ with $r$ nonzero entries.
\State Guess matrix $DU\in \mathbb{R}^{r\times k}$.
\State Find $V_A$ by solving $\min_{V} \| DU V - D(A-B) \|_1$.
\State Find $U_A$ by solving $\min_{U} \| U V_A - (A-B) \|_1$.
\State Take the best solution $\begin{bmatrix} U_A & U_B \end{bmatrix}$, $\begin{bmatrix}V_A \\ V_B \end{bmatrix}$ over all guesses.
\State \Return $\begin{bmatrix} U_A & U_B \end{bmatrix}$, $\begin{bmatrix}V_A \\ V_B \end{bmatrix}$.
\EndProcedure
\end{algorithmic}
\end{algorithm}
\begin{theorem}\label{thm:constant_approx_algorithm}
Given matrix $A\in \mathbb{R}^{n\times d}$, for any $k\geq 1$, there exists an algorithm which takes $(nd)^{\wt{O}(k^2)}$ time to output two matrices $U\in \mathbb{R}^{n\times 2k}$, $V\in \mathbb{R}^{2k\times d}$,
\begin{align*}
\| UV - A \|_1 \lesssim \underset{\rank-k~A_k}{\min} \| A_k - A\|_1,
\end{align*}
holds with probability $9/10$.
\end{theorem}

\begin{proof}

We define $U^*\in \mathbb{R}^{n\times k}, V^*\in \mathbb{R}^{k\times d}$ to be the optimal solution, i.e.,
\begin{align*}
U^* V^* = \underset{U \in \mathbb{R}^{n\times k}, V\in \mathbb{R}^{k\times d} }{\arg\min} \| U V - A \|_1,
\end{align*}
and define $\OPT=\| U^*V^*-A\|_1$.

Solving the Frobenius norm problem, we can find a factorization of a rank-$k$ matrix $B=U_BV_B$ where $U_B\in\mathbb{R}^{n\times k},V_B\in\mathbb{R}^{k\times d}$, and $B$ satisfies
$\| A - B \|_1 \leq \alpha_B \OPT$ for an $\alpha_B=\sqrt{nd}$.

Let $D$ be a sampling and rescaling diagonal matrix corresponding to the Lewis weights of $U^*$, and let the number of nonzero entries on the diagonal be $t=O(k\log k)$.

By Lemma~\ref{lem:con_dil_summary} and Lemma~\ref{lem:general_sketch_SUV}, the solution of $\underset{V\in \mathbb{R}^{k\times d}}{\min}\|DU^* V- D(A-B)\|_1$ together with $U^*$ gives an $O(1)$-approximation to $A-B$. In order to compute $D$ we need to know $U^*$. Although we do not know $U^*$, there still exists a way to figure out the Lewis weights. The idea is the same as in the previous discussion Lemma~\ref{lem:naive_algorithm} ``Guessing Lewis weights''.
By Claim~\ref{cla:number_of_support} and Claim~\ref{cla:number_of_D}, the total number of possible $D$ is $n^{O(t)}$.

Lemma~\ref{lem:naive_algorithm} tries to find a $\rank$-$k$ solution when all the entries in $A$ are integers at most $\poly(nd)$. Here we focus on a bicriteria algorithm that outputs a $\rank$-$2k$ matrix without such a strong bit complexity assumption.
We can show a better claim~\ref{cla:number_of_DU*}, which is that the total number of possible $DU^*$ is $N^{\wt{O}(k^2)}$ where $N=\poly(n)$ is the
number of choices for a single entry in $U^*$.

We explain how to obtain an upper bound on $N$. Consider the optimum $\|U^*V^*-(A-B)\|_1$. We can always change the basis so assume $V^*$ is an Auerbach basis (i.e., an $\ell_1$-well-conditioned basis discussed in Section \ref{sec:intro}), so $e\|x\|_1 \geq \|xV^*\|_1 \geq \|x\|_1/f$, where $e,f = \poly(k)$. Then no entry of $U^*$ is larger than $2f\|A-B\|_1$, otherwise we could replace $U^*$ with $0$ and get a better solution. Also any entry smaller than $\|A-B\|_1/(e nk \alpha_B 100)$ can be replaced with $0$ as this will incur additive error at most $\OPT/100$. So if we round to integer multiples of $\|A-B\|_1/(enk \alpha_B 100)$ we only have $O(enk \alpha_B f)$ possibilities for each entry of $U^*$ and still have an $O(1)$-approximation. We will just refer to this rounded $U^*$ as $U^*$, abusing notation.

Let ${\cal U}$ denote the set of all the matrices $U$ that we guess. From the above discussion, we conclude that, there exists a $U\in {\cal U}$ such that $\| U V^* - (A-B) \|_1\leq O(\OPT)$.

For each guess of $DU^*$ and $D$, we find $V_A,U_A$ in the following way.
We find $V_A$ by using a linear program to solve,
\begin{align*}
\underset{V\in \mathbb{R}^{k\times d}}{ \min} \| DU^* V - D(A-B) \|_1.
\end{align*}

 Given $V_A$ and $A$, we write down a linear program to solve this problem,
\begin{align*}
\underset{U \in \mathbb{R}^{n\times k} }{\min} \| U V_A - (A -B) \|_1,
\end{align*}
which takes $\poly(ndk)$ time. Then we obtain $U_A$.

Recall that $V_B, U_B$ are the two factors of $B$ and it is a $\rank$-$k$, $\alpha_B$-approximation solution to $\underset{U,V}{\min} \| U V - A\|_1$. Then we have

\begin{align*}
\biggl\|  \begin{bmatrix} U_A & U_B \end{bmatrix} \begin{bmatrix} V_A \\ V_B \end{bmatrix} -  A \biggr\|_1 = \biggl\| U_A V_A - (A-U_B V_B) \biggr\|_1 = \biggl\| U_A V_A - (A-B) \biggr\|_1.
\end{align*}
Because there must exist a pair $U_A,V_A$ satisfying $\| U_A V_A - (A-B)\|_1 \leq O(\OPT)$, it follows that by taking the best solution $\begin{bmatrix} U_A & U_B \end{bmatrix} \begin{bmatrix} V_A \\ V_B \end{bmatrix}$ over all guesses, we obtain an $O(1)$-approximation solution.

Overall, the running time is $(nd)^{\wt{O} (k^2)}$.
\end{proof}

\begin{lemma}\label{lem:naive_algorithm}
Given an $n\times d$ matrix $A$ with integers bounded by $\poly(n)$, for any $k\geq 1$, there exists an algorithm which takes $(nd)^{\wt{O}(k^3)}$ time to output two matrices $U\in \mathbb{R}^{n\times k}$, $V\in \mathbb{R}^{k\times d}$, such that
\begin{align*}
\| UV - A \|_1 \lesssim \underset{\rank-k~A_k}{\min} \| A_k - A\|_1,
\end{align*}
%holds with probability at least $9/10$.
\end{lemma}

\begin{proof}

We define $U^*\in \mathbb{R}^{n\times k}, V^*\in \mathbb{R}^{k\times d}$ to be the optimal solution, i.e.,
\begin{align*}
U^*, V^* = \underset{U \in \mathbb{R}^{n\times k}, V\in \mathbb{R}^{k\times d} }{\arg\min} \| U V - A \|_1,
\end{align*}
and define $\OPT=\| U^*V^*-A\|_1$.

%Using Theorem~\ref{thm:k_approx_algorithm} or Corollary~\ref{cor:polyk_approx_algorithm}, we can find a factorization of a rank-$k$ matrix $B=U_BV_B$ where $U_B\in\mathbb{R}^{n\times k},V_B\in\mathbb{R}^{k\times d}$, and $B$ satisfies %two factorization matrices $U_B \in \mathbb{R}^{n\times k}, V_B\in \mathbb{R}^{k\times d}$ ( of rank-$k$ matrix $B$ ) for which
%$\| A - B \|_1 \leq \alpha_B \OPT$ for an $\alpha_B=\poly(k)$.

Let $D$ denote a sampling and rescaling diagonal matrix corresponding to the Lewis weights of $U^*$, and let the number of nonzero entries on the diagonal be $t=O(k\log k)$.

By Lemma~\ref{lem:con_dil_summary} and Lemma~\ref{lem:general_sketch_SUV}, the solution to $\underset{V\in \mathbb{R}^{k\times d}}{\min}\|DU^* V- DA\|_1$ together with $U^*$ gives an $O(1)$-approximation to $A$. In order to compute $D$, we need to know $U^*$. Although we do not know $U^*$, there still exists a way to figure out the Lewis weights. We call the idea ``Guessing Lewis weights''. We will explain this idea in the next few paragraphs.

First, we can guess the nonzero entries on the diagonal, because the number of choices is small.
\begin{claim}\label{cla:number_of_support}
The number of possible choice of $\supp(D)$ is at most $n^{O(t)}$.
\end{claim}
\begin{proof}
The matrix has dimension $n\times n$ and the number of nonzero entries is $t$. Thus the number of possible choices is at most $\sum_{i=1}^t {n \choose t} = n^{O(t)}$.
\end{proof}
Second, we can guess the value of each probability. For each probability, it is trivially at most $1$. If the probability is less than $1/ (\poly(n)  k \log k)$, then we will never sample that row with high probability. It means that we can truncate the probability if it is below that threshold. We can also round each probability to $2^{-i}$ which only loses another constant factor in the approximation ratio. Thus, we have:
\begin{claim}\label{cla:number_of_D}
The total number of possible $D$ is $n^{O(t)}=n^{\wt{O}(k)}$.
\end{claim}

Since $\Delta/\delta \leq \poly(n)$ and the entries of $A$ are in $\{-\Delta,-\Delta+\delta ,\ldots, -2\delta,-\delta,0,\delta,2\delta, \cdots,\Delta-\delta, \Delta \}$, we can lower bound the cost of $\|U^*V-A\|_1$ given that it is non-zero by $(nd\Delta/\delta)^{-O(k)}$ (if it is zero then $A$ has rank at most $k$ and we output $A$) using Lemma 4.1 in \cite{cw09} and relating entrywise $\ell_1$-norm to Frobenius norm. We can assume $V$ is an $\ell_1$ well-conditioned basis, since we can replace $U^*$ with $U^*R^{-1}$ and $V$ with $RV$ for any invertible linear transformation $R$. By properties of such basis, we can discretize the entries of $U^*$ to integer multiples of $(nd\Delta/\delta)^{-O(k)}$ while preserving relative error. Hence we can correctly guess each entry of $DU^*$ in $\left(n^{O(k)}\right)$ time.%Hence, we can correctly guess $D$ and $DU^*$ in $(nd \Delta/\delta)^{O(k)}$ time,
\begin{claim}\label{cla:number_of_DU*}
The total number of possible $DU^*$ is $n^{\wt{O}(k^3)}$. %$(nd \Delta/\delta)^{O(k)}$.
\end{claim}

 %all the different possibilities is at most $n^{t}$ where $t=O(k\log k)$ is the number of nonzeros in matrix $D$. Guessing the probabilities of all the chosen rows takes $(\log n)^{t}$, since the we'll never sample a row with probability less than $1/\Theta(n k \log k)$. Thus, we can try all of them.

%For the Lewis weights sampling-rescaling diagonal matrix $D$ which is able to give an $O(1)$ approximation, we assume that there exists some constant $c\geq 1$ such that $\| D A \|_1 \leq k^c \OPT$. (We will explain how to remove this assumption soon.)
In the following, let $DU$ denote a guess of $DU^*$.
Now the problem remaining is to solve $\underset{V\in \mathbb{R}^{k\times d}}{\min} \| DU V - DA \|_1$. Since we already know $DA$ can be computed, and we know $DU$, we can solve this multiple regression problem by running linear programming. Thus the running time of this step is in $\poly(nd)$. After we get such a solution $V$, we use a linear program to solve $\underset{U\in \mathbb{R}^{n\times k}}{\min} \| DU V - DA \|_1$. Then we can get $U$.

After we guess all the choices of $D$ and $DU^*$, we must find a solution $U,V$ which gives an $O(1)$ approximation. The total running time is $n^{\wt{O}(k)}\cdot n^{\wt{O}(k^3)}\cdot \poly(nd)=(nd)^{\wt{O}(k^3)}$.
\end{proof}

\subsection{CUR decomposition for an arbitrary matrix $A$}\label{sec:cur_decomposition_algorithm}

\begin{algorithm}[h]\caption{CUR Decomposition Algorithm}
\begin{algorithmic}[1]
\Procedure{\textsc{L1LowRankApproxCUR}}{$A,n,d,k$} \Comment{Theorem \ref{thm:cur_decomposition_algorithm}}
\State $U_B, V_B\leftarrow$\textsc{L1LowRankApproxPolykLogd}($A,n,d,k$).
\State Let $D_1\in\mathbb{R}^{n\times n}$ be the sampling and rescaling diagonal matrix corresponding to the Lewis weights of $B_1=U_B\in \mathbb{R}^{n\times k}$, and let $D_1$ have $d_1=O(k\log k)$ nonzero entries. %according to Lewis weights
\State Let $D_2^\top\in\mathbb{R}^{d\times d}$ be the sampling and rescaling diagonal matrix corresponding to the Lewis weights of $B_2^\top=\left((D_1 B_1)^\dagger D_1 A\right)^\top \in \mathbb{R}^{d\times k}$, and let $D_2$ have $d_2=O(k\log k)$ nonzero entries. %according to Lewis weights
%\State Compute Lewis weights matrix $D_1$ according to $B_1=U_B\in \mathbb{R}^{n\times k}$ %which is an orthonormal basis of $\wh{B}$
%\State Compute Lewis weights matrix $D_2$ according to $B_2=(D_1 B_1)^\dagger D_1 A \in \mathbb{R}^{k\times d}$
\State $C\leftarrow AD_2$, $U\leftarrow (B_2 D_2)^\dagger (D_1 B_1)^\dagger$, and $R\leftarrow D_1 A$.
\State \Return $C,U,R$.
\EndProcedure
\end{algorithmic}
\end{algorithm}

\begin{theorem}\label{thm:cur_decomposition_algorithm}
Given matrix $A\in \mathbb{R}^{n\times d}$, for any $k\geq 1$, there exists an algorithm which takes $O(\nnz(A)) +(n+d)\poly(k)$ time to output three matrices $C\in \mathbb{R}^{n\times c}$ with columns from $A$, $U\in \mathbb{R}^{c\times r}$, and $R\in \mathbb{R}^{r\times d}$ with rows from $A$, such that $\rank(CUR)=k$, $c=O(k\log k)$, $r=O(k\log k)$, and
\begin{align*}
\| C U  R - A \|_1 \leq \poly(k) \log d \underset{\rank-k~A_k}{\min} \| A_k - A \|_1,
\end{align*}
holds with probability $9/10$.
\end{theorem}
\begin{proof}
We define
\begin{align*}
\OPT := \underset{\rank-k~A_k}{\min} \| A_k - A \|_1.
\end{align*}
Due to Theorem~\ref{thm:polyklogd_approx_algorithm}, we can output two matrices $U_B\in \mathbb{R}^{n\times k}$, $V_B\in \mathbb{R}^{k\times d}$ such that $U_B V_B$ gives a $\rank$-$k$, and $\poly(k) \log d$-approximation solution to $A$, i.e.,
\begin{align}\label{eq:UBVB_minus_A_is_polyklogd_opt}
\| U_B V_B - A \|_1 \leq \poly(k) \log d \OPT.
\end{align}
By Section~\ref{sec:lewis_weights}, we can compute $D_1 \in \mathbb{R}^{n\times n}$ which is a sampling and rescaling matrix corresponding to the Lewis weights of  $B_1=U_B$ in $O(n\poly(k))$ time, and there are $d_1=O(k\log k)$ nonzero entries on the diagonal of $D_1$.
%Define $B_1 \in \mathbb{R}^{n\times k}$ to be column basis of $\wh{B}$.

\iffalse
We can show, for a fixed $V\in \mathbb{R}^{k\times d}$, with constant probability,  $\| D_1 ( B_1 V - A ) \|_1 \lesssim \| B_1 V - A \|_1 $.
\begin{align*}
& ~ \quad ~\| D_1 ( B_1 V - A ) \|_1 \\
& \leq ~ \| D_1 B_1 ( V- V^*) \|_1 + \| D_1 (B_1 V^* - A) \|_1 & \text{~by~triangle~inequality} \\
& \lesssim ~ \| B_1 (V - V^* ) \|_1 + \| D_1 (B_1 V^* - A) \|_1  &\text{~by~Lemma~\ref{lem:lewis_weights_l1_no_dilation}} \\
& \lesssim ~ \| B_1 (V - V^* ) \|_1 + \|   (B_1 V^* - A) \|_1  &\text{~by~Markov's~inequality} \\
& \leq ~ \| B_1 V - A \|_1 + \| B_1 V^* - A\|_1 + \| B_1 V^* - A \|_1  & \text{~by~triangle~inequality} \\
& \leq ~ 3 \| B_1 V - A \|_1 & \text{~by~} \| B_1 V^*- A\|_1 \leq \| B_1 V - A \|_1.
\end{align*}

On the other side, we can show no contraction bound by Lemma~\ref{lem:lewis_weights_l1_no_contraction}, with constant probability, for all $V\in \mathbb{R}^{k\times d}$,
\begin{align*}
\|  D_1 (B_1 V - A ) \|_1 \gtrsim \| B_1 V - A \|_1.
\end{align*}
\fi

Define $V^* \in \mathbb{R}^{k \times d}$ to be the optimal solution of $\underset{V\in \mathbb{R}^{k\times d}}{\min} \| B_1 V - A\|_1$, $\wh{V}= (D_1 B_1)^\dagger D_1 A \in \mathbb{R}^{k\times d}$, $U_1 \in \mathbb{R}^{n\times k}$ to be the optimal solution of %this problem,
 $\underset{U \in \mathbb{R}^{n\times k}}{\min} \| U \wh{V} - A \|_1$, and $V'$ to be the optimal solution of $\underset{V\in \mathbb{R}^{k\times d}}{\min} \| D_1 A - D_1 B_1 V\|_1$.

By Claim~\ref{cla:ell2_relax_ell1_regression}, we have
\begin{align*}
\| D_1B_1 \wh{V} - D_1A\|_1 \leq \sqrt{d_1}  \| D_1B_1 V' - D_1A\|_1.
\end{align*}

Due to Lemma~\ref{lem:con_dil_summary} and Lemma~\ref{lem:general_sketch_SUV}, with constant probability, we have
\begin{align*}
\| B_1 \wh{V} - A \|_1 \leq \sqrt{d_1}\alpha_{D_1} \| B_1 V^* - A\|_1,
\end{align*}
where $\alpha_{D_1}=O(1)$.

%Recall that $V^* \in \mathbb{R}^{k\times d}$ is the optimal solution of $ \underset{V\in \mathbb{R}^{k\times d}}{\min} \|  B_1 V - A\|_1$.
%Define $Z= (D_1 B_1)^\dagger \in \mathbb{R}^{k \times d_1}$, and $\wt{V}= \underset{V\in \mathbb{R}^{k\times d}}{\arg\min} \| B_1 V - A \|_1$.
Now, we can show,
\begin{align}\label{eq:A_minus_U1_whV_is_polyklogd_opt}
\| U_1\wh{V} - A \|_1 & \leq ~\| B_1 \wh{V} - A\|_1 &\text{~by~}U_1=\underset{U\in \mathbb{R}^{n\times k}}{\arg\min} \| U \wh{V} - A\|_1 \notag \\
%& \lesssim ~\| D_1 B_1 \wh{V} -D_1 A  \|_1 &\text{~by~Lemma~\ref{lem:lewis_weights_l1_no_contraction}} \notag \\
%& \lesssim ~\sqrt{k\log k} \| D_1 B_1 V' -D_1 A  \|_1 &\text{~by~$\ell_2$~relaxation~of~$\ell_1$} \notag \\
%& \leq ~\sqrt{k\log k} \|   D_1 B_1 V^* - D_1 A \|_1 \notag & \text{~by~}V'=\underset{V\in \mathbb{R}^{k\times d}}{\arg\min} \| D_1 A - D_1 B_1 V\|_1\\
%& \lesssim ~\sqrt{k\log k} \| B_1 V^* - A   \|_1 &\text{~by~Lemma~\ref{lem:lewis_weights_l1_no_dilation}} \notag \\
& \lesssim ~ \sqrt{d_1} \| B_1 V^* - A   \|_1 \notag \\ %\min_{V\in \mathbb{R}^{k\times d} }\| B_1 V - A   \|_1 \notag \\
& \leq ~ \sqrt{d_1} \| U_B V_B - A   \|_1  \notag \\
& \leq ~ \poly(k) \log d \OPT. & \text{~by~Equation~(\ref{eq:UBVB_minus_A_is_polyklogd_opt})}
\end{align}

%\begin{align*}
%\| A - U_1 Z D_1 A \|_1 & \leq ~ \| A - B_1 \wt{V}  \|_1 + \| B_1\wt{V} - U_1 Z D_1 A\|_1 \\
%& \leq \poly (k) \log (d) \OPT +  \| B_1 \wt{V} - U_1 Z D_1 A\|_1 \\
%& \leq \poly (k) \log (d) \OPT + O( \sqrt{k\log k} ) \underset{V\in \mathbb{R}^{k \times d}}{\min} \| B_1 V - A\|_1
%\end{align*}

We define $B_2=\wh{V}$, then we replace $ \wh{V}$ by $B_2\in \mathbb{R}^{k\times d}$ and look at this objective function,
\begin{align*}
\underset{U \in \mathbb{R}^{n\times k}}{\min} \| U B_2 - A \|_1,
\end{align*}
where $U^*$ denotes the optimal solution. We use a sketching matrix to sketch the $\RHS$ of matrix $U B_2 - A \in \mathbb{R}^{n\times d}$. Let $D_2^\top\in \mathbb{R}^{d\times d}$ denote a sampling and rescaling diagonal matrix corresponding to the Lewis weights of $B_2^\top\in \mathbb{R}^{d\times k}$, and let the number of nonzero entries on the diagonal of $D_2$ be $d_2=O(k\log k)$.
\iffalse
Then we first show the no dilation bound, for a fixed $U\in \mathbb{R}^{n\times k}$, with constant probability,  $\| (U B_2 - A) D_2 \|_1 \lesssim  \| U B_2 - A \|_1$.
\begin{align*}
 & ~\quad ~\| (U B_2 - A) D_2 \|_1 \\
& \leq ~ \| (U - U^*) B_2 D_2 \|_1 + \| ( U^* B_2 - A )  D_2 \|_1 & \text{~by~triangle~inequality} \\
& \lesssim ~ \| (U - U^*) B_2 \|_1  + \| ( U^* B_2 - A )  D_2 \|_1 & \text{~by~Lemma~\ref{lem:lewis_weights_l1_no_dilation}} \\
& \lesssim ~ \| (U - U^*) B_2 \|_1  + \| U^* B_2 - A  \|_1 & \text{~by~Markov's~inequality} \\
& \leq ~ \| U B_2 - A \|_1 + \| U^* B_2 - A\|_1 + \| U^* B_2 - A  \|_1 & \text{~by~triangle~inequality} \\
& \leq ~ 3 \| U B_2 - A \|_1 & \text{~by~}\| U^* B_2 - A \|_1 \leq \| U B_2 -A \|_1.
\end{align*}
Similar as before, we can show the no contraction bound, with constant probability, for all $U\in \mathbb{R}^{n\times k}$
\begin{align*}
\| (U B_2 - A ) D_2 \|_1 \gtrsim \| U B_2 - A \|_1.
\end{align*}
\fi
We define $\wh{U} = A D_2 (B_2 D_2)^\dagger \in \mathbb{R}^{n\times k}$,
%Define $X$ to be the optimal solution of $\min_X \| \wh{U} X B_2 - A \|_1$.
$U'\in \mathbb{R}^{n\times k}$ to be the optimal solution of $\underset{U\in \mathbb{R}^{n\times k}}{\min}\| (U B_2 - A) D_2\|_1$. Recall that $U_1\in \mathbb{R}^{n\times k}$ is the optimal of $\underset{U \in \mathbb{R}^{n\times k}}{\min} \| U B_2 - A\|_1$.

By Claim~\ref{cla:ell2_relax_ell1_regression}, we have
\begin{align*}
\| \wh{U} B_2 D_2 - A D_2 \|_1 \leq \sqrt{d_2} \| U' B_2 D_2 - A D_2 \|_1.
\end{align*}

According to Lemma~\ref{lem:general_sketch_SUV} and Lemma~\ref{lem:con_dil_summary}, with constant probability,
\begin{align*}
\| \wh{U} B_2 - A \|_1 \leq \alpha_{D_2}\sqrt{d_2} \| U_1 B_2 - A \|_1,
\end{align*}
where $\alpha_{B_2}=O(1)$.

%Using Lewis weights with Lemma~\ref{lem:lewis_weights_l1_no_contraction}, Lemma~\ref{lem:lewis_weights_l1_no_dilation}, Lemma~\ref{lem:general_sketch_SUV}, with constant probability, we have
%\begin{align*}
%\| U' B_2 - A \|_1 \leq \alpha_{D_2} \min_{U\in \mathbb{R}^{n\times k}} \| U B_2 - A \|_1
%\end{align*}
%where $\alpha_{B_2}=O(1)$. Using Claim~\ref{cla:ell2_relax_ell1_regression}, we have
%\begin{align*}
%\| \wh{U} B_2 - A \|_1 \leq \sqrt{d_2} \alpha_{D_2}   \min_{U\in \mathbb{R}^{n\times k}} \| U B_2 - A \|_1
%\end{align*}

We have
\begin{align*}
%\| \wh{U} X B_2 - A \|_1 & \leq ~ \|\wh{U}B_2 - A \|_1 \\
 & ~\quad ~\|\wh{U }B_2 - A \|_1 \\
%& \lesssim ~ \|(\wh{U}B_2 - A ) D_2\|_1  & \text{~by~Lemma~\ref{lem:lewis_weights_l1_no_contraction}} \\
%& \lesssim ~ \sqrt{k\log k} \| (U' B_2 - A) D_2\|_1 & \text{~by~$\ell_2$~relaxation~$\ell_1$}\\
%& \leq ~ \sqrt{k\log k} \| (U^* B_2 - A) D_2\|_1 & \text{~by~} U'=\underset{U\in \mathbb{R}^{n\times k}}{\arg\min} \| (UB_2 - A) D_2\|_1 \\
%& \lesssim ~ \sqrt{k\log k} \| U^* B_2 - A\|_1 & \text{~by~Lemma~\ref{lem:lewis_weights_l1_no_dilation}} \\
%& \leq ~\sqrt{d_2} \alpha_{D_2} \min_{U \in \mathbb{R}^{n\times k}} \| U B_2 - A \|_1 \\
& \leq ~\sqrt{d_2} \alpha_{D_2}\| U_1 B_2 - A\|_1   \\
& = ~ \sqrt{d_2} \alpha_{D_2} \| U_1 \wh{V} - A\|_1  &\text{~by~} B_2=\wh{V} \\
& \leq ~ \poly(k) \log (d) \OPT. & \text{~by~Equation~(\ref{eq:A_minus_U1_whV_is_polyklogd_opt})}
\end{align*}

Notice that $\wh{U}B_2 =  A D_2 (B_2 D_2)^\dagger  (D_1 B_1)^\dagger D_1 A  $.
Setting
\begin{align*}
C= AD_2 \in \mathbb{R}^{n\times d_2}, U= (B_2D_2)^\dagger  (D_1 B_1)^\dagger \in \mathbb{R}^{d_2 \times d_1}, \text{~and~} R=D_1 A \in \mathbb{R}^{d_1 \times d},
\end{align*} we get the desired CUR decomposition,
\begin{align*}
\| \underbrace{AD_2}_{C}\cdot \underbrace{ (B_2 D_2)^\dagger (D_1 B_1)^\dagger}_{U} \cdot \underbrace{D_1 A}_{R} - A \|_1 \leq \poly(k) \log(d) \OPT.
\end{align*}
with $\rank(CUR)=k$. Overall, the running time is $O(\nnz(A))+ (n+d)\poly(k)$.

\end{proof}

\subsection{Rank-$r$ matrix $B$}\label{sec:rankr_B_algorithm}

\subsubsection{Properties}

\begin{lemma}\label{lem:solution_to_B_is_solution_to_A}
Given matrix $A\in \mathbb{R}^{n\times d}$, let $\OPT = \underset{\rank-k~A_k}{\min} \| A - A_k \|_1 $. For any $r\geq k$, if rank-$r$ matrix $B\in \mathbb{R}^{n\times d}$ is an $f$-approximation to $A$, i.e.,
\begin{equation*}
\| B - A \|_1 \leq f \cdot \OPT,
\end{equation*}
and $U\in \mathbb{R}^{n\times k}, V\in \mathbb{R}^{k\times d}$ is a $g$-approximation to $B$, i.e.,
\begin{align*}
\| UV - B \|_1 \leq g \cdot \underset{\rank-k~B_k}{\min} \| B_k - B \|_1,
\end{align*}
then,
\begin{equation*}
 \| U V - A \|_1 \lesssim gf \cdot \OPT.
\end{equation*}
\end{lemma}
\begin{proof}
We define $\wt{U}\in\mathbb{R}^{n\times k},\wt{V}\in\mathbb{R}^{k\times d}$ to be two matrices, such that
\begin{align*}
\| \wt{U} \wt{V} - B \|_1 \leq g \min_{\rank-k~B_k} \| B_k -B \|_1,
\end{align*}
and also define,
\begin{align*}
\wh{U}, \wh{V} = \underset{U\in \mathbb{R}^{n\times k},V\in\mathbb{R}^{k\times d}}{\arg\min} \| U V - B \|_1 \text{~and~} U^*, V^* = \underset{U\in \mathbb{R}^{n\times k},V\in\mathbb{R}^{k\times d}}{\arg\min} \| U V - A \|_1,
\end{align*}
Then,
\begin{align*}
\| \wt{U}\wt{V} -A  \|_1  \leq ~ &   \| \wt{U}\wt{V} - B \|_1 + \| B -A \|_1 &\text{~by~triangle~inequality} \\
\leq ~ & g \| \wh{U} \wh{V} - B \|_1 + \| B -A \|_1 &\text{~by~definition} \\
\leq ~& g \| U^* V^* - B \|_1 + \| B - A\|_1 & \text{~by~} \| \wh{U} \wh{V} - B \|_1 \leq \| U^* V^* - B\|_1 \\
\leq ~& g \| U^* V^* - A \|_1 + g\| B- A \|_1 + \| B - A\|_1 & \text{~by~triangle~inequality}\\
 = ~ & g\OPT + (g+1)\| B- A\|_1 & \text{~by~definition~of~}\OPT\\
 \leq ~ & g\OPT + (g+1) f\cdot \OPT & \text{~by~$B$~is~$f$-approximation~to~$A$} \\
\lesssim ~& gf \OPT.
\end{align*}
This completes the proof.
\end{proof}

\begin{lemma}\label{lem:rankr_B_l1_no_contraction}
Given a matrix $B \in \mathbb{R}^{n\times d}$ with rank $r$, for any $1 \leq k< r$, for any fixed $U^* \in \mathbb{R}^{n\times k}$, choose a Cauchy matrix $S$ with $m = O(r\log r)$ rows and rescaled by $\Theta(1/m)$. With probability $.999$ for all $ V \in \mathbb{R}^{k\times d}$, we have
\begin{align*}
\| SU^* V - SB\|_1 \geq \| U^* V - B\|_1.
\end{align*}
\end{lemma}
\begin{proof}
This follows by definitions in Section~\ref{sec:l1} and Lemma~\ref{lem:dense_cauchy_l1_k_subspace_general}.
\end{proof}

\begin{lemma}\label{lem:rankr_B_l1_no_dilation}
Given a matrix $B \in \mathbb{R}^{n\times d}$ with rank $r$, for any $1 \leq k< r$, for any fixed $U^* \in \mathbb{R}^{n\times k}, V^* \in \mathbb{R}^{k\times d}$, choose a Cauchy matrix $S$ with $m$ rows and rescaled by $\Theta(1/m)$. We have
\begin{align*}
\| SU^* V^* - SB\|_1 \leq O(r\log r) \| U^* V^* - B\|_1,
\end{align*}
with probability $.999$.
\end{lemma}

%%% David said, there is a tiny bug in SW11, I shouldn't copy the proof in SW11 here directly.
%%% I should follow the proof of ``fast cauchy'' SODA paper
\begin{proof}
Let $U$ denote a well-conditioned basis of $[ U^* V^*, B]$, then $U$ has $\wt{d} = O(r)$ columns. We have
\begin{align*}
\| S U \|_1 & = ~ \sum_{i=1}^{\wt{d}}\sum_{j=1}^m | (SU_i)_j | \\
& = ~ \sum_{i=1}^{\wt{d}} \sum_{j=1}^m | \frac{1}{m} \sum_{l=1}^n S_{j,l} U_{l,i} | & \text{~by~}S_{j,l}\sim C(0,1)\\
& = ~  \frac{1}{m} \sum_{i=1}^{\wt{d}}  \sum_{j=1}^m |c_{i,j}| & \text{~by~} c_{i,j} \sim C(0, \| U_i\|_1) \\
& = ~   \frac{1}{m} \sum_{i=1}^{\wt{d}} \sum_{j=1}^m \| U_i\|_1 \cdot w_{i+(j-1)d},  & \text{~by~}w_{i,j} \sim | C(0, 1) |
\end{align*}
where the last step follows since each $w_i$ can be thought of as a clipped half-Cauchy random variable. Define $d'=m\wt{d}$. Define
event $\xi_i$ to be the situation when $w_i< D$ (we will choose $D$ later), and define event
$\xi = \xi_1 \cap \xi_2 \cap \cdots \cap \xi_{d'}$. Using a similar proof as Lemma \ref{lem:dense_cauchy_l1_no_dilation},
which is also similar to previous work \cite{i06,sw11,cdmmmw13}, we obtain that
\begin{align*}
\Pr \biggl[ \sum_{i=1}^m \| SU_i \|_1 \geq \sum_{i=1}^m \|U_i\|_1 t \biggr] \lesssim \frac{\log d'}{t} + \frac{d'}{ D }.
\end{align*}
Choosing $t= \Theta (\log d')$ and $D = \Theta( d')$, we have
\begin{align*}
\Pr\biggl[ \sum_{i=1}^m \| SU_i \|_1 \geq \sum_{i=1}^m \|U_i\|_1 O(\log d') \biggr] \leq \frac{1}{C},
\end{align*}
for a constant $C$.
%Taking the union bound over all $i\in [d']$, we have with probability $1-\frac{1}{C}$ such that,
%for all $i\in [d']$,
%\begin{equation*}
%\| SU_i\|_1 \lesssim C \log d' \| U_i\|_1
%\end{equation*}
Condition on the above event.
Let $y = U x$, for some $x\in \mathbb{R}^{d}$. Then for any $y$, 
\begin{align*}
\| S y \|_1 & = ~ \| S U x \|_1 \\
& \leq ~ \sum_{j=1}^{d'} \| S U_j x_j \|_1 &\text{~by~triangle~inequality}\\
& = ~ \sum_{j=1}^{d'}  |x_j| \cdot \| S U_j \|_1 \\
& \lesssim \|x\|_{\infty} \log(d') \sum_{j=1}^{d'} \| U_j \|_1 \\
& \lesssim r\log r \| y \|_1,
\end{align*}
where the last step follows by $\sum_{j=1}^{d'}\|U_j\|_1 \leq d'$ and $\|x\|_{\infty} \leq \| U x\|_1= \|y\|_1$. Choosing $C=1000$ completes the proof.
\end{proof}

\begin{lemma}\label{lem:rankr_B_l1_no_dilation_general}
  Given a matrix $M \in \mathbb{R}^{n\times d}$ with rank $O(r)$, choose a random matrix $S\in\mathbb{R}^{m\times n}$ with each entry drawn from
  a standard Cauchy distribution and scaled by $\Theta(1/m)$. We have that 
\begin{align*}
\| SM\|_1 \leq O(r\log r) \| M\|_1,
\end{align*}
holds with probability $.999$.
\end{lemma}
\begin{proof}
  Let $U\in\mathbb{R}^{O(r)}$ be the well-conditioned basis of $M$. Then each column of $M$ can be expressed by $Ux$ for some $x$.
  We then follow the same proof as that of Lemma~\ref{lem:rankr_B_l1_no_dilation}.
\end{proof}

\subsubsection{$\poly(k,r)$-approximation for rank-$r$ matrix $B$}

\begin{algorithm}[h]\caption{$\poly(r,k)$-approximation Algorithm for Rank-$r$ matrix $B$}
\begin{algorithmic}[1]
\Procedure{\textsc{L1LowRankApproxB}}{$U_B,V_B,n,d,k,r$} \Comment{Theorem \ref{thm:rank_r_approx_polyr_B}}
\State Set $s\leftarrow \wt{O}(r), r'\leftarrow \wt{O}(r)$, $t_1\leftarrow \wt{O}(r)$, $t_2\leftarrow \wt{O}(r)$.
\State Choose dense Cauchy matrices $S\in \mathbb{R}^{s\times n}$, $R\in \mathbb{R}^{d\times r'}$, $T_1\in \mathbb{R}^{t_1\times n}$, $T_2\in \mathbb{R}^{d\times t_2}$.
\State Compute $S\cdot U_B \cdot V_B$, $ U_B \cdot V_B \cdot R$ and $T_1\cdot U_B \cdot V_B \cdot T_2$.
\State Compute $XY = \arg\min_{X,Y}\| T_1 U_B V_B R XY S U_B V_B T_2 - T_1 U_B V_B T_2 \|_F$.
\State \Return $U_B V_B RX,YS U_B V_B$.
\EndProcedure
\end{algorithmic}
\end{algorithm}

\begin{theorem}\label{thm:rank_r_approx_polyr_B}
Given a factorization of a $\rank$-$r$ matrix $B=U_BV_B\in \mathbb{R}^{n\times d}$, where $U_B\in \mathbb{R}^{n\times r}, V_B\in \mathbb{R}^{r\times d}$, for any $1\leq k\leq r$ there exists an algorithm which takes $ (n+d) \cdot \poly(k)$ time to output two matrices $U\in \mathbb{R}^{n\times k} , V\in \mathbb{R}^{k\times d}$ such that
\begin{equation*}
\| UV - B \|_1 \leq \poly(r) \underset{\rank-k~ B_k}{\min}\| B_k - B\|_1,
\end{equation*}
holds with probability $9/10$.
\end{theorem}
\begin{proof}
We define
\begin{align*}
\OPT = \underset{\rank-k~ B_k}{\min}\| B_k - B\|_1.
\end{align*}
Choose $S\in \mathbb{R}^{s\times n}$ to be a dense Cauchy transform matrix with $s=O(r\log r)$. %Using Part (\RN{1}) of Corollary~\ref{cor:three_existence_results}, we have
Using Lemma~\ref{lem:rankr_B_l1_no_dilation_general}, Lemma~\ref{lem:dense_cauchy_l1_k_subspace_general}, and combining with Equation~(\ref{eq:existance_result}), we have
\begin{align*}
\underset{U\in \mathbb{R}^{n\times k}, Z\in \mathbb{R}^{k\times s}}{\min} \| UZSB - B\|_1 \leq \sqrt{s}O(r\log r) \OPT=O(r^{1.5}\log^{1.5} r)\OPT.
\end{align*}
Let $\alpha_s=O(r^{1.5}\log^{1.5} r)$.

We define $U^*,Z^* = \underset{U\in\mathbb{R}^{n\times k},Z\in\mathbb{R}^{k\times d}}{\arg\min} \| U Z S B - B\|_1$. For the fixed $Z^*\in \mathbb{R}^{k\times s}$, choose a dense Cauchy transform matrix $R\in \mathbb{R}^{d\times r'}$ with $r'=O(r\log r)$ and sketch on the right of $ (U Z SB - B )$. We obtain the minimization problem, $\underset{U\in \mathbb{R}^{n\times k}}{\min} \| UZ^* SBR - BR \|_1$.

Define $\widehat{U}^j =  B^j R ( (Z^*SB) R)^\dagger \in \mathbb{R}^k, \forall j\in [n]$. Then  $\widehat{U} = B R ( (Z^*SB) R)^\dagger \in \mathbb{R}^{n\times k}$. Due to Claim \ref{cla:ell2_relax_ell1_regression},
\begin{align*}
\sum_{j=1}^n  \| B^j R ( (Z^*SB) R )^\dagger Z^*SBR - B^j R\|_1 \leq O(\sqrt{r'}) \sum_{j=1}^n \min_{ U^j \in \mathbb{R}^k } \| U^j Z^*S B R - B^jR\|_1,
\end{align*}
which is equivalent to
\begin{align*}
 \| B  R ( (Z^*S B) R )^\dagger Z^*S B R - B R \|_1 &\leq O(\sqrt{r'})  \min_{ U\in \mathbb{R}^{n\times k} }\| U Z^*S B R - B R\|_1,
\end{align*}
where $BR$ is an $n \times r'$ matrix and $SB$ is an $s \times d$ matrix. Both of them can be computed in $(n+d)\poly(r)$ time.

Using Lemma \ref{lem:general_sketch_SUV}, Lemma \ref{lem:dense_cauchy_l1_k_subspace_general}, Lemma \ref{lem:rankr_B_l1_no_dilation_general}, we obtain,
\begin{align*}
 \| B  R ( (Z^*SB) R )^\dagger Z^*S B  - B \|_1 &\leq O(\sqrt{r'} \alpha_{r'})  \min_{ U\in \mathbb{R}^{n\times k} }\| U Z^*S B  - B \|_1
\end{align*}
where $\alpha_{r'} = \wt{O}(r)$.

We define $X^*\in \mathbb{R}^{r\times k}$, $Y^*\in \mathbb{R}^{k\times s}$,
\begin{equation*}
X^*, Y^* = \underset{X \in \mathbb{R}^{r' \times k}, Y \in \mathbb{R}^{k\times s} }{\arg\min} \| BR X Y SB -B\|_1.
\end{equation*}
Then,
\begin{align*}
\| B RX^* Y^*SB - B \|_1 & \leq \| B  R ( (Z^*SB) R )^\dagger Z^*SB  - B \|_1 \\
& \leq O(\sqrt{r'} \alpha_{r'})  \min_{ U\in \mathbb{R}^{n\times k} }\| U Z^*SB  - B \|_1 \\
& =  O(\sqrt{r'} \alpha_{r'})  \min_{ U\in \mathbb{R}^{n\times k}, Z\in \mathbb{R}^{k\times s} }\| U ZSB  - B \|_1 \\
& \leq O(\sqrt{r'}\alpha_{r'} \alpha_s)\OPT.
\end{align*}
It means that $B RX$, $YSB$ gives an $O(\alpha_{r'} \alpha_s \sqrt{r'})$-approximation to the original problem.

Thus it suffices to use Lemma \ref{lem:solve_BRXYSB} to solve
\begin{equation*}
 \underset{X \in \mathbb{R}^{r \times k}, Y \in \mathbb{R}^{k\times s} }{\min} \| B R X Y S B -B\|_1,
\end{equation*}
by losing an extra $\poly(r)$ approximation ratio. Therefore, we finish the proof.
\end{proof}

\begin{lemma}\label{lem:solve_BRXYSB}
Suppose we are given $S\in \mathbb{R}^{s \times n}$, $R\in \mathbb{R}^{d\times r'}$, and a factorization of a $\rank$-$r$ matrix $B=U_BV_B\in \mathbb{R}^{n\times d}$, where $U_B\in \mathbb{R}^{n\times r}, V_B\in \mathbb{R}^{r\times d}$. Then for any $1\leq k\leq r$, there exists an algorithm which takes $(n+d)\poly(r,r',s) $ time to output two matrices $X' \in \mathbb{R}^{r' \times k}, Y'\in\mathbb{R}^{k\times s}$ such that
\begin{align*}
\| BR X' \cdot Y' SB - B\|_1 \leq \poly(r) \min_{X \in \mathbb{R}^{r' \times k}, Y\in\mathbb{R}^{k\times s}} \| BR X Y SB  -  B \|_1
\end{align*}
holds with probability at least $.999$.
\end{lemma}
\begin{proof}
Choosing dense Cauchy matrices $T_1 \in \mathbb{R}^{t_1 \times n},T_2^\top \in \mathbb{R}^{t_2 \times d}$ to sketch on both sides, we get the problem
\begin{equation}\label{eq:TBRXYSBT_minus_TBT_ell1}
\min_{X \in \mathbb{R}^{r'\times k}, Y \in \mathbb{R}^{k\times s} } \| T_1 BR X Y SB T_2 - T_1 B T_2\|_1,
\end{equation}
where $t_1=\wt{O}(r)$ and $t_2=\wt{O}(r)$.

Define $X',Y'$ to be the optimal solution of
\begin{equation*}
\min_{X \in \mathbb{R}^{r' \times k}, Y\in\mathbb{R}^{k\times s}} \| T_1 BR X Y SB T_2 - T_1 B T_2\|_F.
\end{equation*}
Define $\wt{X}, \wt{Y}$ to the be the optimal solution of
\begin{equation*}
\min_{X \in \mathbb{R}^{r' \times k}, Y\in\mathbb{R}^{k\times s}} \| BR X Y SB -B\|_1.
\end{equation*}

By Claim~\ref{cla:frobenius_relax_ell1_lowrank},
\begin{align*}
\| T_1 BR X' Y' SB T_2 - T_1 B T_2\|_1\leq \sqrt{t_1t_2}\min_{X \in \mathbb{R}^{r' \times k}, Y\in\mathbb{R}^{k\times s}} \| T_1 BR X Y SB T_2 - T_1 B T_2\|_1.
\end{align*}

By Lemma~\ref{lem:general_sketch_T1BXYCT2}, Lemma~\ref{lem:rankr_B_l1_no_dilation_general} and Lemma~\ref{lem:dense_cauchy_l1_k_subspace_general}, we have
\begin{align*}
\| BR X' \cdot Y' S B - B \|_1\leq\sqrt{t_1 t_2} \cdot \wt{O}(r^2) \|  BR\wt{X} \cdot \wt{Y} S B - B \|_1.
\end{align*}

\iffalse
Then using Lemma~\ref{lem:general_sketch_T1BXYCT2}, Lemma~\ref{lem:rankr_B_l1_no_dilation}, Lemma~\ref{lem:rankr_B_l1_no_contraction} and Claim~\ref{cla:frobenius_relax_ell1_lowrank}, we have
\begin{align*}
\| BR X' \cdot Y' S B - B \|_1 & \lesssim ~ \| T_1 BR X' \cdot Y' S B - T_1 B  \|_1 &\text{~by~no~contraction} \\
& \lesssim ~ \| T_1 BR X' \cdot Y' S B T_2 - T_1 B T_2 \|_1 & \text{~by~no~contraction} \\
& \leq ~ \sqrt{t_1 t_2} \| T_1 BR X' \cdot Y' S B T_2 - T_1 B T_2 \|_F \\
& \leq ~ \sqrt{t_1 t_2} \| T_1 BR \wt{X} \cdot \wt{Y} S B T_2 - T_1 B T_2 \|_F \\
& \leq ~ \sqrt{t_1 t_2} \| T_1 BR \wt{X} \cdot \wt{Y} S B T_2 - T_1 B T_2 \|_1 \\
& \lesssim ~ \sqrt{t_1 t_2} \cdot \wt{O}(r) \| T_1 BR\wt{X} \cdot \wt{Y} S B - T_1 B \|_1 & \text{~by~no~dilation}\\
& \lesssim ~ \sqrt{t_1 t_2} \cdot \wt{O}(r^2) \|  BR\wt{X} \cdot \wt{Y} S B - B \|_1 & \text{~by~no~dilation}\\
& = ~ \poly(r) \min_{X \in \mathbb{R}^{r' \times k}, Y\in\mathbb{R}^{k\times s}} \| BR X Y SB -B\|_1
\end{align*}
\fi
This completes the proof.
\end{proof}

\section{Contraction and Dilation Bound for $\ell_1$} \label{sec:l1}
This section presents the essential lemmas for $\ell_1$-low rank approximation. Section~\ref{sec:definitions_of_con_dil} gives some basic definitions. Section~\ref{sec:properties_of_con_dil} shows some properties implied by contraction and dilation bounds. Section~\ref{sec:dense_cauchy_l1_no_dilation} presents the no dilation lemma for a dense Cauchy transform. Section~\ref{sec:dense_cauchy_l1_no_contraction} and \ref{sec:dense_cauchy_l1_k_subspace} presents the no contraction lemma for dense Cauchy transforms. Section~\ref{sec:sparse_cauchy_l1} and \ref{sec:lewis_weights_l1} contains the results for sparse Cauchy transforms and Lewis weights.

\subsection{Definitions}\label{sec:definitions_of_con_dil}

\begin{definition}\label{def:no_dilation_of_a_sketch}
Given a matrix $M\in\mathbb{R}^{n\times d}$, if matrix $S\in\mathbb{R}^{m\times n}$ satisfies
\begin{align*}
\|SM\|_1\leq c_1\|M\|_1,
\end{align*}
then $S$ has at most $c_1$-dilation on $M$.
\end{definition}

\begin{definition}\label{def:no_contraction_for_general_vectors}
Given a matrix $U\in\mathbb{R}^{n\times k}$, if matrix $S\in\mathbb{R}^{m\times n}$ satisfies
\begin{align*}
\forall x\in\mathbb{R}^{k}, \|SUx\|_1\geq \frac1{c_2}\|Ux\|_1,
\end{align*}
then $S$ has at most $c_2$-contraction on $U$.
\end{definition}

\begin{definition}\label{def:no_contraction_of_a_sketch}
Given matrices $U\in\mathbb{R}^{n\times k},A\in\mathbb{R}^{n\times d}$, let $V^*=\arg \min_{V\in\mathbb{R}^{k\times d}}\|UV-A\|_1$. If matrix $S\in\mathbb{R}^{m\times n}$ satisfies
\begin{align*}
\forall V\in\mathbb{R}^{k\times d}, \|SUV-SA\|_1\geq \frac{1}{c_3}\|UV-A\|_1-c_4\|UV^*-A\|_1,
\end{align*}
then $S$ has at most $(c_3,c_4)$-contraction on $(U,A)$.
\end{definition}

\begin{definition}\label{def:l1_subspace_embedding}
A $(c_5,c_6)$ $\ell_1$-subspace embedding for the column space of an $n\times k$ matrix $U$ is a matrix $S\in\mathbb{R}^{m\times n}$ for which all $x\in\mathbb{R}^k$
\begin{align*}
\frac{1}{c_5}\|Ux\|_1\leq\|SUx\|_1\leq c_6\|Ux\|_1.
\end{align*}
\end{definition}

\begin{definition}\label{def:sketch_for_l1_multiple_regression}
Given matrices $U\in\mathbb{R}^{n\times k},A\in\mathbb{R}^{n\times d}$, let $V^*=\arg\min_{V\in\mathbb{R}^{k\times d}}\|UV-A\|_1$. Let $S\in\mathbb{R}^{m\times n}.$ If for all $c\geq 1$, and if for any $\wh{V}\in\mathbb{R}^{k\times d}$ which satisfies
\begin{align*}
\|SU\wh{V}-SA\|_1\leq c\cdot \min_{V\in\mathbb{R}^{k\times d}}\|SUV-SA\|_1,
\end{align*}
it holds that
\begin{align*}
\|U\wh{V}-A\|_1\leq c\cdot c_7\cdot\|UV^*-A\|_1,
\end{align*}
then $S$ provides a $c_7$-multiple-regression-cost preserving sketch of $(U,A)$.
\end{definition}

\begin{definition}
Given matrices $L\in\mathbb{R}^{n\times m_1},N\in\mathbb{R}^{m_2\times d},A\in\mathbb{R}^{n\times d},k\geq 1$, let 
\begin{align*}
X^*=\arg\min_{\rank-k\ X}\|LXN-A\|_1.
\end{align*}
Let $S\in\mathbb{R}^{m\times n}.$ If for all $c\geq 1$, and if for any $\rank-k\ \wh{X}\in\mathbb{R}^{m_1\times m_2}$ which satisfies
\begin{align*}
\|SL\wh{X}N-SA\|_1\leq c\cdot \min_{\rank-k\ X}\|SLXN-SA\|_1,
\end{align*}
it holds that
\begin{align*}
\|L\wh{X}N-A\|_1\leq c\cdot c_8\cdot\|LX^*N-A\|_1,
\end{align*}
then $S$ provides a $c_8$-restricted-multiple-regression-cost preserving sketch of $(L,N,A,k)$.
\end{definition}

\subsection{Properties}\label{sec:properties_of_con_dil}

\begin{lemma}\label{lem:no_dialation_and_no_contraction_imply_general_no_contraction}
Given matrices $A\in\mathbb{R}^{n\times d},U\in\mathbb{R}^{n\times k}$, let $V^*=\arg \min_{V\in\mathbb{R}^{k\times d}}\|UV-A\|_1$. If $S\in\mathbb{R}^{m\times n}$ has at most $c_1$-dilation on $UV^*-A$, i.e.,
\begin{align*}
\|S(UV^*-A)\|_1\leq c_1\|UV^*-A\|_1,
\end{align*}
and it has at most $c_2$-contraction on $U$, i.e.,
\begin{align*}
\forall x\in\mathbb{R}^{k}, \|SUx\|_1\geq \frac{1}{c_2}\|Ux\|_1,
\end{align*}
 then $S$ has at most $(c_2,c_1+\frac{1}{c_2})$-contraction on $(U,A)$, i.e.,
\begin{align*}
\forall V\in\mathbb{R}^{k\times d}, \|SUV-SA\|_1\geq \frac{1}{c_2}\|UV-A\|_1-(c_1+\frac{1}{c_2})\|UV^*-A\|_1,
\end{align*}
\end{lemma}
\begin{proof}
Let $A\in\mathbb{R}^{n\times d},U\in\mathbb{R}^{n\times k},S\in\mathbb{R}^{m\times n}$ be the same as that described in the lemma. Then $\forall V\in\mathbb{R}^{k\times d}$
\begin{align*}
\|SUV-SA\|_1&\geq \|SUV-SUV^*\|_1-\|SUV^*-SA\|_1\\
&\geq \|SUV-SUV^*\|_1-c_1\|UV^*-A\|_1\\
&= \|SU(V-V^*)\|_1-c_1\|UV^*-A\|_1\\
&= \sum_{j=1}^d\|SU(V-V^*)_j\|_1-c_1\|UV^*-A\|_1\\
&\geq \sum_{j=1}^d \frac{1}{c_2} \|U(V-V^*)_j\|_1-c_1\|UV^*-A\|_1\\
&=\frac{1}{c_2}\|UV-UV^*\|_1-c_1\|UV^*-A\|_1\\
&\geq \frac{1}{c_2}\|UV-A\|_1-\frac{1}{c_2}\|UV^*-A\|_1-c_1\|UV^*-A\|_1\\
&= \frac{1}{c_2}\|UV-A\|_1-\left((\frac{1}{c_2}+c_1)\|UV^*-A\|_1\right).
\end{align*}
The first inequality follows by the triangle inequality. The second inequality follows since $S$ has at most $c_1$ dilation on $UV^*-A$. The third inequality follows since $S$ has at most $c_2$ contraction on $U$. The fourth inequality follows by the triangle inequality.
\end{proof}

\begin{lemma} \label{lem:general_sketch_SUV}
Given matrices $A\in\mathbb{R}^{n\times d},U\in\mathbb{R}^{n\times k}$, let $V^*=\arg \min_{V\in\mathbb{R}^{k\times d}}\|UV-A\|_1$. If $S\in\mathbb{R}^{m\times n}$ has at most $c_1$-dilation on $UV^*-A$, i.e.,
 \begin{align*}
 \|S(UV^*-A)\|_1\leq c_1\|UV^*-A\|_1,
 \end{align*}
 and has at most $c_2$-contraction on $U$, i.e., %$(c_3,c_4)$-contraction on $(U,A)$, i.e.
 \begin{align*}
 \forall x\in\mathbb{R}^{k}, \|SUx\|_1\geq \frac{1}{c_2}\|Ux\|_1,
 \end{align*}
 then $S$ provides a $(2c_1c_2+1)$-multiple-regression-cost preserving sketch of $(U,A)$, i.e.,
 for all $c\geq 1$, for any $\wh{V}\in\mathbb{R}^{k\times d}$ which satisfies
\begin{align*}
\|SU\wh{V}-SA\|_1\leq c\cdot \min_{V\in\mathbb{R}^{k\times d}}\|SUV-SA\|_1,
\end{align*}
it has
\begin{align*}
\|U\wh{V}-A\|_1\leq c\cdot (2c_1c_2+1)\cdot\|UV^*-A\|_1,
\end{align*}
\end{lemma}
\begin{proof}
Let $S\in\mathbb{R}^{m\times n},A\in\mathbb{R}^{n\times d},U\in\mathbb{R}^{n\times k},V^*,\wh{V}\in\mathbb{R}^{k\times d},$ and $c$ be the same as stated in the lemma.
\begin{align*}
\|U\wh{V}-A\|_1&\leq c_2\|SU\wh{V}-SA\|_1+(1+c_1c_2)\|UV^*-A\|_1\\
&\leq c_2 c\min_{V\in\mathbb{R}^{k\times d}}\|SUV-SA\|_1+(1+c_1c_2)\|UV^*-A\|_1\\
&\leq c_2 c\|SUV^*-SA\|_1+(1+c_1c_2)\|UV^*-A\|_1\\
&\leq c_1c_2c\|UV^*-A\|_1+(1+c_1c_2)\|UV^*-A\|_1\\
&\leq c\cdot (1+2c_1c_2)\|UV^*-A\|_1.
\end{align*}
The first inequality follows by Lemma~\ref{lem:no_dialation_and_no_contraction_imply_general_no_contraction}. The second inequality follows by the guarantee of $\wh{V}$. The fourth inequality follows since $S$ has at most $c_1$-dilation on $UV^*-A$. The fifth inequality follows since $c\geq 1$.
\end{proof}

\begin{lemma}\label{lem:general_no_contraction_and_no_dialation_imply_restricted_regression_sketch}
Given matrices $L\in\mathbb{R}^{n\times m_1},N\in\mathbb{R}^{m_2\times d},A\in\mathbb{R}^{n\times d},k\geq 1$, let 
\begin{align*}
X^*=\arg\min_{\rank-k\ X}\|LXN-A\|_1.
\end{align*}
If $S\in\mathbb{R}^{m\times n}$ has at most $c_1$-dilation on $LX^*N-A$, i.e.,
\begin{align*}
\|S(LX^*N-A)\|_1\leq c_1\|LX^*N-A\|_1,
\end{align*}
and has at most $c_2$-contraction on $L$, i.e.,
\begin{align*}
\forall x\in\mathbb{R}^{m_1} \|SLx\|_1\geq \|Lx\|_1,
\end{align*}
then $S$ provides a $(2c_1c_2+1)$-restricted-multiple-regression-cost preserving sketch of $(L,N,A,k)$, i.e., for all $c\geq 1$, for any $\rank-k\ \wh{X}\in\mathbb{R}^{m_1\times m_2}$ which satisfies
\begin{align*}
\|SL\wh{X}N-SA\|_1\leq c\cdot \min_{\rank-k\ X}\|SLXN-SA\|_1,
\end{align*}
it holds that 
\begin{align*}
\|L\wh{X}N-A\|_1\leq c\cdot (2c_1c_2+1)\cdot\|LX^*N-A\|_1.
\end{align*}
\end{lemma}

\begin{proof}
Let $S\in\mathbb{R}^{m\times n},L\in\mathbb{R}^{n\times m_1},\wh{X}\in\mathbb{R}^{m_1\times m_2},X^*\in\mathbb{R}^{m_1\times m_2},N\in\mathbb{R}^{m_2\times d},A\in\mathbb{R}^{n\times d}$, and $c\geq 1$ be the same as stated in the lemma.
\begin{align*}
\|SL\wh{X}N-SA\|_1&\geq\|SL\wh{X}N-SLX^*N\|_1-\|SLX^*N-SA\|_1\\
&\geq \frac{1}{c_2}\|L(\wh{X}N-X^*N)\|_1-c_1\|LX^*N-A\|_1\\
&\geq \frac{1}{c_2}\|L\wh{X}N-A\|_1-\frac{1}{c_2}\|LX^*N-A\|_1-c_1\|LX^*N-A\|_1\\
&= \frac{1}{c_2}\|L\wh{X}N-A\|_1-(\frac{1}{c_2}+c_1)\|LX^*N-A\|_1.\\
\end{align*}
The inequality follows by the triangle inequality. The second inequality follows since $S$ has at most $c_2$-contraction on $L$, and it has at most $c_1$-dilation on $LX^*N-A$. The third inequality follows by the triangle inequality.

It follows that
\begin{align*}
\|L\wh{X}N-A\|_1&\leq c_2\|SL\wh{X}N-SA\|_1+(1+c_1c_2)\|LX^*N-A\|_1\\
&\leq c_2c\cdot \min_{\rank-k\ X}\|SLXN-SA\|_1+(1+c_1c_2)\|LX^*N-A\|_1\\
&\leq c_2c\cdot \|SLX^*N-SA\|_1+(1+c_1c_2)\|LX^*N-A\|_1\\
&\leq cc_1c_2\cdot \|LX^*N-A\|_1+(1+c_1c_2)\|LX^*N-A\|_1\\
&\leq c\cdot (1+2c_1c_2)\|LX^*N-A\|_1.
\end{align*}
The first inequality directly follows from the previous one. The second inequality follows from the guarantee of $\wh{X}$. The fourth inequality follows since $S$ has at most $c_1$ dilation on $LX^*N-A$. The fifth inequality follows since $c\geq 1$.
\end{proof}

\begin{lemma}\label{lem:general_sketch_T1BXYCT2}
Given matrices $L\in\mathbb{R}^{n\times m_1},N\in\mathbb{R}^{m_2\times d},A\in\mathbb{R}^{n\times d},k\geq 1$, let 
\begin{align*}
X^*=\arg\min_{\rank-k\ X}\|LXN-A\|_1.
\end{align*}
Let $T_1\in\mathbb{R}^{t_1\times n}$ have at most $c_1$-dilation on $LX^*N-A$, %i.e.
%\begin{align*}
%\|T_1(LX^*N-A)\|_1\leq c_1\|LX^*N-A\|_1,
%\end{align*}
and at most $c_2$-contraction on $L$. %i.e.
%\begin{align*}
%\forall x\in\mathbb{R}^{m_1}, \|T_1Lx\|_1\geq \frac{1}{c_2}\|Lx\|_1.
%\end{align*}
Let 
\begin{align*}
\wt{X}=\arg\min_{\rank-k\ X}\|T_1LXN-T_1A\|_1.
\end{align*}
Let $T_2^\top\in\mathbb{R}^{t_2\times d}$ have at most $c'_1$-dilation on $(T_1L\wt{X}N-T_1A)^\top$, %i.e.
%\begin{align*}
%\|T_1(L\wt{X}N-A)T_2\|_1\leq c'_1\|T_1(L\wt{X}N-A)\|_1,
%\end{align*}
and at most $c'_2$-contraction on $N^\top$. %i.e.
%\begin{align*}
%\forall x\in\mathbb{R}^{m_2},\|x^\top NT_2\|_1\geq \frac{1}{c'_2}\|x^\top N\|_1.
%\end{align*}
Then, for all $c\geq 1$, for any $\rank-k\ \wh{X}\in\mathbb{R}^{m_1\times m_2}$ which satisfies
\begin{align*}
\|T_1(L\wh{X}N-SA)T_2\|_1\leq c\cdot \min_{\rank-k\ X}\|T_1(LXN-A)T_2\|_1,
\end{align*}
it has
\begin{align*}
\|L\wh{X}N-A\|_1\leq c\cdot (2c_1c_2+1)(2c'_1c'_2+1)\cdot\|LX^*N-A\|_1.
\end{align*}
\end{lemma}

\begin{proof}
Apply Lemma~\ref{lem:general_no_contraction_and_no_dialation_imply_restricted_regression_sketch} for sketch matrix $T_2$. Then for any $c\geq 1$, any $\rank-k\ \wh{X}\in\mathbb{R}^{m_1\times m_2}$ which satisfies
\begin{align*}
\|T_1(L\wh{X}N-A)T_2\|_1\leq c\cdot \min_{\rank-k\ X}\|T_1(LXN-A)T_2\|_1,
\end{align*}
has
\begin{align*}
\|T_1(L\wh{X}N-A)\|_1\leq c\cdot (2c'_1c'_2+1)\cdot\|T_1(L\wt{X}N-A)\|_1.
\end{align*}

Apply Lemma~\ref{lem:general_no_contraction_and_no_dialation_imply_restricted_regression_sketch} for sketch matrix $T_1$. Then for any $c\geq 1$, any $\rank-k\ \wh{X}\in\mathbb{R}^{m_1\times m_2}$ which satisfies
\begin{align*}
\|T_1(L\wh{X}N-A)\|_1\leq c(2c'_1c'_2+1)\cdot \min_{\rank-k\ X}\|T_1(L\wt{X}N-A)\|_1,
\end{align*}
has
\begin{align*}
\|L\wh{X}N-A\|_1\leq c\cdot (2c_1c_2+1)(2c'_1c'_2+1)\cdot\|LX^*N-A\|_1.
\end{align*}
\end{proof}

\begin{lemma}\label{lem:con_dil_summary}
Given matrices $M\in\mathbb{R}^{n\times d},U\in\mathbb{R}^{n\times t}$, $d\geq t=\rank(U), n\geq d\geq r=\rank(M)$, if sketching matrix $S\in\mathbb{R}^{m\times n}$ is drawn from any of the following probability distributions of matrices, with $.99$ probability $S$ has at most $c_1$-dilation on $M$, i.e.,
\begin{align*}
\|SM\|_1\leq c_1 \|M\|_1,
\end{align*}
and $S$ has at most $c_2$-contraction on $U$, i.e.,
\begin{align*}
\forall x\in\mathbb{R}^{t},\ \|SUx\|_1\geq \frac{1}{c_2}\|Ux\|_1,
\end{align*}
where $c_1,\ c_2$ are parameters depend on the distribution over $S$.
\begin{enumerate}
\item[\rm{(\RN{1})}] $S\in\mathbb{R}^{m\times n}$ is a dense Cauchy matrix: a matrix with i.i.d. entries from the standard Cauchy distribution. If $m=O(t\log t)$, then $c_1c_2=O(\log d)$. If $m=O((t+r)\log (t+r))$, then $c_1c_2=O(\min(\log d,r\log r))$.

\item[\rm{(\RN{2})}] $S\in\mathbb{R}^{m\times n}$ is a sparse Cauchy matrix: $S=TD$, where $T\in\mathbb{R}^{m\times n}$ has each column i.i.d. from the uniform distribution on standard basis vectors of $\mathbb{R}^{m}$, and $D\in\mathbb{R}^{n\times n}$ is a diagonal matrix with i.i.d. diagonal entries following a standard Cauchy distribution. If $m=O(t^5\log^5 t)$, then $c_1c_2=O(t^2\log^2 t\log d)$. If $m=O((t+r)^5\log^5(t+r))$, then $c_1c_2=O(\min(t^2\log^2 t\log d,r^3\log^3 r))$.

\item[\rm{(\RN{3})}] $S\in\mathbb{R}^{m\times n}$ is a sampling and rescaling matrix (notation $S\in\mathbb{R}^{n\times n}$ denotes a diagonal sampling and rescaling matrix with $m$ non-zero entries): If $S$ samples and reweights $m=O(t\log t)$ rows of $U$, selecting each with probability proportional to the $i^{\text{th}}$ row's $\ell_1$ Lewis weight and reweighting by the inverse probability, then $c_1c_2=O(1)$.

\item[\rm{(\RN{4})}] $S\in\mathbb{R}^{m\times n}$ is a dense Cauchy matrix with limited independence: $S$ is a matrix with each entry drawn from a standard Cauchy distribution. Entries from different rows of $S$ are fully independent, and entries from the same row of $S$ are $W$-wise independent. If $m=O(t\log t)$, and $W=\wt{O}(d)$, then $c_1c_2=O(t\log d)$. If $m=O(t\log t)$, and $W=\wt{O}(td)$, then $c_1c_2=O(\log d)$.
\end{enumerate}
In the above, if we replace $S$ with $\sigma\cdot S$ where $\sigma\in\mathbb{R}\backslash\{0\}$ is any scalar, then the relation between $m$ and $c_1c_2$ can be preserved.
%\rm{\RN{1}} $S\in\mathbb{R}^{n\times m}$ is a dense Cauchy transform\\
%$m=O(t\log t),c_1=O(\log d),c_2=O(1)$

%\rm{\RN{2}} $S\in\mathbb{R}^{n\times m}$ is a dense Cauchy transform\\
%$m=O((t+r)\log (t+r)),c_1=O(\min(r\log r,\log d)),c_2=O(1)$

%\rm{\RN{3}} $S\in\mathbb{R}^{n\times m}$ is a sparse Cauchy transform\\
%$m=O(t^5\log t^5),c_1=O(\log d),c_2=O(t^2\log^2 t)$

\end{lemma}

For $\rm{(\RN{1})}$, if $m=O(t\log t)$, then $c_1c_2=O(\log d)$ is implied by Lemma~\ref{lem:dense_cauchy_l1_k_subspace_general} and Lemma~\ref{lem:dense_cauchy_l1_no_dilation_general}. If $m=O((t+r)\log (t+r))$, $c_1c_2=O(r\log r)$ is implied by~\cite{sw11}.

For $\rm{(\RN{2})}$, if $m=O(t^5\log^5 t)$, then $c_1c_2=O(t^2\log^2t\log d)$ is implied by Corollary~\ref{cor:sparse_cauchy_k_subspace} and Lemma~\ref{lem:sparse_cauchy_l1_no_dilation_general}. If $m=O((t+r)^5\log^5 (t+r))$, $c_1c_2=O(r^3\log^3 r)$ is implied by~\cite{mm13}.

For $\rm{(\RN{3})}$, it is implied by~\cite{cp15} and Lemma~\ref{lem:lewis_weights_l1_no_dilation_general}.

For $\rm{(\RN{4})}$, it is implied by Lemma~\ref{lem:limited_no_dilation}, Corollary~\ref{cor:limited_no_contraction}.

\subsection{Cauchy embeddings, no dilation}\label{sec:dense_cauchy_l1_no_dilation}
\begin{lemma}\label{lem:dense_cauchy_l1_no_dilation}
Define $U^*\in \mathbb{R}^{n\times k}, V^*\in \mathbb{R}^{k\times d}$ to be the optimal solution of $\underset{{U \in \mathbb{R}^{n\times k}, V\in \mathbb{R}^{k\times d}}}{\min}\| UV -A\|_1$. Choose a Cauchy matrix $S$ with $m$ rows and rescaled by $\Theta(1/m)$. We have that 
\begin{align*}
\| S U^* V^* - S A\|_1  \leq O(\log d) \| U^* V^* - A\|_1
\end{align*}
holds with probability at least $99/100$.
\end{lemma}
\begin{proof}
The proof technique has been used in \cite{i06} and \cite{cdmmmw13}. Fix the optimal $U^*$ and $V^*$, then
\begin{align}
\| S U^* V^* - S A\|_1 = ~ &\sum_{i=1}^d \| S (U^* V_i^* - A_i) \|_1  \notag \\
= ~ &  \sum_{i=1}^d  \sum_{j=1}^m | \sum_{l=1}^n \frac{1}{m} S_{j,l} \cdot  ( U^* V_i^* - A_i )_l |  &\text{~where~}S_{j,l}\sim C(0,1)\notag \\
= ~ &  \sum_{i=1}^d  \sum_{j=1}^m \frac{1}{m} |c_{i,j}|  & \text{~where~} c_{i,j} \sim C(0, \| U^* V_i^* - A_i  \|_1 ) \notag \\
= ~ &  \frac{1}{m} \sum_{i=1}^d \sum_{j=1}^m  \| U^* V_i^* - A_i \|_1 \cdot  w_{i+ d(j-1)}.  & \text{~where~} w_{i+ d(j-1)} \sim | C(0,1) |
\end{align}
where the last step follows since each $w_i$ can be thought of as a clipped half-Cauchy random variable. Define $d'=md$.
Define event $\xi_i$ to be the situation in which $ w_i < D$ (we will decide upon $D$ later), and define event $\xi = \xi_1 \cap \xi_2 \cap \cdots \cap \xi_{d'}$. Then it is clear that $\xi \cap \xi_i = \xi,\forall i\in [d']$.

Using the probability density function (pdf) of a
Cauchy and because $\tan^{-1} x \leq x$, we can lower bound $\Pr[\text{event~}\xi_i~\text{holds}]$ in the following sense,
\begin{align*}
\Pr[\xi_i]= \frac{2}{\pi} \tan^{-1}(D) = 1-\frac{2}{\pi} \tan^{-1} (1/D) \geq 1-\frac{2}{\pi D}.
\end{align*}
By a union bound over all $i\in [d']$, we can lower bound $\Pr[\text{event~}\xi~\text{holds}]$,
\begin{align}\label{eq:lower_bound_pr_xi}
\Pr[\xi] \geq 1- \sum_{i=1}^{d'} \Pr[\overline{\xi}_i] \geq 1- \frac{2d'}{\pi D}.
\end{align}
By Bayes rule and $\xi = \xi \cap \xi_i$, $\Pr[\xi | \xi_i] \Pr[\xi_i] = \Pr[\xi \cap \xi_i] = \Pr[\xi]$, which implies that $\Pr[\xi | \xi_i ] = \Pr[\xi] /\Pr[\xi_i]$. First, we can lower bound $\E [ w_i | \xi_i ]$,
\begin{align*}
\E[w_i | \xi_i ] = ~& \E[ w_i | \xi_i \cap \xi ] \Pr[\xi | \xi_i] + \E[w_i | \xi \cap \overline{\xi} ] \Pr[\overline{\xi}| \xi_i] \\
\geq ~ & \E[ w_i | \xi_i \cap \xi ] \Pr[\xi | \xi_i] & \text{~by~} w_i\geq 0 \text{~and~} \Pr[]\geq 0 \\
= ~ & \E[w_i | \xi] \Pr[\xi | \xi_i] & \text{~by~} \xi=\xi\cap \xi_i.
\end{align*}
The above equation implies that
\begin{align*}
\E[w_i |\xi] \leq ~ & \frac{\E[w_i |\xi_i]}{\Pr[\xi | \xi_i]} \\
= ~ & \frac{\E[w_i |\xi_i] \Pr[\xi_i]}{\Pr[\xi \cap \xi_i]} &  \text{~by~Bayes~rule~} \Pr[\xi | \xi_i] \Pr[\xi_i] = \Pr[\xi \cap \xi_i] \\
= ~ & \frac{\E[w_i |\xi_i] \Pr[\xi_i]}{\Pr[\xi ]} & \text{~by~} \xi = \xi \cap \xi_i.
\end{align*}

Using the pdf of a Cauchy, $\E[w_i | \xi_i] = \frac{1}{\pi} \log(1+D^2)/\Pr[\xi_i]$ and plugging it into the lower bound of $\E[w_i | \xi]$,
\begin{align*}
\E[w_i |\xi] \leq \frac{\E[w_i |\xi_i] \Pr[\xi_i]}{\Pr[\xi ]} = \frac{\frac{1}{\pi} \log(1+D^2) }{\Pr[\xi]} \leq \frac{\frac{1}{\pi} \log(1+D^2)}{1-\frac{2d}{\pi D}} \lesssim \log(D),
\end{align*}
where the third step follows since $\Pr[\xi] \geq 1-\frac{2d'}{\pi D}$ and the last step follows by choosing $D = \Theta(d') $.

We can conclude
\begin{align}\label{eq:E_bound_for_optimal_conditioned_on_xi}
\E[\| SU^* V^* - SA\|_1 | \xi ] = \frac{1}{m} \sum_{i=1}^d \sum_{j=1}^m \| U^* V_i^* - A_i \|_1 \cdot \E[w_{i+d(j-1)} |\xi] \lesssim ( \log d') \cdot \| U^* V^* - A\|_1.
\end{align}
For simplicity, define $X=\| SU^* V^* - SA\|_1 $ and $\gamma = \| U^* V^* - A \|_1$.
By Markov's inequality and because $\Pr[ X \geq \gamma t| \overline{\xi} ]\leq 1$, we have
\begin{align*}
 ~ &  \Pr[ X \geq \gamma t ] \\
= ~ & \Pr[ X \geq \gamma t | \xi] \Pr[\xi] + \Pr[ X \geq \gamma t | \overline{\xi}] \Pr[\overline{\xi}] \\
\leq ~ & \Pr[ X \geq \gamma t | \xi]  + \Pr[\overline{\xi}]  \\
\leq ~ & \frac{\E[X| \xi ] }{ \gamma t }  + \Pr[\overline{\xi}]  & \text{~by~Markov's~inequality} \\
\leq ~ & \frac{\E[X| \xi ] }{ \gamma t }  + \frac{2d'}{\pi D} &\text{~by~Equation~(\ref{eq:lower_bound_pr_xi})} \\
\lesssim ~ & \frac{\log d'}{t}+ \frac{2d'}{\pi D} &\text{~by~Equation~(\ref{eq:E_bound_for_optimal_conditioned_on_xi})}\\
\leq ~ & .01 ,
\end{align*}
where choosing $t= \Theta(\log d')$ and $D = \Theta(d')$. Since $k\leq d$ and $m=\poly(k)$, we have $t=\Theta(\log d)$, which completes the proof.

\end{proof}

\begin{lemma}\label{lem:dense_cauchy_l1_no_dilation_general}
Given any matrix $M\in\mathbb{R}^{n\times d}$, if matrix $S\in\mathbb{R}^{m\times n}$ has each entry drawn from an i.i.d. standard Cauchy distribution and is rescaled by $\Theta(1/m)$, then
\begin{align*}
\| S M\|_1  \leq O(\log d) \| M\|_1
\end{align*}
holds with probability at least $99/100$.
\end{lemma}
\begin{proof}
Just replace the matrix $U^*V^*-A$ in the proof of Lemma~\ref{lem:dense_cauchy_l1_no_dilation} with $M$. Then we can get the result directly.
\end{proof}

\subsection{Cauchy embeddings, no contraction}\label{sec:dense_cauchy_l1_no_contraction}
We prove that if we choose a Cauchy matrix $S$, then for a fixed optimal solution $U^*$ of $\min_{U,V}\| UV -A\|_1$, and for all $V$, we have that with high probability $\| S U^* V- SA \|_1 $ is lower bounded by $\| U^* V -A \|_1$ up to some constant. %We first define $\OPT_i$ to be cost of $\| U^* V_i^* - A_i \|_1$ where $U^?$, $V^?$ is the optimal solution. Define $\OPT= \sum_{i=1}^d \OPT_i$. Let $Z_i$ denote the indicator variable, where $Z_i$ = 1 if $\|SU^* V_i-SA_i\|_1 \geq C_1 \| U^* V_i - A_i \|_1 $ where $C_1$ will be decided later; $Z_i = 0$, otherwise.

\begin{lemma}\label{lem:dense_cauchy_l1_no_contraction}
Define $U^*\in \mathbb{R}^{n\times k},V^*\in\mathbb{R}^{k\times d}$ to be the optimal solution of $\underset{U \in \mathbb{R}^{n\times k} V\in \mathbb{R}^{k\times d}}{\min} \| U V - A \|_1$, where $A\in\mathbb{R}^{n\times d}$. Let $m=O(k\log k),S\in\mathbb{R}^{m\times n}$ be a random matrix with each entry an i.i.d. standard Cauchy random variable, scaled by $\Theta(1/m)$. Then with probability at least $0.95$,
\begin{align*}
\forall V\in\mathbb{R}^{k\times d}, \| U^* V - A\|_1 \lesssim \| S U^* V- SA \|_1 + O(\log (d))\|U^*V^*-A\|_1.
\end{align*}

%If there is a matrix $S\in\mathbb{R}^{m\times n}$ which satisfies \rm{(\RN{1})} $\|SU^*V^*-SA\|_1\leq c\|U^*V^*-A\|_1$, where $c=\Omega(1)$, and \rm{(\RN{2})} $\forall x\in\mathbb{R}^{k},\|SU^*x\|_1\gtrsim\|x\|_1$, then for all $V\in \mathbb{R}^{k\times d}$ we have
%\begin{equation*}
% \| U^* V - A\|_1 \lesssim \| S U^* V- SA \|_1 + O(c)\|U^*V^*-A\|_1.
%\end{equation*}
%If $\|U^*V-A\|_1\geq \omega(c)\|U^*V^*-A\|_1$, we can rewrite it as:
%\begin{equation*}
%\| S U^* V- SA \|_1  \gtrsim \| U^* V - A\|_1.
%\end{equation*}
%Furthermore, if $m=\Omega(k\log k)$, $c=O(\log (md))$, and $S\in\mathbb{R}^{m\times n}$ is a random matrix with each entry drawn from i.i.d. Cauchy distribution $C(0,1/m)$, then \rm{(\RN{1})} and \rm{(\RN{2})} hold with probability at least $.98$.
\end{lemma}

\begin{proof}
This follows by Lemmas \ref{lem:dense_cauchy_l1_no_dilation}, \ref{lem:dense_cauchy_l1_k_subspace_general}, \ref{lem:no_dialation_and_no_contraction_imply_general_no_contraction}.

\end{proof}

%\begin{lemma}\label{lem:dense_cauchy_l1_multiple_regression}
%Define $U^*\in \mathbb{R}^{n\times k},V^*\in\mathbb{R}^{k\times d}$ to be the optimal solution of $\underset{U \in \mathbb{R}^{n\times k} V\in \mathbb{R}^{k\times d}}{\min} \| U V - A \|_1$. Fix a matrix $\wh{V}\mathbb{R}^{n\times k}$, if matrix $S\in\mathbb{R}^{m\times n}$ satisfies the following two properties:

%\rm{(\RN{1})}: $\|S\wh{V}\|$

%\rm{(\RN{2})}:  $ \| x \|_1 \leq O(\log d)$.
%\end{lemma}

%\newpage

%\begin{definition}[$\eps$-contraction, $\eps$-dilation, $\eps$-embedding]
%For a matrix measure $\| \|$, and ${\cal T }\in \mathbb{R}^{n\times d'}$, call matrix $S$ an $\eps$-contraction for
%${\cal T}$ with respect to $\| \|$ if $\| S Y \| \geq (1-\eps) \| Y \|$ for all $Y \in {\cal T}$.
%\end{definition}

%\begin{definition}[well-conditioned subspace]
%For a basis $U \in \mathbb{R}^{n\times k}$, if
%\begin{equation*}
%\forall x \in \mathbb{R}^{k}, \alpha \cdot \| x \|_1 \leq \| U x \|_1 \leq \beta \cdot \| x\|_1
%\end{equation*}
%then we say $U$ is a well-conditioned basis for $\ell_1$.
%\end{definition}
%We choose $\beta = O(\poly(k))$ and $\alpha = O(1/\poly(k))$.
%\begin{lemma}
%There exists some constant $C_2$, with probability $99/100$, for all $ w\in \mathbb{R}^k$, $C_2 \| SU^* w\|_1 \geq \| U^* w\|_1 $.
%\end{lemma}

%Let $U$ denote the well-conditioned basis of $[U^*, A_i] \in \mathbb{R}^{n\times (k+1)}$ such that,
%\begin{equation*}
%\forall x \in \mathbb{R}^{k+1}, \alpha \cdot \| x \|_1 \leq \| U x \|_1 \leq \beta \cdot \| x\|_1
%\end{equation*}
%where $\beta = O(\poly(k))$ and $\alpha = O(1/\poly(k))$.

\subsection{Cauchy embeddings, $k$-dimensional subspace}\label{sec:dense_cauchy_l1_k_subspace}
The goal of this section is to prove Lemma \ref{lem:dense_cauchy_l1_k_subspace_general}.

Before getting into the details, we first give the formal definition of an $(\alpha,\beta)$ $\ell_1$ well-conditioned basis and $\eps$-net.

\begin{definition}[$\ell_1$ Well-conditioned basis]\cite{ddhkm09}
A basis $U$ for the range of $A$ is $(\alpha,\beta)$-conditioned if $\| U \|_1 \leq \alpha$ and for all $x\in \mathbb{R}^k$, $\| x\|_{\infty}\leq \beta \| U x\|_1$. We will say $U$ is well-conditioned if $\alpha$ and $\beta$ are low-degree polynomials in $k$, independent of $n$.
\end{definition}
%$(\alpha,\beta)$-well-conditioned basis : (1) $\| U \|_1 \alpha_1$ and (2) for all $z\in \mathbb{R}^k$, $\| z\|_1 \leq \beta \|U z \|_1$

Note that a well-conditioned basis implies the following result.
\begin{fact}
There exist $\alpha,\beta \geq 1$ such that,
\begin{equation*}
\forall x\in \mathbb{R}^k, \frac{1}{k\beta} \| x \|_1 \leq \| U x\|_1 \leq \alpha \| x\|_1.
\end{equation*}
\end{fact}
\begin{proof}
The lower bound can be proved in the following sense,
\begin{equation*}
\| U x\|_1 \geq \frac{1}{\beta} \| x \|_{\infty} \geq\frac{1}{\beta} \frac{1}{k} \| x\|_1,
\end{equation*}
where the first step follows by the properties of a well-conditioned basis, and the second step follows since $k \| x\|_{\infty} \geq \| x\|_1$.
Then we can show an upper bound,
\begin{equation*}
\| U x\|_1 \leq \| U \|_1 \cdot \|x\|_1 \leq \alpha \| x\|_1,
\end{equation*}
where the first step follows by $\| U x\|_1\leq \| U \|_1 \| x\|_1$, and the second step follows using 
$ \|U \|_1 \leq \alpha $.
\end{proof}

\begin{definition}[$\eps$-net]
Define ${\cal N}$ to an $\eps$-net where, for all the $x\in \mathbb{R}^k$ that $\| x \|_1 = 1$, for any vectors two $x,x'\in {\cal N}$, $\| x - x' \|_1 \geq \eps$, for any vector $y\notin {\cal N}$, there exists an $x\in {\cal N}$ such that $\| x-y \|_1 \leq \eps$. Then the size of ${\cal N}$ is $(\frac{1}{\epsilon})^{O(k)} = 2^{O(k\log(1/\epsilon))}$.
\end{definition}

\begin{lemma}[Lemma 6 in \cite{sw11}]\label{lem:Sy_at_least_y_with_exp_prob}
There is a constant $c_0>0$ such that for any $t\geq 1$ and any constant $c> c_0$, if $S$ is a $t\times n$ matrix whose entries are i.i.d. standard Cauchy random variables scaled by $c/t$, then, for any fixed $y \in \mathbb{R}^n$,
\begin{equation*}
\Pr[ \| S y \|_1 < \| y \|_1 ] \leq 1/2^t
\end{equation*}
\end{lemma}

\begin{lemma}
Suppose we are given a well-conditioned basis $U\in \mathbb{R}^{n\times k}$. If $S$ is a $t\times n$ matrix whose entries are i.i.d. standard Cauchy random variables scaled by $\Theta(1/t)$, then with probability $1-2^{-\Omega(t)}$, for all vectors $x\in {\N}$ we have that $\| SUx\|_1 \geq   \| Ux\|_1$.
\end{lemma}
\begin{proof}
First, using Lemma \ref{lem:Sy_at_least_y_with_exp_prob}, we have for any fixed vector $y\in \mathbb{R}^n$, $\Pr[\| S y\| < \| y \|_1 ] \leq 1/2^t$. Second, we can rewrite $y = Ux$. Then for any fixed $x\in \N$, $\Pr[\| S U x\| < \| U x \|_1 ] \leq 1/2^t$.
Third, choosing $t\gtrsim k\log(1/\epsilon)$ and taking a union bound over all the vectors in the $\epsilon$-net $\N$ completes the proof.
\end{proof}

\begin{lemma}[Lemma 7 in \cite{sw11}]\label{lem:SW_lemma7}
Let $S$ be a $t\times n$ matrix whose entries are i.i.d standard Cauchy random variables, scaled by $c/t$ for a constant $c$, and $t\geq 1$. Then there is a constant $c' = c'(c) > 0$ such that for any fixed set of $\{ y_1, y_2,\cdots, y_k\}$ of $d$ vectors in $\{ y \in \mathbb{R}^n : x\in \mathbb{R}^k, U x = y\}$,
\begin{equation*}
\Pr[ \sum_{i=1}^k \| S y_i \|_1 \geq c' \log(t k) \sum_{i=1}^k \| y_i \|_1 ] \leq \frac{1}{1000}.
\end{equation*}
\end{lemma}

Using Lemma \ref{lem:SW_lemma7} and the definition of an $(\alpha,\beta)$ well-conditioned basis, we can show the following corollary.
\begin{corollary}\label{cor:no_dilation_for_k_dim_subspace}
Suppose we are given an $(\alpha,\beta)$ $\ell_1$ well-conditioned basis $U\in \mathbb{R}^{n\times k}$. Choose $S$ to be an i.i.d. Cauchy matrix with $t = O(k\log k)$ rows, and rescale each entry by $\Theta(1/t)$. Then with probability $99/100$, for all vectors $x\in \mathbb{R}^{k}$,
\begin{equation*}
 \| S Ux\|_1 \leq O(\alpha \beta \log(k)) \cdot \| Ux \|_1,
\end{equation*}
\end{corollary}

\begin{proof}
Define event $E$ to be the situation when
\begin{equation*}
\sum_{i=1}^k \| S U_i \|_1 < c' \log (t k) \sum_{i=1}^k \|  U_i\|_1
\end{equation*}
holds, and $U$ is a well-conditioned basis. Using Lemma \ref{lem:SW_lemma7}, we can show that event $E$ holds with probability $999/1000$. We condition on Event $E$ holding. Then for any $y= U x$ for an $x\in \mathbb{R}^k$, we have
\begin{align*}
\| S y \|_1 ~ & = \| S U x\|_1 \\
~ & \leq \sum_{j=1}^k \| S U_j x_j \|_1 &\text{~by~triangle~inequality}\\
~ & = \sum_{j=1}^k |x_j | \cdot \| SU_j \|_1 \\
~ & \leq \| x \|_{\infty} c' \log(t k) \sum_{j=1}^k \| U_j \|_1 & \text{~by~Lemma~\ref{lem:SW_lemma7}} \\
~ & = \| x \|_{\infty} c' \log(t k) \alpha & \text{~by~} \| U \|_1 =\sum_{j=1}^k \| U_j \|_1 \leq \alpha \\
~ & \leq  \beta \| U x\|_1 c' \log(t k) \alpha & \text{~by~} \| x\|_{\infty} \leq\beta \| Ux \|_1 \\
~ & \leq  \beta \| y \|_1 c' \log(t k) \alpha & \text{~by~} Ux = y.
\end{align*}
This completes the proof.
\end{proof}

  %We set $\eps = 1/\poly(k)$, thus we obtain $|{\cal N}| = 2^{O(k\log k)}$
\begin{lemma}\label{lem:dense_cauchy_l1_k_subspace}
Given an $(\alpha,\beta)$ $\ell_1$ well-conditioned basis, condition on the following two events,

1. For all  $x\in {\N}$, $\| SU x\|_1 \geq  \| Ux\|_1 $. (Lemma~\ref{lem:Sy_at_least_y_with_exp_prob})

2. For all  $x\in \mathbb{R}^k$, $\| SU x\|_1 \leq O(\alpha \beta \log k) \| Ux\|_1 $. (Corollary \ref{cor:no_dilation_for_k_dim_subspace})

Then, for all  $w\in \mathbb{R}^k$, $\| SU w\|_1 \gtrsim \| Uw\|_1 $.
\end{lemma}
\begin{proof}
For any $w\in \mathbb{R}^{k}$ we can write it as $w = \ell \cdot z$ where $\ell$ is some scalar and $z$ has $\| z \|_1=1$. Define $y =  \underset{ y' \in {\cal N}  }{\arg\min} \| y' - z \|_1$.

We first show that if $U$ is an $(\alpha,\beta)$ well-conditioned basis for $\ell_1$, then $\|U(y-z)\|_1 \leq \alpha \beta k \epsilon \| U y \|_1$,
\begin{align*}
 ~& \| U (y-z ) \|_1 \\
 \leq~ &\alpha \| y - z\|_1 &\text{~by~} \| U(y-z)\|_1 \leq \alpha \| y-z\|_1 \\
 \leq~ &\alpha \epsilon & \text{~by~} \| y-z\|_1 \leq \epsilon \\
= ~ &\alpha \epsilon \| y \|_1 & \text{~by~} \| y\|_1 =1 \\
 \leq ~& \alpha \epsilon \| Uy \|_1 \beta k & \text{~by~} \| Uy\|_1 \geq \frac{1}{\beta k} \| y\|_1 .
\end{align*}
Because $\eps < 1/(\alpha \beta k^{c+1})$, we have 
\begin{equation}\label{eq:Uyz_is_bounded_by_Uz}
\| U(y-z)\|_1 \leq \frac{1}{k^c} \| U y\|_1.
\end{equation}
Using the triangle inequality, we can lower bound $\| U y\|_1$ by $\| Uz\|_1$ up to some constant,
\begin{equation}
\| U y\|_1 \geq \| U z \|_1 - \| U(y-z)\|_1 \geq  \| U z \|_1 - \frac{1}{k^c} \|U y\|_1,
\end{equation}
which implies
\begin{equation}\label{eq:Uy_geq_Uz}
\| U y\|_1 \geq .99 \| U z\|_1.
\end{equation}

%we first show that $\| U y\| \gtrsim \|U z\|_1$,
%\begin{align}\label{eq:Uy_geq_Uz}
%\| U y \|_1 ~ \geq & \| U z\|_1 - \| U(y-z) \|_1 &\text{~by~ triangle~ inequality} \notag \\
%~ \geq & \| U z \|_1 - \beta \| y-z\|_1 & \text{~by~} \| U (y-z)\|_1 \leq \beta \| y-z\|_1 \notag \\
%~  \geq & \| U z\|_1 - \beta \epsilon & \text{~by~} \| y -z \|_1 \leq \epsilon \notag \\
%~  = & \| U z\|_1 - \beta \epsilon \| z\|_1 & \text{~by~} \|z\|_1 =1 \\
%~\geq & \| U z\|_1 - \beta \epsilon \| U z\|_1 /\alpha & \text{~by~} \alpha \| z\|_1 \leq \| U z\|_1
%\end{align}
%where the first inequality follows by triangle inequality, and the second inequality follows by $\| U(y-z) \|_1 \leq \beta \|U z\|_1$. Then  we have,
Thus,
\begin{align*}
 ~ & \| S U z \|_1 \\
\geq ~ & \| S U y \|_1 - \| SU (z-y)\|_1 & \text{~by~triangle~inequality} \\
\geq ~ &  \| U y \|_1 -   \| S U (z-y)\|_1 &~\text{~by~Lemma~\ref{lem:Sy_at_least_y_with_exp_prob}} \\
\geq ~ &  \| U y \|_1 - \alpha\beta \log(k) \cdot \| U (z-y)\|_1  &~\text{~by~Corollary~\ref{cor:no_dilation_for_k_dim_subspace}} \\
\geq ~ &  \| U y \|_1 - \alpha\beta \log(k) \cdot \frac{1}{k^c} \| Uy\|_1  &~\text{~by~Equation~(\ref{eq:Uyz_is_bounded_by_Uz})}  \\
\gtrsim ~ &  \| U y \|_1 &\text{~by~} k^c \gtrsim \alpha \beta \log(k) \\
\gtrsim ~ &  \| U z \|_1 , & \text{~by~Equation~(\ref{eq:Uy_geq_Uz})}
\end{align*}
by rescaling $z$ to $w$, we complete the proof.
\end{proof}
%%% some scaling issues.
%%% the last inequality, if \| U*V - A\|_1 \leq O(log d) OPT, it means that V actually gives an O(log d) approximation, which is fine for us.

\begin{lemma}\label{lem:dense_cauchy_l1_k_subspace_general}
Given matrix $U\in\mathbb{R}^{n\times k}$, let $t=O(k\log k)$, and let $S\in\mathbb{R}^{t\times n}$ be a random matrix with entries drawn i.i.d. from a standard Cauchy distribution, where each entry is rescaled by $\Theta(1/t)$. With probability $.99$,
\begin{align*}
\forall x\in\mathbb{R}^{k}, \|SUx\|_1\gtrsim\|Ux\|_1.
\end{align*}
\end{lemma}

\begin{proof}
We can compute a well-conditioned basis for $U$, and denote it $U'$. Then, $\forall x\in\mathbb{R}^{k}$, there exists $y \in \mathbb{R}^{k}$ such that $Ux=U'y$. Due to Lemma~\ref{lem:dense_cauchy_l1_k_subspace}, with probability $.99$, we have
\begin{align*}
\forall y\in\mathbb{R}^k, \|SU'y\|_1\gtrsim\|U'y\|_1.
\end{align*}
\end{proof}

\subsection{Sparse Cauchy transform}\label{sec:sparse_cauchy_l1}
This section presents the proof of two lemmas related to the sparse Cauchy transform. We first prove the no dilation result in Lemma~\ref{lem:sparse_cauchy_l1_no_dilation}. Then we show how to get the no contraction result in Lemma~\ref{lem:sparse_cauchy_l1_no_contraction}.

\begin{lemma}\label{lem:sparse_cauchy_l1_no_dilation}
Given matrix $A\in \mathbb{R}^{n\times d}$, let $U^*,V^*$ be the optimal solutions of $\min_{U,V}\| UV - A \|_1$. Let $\Pi = \sigma \cdot SC\in \mathbb{R}^{m\times n}$, where $S\in \mathbb{R}^{m\times n}$ has each column vector chosen independently and uniformly from the $m$ standard basis vectors of $\mathbb{R}^m$, where $C$ is a diagonal matrix with diagonals chosen independently from the standard Cauchy distribution, and $\sigma$ is a scalar. Then
\begin{align*}
\| \Pi U^* V^* - \Pi A\|_1 \lesssim \sigma \cdot \log (md) \cdot \| U^* V^* - A \|_1
\end{align*}
holds with probability at least $.999$.
\end{lemma}

\begin{proof}
  We define $\Pi=\sigma \cdot SC \in \mathbb{R}^{m\times n}$, as in the statement
  of the lemma. Then, 
  %in Lemma~\ref{lem:sparse_cauchy_l1_no_dilation}.
%Then by the definition of $\Pi$, we have,
\begin{align*}
 & ~\| \Pi (U^*V^* - A ) \|_1 \\
=& ~ \sum_{i=1}^d \| S D (U^*V^*_i - A_i ) \|_1 \\
=& ~ \sum_{i=1}^d \biggl\|
\begin{bmatrix}
S_{11} & S_{12} & \cdots & S_{1n} \\
S_{21} & S_{22} & \cdots & S_{2n} \\
\cdots & \cdots & \cdots & \cdots \\
S_{m1} & S_{m2} & \cdots & S_{mn} \\
\end{bmatrix}
\cdot
\begin{bmatrix}
c_{1} & 0 & 0 & 0 \\
0 & c_{2} & 0 & 0 \\
0 & 0 & \cdots & 0 \\
0 & 0 & 0 & c_{n} \\
\end{bmatrix}
\cdot
(U^*V^*_i - A_i) \biggr\|_1~\\
=&~  \sum_{i=1}^d \biggl\|
\begin{bmatrix}
c_1 S_{11} & c_2 S_{12} & \cdots & c_n S_{1n} \\
c_1 S_{21} & c_2 S_{22} & \cdots & c_n S_{2n} \\
\cdots & \cdots & \cdots & \cdots \\
c_1 S_{m1} & c_2 S_{m2} & \cdots & c_n S_{mn} \\
\end{bmatrix}
\cdot
(U^*V^*_i - A_i) \biggr\|_1~\\
= & ~  \sum_{i=1}^d \biggl\| \sum_{l=1}^n c_l S_{1l} \cdot (U^*V^*_i-A_i)_l,  \sum_{l=1}^n c_l S_{2l} \cdot (U^*V^*_i-A_i)_l, \cdots,  \sum_{l=1}^n c_l S_{ml} \cdot (U^*V^*_i-A_i)_l \biggr\|_1 \\
= & ~ \sum_{i=1}^d \sum_{j=1}^m |\sum_{l=1}^n c_l S_{jl} \cdot (U^*V^*_i-A_i)_l | \\
= & ~ \sum_{i=1}^d \sum_{j=1}^m |\wt{w}_{ij} \cdot \sum_{l=1}^n | S_{jl} (U^* V^*_i - A_i)_l | | \text{~where~}\wt{w}_{ij} \sim C(0,1) \\
= & ~ \sum_{i=1}^d \sum_{j=1}^m  \sum_{l=1}^n | S_{jl} (U^* V^*_i - A_i)_l | \cdot |\wt{w}_{ij}|\\
= & ~ \sum_{i=1}^d \sum_{j=1}^m  \sum_{l=1}^n | S_{jl} (U^* V^*_i - A_i)_l | \cdot w_{i+ (j-1)d} \text{~where~} w_{i+(j-1) d} \sim |C(0,1)|,
\end{align*}
where the last step follows since each $w_i$ can be thought of as a clipped half-Cauchy random variable.
Define $d'= md$. Define event $\xi_i$ to be the situation when $w_i < D$ (we will decide upon $D$ later).
Define event $\xi = \xi_1 \cap \xi_2 \cap \cdots \cap \xi_{d'}$.
By choosing $D = \Theta(d')$, we can conclude that,
\begin{align*}
\E[ \| \Pi U^* V^* - \Pi A \|_1 | \xi ] & = ~ \sum_{i=1}^d \sum_{j=1}^m  \sum_{l=1}^n | S_{jl} (U^* V^*_i - A_i)_l | \cdot \E[w_{i+ (j-1)d} | \xi] \\
& \lesssim ~  \sum_{i=1}^d \sum_{j=1}^m  \sum_{l=1}^n | S_{jl} (U^* V^*_i - A_i)_l | \cdot \log(d') \\
& = ~ \sum_{i=1}^d  \sum_{l=1}^n  |  (U^* V^*_i - A_i)_l  |_1 \cdot \log(d') & \text{~by~} \sum_{j=1}^m |S_{jl}|=1,\forall l\in [n] \\
& = ~ \log(d') \cdot \| U^* V^* - A\|_1.
\end{align*}
Thus, we can show that
\begin{align*}
\Pr[\| \Pi U^*V^* - \Pi A \|_1 \lesssim \log(d') \| U^* V^* - A\|_1 ] \geq 0.999.
\end{align*}

\end{proof}

\begin{lemma}\label{lem:sparse_cauchy_l1_no_dilation_general}
Given any matrix $M\in \mathbb{R}^{n\times d}$. Let $\Pi = \sigma \cdot SC\in \mathbb{R}^{m\times n}$, where $S\in \mathbb{R}^{m\times n}$ has each column vector chosen independently and uniformly from the $m$ standard basis vectors of $\mathbb{R}^m$, where $C$ is a diagonal matrix with diagonals chosen independently from the standard Cauchy distribution, and $\sigma$ is a scalar. Then
\begin{align*}
\| \Pi M\|_1 \lesssim \sigma \cdot \log (md) \cdot \| M \|_1
\end{align*}
holds with probability at least $.999$.
\end{lemma}
\begin{proof}
Just replace the matrix $U^*V^*-A$ in the proof of Lemma~\ref{lem:sparse_cauchy_l1_no_dilation} with $M$. Then we can get the result directly.
\end{proof}

We already provided the proof of Lemma~\ref{lem:sparse_cauchy_l1_no_dilation}. It remains to prove Lemma~\ref{lem:sparse_cauchy_l1_no_contraction}.
\begin{lemma}\label{lem:sparse_cauchy_l1_no_contraction}
%Given matrix $A\in \mathbb{R}^{n\times d}$ with $U^*,V^*$ to be the optimal solution of $\min_{U,V}\| UV - A \|_1$. Let $\Pi = \sigma \cdot SC\in \mathbb{R}^{m\times n}$, where $S\in \mathbb{R}^{m\times n}$ has each column vectors chosen independently and uniformly from the $m$ standard basis vectors of $\mathbb{R}^m$, and where $C$ is a diagonal matrix with diagonals chosen independently from the standard Cauchy distribution. Then with probability at least $.999$, for all $V\in \mathbb{R}^{k\times d}$,
%\begin{align*}
%\| \Pi U^* V - \Pi A\|_1 \geq \| U^* V - A \|_1.
%\end{align*}
%Notice that $m=O(k^5 \log^5 k)$ and $\sigma = O(k^2 \log^2 k)$ according to Theorem 2 in \cite{mm13}.

Given matrix $A\in \mathbb{R}^{n\times d}$, let $U^*,V^*$ be the optimal solution of $\min_{U,V}\| UV - A \|_1$. Let $\Pi = \sigma \cdot SC\in \mathbb{R}^{m\times n}$, where $S\in \mathbb{R}^{m\times n}$ has each column vector chosen independently and uniformly from the $m$ standard basis vectors of $\mathbb{R}^m$, and where $C$ is a diagonal matrix with diagonals chosen independently from the standard Cauchy distribution. Then with probability at least $.999$, for all $V\in \mathbb{R}^{k\times d}$,
\begin{align*}
\| \Pi U^* V - \Pi A\|_1 \geq \| U^* V - A \|_1-O(\sigma\cdot\log(md))\|U^*V^*-A\|_1.
\end{align*}
Notice that $m=O(k^5 \log^5 k)$ and $\sigma = O(k^2 \log^2 k)$ according to Theorem 2 in \cite{mm13}.
\end{lemma}

We start by using Theorem 2 in \cite{mm13} to generate the following Corollary.
\begin{corollary}\label{cor:sparse_cauchy_k_subspace}
Given $U \in \mathbb{R}^{n\times k}$ with full column rank, let $\Pi = \sigma \cdot SC \in \mathbb{R}^{m\times n}$ where $S\in \mathbb{R}^{m\times n}$ has each column chosen independently and uniformly from the $m$ standard basis vectors of $\mathbb{R}^{m}$, and where $C\in \mathbb{R}^{n\times n}$ is a diagonal matrix with diagonals chosen independently from the standard Cauchy distribution. Then with probability $.999$, for all $x\in \mathbb{R}^k$, we have
\begin{align*}
\| \Pi U x \|_1 \geq \| U x \|_1.
\end{align*}
Notice that $m=O(k^5 \log^5 k)$ and $\sigma = O(k^2 \log^2 k)$ according to Theorem 2 in \cite{mm13}.
\end{corollary}

We give the proof of Lemma \ref{lem:sparse_cauchy_l1_no_contraction}.

\begin{proof}
Follows by Corollary~\ref{cor:sparse_cauchy_k_subspace}, Lemma~\ref{lem:sparse_cauchy_l1_no_dilation}, Lemma~\ref{lem:no_dialation_and_no_contraction_imply_general_no_contraction}.
%We define $\Pi=\sigma \cdot SC \in \mathbb{R}^{m\times n}$ as the statement in Lemma~\ref{lem:sparse_cauchy_l1_no_contraction}.
%\begin{align*}
%\| \Pi U^* V - \Pi A\|_1 & \geq ~ \| \Pi U^* V - \Pi U^* V^* \|_1 - \| \Pi A - \Pi U^* V^* \|_1 & \text{~by~triangle~inequality}\\
%& \geq ~\| \Pi U^* V - \Pi U^* V^* \|_1  - c\| A - U^* V^* \|_1 & \text{~by~Lemma~\ref{lem:sparse_cauchy_l1_no_contraction}},c=O(\sigma \log(md)) \\
%& \geq ~\| U^* V -  U^* V^* \|_1  - c\| A - U^* V^* \|_1 & \text{~by~Corollary~\ref{cor:sparse_cauchy_k_subspace}}
%\end{align*}
%
%We consider two cases separately. Case I, if $\| U^* V - A \|_1 \geq 100 c \| U^* V^* - A \|_1$ with $c> 1$($c$ is from Lemma~\ref{lem:sparse_cauchy_l1_no_dilation}).
%We have
%
%\begin{align*}
%&~\| U^* V- U^* V^*\|_1 - c\| A - U^* V^* \|_1\\
%\geq & ~ \| U^* V- A\|_1 - (c+1)\| A - U^* V^* \|_1 &\text{~by~triangle~inequality}\\
%\geq & ~ \| U^* V- A\|_1 - \frac{ (c+1) }{100c} \| A - U^* V \|_1\\
%\geq & ~0.98  \| U^* V- A\|_1
%\end{align*}
%Case II, if $\| U^* V - A \|_1 < 100 c \| U^* V^* - A \|_1$, then outputting such $V$ also provides a good solution. Thus in both cases, we are fine.
%
\end{proof}

%\subsection{Sparse Cauchy transform, $\eps$-net argument for $k$-dimensional subspace}

\subsection{$\ell_1$-Lewis weights}\label{sec:lewis_weights_l1}
In this section we show how to use Lewis weights to get a better dilation bound. The goal of this section is to prove Lemma \ref{lem:lewis_weights_l1_no_contraction} and Lemma \ref{lem:lewis_weights_l1_no_dilation}. Notice that our algorithms in Section~\ref{sec:alg} use Lewis weights in several different ways. The first way is using Lewis weights to show an existence result, which means we only need an existential result for Lewis weights. The second way is only guessing the nonzero locations on the diagonal of a sampling and rescaling matrix according to the Lewis weights. The third way is guessing the values on the diagonal of a sampling and rescaling matrix according to the Lewis weights. The fourth way is computing the Lewis weights for a known low dimensional matrix($n\times\poly(k)$ or $\poly(k)\times d$). We usually do not need to optimize the running time of computing Lewis weights for a low-rank matrix to have input-sparsity time.

\begin{claim}\label{cla:lewis_weight_upper_bound_DU*V*_minus_DA}
Given matrix $A\in \mathbb{R}^{n\times d}$, let $ B = U^* V^* - A$. For any distribution $p=(p_1, p_2, \dotsc, p_n)$ define random variable $X$ such that $X = \|B_i\|_1 / p_i$ with probability $p_i$. Then take any $m$ independent samples $X^1, X^2,\dotsc, X^m$, let $Y= \frac{1}{m} \sum_{j=1}^m X^j$. We have
\begin{equation*}
\Pr[ Y \leq 1000 \| B\|_1 ] \geq .999
\end{equation*}
\end{claim}
\begin{proof}
We can compute the expectation of $X^j$, for any $j\in [m]$
\begin{equation*}
\E [X^j]  = \sum_{i=1}^n \frac{\| B_i \|_1}{ p_i} \cdot p_i = \| B\|_1.
\end{equation*}
Then, $\E[Y] = \frac{1}{m} \sum_{j=1}^m \E[X^j] = \| B \|_1$. Using Markov's inequality, we have
\begin{equation*}
\Pr[ Y \geq 1000 \| B\|_1 ] \leq .001
\end{equation*}
\end{proof}

\begin{lemma}\label{lem:lewis_weights_l1_no_dilation_general}
Given matrix $M\in \mathbb{R}^{n\times d}$, let $S\in\mathbb{R}^{n\times n}$ be any sampling and rescaling diagonal matrix. Then with probability at least $.999$,
\begin{align*}
\|SM\|_1\lesssim \|M\|_1.
\end{align*}
\end{lemma}
\begin{proof}
Just replace the matrix $B$ in the proof of Claim~\ref{cla:lewis_weight_upper_bound_DU*V*_minus_DA} with $M$. Then we can get the result directly.
\end{proof}

Using Theorem 1.1 of \cite{cp15}, we have the following result,
\begin{claim}\label{cla:lewis_weight_upper_bound_DU*V*_minus_DU*V}
Given matrix $A\in \mathbb{R}^{n\times d}$, for any fixed $U^*\in\mathbb{R}^{n\times k}$ and $V^*\in\mathbb{R}^{k\times d}$, choose $D\in \mathbb{R}^{n\times n}$ to be the sampling and rescaling diagonal matrix with $m=O(k\log k)$ nonzeros according to the Lewis weights of $U^*$. Then with probability $.999$, for all $V$,
\begin{equation*}
 \| U^* V^* - U^* V \|_1 \leq \| D U^* V^* - D U^* V \|_1 \lesssim \| U^* V^* - U^* V \|_1.
\end{equation*}
\end{claim}

\begin{lemma}\label{lem:lewis_weights_l1_no_dilation}
Given matrix $A\in \mathbb{R}^{n\times d},U^*\in \mathbb{R}^{n\times k}$, define $V^*\in \mathbb{R}^{k\times d}$ to be the optimal solution of
$\underset{ V\in \mathbb{R}^{k\times d} }{\min} \| U^* V - A \|_1$. Choose a sampling and rescaling diagonal matrix $D\in \mathbb{R}^{n\times n}$ with $m=O(k\log k)$ non-zero entries according to the Lewis weights of $U^*$. Then with probability at least $.99$, we have: for all $V\in\mathbb{R}^{k\times d}$,
\begin{equation*}
\| D U^* V - D A \|_1 \lesssim \| U^* V^* - U^*V \|_1+O(1)\|U^* V^*-A\|_1 \lesssim \|U^* V - A\|_1,
\end{equation*}
holds with probability at least $.99$.
\end{lemma}

\begin{proof}
Using the above two claims, we have with probability at least $.99$, for all $V\in\mathbb{R}^{k\times d}$,
\begin{align*}
\| DU^* V - D A\|_1 & \leq ~ \| DU^* V - D U^* V^* \|_1 + \| D U^* V^* - DA \|_1 & \text{~by~triangle~inequality}\\
& \lesssim ~  \| DU^* V - D U^* V^* \|_1 + O(1)\|  U^* V^* - A \|_1 & \text{~by~Claim~\ref{cla:lewis_weight_upper_bound_DU*V*_minus_DA} } \\
& \lesssim ~  \| U^* V -  U^* V^* \|_1 + O(1)\|  U^* V^* - A \|_1 & \text{~by~Claim~\ref{cla:lewis_weight_upper_bound_DU*V*_minus_DU*V} } \\
& \leq ~ \| U^* V - A \|_1 + \| U^* V^* - A\|_1 + O(1)\|  U^* V^* - A \|_1 & \text{~by~triangle~inequality} \\
& \lesssim ~  \| U^* V - A \|_1.
\end{align*}
\end{proof}

\begin{lemma}\label{lem:lewis_weights_l1_no_contraction}
Given matrix $A\in \mathbb{R}^{n\times d}$, define $U^*\in \mathbb{R}^{n\times k},V^*\in\mathbb{R}^{k\times d}$ to be the optimal solution of
$\underset{U \in \mathbb{R}^{n\times k}, V\in \mathbb{R}^{k\times d} }{\min} \| UV - A \|_1$. Choose a sampling and rescaling diagonal matrix $D\in \mathbb{R}^{n\times n}$ with $m=O(k\log k)$ non-zero entries according to the Lewis weights of $U^*$. For all $V\in \mathbb{R}^{k\times d}$ we have
%\begin{equation*}
%\| D U^* V - D A\|_1 \gtrsim \| U^* V - A \|_1,
%\end{equation*}
\begin{align*}
\| U^* V - A \|_1\lesssim \| D U^* V - D A\|_1+O(1)\|U^*V^*-A\|_1,
\end{align*}
holds with probability at least $.99$.
\end{lemma}
\begin{proof}

Follows by Claim~\ref{cla:lewis_weight_upper_bound_DU*V*_minus_DA}, Lemma~\ref{lem:lewis_weights_l1_no_dilation}, Lemma~\ref{lem:no_dialation_and_no_contraction_imply_general_no_contraction}.

\end{proof}

\section{$\ell_p$-Low Rank Approximation}\label{sec:lp}
This section presents some fundamental lemmas for $\ell_p$-low rank approximation problems. Using these lemmas, all the algorithms described for $\ell_1$-low rank approximation problems can be extended to $\ell_p$-low rank approximation directly. We only state the important Lemmas in this section, due to most of the proofs in this section being identical to the proofs in Section~\ref{sec:l1}.

\subsection{Definitions}\label{sec:definitions_of_con_dil_lp}

This section is just a generalization of Section~\ref{sec:definitions_of_con_dil} to the $\ell_p$ setting when $1<p<2$.

\begin{definition}\label{def:no_dilation_of_a_sketch_lp}
Given a matrix $M\in\mathbb{R}^{n\times d}$, if matrix $S\in\mathbb{R}^{m\times n}$ satisfies
\begin{align*}
\|SM\|_p^p\leq c_1\|M\|_p^p,
\end{align*}
then $S$ has at most $c_1$-dilation on $M$.
\end{definition}

\begin{definition}\label{def:no_contraction_for_general_vectors_lp}
Given a matrix $U\in\mathbb{R}^{n\times k}$, if matrix $S\in\mathbb{R}^{m\times n}$ satisfies
\begin{align*}
\forall x\in\mathbb{R}^{k}, \|SUx\|_p^p\geq \frac1{c_2}\|Ux\|_p^p,
\end{align*}
then $S$ has at most $c_2$-contraction on $U$.
\end{definition}

\begin{definition}\label{def:no_contraction_of_a_sketch_lp}
Given matrices $U\in\mathbb{R}^{n\times k},A\in\mathbb{R}^{n\times d}$, denote $V^*=\arg \min_{V\in\mathbb{R}^{k\times d}}\|UV-A\|_p^p$. If matrix $S\in\mathbb{R}^{m\times n}$ satisfies
\begin{align*}
\forall V\in\mathbb{R}^{k\times d}, \|SUV-SA\|_p^p\geq \frac{1}{c_3}\|UV-A\|_p^p-c_4\|UV^*-A\|_p^p,
\end{align*}
then $S$ has at most $(c_3,c_4)$-contraction on $(U,A)$.
\end{definition}

\begin{definition}\label{def:l1_subspace_embedding_lp}
A $(c_5,c_6)$ $\ell_p$-subspace embedding for the column space of an $n\times k$ matrix $U$ is a matrix $S\in\mathbb{R}^{m\times n}$ for which all $x\in\mathbb{R}^k$
\begin{align*}
\frac{1}{c_5}\|Ux\|_p^p\leq\|SUx\|_p^p\leq c_6\|Ux\|_p^p.
\end{align*}
\end{definition}

\begin{definition}\label{def:sketch_for_l1_multiple_regression_lp}
Given matrices $U\in\mathbb{R}^{n\times k},A\in\mathbb{R}^{n\times d}$, denote $V^*=\arg\min_{V\in\mathbb{R}^{k\times d}}\|UV-A\|_p^p$. Let $S\in\mathbb{R}^{m\times n}.$ If for all $c\geq 1$, and if for any $\wh{V}\in\mathbb{R}^{k\times d}$ which satisfies
\begin{align*}
\|SU\wh{V}-SA\|_p^p\leq c\cdot \min_{V\in\mathbb{R}^{k\times d}}\|SUV-SA\|_p^p,
\end{align*}
it holds that
\begin{align*}
\|U\wh{V}-A\|_p^p\leq c\cdot c_7\cdot\|UV^*-A\|_p^p,
\end{align*}
then $S$ provides a $c_7$-multiple-regression-cost preserving sketch of $(U,A)$.
\end{definition}

\begin{definition}
Given matrices $L\in\mathbb{R}^{n\times m_1},N\in\mathbb{R}^{m_2\times d},A\in\mathbb{R}^{n\times d},k\geq 1$, let 
\begin{align*}
X^*=\arg\min_{\rank-k\ X}\|LXN-A\|_p^p.
\end{align*}
Let $S\in\mathbb{R}^{m\times n}.$ If for all $c\geq 1$, and if for any $\rank-k\ \wh{X}\in\mathbb{R}^{m_1\times m_2}$ which satisfies
\begin{align*}
\|SL\wh{X}N-SA\|_p^p\leq c\cdot \min_{\rank-k\ X}\|SLXN-SA\|_p^p,
\end{align*}
it holds that
\begin{align*}
\|L\wh{X}N-A\|_p^p\leq c\cdot c_8\cdot\|LX^*N-A\|_p^p,
\end{align*}
then $S$ provides a $c_8$-restricted-multiple-regression-cost preserving sketch of $(L,N,A,k)$.
\end{definition}

\subsection{Properties}\label{sec:properties_of_con_dil_lp}

\begin{lemma}\label{lem:no_dialation_and_no_contraction_imply_general_no_contraction_lp}
Given matrices $A\in\mathbb{R}^{n\times d},U\in\mathbb{R}^{n\times k}$, let $V^*=\arg \min_{V\in\mathbb{R}^{k\times d}}\|UV-A\|_p^p$. If $S\in\mathbb{R}^{m\times n}$ has at most $c_1$-dilation on $UV^*-A$, i.e.,
\begin{align*}
\|S(UV^*-A)\|_p^p\leq c_1\|UV^*-A\|_p^p,
\end{align*}
and it has at most $c_2$-contraction on $U$, i.e.,
\begin{align*}
\forall x\in\mathbb{R}^{k}, \|SUx\|_p^p\geq \frac{1}{c_2}\|Ux\|_p^p,
\end{align*}
 then $S$ has at most $(2^{2p-2}c_2,c_1+2^{1-p}\frac{1}{c_2})$-contraction on $(U,A)$, i.e.,
\begin{align*}
\forall V\in\mathbb{R}^{k\times d}, \|SUV-SA\|_p^p\geq \frac{1}{c_2}\|UV-A\|_p^p-(c_1+\frac{1}{c_2})\|UV^*-A\|_p^p,
\end{align*}
\end{lemma}
\begin{proof}
Let $A\in\mathbb{R}^{n\times d},U\in\mathbb{R}^{n\times k}$, and $S\in\mathbb{R}^{m\times n}$ be the same as that described in the lemma. Then $\forall V\in\mathbb{R}^{k\times d}$
\begin{align*}
\|SUV-SA\|_p^p&\geq 2^{1-p}\|SUV-SUV^*\|_p^p-\|SUV^*-SA\|_p^p\\
&\geq 2^{1-p}\|SUV-SUV^*\|_p^p-c_1\|UV^*-A\|_p^p\\
&= 2^{1-p}\|SU(V-V^*)\|_p^p-c_1\|UV^*-A\|_p^p\\
&= 2^{1-p}\sum_{j=1}^d\|SU(V-V^*)_j\|_p^p-c_1\|UV^*-A\|_p^p\\
&\geq 2^{1-p}\sum_{j=1}^d \frac{1}{c_2} \|U(V-V^*)_j\|_p^p-c_1\|UV^*-A\|_p^p\\
&=2^{1-p}\frac{1}{c_2}\|UV-UV^*\|_p^p-c_1\|UV^*-A\|_p^p\\
&\geq 2^{2-2p}\frac{1}{c_2}\|UV-A\|_p^p-2^{1-p}\frac{1}{c_2}\|UV^*-A\|_p^p-c_1\|UV^*-A\|_p^p\\
&= 2^{2-2p}\frac{1}{c_2}\|UV-A\|_p^p-\left((2^{1-p}\frac{1}{c_2}+c_1)\|UV^*-A\|_p^p\right).
\end{align*}
The first inequality follows by Fact~\ref{fac:tri_for_lp}. The second inequality follows since $S$ has at most $c_1$ dilation on $UV^*-A$. The third inequality follows since $S$ has at most $c_2$ contraction on $U$. The fourth inequality follows by Fact~\ref{fac:tri_for_lp}.
\end{proof}

\begin{lemma} \label{lem:general_sketch_SUV_lp}
Given matrices $A\in\mathbb{R}^{n\times d},U\in\mathbb{R}^{n\times k}$, let $V^*=\arg \min_{V\in\mathbb{R}^{k\times d}}\|UV-A\|_p^p$. If $S\in\mathbb{R}^{m\times n}$ has at most $c_1$-dilation on $UV^*-A$, i.e.,
 \begin{align*}
 \|S(UV^*-A)\|_p^p\leq c_1\|UV^*-A\|_p^p,
 \end{align*}
 and has at most $c_2$-contraction on $U$, i.e.,%$(c_3,c_4)$-contraction on $(U,A)$, i.e.
 \begin{align*}
 \forall x\in\mathbb{R}^{k}, \|SUx\|_p^p\geq \frac{1}{c_2}\|Ux\|_p^p,
 \end{align*}
 then $S$ provides a $2^{p-1}(2c_1c_2+1)$-multiple-regression-cost preserving sketch of $(U,A)$, i.e.,
 for all $c\geq 1$, for any $\wh{V}\in\mathbb{R}^{k\times d}$ which satisfies
\begin{align*}
\|SU\wh{V}-SA\|_p^p\leq c\cdot \min_{V\in\mathbb{R}^{k\times d}}\|SUV-SA\|_p^p,
\end{align*}
it has
\begin{align*}
\|U\wh{V}-A\|_p^p\leq c\cdot 2^{p-1}(2c_1c_2+1)\cdot\|UV^*-A\|_p^p,
\end{align*}
\end{lemma}
\begin{proof}
Let $S\in\mathbb{R}^{m\times n},A\in\mathbb{R}^{n\times d},U\in\mathbb{R}^{n\times k},V^*,\wh{V}\in\mathbb{R}^{k\times d},$ and $c$ be the same as stated in the lemma.
\begin{align*}
\|U\wh{V}-A\|_p^p&\leq 2^{2p-2}c_2\|SU\wh{V}-SA\|_p^p+(2^{p-1}+2^{2p-2}c_1c_2)\|UV^*-A\|_p^p\\
&\leq 2^{2p-2}c_2 c\min_{V\in\mathbb{R}^{k\times d}}\|SUV-SA\|_p^p+(2^{p-1}+2^{2p-2}c_1c_2)\|UV^*-A\|_p^p\\
&\leq 2^{2p-2}c_2 c\|SUV^*-SA\|_p^p+(2^{p-1}+2^{2p-2}c_1c_2)\|UV^*-A\|_p^p\\
&\leq 2^{2p-2}c_1c_2c\|UV^*-A\|_p^p+(2^{p-1}+2^{2p-2}c_1c_2)\|UV^*-A\|_p^p\\
&\leq c\cdot 2^{p-1}(1+2c_1c_2)\|UV^*-A\|_p^p.
\end{align*}
The first inequality follows by Lemma~\ref{lem:no_dialation_and_no_contraction_imply_general_no_contraction_lp}. The second inequality follows by the guarantee of $\wh{V}$. The fourth inequality follows since $S$ has at most $c_1$-dilation on $UV^*-A$. The fifth inequality follows since $c\geq 1$.
\end{proof}

\begin{lemma}\label{lem:general_no_contraction_and_no_dialation_imply_restricted_regression_sketch_lp}
Given matrices $L\in\mathbb{R}^{n\times m_1},N\in\mathbb{R}^{m_2\times d},A\in\mathbb{R}^{n\times d},k\geq 1$, let 
\begin{align*}
X^*=\arg\min_{\rank-k\ X}\|LXN-A\|_p^p.
\end{align*}
If $S\in\mathbb{R}^{m\times n}$ has at most $c_1$-dilation on $LX^*N-A$, i.e.,
\begin{align*}
\|S(LX^*N-A)\|_p^p\leq c_1\|LX^*N-A\|_p^p,
\end{align*}
and has at most $c_2$-contraction on $L$, i.e.,
\begin{align*}
\forall x\in\mathbb{R}^{m_1} \|SLx\|_p^p\geq \|Lx\|_p^p,
\end{align*}
then $S$ provides a $2^{p-1}(2c_1c_2+1)$-restricted-multiple-regression-cost preserving sketch of $(L,N,A,k)$, i.e., for all $c\geq 1$, for any $\rank-k\ \wh{X}\in\mathbb{R}^{m_1\times m_2}$ which satisfies
\begin{align*}
\|SL\wh{X}N-SA\|_p^p\leq c\cdot \min_{\rank-k\ X}\|SLXN-SA\|_p^p,
\end{align*}
it has
\begin{align*}
\|L\wh{X}N-A\|_p^p\leq c\cdot 2^{p-1}(2c_1c_2+1)\cdot\|LX^*N-A\|_p^p.
\end{align*}
\end{lemma}

\begin{proof}
Let $S\in\mathbb{R}^{m\times n},L\in\mathbb{R}^{n\times m_1},\wh{X}\in\mathbb{R}^{m_1\times m_2},X^*\in\mathbb{R}^{m_1\times m_2},N\in\mathbb{R}^{m_2\times d},A\in\mathbb{R}^{n\times d}$ and $c\geq 1$ be the same as stated in the lemma.
\begin{align*}
\|SL\wh{X}N-SA\|_p^p&\geq2^{1-p}\|SL\wh{X}N-SLX^*N\|_p^p-\|SLX^*N-SA\|_p^p\\
&\geq 2^{1-p}\frac{1}{c_2}\|L(\wh{X}N-X^*N)\|_p^p-c_1\|LX^*N-A\|_p^p\\
&\geq 2^{2-2p}\frac{1}{c_2}\|L\wh{X}N-A\|_p^p-2^{1-p}\frac{1}{c_2}\|LX^*N-A\|_1-c_1\|LX^*N-A\|_p^p\\
&= 2^{2-2p}\frac{1}{c_2}\|L\wh{X}N-A\|_p^p-(2^{1-p}\frac{1}{c_2}+c_1)\|LX^*N-A\|_p^p.\\
\end{align*}
The inequality follows from the Fact~\ref{fac:tri_for_lp}. The second inequality follows since $S$ has at most $c_2$-contraction on $L$, and it has at most $c_1$-dilation on $LX^*N-A$. The third inequality follows by Fact~\ref{fac:tri_for_lp}.

It follows that
\begin{align*}
\|L\wh{X}N-A\|_p^p&\leq 2^{2p-2}c_2\|SL\wh{X}N-SA\|_p^p+(2^{p-1}+2^{2p-2}c_1c_2)\|LX^*N-A\|_p^p\\
&\leq 2^{2p-2} c_2c\cdot \min_{\rank-k\ X}\|SLXN-SA\|_p^p+(2^{p-1}+2^{2p-2}c_1c_2)\|LX^*N-A\|_p^p\\
&\leq 2^{2p-2} c_2c\cdot \|SLX^*N-SA\|_p^p+(2^{p-1}+2^{2p-2}c_1c_2)\|LX^*N-A\|_p^p\\
&\leq 2^{2p-2} cc_1c_2\cdot \|LX^*N-A\|_p^p+(2^{p-1}+2^{2p-2}c_1c_2)\|LX^*N-A\|_p^p\\
&\leq c\cdot 2^{p-1}(1+2c_1c_2)\|LX^*N-A\|_p^p.
\end{align*}
The first inequality directly follows from the previous one. The second inequality follows from the guarantee of $\wh{X}$. The fourth inequality follows since $S$ has at most $c_1$ dilation on $LX^*N-A$. The fifth inequality follows since $c\geq 1$.
\end{proof}

\begin{lemma}\label{lem:general_sketch_T1BXYCT2_lp}
Given matrices $L\in\mathbb{R}^{n\times m_1},N\in\mathbb{R}^{m_2\times d},A\in\mathbb{R}^{n\times d},k\geq 1$, let 
\begin{align*}
X^*=\arg\min_{\rank-k\ X}\|LXN-A\|_p^p.
\end{align*}
Let $T_1\in\mathbb{R}^{t_1\times n}$ have at most $c_1$-dilation on $LX^*N-A$, %i.e.
%\begin{align*}
%\|T_1(LX^*N-A)\|_1\leq c_1\|LX^*N-A\|_1,
%\end{align*}
and have at most $c_2$-contraction on $L$. %i.e.
%\begin{align*}
%\forall x\in\mathbb{R}^{m_1}, \|T_1Lx\|_1\geq \frac{1}{c_2}\|Lx\|_1.
%\end{align*}
Let
\begin{align*}
\wt{X}=\arg\min_{\rank-k\ X}\|T_1LXN-T_1A\|_p^p.
\end{align*}
Let $T_2^\top\in\mathbb{R}^{t_2\times d}$ have at most $c'_1$-dilation on $(T_1L\wt{X}N-T_1A)^\top$, %i.e.
%\begin{align*}
%\|T_1(L\wt{X}N-A)T_2\|_1\leq c'_1\|T_1(L\wt{X}N-A)\|_1,
%\end{align*}
and at most $c'_2$-contraction on $N^\top$. %i.e.
%\begin{align*}
%\forall x\in\mathbb{R}^{m_2},\|x^\top NT_2\|_1\geq \frac{1}{c'_2}\|x^\top N\|_1.
%\end{align*}
Then, for all $c\geq 1$, for any $\rank-k\ \wh{X}\in\mathbb{R}^{m_1\times m_2}$ which satisfies
\begin{align*}
\|T_1(L\wh{X}N-SA)T_2\|_p^p\leq c\cdot \min_{\rank-k\ X}\|T_1(LXN-A)T_2\|_p^p,
\end{align*}
it holds that 
\begin{align*}
\|L\wh{X}N-A\|_p^p\leq c\cdot 2^{2p-2}(2c_1c_2+1)(2c'_1c'_2+1)\cdot\|LX^*N-A\|_p^p.
\end{align*}
\end{lemma}

\begin{proof}
Apply Lemma~\ref{lem:general_no_contraction_and_no_dialation_imply_restricted_regression_sketch} for sketching matrix $T_2$. Then for any $c\geq 1$, any $\rank-k\ \wh{X}\in\mathbb{R}^{m_1\times m_2}$ which satisfies
\begin{align*}
\|T_1(L\wh{X}N-A)T_2\|_p^p\leq c\cdot \min_{\rank-k\ X}\|T_1(LXN-A)T_2\|_p^p,
\end{align*}
it has
\begin{align*}
\|T_1(L\wh{X}N-A)\|_p^p\leq c\cdot 2^{p-1}(2c'_1c'_2+1)\cdot\|T_1(L\wt{X}N-A)\|_p^p.
\end{align*}

Apply Lemma~\ref{lem:general_no_contraction_and_no_dialation_imply_restricted_regression_sketch} for sketch matrix $T_1$. Then for any $c\geq 1$, any $\rank-k\ \wh{X}\in\mathbb{R}^{m_1\times m_2}$ which satisfies
\begin{align*}
\|T_1(L\wh{X}N-A)\|_p^p\leq c2^{p-1}(2c'_1c'_2+1)\cdot \min_{\rank-k\ X}\|T_1(L\wt{X}N-A)\|_p^p,
\end{align*}
it has
\begin{align*}
\|L\wh{X}N-A\|_p^p\leq c\cdot 2^{2p-2}(2c_1c_2+1)(2c'_1c'_2+1)\cdot\|LX^*N-A\|_p^p.
\end{align*}
\end{proof}

\begin{lemma}\label{lem:con_dil_summary_lp}
Given matrices $M\in\mathbb{R}^{n\times d},U\in\mathbb{R}^{n\times t}$, $d\geq t=\rank(U), n\geq d\geq r=\rank(M)$. If sketching matrix $S\in\mathbb{R}^{m\times n}$ is drawn from any of the following probability distributions on matrices, with $.99$ probability, $S$ has at most $c_1$-dilation on $M$, i.e.,
\begin{align*}
\|SM\|_p^p\leq c_1 \|M\|_p^p,
\end{align*}
and $S$ has at most $c_2$-contraction on $U$, i.e.,
\begin{align*}
\forall x\in\mathbb{R}^{t},\ \|SUx\|_p^p\geq \frac{1}{c_2}\|Ux\|_p^p,
\end{align*}
where $c_1,\ c_2$ are parameters depend on the distribution over $S$.
\begin{enumerate}
\item[\rm{(\RN{1})}] $S\in\mathbb{R}^{m\times n}$ is a dense matrix with entries drawn from a $p$-stable distribution: a matrix with i.i.d. standard $p$-stable random variables. If $m=O(t\log t)$, then $c_1c_2=O(\log d)$. %If $m=O((t+r)\log (t+r))$, then $c_1c_2=O(\min(\log d,r\log r))$.

\item[\rm{(\RN{2})}] $S\in\mathbb{R}^{m\times n}$ is a sparse matrix with some entries drawn from a $p$-stable distribution: $S=TD$, where $T\in\mathbb{R}^{m\times n}$ has each column drawn i.i.d. from the uniform distribution over standard basis vectors of $\mathbb{R}^{m}$, and $D\in\mathbb{R}^{n\times n}$ is a diagonal matrix with each diagonal entry drawn from i.i.d. from the standard $p$-stable distribution. If $m=O(t^5\log^5 t)$, then $c_1c_2=O(t^{2/p}\log^{2/p} t\log d)$. If $m=O((t+r)^5\log^5(t+r))$, then \\$c_1c_2=O(\min(t^{2/p}\log^{2/p} t\log d,r^{3/p}\log^{3/p} r))$.

\item[\rm{(\RN{3})}] $S\in\mathbb{R}^{m\times n}$ is a sampling and rescaling matrix (notation $S\in\mathbb{R}^{n\times n}$ denotes a diagonal sampling and rescaling matrix with $m$ non-zero entries): If $S$ samples and reweights $m=O(t\log t\log\log t)$ rows of $U$, selecting each with probability proportional to the $i^{\text{th}}$ row's $\ell_p$ Lewis weight and reweighting by the inverse probability, then $c_1c_2=O(1)$.
\end{enumerate}
In the above, if we replace $S$ with $\sigma\cdot S$ where $\sigma\in\mathbb{R}\backslash\{0\}$ is any scalar, then the relation between $m$ and $c_1c_2$ can be preserved.
%\rm{\RN{1}} $S\in\mathbb{R}^{n\times m}$ is a dense Cauchy transform\\
%$m=O(t\log t),c_1=O(\log d),c_2=O(1)$

%\rm{\RN{2}} $S\in\mathbb{R}^{n\times m}$ is a dense Cauchy transform\\
%$m=O((t+r)\log (t+r)),c_1=O(\min(r\log r,\log d)),c_2=O(1)$

%\rm{\RN{3}} $S\in\mathbb{R}^{n\times m}$ is a sparse Cauchy transform\\
%$m=O(t^5\log t^5),c_1=O(\log d),c_2=O(t^2\log^2 t)$

\end{lemma}

For $\rm{(\RN{1})}$, it is implied by Lemma E.17, Lemma E.19. Also see from \cite{sw11}\footnote{Full version.}.

For $\rm{(\RN{2})}$, if $m=O(t^5\log^5 t)$, then $c_1c_2=O(t^{2/p}\log^{2/p}t\log d)$ is implied by Corollary~\ref{cor:sparse_cauchy_k_subspace} and Lemma~\ref{lem:sparse_pstable_lp_no_dilation} and Theorem 4 in~\cite{mm13}. If $m=O((t+r)^5\log^5 (t+r))$, $c_1c_2=O(r^{3/p}\log^{3/p} r)$ is implied by~\cite{mm13}.

For $\rm{(\RN{3})}$, it is implied by~\cite{cp15} and Lemma~\ref{lem:lewis_weights_l1_no_dilation_general}.

%For $\rm{(\RN{4})}$, it is implied by Lemma~\ref{lem:limited_no_dilation}, Corollary~\ref{cor:limited_no_contraction}.

\subsection{Tools and inequalities}

\begin{lemma}[Lemma 9 in \cite{mm13}, Upper Tail Inequality for $p$-stable Distributions]\label{lem:upper_tail_p_stable}
Let $p\in (1,2)$ and $m\geq 3$. $\forall i\in[m]$, let $X_i$ be $m$(not necessarily independent) random variables sampled from $D_p$, and let $\gamma_i >0$ with $\gamma = \sum_{i=1}^m \gamma_i$. Let $X= \sum_{i=1}^m \gamma_i |X_i|^p$. Then for any $t\geq 1$,
\begin{align*}
\Pr[ X \geq t \alpha_p \gamma ] \leq \frac{2\log(mt)}{t}.
\end{align*}
\end{lemma}

We first review some facts about the $p$-norm and $q$-norm,
\begin{fact}
For any $p\geq q>0$ and any $x\in \mathbb{R}^k$,
\begin{equation*}
\| x \|_p \leq \| x \|_q \leq k^{\frac{1}{q}-\frac{1}{p} } \| x \|_p.
\end{equation*}
\end{fact}

We provide the triangle inequality for the $p$-norm,
\begin{fact}\label{fac:tri_for_lp}
For any $p\in (1,2)$, for any $x,y\in \mathbb{R}^k$,
 \begin{align*}
\| x + y\|_p \leq \| x \|_p + \| y \|_p, \mathrm{~and~}
\| x + y\|_p^p \leq 2^{p-1} ( \| x \|_p^p + \| y \|_p^p ).
\end{align*}
\end{fact}
\begin{fact}[H\"{o}lder's inequality]
For any $x,y\in \mathbb{R}^k$, if $\frac{1}{p}+\frac{1}{q}=1$, then $| x^\top y | \leq \| x \|_p \| y \|_q$.
\end{fact}

We give the definition of a well-conditioned basis for $\ell_p$,
\begin{definition}
Let $p\in (1,2)$. A basis $U$ for the range of $A$ is $(\alpha,\beta,p)$-conditioned if $\| U \|_p \leq \alpha$
and for all $x\in \mathbb{R}^k$, $\| x\|_{q} \leq \beta \| U x \|_p$. We will say $U$ is well-conditioned
if $\alpha$ and $\beta$ are low-degree polynomials in $k$, independent of $n$.
\end{definition}

\begin{proof}
We first show an upper bound,
\begin{equation*}
\| U x\|_p \leq \| U \|_p \cdot \| x \|_p \leq \alpha \| x \|_p
\end{equation*}
Then we show a lower bound,
\begin{equation*}
\| U x \|_p \geq \frac{1}{\beta} \| x \|_q
\end{equation*}

For any $p$ and $q$  with $1/p + 1/q = 1$, by H\"{o}lder's inequality we have
\begin{equation*}
| x^\top y| \leq \| x \|_p \cdot \| y \|_q
\end{equation*}
choosing $y$ to be the vector that has $1$ everywhere, $\| x \|_1 \leq \| x\|_p k^{1/q}$
%%% | x |_1 <= sqrt{k} | x |_2
%%% 1/sqrt{k} | x |_1 <= | x |_2 <= | x |_1
%%% 1/(k^1/q) \| x \|_p <=
\end{proof}

%\subsection{Algorithm}
\subsection{Dense $p$-stable transform}
This section states the main tools for the dense $p$-stable transform. The proof is identical to that for the dense Cauchy transform.
%%% for the regression problem, the whole point is, it is not contracted for all, and not dilated for one.
\begin{lemma}\label{lem:dense_pstable_lp_no_dilation}
Given matrix $A\in \mathbb{R}^{n\times d}$ and $p\in (1,2)$, define $U^*\in \mathbb{R}^{n\times k}$, $V^*\in \mathbb{R}^{k\times d}$ to be an optimal solution of $\underset{U \in \mathbb{R}^{n\times k} , V \in \mathbb{R}^{k\times d} }{\min} \| U V - A \|_p$.
Choose a $p$-stable distribution matrix $S\in \mathbb{R}^{m\times n}$, rescaled by $\Theta(1/m^{1/p})$. Then we have
\begin{equation*}
\| S U^* V^* - SA\|_p^p \lesssim \log(md) \| U^* V^* - A \|_p^p
\end{equation*}
with probability at least $99/100$.
\end{lemma}
\begin{proof}
Let $P(0,1)$ denote the $p$-stable distribution. Then, 

\begin{align*}
\| S U^* V^* - S A\|_p^p  \leq ~ & \sum_{i=1}^d \| S (U^* V_i^* - A_i ) \|_p^p\\
=~& \sum_{i=1}^d \sum_{j=1}^m | \sum_{l=1}^n \frac{1}{m} S_{j,l} (U^* V_i^* - A_i )_l |^p &\text{~where~}S_{j,l}\sim P(0,1)\\
=~& \frac{1}{m} \sum_{i=1}^d \sum_{j=1}^m  | \wt{w}_{ij} ( \sum_{l=1}^n  | (U^* V_i^* - A_i)_l |^p )^{1/p} |^p &\text{~where~} \wt{w}_{ij} \sim P(0,1) \\
=~& \frac{1}{m} \sum_{i=1}^d \sum_{j=1}^m  \sum_{l=1}^n | (U^* V_i^* - A_i)_l |^p \cdot |\wt{w}_{ij}|^p \\
=~& \frac{1}{m} \sum_{i=1}^d \sum_{j=1}^m \| U^* V_i^* - A_i \|_p^p \cdot w_{i+(j-1)d}^p, & \text{~where~} w_{i+(j-1)d} \sim |P(0,1)|
\end{align*}
where the last step follows since each $w_i$ can be thought of as a clipped half-$p$-stable random variable. Define $X$ to be $\sum_{i=1}^d \sum_{j=1}^m   \|  U^* V^*_i - A_i \|_p^p \cdot w_{i+ (j-1)d}^p$ and $\gamma$ to be $\sum_{i=1}^d \sum_{j=1}^m   \|  U^* V^*_i - A_i  \|_p^p $. Then applying Lemma \ref{lem:upper_tail_p_stable},
\begin{align*}
\Pr[ X \geq t \alpha_p \gamma] \leq \frac{2\log(md t)}{t}.
\end{align*}
Choosing $t= \Theta(\log(md))$, we have with probability $.999$,
\begin{align*}
X \lesssim \log(md) \alpha_p \gamma = \log(md) \alpha_p \sum_{i=1}^d  \|  U^* V^*_i - A_i\|_p^p,
\end{align*}
where the last steps follows by definition of $\gamma$. Thus, we can conclude that with probability $.999$, $\| \Pi (U^* V^* - A ) \|_p^p \lesssim \log(md) \| U^* V^* - A \|_p^p $.
%\begin{align*}
%\| S U^* V^* - S A\|_p  \leq ~& \sum_{i=1}^d \| S (U^* V_i^* - A_i ) \|_1 \\
%=~& \sum_{i=1}^d \sum_{j=1}^m | \sum_{l=1}^n \frac{1}{m} S_{j,l} (U^* V_i^* - A_i )_l | &\text{~where~}S_{j,l}\sim P(0,1)\\
%=~& \sum_{i=1}^d \sum_{j=1}^m \frac{1}{m} P(0,\| U^* V_i^* - A_i  \|_p)  \\
%=~& \sum_{i=1}^d  P(0, \| (U^* V_i^* - A_i ) \|_p) \\
%=~& \sum_{i=1}^d  \| (U^* V_i^* - A_i ) \|_p \cdot w_i & \text{~where~} w_i \sim P(0,1) \\
%\leq~ & \log d \cdot \| U^* V^* - A\|_p
%\end{align*}
\end{proof}

\begin{lemma}\label{lem:dense_pstable_lp_no_contraction}
Given matrix $A\in \mathbb{R}^{n\times d}$ and $p\in (1,2)$, define $U^*\in \mathbb{R}^{n\times k}$ to be the optimal solution of $\underset{U \in \mathbb{R}^{n\times k} , V \in \mathbb{R}^{k\times d} }{\min} \| U V - A \|_p$. Choose a matrix of i.i.d. $p$-stable random variables $S\in \mathbb{R}^{m\times n}$. Then for all $V\in \mathbb{R}^{k\times n}$, we have
\begin{equation*}
\| S U^* V - S A\|_p^p \gtrsim \| U^* V - A \|_p^p - O(\log(md)) \| U^*V^* - A \|_p^p.
\end{equation*}
\end{lemma}

%\paragraph{$p$-stable distribution, $\eps$-net argument for $k$-dimensional subspace}

\begin{lemma}\label{lem:dense_pstable_lp_k_subspace}
Let $p\in (1,2)$. Given an $(\alpha,\beta)$ $\ell_p$ well-conditioned basis, condition on the following
two events,

1. For all $x\in {\cal N}$, $\| SU x\|_p \gtrsim \| U x\|_p$.

2. For all $x\in \mathbb{R}^k$, $\| S U x \|_p \leq \poly(k) \| U x \|_p$.  %%% The poly(k) need to be verified.

Then for all $x\in \mathbb{R}^k$, $\| S U x \|_p \gtrsim \| U x \|_p$.
\end{lemma}
\begin{proof}
The proof is identical to Lemma~\ref{lem:dense_cauchy_l1_k_subspace} in Section~\ref{sec:l1}.
\end{proof}

\subsection{Sparse $p$-stable transform}
This section states the main tools for the sparse $p$-stable transform. The proof is identical to that of the sparse Cauchy transform.
\begin{lemma}\label{lem:sparse_pstable_lp_no_dilation}
Let $p\in (1,2)$. Given matrix $A\in \mathbb{R}^{n\times d}$ with $U^*,V^*$ an optimal solution of $\min_{U,V}\| UV - A \|_p$, let $\Pi = \sigma \cdot SC\in \mathbb{R}^{m\times n}$, where $S\in \mathbb{R}^{m\times n}$ has each column vector chosen independently and uniformly from the $m$ standard basis vectors of $\mathbb{R}^m$, where $C$ is a diagonal matrix with diagonals chosen independently from the standard $p$-stable distribution, and $\sigma$ is a scalar. Then
\begin{align*}
\| \Pi U^* V^* - \Pi A\|_p^p \lesssim \sigma \cdot \log (md) \cdot \| U^* V^* - A \|_p^p
\end{align*}
holds with probability at least $.999$.
\end{lemma}

\begin{proof}
  We define $\Pi=\sigma \cdot SC \in \mathbb{R}^{m\times n}$ as in the statement
  of the lemma.
Then by the definition of $\Pi$, we have,
\begin{align*}
 & ~\| \Pi (U^*V^* - A ) \|_p^p \\
=& ~ \sum_{i=1}^d \| S C (U^*V^*_i - A_i ) \|_p^p \\
=& ~ \sum_{i=1}^d \biggl\|
\begin{bmatrix}
S_{11} & S_{12} & \cdots & S_{1n} \\
S_{21} & S_{22} & \cdots & S_{2n} \\
\cdots & \cdots & \cdots & \cdots \\
S_{m1} & S_{m2} & \cdots & S_{mn} \\
\end{bmatrix}
\cdot
\begin{bmatrix}
c_{1} & 0 & 0 & 0 \\
0 & c_{2} & 0 & 0 \\
0 & 0 & \cdots & 0 \\
0 & 0 & 0 & c_{n} \\
\end{bmatrix}
\cdot
(U^*V^*_i - A_i) \biggr\|_p^p~\\
=&~  \sum_{i=1}^d \biggl\|
\begin{bmatrix}
c_1 S_{11} & c_2 S_{12} & \cdots & c_n S_{1n} \\
c_1 S_{21} & c_2 S_{22} & \cdots & c_n S_{2n} \\
\cdots & \cdots & \cdots & \cdots \\
c_1 S_{m1} & c_2 S_{m2} & \cdots & c_n S_{mn} \\
\end{bmatrix}
\cdot
(U^*V^*_i - A_i) \biggr\|_p^p~\\
= & ~  \sum_{i=1}^d \biggl\| \sum_{l=1}^n c_l S_{1l} \cdot (U^*V^*_i-A_i)_l,  \sum_{l=1}^n c_l S_{2l} \cdot (U^*V^*_i-A_i)_l, \cdots,  \sum_{l=1}^n c_l S_{ml} \cdot (U^*V^*_i-A_i)_l \biggr\|_p^p \\
= & ~ \sum_{i=1}^d \sum_{j=1}^m \biggl|\sum_{l=1}^n c_l S_{jl} \cdot (U^*V^*_i-A_i)_l \biggr|^p \text{~by~}aX + bY \text{~and~} (|a|^p + |b|^p)^{1/p} Z \text{~are~identically distributed}\\
= & ~ \sum_{i=1}^d \sum_{j=1}^m \biggl|\wt{w}_{ij} \cdot (\sum_{l=1}^n | S_{jl} (U^* V^*_i - A_i)_l |^p )^{1/p} \biggr|^p \text{~where~}\wt{w}_{ij} \sim P(0,1) \\
= & ~ \sum_{i=1}^d \sum_{j=1}^m  \sum_{l=1}^n | S_{jl} (U^* V^*_i - A_i)_l |^p \cdot |\wt{w}_{ij}|^p\\
= & ~ \sum_{i=1}^d \sum_{j=1}^m  \sum_{l=1}^n | S_{jl} (U^* V^*_i - A_i)_l |^p \cdot w_{i+ (j-1)d}^p, \text{~where~} w_{i+(j-1) d} \sim |P(0,1)|
\end{align*}
where the last step follows since each $w_i$ can be thought of as a clipped half-$p$-stable random variable. Define $X$ to be $\sum_{i=1}^d \sum_{j=1}^m  \sum_{l=1}^n | S_{jl} (U^* V^*_i - A_i)_l |^p \cdot w_{i+ (j-1)d}^p$ and $\gamma$ to be $\sum_{i=1}^d \sum_{j=1}^m  \sum_{l=1}^n | S_{jl} (U^* V^*_i - A_i)_l |^p $. Then applying Lemma \ref{lem:upper_tail_p_stable},
\begin{align*}
\Pr[ X \geq t \alpha_p \gamma] \leq \frac{2\log(md t)}{t}.
\end{align*}
Choosing $t= \Theta(\log(md))$, we have with probability $.999$,
\begin{align*}
X \lesssim \log(md) \alpha_p \gamma = \log(md) \alpha_p \sum_{i=1}^d \sum_{l=1}^n |  (U^* V^*_i - A_i)_l |^p,
\end{align*}
where the last steps follows by
\begin{align*}
\gamma=\sum_{i=1}^d \sum_{j=1}^m  \sum_{l=1}^n | S_{jl} (U^* V^*_i - A_i)_l |^p = \sum_{i=1}^d \sum_{j=1}^m  \sum_{l=1}^n | S_{jl} |^p | (U^* V^*_i - A_i)_l |^p = \sum_{i=1}^d \sum_{l=1}^n | (U^* V^*_i - A_i)_l |^p.
\end{align*}
Thus, we can conclude that with probability $.999$, $\| \Pi (U^* V^* - A ) \|_p^p \lesssim \log(md) \| U^* V^* - A \|_p^p $.
\end{proof}

\begin{lemma}\label{lem:sparse_pstable_lp_no_contraction}
Given matrix $A\in \mathbb{R}^{n\times d}$ with $U^*,V^*$ an optimal solution of $\min_{U,V}\| UV - A \|_p$, let $\Pi = \sigma \cdot SC\in \mathbb{R}^{m\times n}$, where $S\in \mathbb{R}^{m\times n}$ has each column vector chosen independently and uniformly from the $m$ standard basis vectors of $\mathbb{R}^m$, and where $C$ is a diagonal matrix with diagonals chosen independently from the standard $p$-stable distribution. Then with probability at least $.999$, for all $V\in \mathbb{R}^{k\times d}$,
\begin{align*}
\| \Pi U^* V - \Pi A\|_p^p \geq \| U^* V - A \|_p^p - O(\sigma \log(md)) \|U^* V^* - A \|_p^p.
\end{align*}
Notice that $m=O(k^5 \log^5 k)$ and $\sigma = O( (k \log k)^{2/p})$ according to Theorem 4 in \cite{mm13}.
\end{lemma}

\subsection{$\ell_p$-Lewis weights}
This section states the main tools for $\ell_p$-Lewis weights. The proof is identical to $\ell_1$-Lewis weights.

\begin{lemma}\label{lem:lewis_weights_lp_no_dilation}
For any $p\in (1,2)$. Given matrix $A\in \mathbb{R}^{n\times d}$, define $U^*\in \mathbb{R}^{n\times k}, V^*\in \mathbb{R}^{k\times d}$ to be an optimal solution of
$\underset{U \in \mathbb{R}^{n\times k}, V\in \mathbb{R}^{k\times d} }{\min} \| U V - A \|_p$. Choose a diagonal matrix $D\in \mathbb{R}^{n\times n}$ according to the Lewis weights of $U^*$. We have that 
\begin{align*}
\| D U^* V^* - D A \|_p^p \lesssim \| U^* V^* - A \|_p^p,
\end{align*}
holds with probability at least $.99$.
\end{lemma}

\begin{lemma}\label{lem:lewis_weights_lp_no_contraction}
Let $p\in (1,2)$. Given matrix $A\in \mathbb{R}^{n\times d}$, define $U^*\in \mathbb{R}^{n\times k}, V^*\in \mathbb{R}^{k\times d}$ to be an optimal solution of
$\underset{U \in \mathbb{R}^{n\times k}, V\in \mathbb{R}^{k\times d} }{\min} \| UV - A \|_p$. Choose a sampling and rescaling matrix $D\in \mathbb{R}^{n\times n}$ according to the Lewis weights of $U^*$. For all $V\in \mathbb{R}^{k\times d}$ we have
\begin{align*}
\| D U^* V - D A\|_p^p \gtrsim \| U^* V - A \|_p^p - O(1)\|U^* V^* - A \|_p^p,
\end{align*}
holds with probability at least $.99$.
\end{lemma}

\section{$\EMD$-Low Rank Approximation}\label{sec:emd}
In this section we explain how to embed EMD to $\ell_1$. For more detailed background on the Earth-Mover Distance(EMD) problem, we refer the reader to \cite{it03,aik08,akiw09,ip11,birw16} and \cite{sl09,log16}. Section~\ref{sec:emd_definition} introduces some necessary notation and definitions for Earth-Mover Distance. Section~\ref{sec:emd_main} presents the main result for the Earth-Mover distance low rank approximation problem.

\subsection{Definitions}\label{sec:emd_definition}
%For any two multisets $A,B$ of points in $\mathbb{R}^2$, $|A| = |B| = N$, the (planar) Earth-Mover Distance between $A$ and $B$ is defined as the minimum cost of a perfect matching with edges between $A$ and $B$, i.e.,
%\begin{equation*}
%\EMD(A,B) = \underset{\pi : A \rightarrow B }{\min} \sum_{a \in A} \| a - \pi(a) \|
%\end{equation*}
Consider any two non-negative vectors $x,y \in \mathbb{R}_{+}^{ [\Delta]^2 }$ such that $\| x \|_1 = \| y\|_1$. Let $\Gamma(x,y)$ be the set of functions $\gamma: [\Delta]^2 \times [\Delta]^2 \rightarrow \mathbb{R}_{+}$, such that for any $i,j\in [\Delta]^2$ we have $\sum_{l} \gamma(i, l) = x_i$ and $\sum_{l} \gamma (l,j) = y_j$; that is, $\Gamma$ is the set of possible ``flows'' from $x$ to $y$. Then we define
\begin{equation*}
\EMD(x,y) = \underset{\gamma \in \Gamma }{\min} \sum_{i,j \in [\Delta]^2 } \gamma(i,j) \| i - j \|_1
\end{equation*}
to be the min cost flow from $x$ to $y$, where the cost of an edge is its $\ell_1$ distance.

Using the $\EMD(\cdot, \cdot)$ metric, for general vectors $w$, we define $\| \cdot \|_{\EEMD}$ distance (which is the same as \cite{sl09}),
\begin{equation*}
\| w \|_{\EEMD} = \underset{ \substack{ x-y+z=w\\ \| x\|_1= \|y \|_1 \\ x,y \geq 0 } }{\min} \EMD(x,y) + 2\Delta \| z \|_1.
\end{equation*}
Using $\| \cdot \|_{\EEMD}$ distance, for general matrices $X\in \mathbb{R}^{n\times d}$, we define the $\| \cdot \|_{1,\EEMD}$ distance,
\begin{equation*}
\| X \|_{1,\EEMD} = \sum_{i=1}^d \| X_i\|_{\EEMD},
\end{equation*}
where $X_i$ denotes the $j$-th column of matrix $X$.

\subsection{Analysis of no contraction and no dilation bound}\label{sec:emd_main}
\begin{lemma}\label{lem:emd_no_contraction}
Given matrix $A\in \mathbb{R}^{n\times d}$ and $U^*, V^* = \underset{U \in \mathbb{R}^{n\times k}, V \in \mathbb{R}^{k\times d} }{\arg\min} \| UV - A\|_{1,\EEMD}$, there exist sketching matrices $S\in \mathbb{R}^{m\times n}$ such that, with probability $.999$, for all $V\in \mathbb{R}^{k\times d}$,
\begin{align*}
\| S (U^* V - A ) \|_1 \geq \| U^* V - A \|_{1,\EEMD}
\end{align*}
holds.
\end{lemma}

\begin{proof}
Using Lemma 1 in \cite{it03}, there exists a constant $C>0$ such that for all $i\in [d]$,
\begin{align}
C \| SU^* V_i - SA_i \|_1 \geq \| U^* V_i - A_i \|_{\EEMD}.
\end{align}
Then taking a summation over all $d$ terms and rescaling the matrix $S$, we obtain,
\begin{align*}
\sum_{i=1}^d \| S (U^* V_i - A_i ) \|_1 \geq \| U^* V_i - A_i \|_{\EEMD}
\end{align*}
which completes the proof.
\end{proof}

\begin{lemma}\label{lem:emd_no_dilation}
 Given matrix $A\in \mathbb{R}^{n\times d}$ and $U^*, V^* = \underset{U\in \mathbb{R}^{n\times k}, V \in \mathbb{R}^{k\times d}  }{\arg\min} \| UV - A\|_{1,\EEMD}$, there exist sketching matrices $S\in \mathbb{R}^{m\times n}$ such that
\begin{equation*}
\| S (U^* V^* - A ) \|_1 \leq O(\log n ) \| U^* V^* - A \|_{1,\EEMD}
\end{equation*}
holds with probability at least $.999$.
\end{lemma}

\begin{proof}

Using Lemma 2 in \cite{it03}, we have for any $i\in [d]$,
\begin{equation}\label{eq:emd_single_column_no_dilation}
\E[ \| S U^* V^*_i - S A_i \|_1 ] \leq O(\log n) \| U^* V^*_i - A_i \|_{\EEMD}.
\end{equation}
Then using that the expectation is linear, we have
\begin{align*}
\E[ \| S U^* V^* - S A \|_1 ] & = ~ \E[ \sum_{i=1}^d \| S U^* V^*_i - S A_i \|_1 ] \\
& ~ = \sum_{i=1}^d \E[ \| S U^* V^*_i - S A_i \|_1 ] \\
& ~ \leq \sum_{i=1}^d O(\log n)\| U^* V^*_i - A_i \|_{\EEMD} & \text{~by~Equation~(\ref{eq:emd_single_column_no_dilation})} \\
& ~ = O(\log n) \| U^* V^* - A \|_{1,\EEMD}.
\end{align*}
Using Markov's inequality, we can complete the proof.
\end{proof}

\begin{theorem}
Given a matrix $A\in \mathbb{R}^{n\times d}$, there exists an algorithm running in $\poly(k,n,d)$ time that is able to output $U\in \mathbb{R}^{n\times k}$ and $V\in \mathbb{R}^{k\times d}$ such that
\begin{align*}
\| U V - A \|_{1,\EEMD} \leq \poly(k) \cdot \log d\cdot \log n \min_{\rank-k~A_k}\| A_k -A \|_{1,\EEMD}
\end{align*}
holds with probability $.99$.
\end{theorem}
\begin{proof}
First using Lemma~\ref{lem:emd_no_contraction} and Lemma~\ref{lem:emd_no_dilation}, we can reduce the original problem into an $\ell_1$-low rank approximation problem by choosing $m=\poly(n)$. Second, we can use our $\ell_1$-low rank approximation algorithm to solve it. Notice that all of our $\ell_1$-low rank approximation algorithms can be applied here. If we apply Theorem~\ref{thm:polyklogd_approx_algorithm}, we complete the proof.
%Let $U^*, V^* = \underset{U \in \mathbb{R}^{n\times k}, V \in \mathbb{R}^{k\times d} }{\arg\min} \| UV - A\|_{1,\EEMD}$.
%Let $S\in\mathbb{R}^{m\times n}$ be the same matrix in Lemma~\ref{lem:emd_no_contraction} and Lemma~\ref{lem:emd_no_dilation}. Let $D\in\mathbb{R}^{\poly(k)\times m}$ be one of the matrix in 

\end{proof}
Our current $\| \cdot \|_{1,\EEMD}$ is column-based. We can also define it to be row-based. Then we get a slightly better result by applying the $\ell_1$-low rank algorithm.
\begin{corollary}
Given a matrix $A\in \mathbb{R}^{n\times d}$, there exists an algorithm running in $\poly(k,n,d)$ time that is able to output $U\in \mathbb{R}^{n\times k}$ and $V\in \mathbb{R}^{k\times d}$ such that
\begin{align*}
\| U V - A \|_{1,\EEMD} \leq \poly(k) \cdot \log^2 d \min_{\rank-k~A_k}\| A_k -A \|_{1,\EEMD}
\end{align*}
holds with probability $.99$.
\end{corollary}

\section{Hardness Results for Cauchy Matrices, Row Subset Selection, OSE} \label{sec:hardinstance}
Section~\ref{sec:hardinstance_cauchy} presents some inapproximability results by using random Cauchy matrices. Section~\ref{sec:hardinstance_row} is a warmup for inapproximability results for row subset selection problems. Section~\ref{sec:hardinstance_ose} shows the inapproximability results by using any linear oblivious subspace embedding(OSE), and also shows inapproximability results for row subset selection.
\subsection{Hard instance for Cauchy matrices}\label{sec:hardinstance_cauchy}
The goal of this section is to prove Theorem \ref{thm:hardinstance_cauchy}.
Before stating the result, we first introduce some useful tools in our analysis.
%hello --- check
\begin{lemma}[Cauchy Upper Tail Inequality, Lemma~3 of \cite{cdmmmw13}]\label{lem:cauchy_tail}
For $i\in[m]$, let $C_i$ be $m$ random Cauchy variables from $C(0,1)$ (not necessarily independent), and $\gamma_i>0$ with $\gamma=\sum_{i\in [m]}\gamma_i$. Let $X=\sum_{i\in [m]}\gamma_i|C_i|$. Then, for any $t\geq 1$,
\begin{align*}
\Pr[X>\gamma t]\leq O(\log (mt)/t).
\end{align*}
\end{lemma}

\begin{theorem}\label{thm:hardinstance_cauchy}
Let $k\geq 1$. There exist matrices $A\in \mathbb{R}^{d\times d}$ such that for any $o(\log d)\geq t\geq 1$, where $c$ can be any constant smaller than $1/3$, for random Cauchy matrices $S\in \mathbb{R}^{t\times d}$ where each entry is sampled from an i.i.d. Cauchy distribution $C(0,\gamma)$ where $\gamma$ is an arbitrary real number, with probability $.99$ we have
\begin{align*}
 \underset{U \in \mathbb{R}^{d\times t} }{\min} \| U SA - A\|_1 \geq \Omega(\log d/(t\log t)) \underset{\rank-k~A'}{\min} \| A' -  A\|_1.
 \end{align*}
\end{theorem}

\begin{proof}
We define matrix $A \in \mathbb{R}^{d\times d}$
\begin{align*}
A = B + I =
\alpha \cdot
\begin{bmatrix}
1 & 1 & 1 & \cdots & 1 \\
0 & 0 & 0 & \cdots & 0 \\
0 & 0 & 0 & \cdots & 0 \\
\cdots & \cdots & \cdots & \cdots & \cdots \\
0 & 0 & 0 & \cdots & 0 \\
\end{bmatrix}
+
\begin{bmatrix}
1 & 0 & 0 & \cdots & 0 \\
0 & 1 & 0 & \cdots & 0 \\
0 & 0 & 1 & \cdots & 0 \\
\cdots & \cdots & \cdots & \cdots & \cdots \\
0 & 0 & 0 & \cdots & 1 \\
\end{bmatrix},
\end{align*}
where $\alpha=\Theta(\log d)$.% will be decided later.
 So, if we only fit the first row of $A$, we can get approximation cost at most $O(d)$.

For $t>0$, let $S\in\mathbb{R}^{t\times d}$ denote a random Cauchy matrix where $S_{i,j}$ denotes the entry in the $i$th row and $j$th column. Then $SA$ is
\begin{align*}
S A =SB + S I =
\alpha \cdot
\begin{bmatrix}
S_{1,1} & S_{1,1} & S_{1,1} & \cdots & S_{1,1} \\
S_{2,1} & S_{2,1} & S_{2,1} & \cdots & S_{2,1} \\
S_{3,1} & S_{3,1} & S_{3,1} & \cdots & S_{3,1} \\
\cdots & \cdots & \cdots & \cdots & \cdots \\
S_{t,1} & S_{t,1} & S_{t,1} & \cdots & S_{t,1} \\
\end{bmatrix}
+
\begin{bmatrix}
S_{1,1} & S_{1,2} & S_{1,3} & \cdots & S_{1,d} \\
S_{2,1} & S_{2,2} & S_{2,3} & \cdots & S_{2,d} \\
S_{3,1} & S_{3,2} & S_{3,3} & \cdots & S_{3,d} \\
\cdots & \cdots & \cdots & \cdots & \cdots \\
S_{t,1} & S_{t,2} & S_{t,3} & \cdots & S_{t,d} \\
\end{bmatrix}.
\end{align*}
Since $\forall \gamma,$
\begin{align*}
\underset{U \in \mathbb{R}^{d\times t} }{\min} \| U SA - A\|_1=\underset{U \in \mathbb{R}^{d\times t} }{\min} \|\gamma U SA - A\|_1,
\end{align*}
without loss of generality, we can let $S_{i,j}\sim C(0,1)$. Then we want to argue that, if we want to use $SA$ to fit the first row of $A$, with high probability, the cost will be $\Omega(d\log d)$.

Let $b\in\mathbb{R}^d$ denote the first row of $A$. Then, we want to use $SA$ to fit the first row of $A$, which is a $d$-dimensional vector that has entry $\Theta(\log d)$ on each position. The problem is equivalent to
\begin{align*}
\min_{x \in \mathbb{R}^{t} } \|  (SA)^\top x - b \|_1.
\end{align*}
%We denote $A$ as $A=  B+I \in \mathbb{R}^{d\times d}$ , and the first row of B is all $\alpha=\Theta(\log d)$, and is all zero elsewhere.
%Suppose $S$ is a random Cauchy matrix with size $t\times d$ where $t$ is a constant.
First, we want to show that for any $x$ in $\mathbb{R}^t$, if $x^\top SA$ fits the first row of $A$ very well, then the $\ell_1$ and $\ell_2$ norm of vector $x$ must have reasonable size.
\begin{claim}\label{cla:upper_lower_x}
Define $A^1$ to be the first row of matrix $A$. With probability $.999$, for any column vector $x\in \mathbb{R}^{t\times 1}$, if $ \| x^\top SA - A^1  \|_1  \leq o(d \log d)$, then

\rm{~Property~(\RN{1})} : $\| x \|_2 \geq \Omega(1/t\log t)$;

\rm{~Property~(\RN{2})} :  $ \| x \|_1 \leq O(\log d)$.
\end{claim}
\begin{proof}
Consider the absolute value of the $i$-th coordinate of $x^\top SA$. We can rewrite $| \langle (SA)_i,x \rangle|$ in the following sense,
\begin{align}\label{eq:rewrite_SA_i_dot_x}
| \langle (SA)_i , x \rangle | = |\langle (SB)_i,x \rangle + \langle (SI)_i,x \rangle | = | \langle (SB)_1,x \rangle + \langle (SI)_i,x \rangle | = | \langle (SB)_1,x \rangle  + \langle S_i,x \rangle |,
\end{align}
where the second step follows because $(SB)_i = (SB)_1, \forall i\in [n]$, and the last step follows because $(SI)_i = S_i, \forall i \in [n]$.

We start by proving Property \RN{1}. Using the triangle inequality and Equation~(\ref{eq:rewrite_SA_i_dot_x}),
\begin{align*}
 &  ~| \langle (SA)_i,x \rangle | \\
\leq & ~ | \langle (SB)_1,x \rangle | + |\langle S_i,x \rangle | \\
\leq & ~ \| (SB)_1 \|_2 \|x\|_2 + \| S_i \|_2  \|x \|_2  & \text{~by~Cauchy-Schwarz~inequality} \\
\leq & ~ \| (SB)_1 \|_1 \|x\|_2 + \| S_i \|_1  \|x \|_2.  & \text{~by~} \| \cdot \|_2 \leq \| \cdot \|_1
\end{align*}

%Because $t$ is a constant,
Then, according to Lemma~\ref{lem:cauchy_tail}, with probability $.99999$, we have $\| (SB)_1\|_1\leq O(t\log t\log d)$ and for a fixed $i\in[d]$, $\|S_i\|_1\leq O(t\log t)$. Applying the Chernoff bound, with probability $1-2^{-\Omega(d)}$,  there are a constant fraction of $i$ such that $\|S_i\|_1=O(t\log t)$. Taking the union bound, with probability $.9999$, there exists a constant fraction of $i$ such that $ | \langle (SA)_i,x \rangle| \leq O(t\log t\log d) \|x\|_2$. Because $A_{1,i} \geq \alpha$, $\forall i\in [d]$ where $\alpha=\Theta(\log d)$, we need $\|x\|_2\geq \Omega(1/t\log t)$. Otherwise, the total cost on this constant fraction of coordinates will be at least $\Omega(d\log d)$.

For Property \RN{2}, for a fixed $x\in \mathbb{R}^t$, we have
\begin{align*}
&|\langle (SA)_i,x \rangle | \\
=&| \langle (SB)_1,x \rangle + \langle S_i,x \rangle | &\text{~by~Equation~(\ref{eq:rewrite_SA_i_dot_x})}  \\
=&\biggl| \Theta(\log d) \|x\|_1 w'_1(x) + \|x\|_1 w'_i(x) \biggr|.
\end{align*}
where $w'_i(x) \sim C(0,1)$, and for different $x$, $w'_i(x)$ are different.
Then for a fixed $x\in\mathbb{R}^t$ with probability at least $0.9$, $|w'_i(x)|=\Omega(1)$, and with probability $0.5$, $w'_1(x)$ and $w'_i(x)$ have the same sign. Since these two events are independent, with probability at least $0.45$, we have
\begin{align*}
&|\langle (SA)_i,x \rangle |\geq \|x\|_1\cdot\Omega(1).
\end{align*}
Applying the Chernoff bound, with probability at least $1-2^{\Theta(d)}$, there exists a $3/10$ fraction of $i$ such that $|\langle (SA)_i,x \rangle | \geq \|x\|_1\Omega(1)$. %Because of $A_{1,i} \leq \alpha+1$, $\forall i\in [d]$ where $\alpha=\Theta(\log d)$, we need $\|x\|_1=O(\log d)$. Otherwise, the total cost on those constant fraction of coordinates will be at least $\Omega(d\log d)$.

%Using triangle inequality, we can lower bound $|\langle (SA)_i, x \rangle |$,
%\begin{align*}
%|\langle (SA)_i, x \rangle | \geq &~ \Theta( \log d) \|x\|_1 |w'_1(x)| - \|x\|_1 |w'_i(x)| \\
%= & ~ \|x\|_1 \biggl( \Theta(\log d)|w'_1(x)| - |w'_i(x)| \biggr)
%\end{align*}

%Since $x$ is fixed, then with probability $1-2^{-\Omega(t)}$, then $|w_1'(x)| = \Omega(1)$. With probability $.9999$, for a fixed $i\in [d]$, $|w_i'(x)| =O(1)$. Applying the Chernoff bound again, with probability $1-2^{\Omega(d)}$, there exists a constant fraction of $i$ such that $|w_i'(x)| =O(1)$. By taking union, with probability $1-2^{-\Omega(t)}$, there exists a constant fraction of $i$ such that $|\langle (SA)_i, x \rangle | \geq \Theta(\log d) \|x\|_1$. Because of $A_{1,i} \leq \alpha+1$, $\forall i\in [d]$ where $\alpha=\Theta(\log d)$, we need $\|x\|_1=O(1)$. Otherwise, the total cost on those constant fraction of coordinates will be at least $\omega(d\log d)$.

%With constant probability (we regard t as constant), $|x|_1$ should be $O(1)$.
We build an $\eps$-net $\N$ for $x \in \mathbb{R}^t$ on an $\ell_1$-norm unit ball, where $\eps = 1/(t^2\log^2 d)$. Thus the size of the net is $|\N| = 2^{\wt{\Theta}(t)}$. Consider $y$ to be an arbitrary vector, let $y /\| y \|_1=x+\delta$, where $x$ is the closest point to $y/\|y\|_1$ and $x \in \N$. For any $\delta \in \mathbb{R}^t$,
\begin{align*}
| \langle (SA)_i, \delta \rangle | = & ~ |\langle (SB)_1, \delta \rangle + \langle (SI)_i,\delta \rangle | \\
\leq & ~ |\langle (SB)_1, \delta \rangle | +  | \langle (SI)_i,\delta \rangle | &\text{~by~triangle~inequality} \\
\leq  & ~ \|(SB)_1\|_2 \|\delta\|_2 + \|S_i\|_2 \|\delta\|_2 & \text{~by~Cauchy-Schwarz~inequality} \\
\leq  & ~ \|(SB)_1\|_2 \|\delta\|_1 + \|S_i\|_2 \|\delta\|_1. & \text{~by~} \| \cdot \|_2 \leq \| \cdot \|_1
\end{align*}
As we argued before, With probability $.99999$, we have
\begin{align}\label{eq:upper_on_SB}
\| (SB)_1 \|_2 \leq \|(SB)_1\|_1\leq O(t\log t\log d),
\end{align}
and with probability $1-2^{-\Theta(d)}$, there is a $9/10$ fraction of $i\in[d]$, $\| S_i\|_2\leq \|S_i\|_1 = O(t\log t)$. Therefore, with probability $.999$, for any $\delta\in\mathbb{R}^t$, there exists a $9/10$ fraction of $i$ such that $| \langle (SA)_i,\delta \rangle | \leq \Theta(t\log t\log d) \|\delta\|_1$.

%Applying the Chernoff bound, there exists a constant fraction of $i$ such that $\| S_i\|_2 = \Omega(1)$. Then taking the union bound, with probability $.999$, there exists a constant fraction of $i$ such that $| \langle (SA)_i,\delta \rangle | \leq \Theta(\log d) \|\delta\|_1$.
Therefore, with probability $.99$, $\forall y\in\mathbb{R}^t$, due to the pigeonhole principle, there is a $3/10+9/10-1=1/5$ fraction of $i$ such that
\begin{align*}
 & ~ | \langle (SA)_i,y \rangle | \\
 = & ~ \| y\|_1 \cdot  | \langle (SA)_i, y/\| y\|_1 \rangle | \\
 = & ~ \| y\|_1 \cdot  | \langle (SA)_i, x + \delta \rangle | & \text{~by~} y/\| y\|_1 = x+\delta \\
\geq & ~ \| y\|_1 \cdot \bigl( | \langle (SB)_i,x \rangle + \langle (SI)_i,x \rangle |-| \langle (SB)_1, \delta \rangle+\langle(SI)_i, \delta\rangle | \bigr) & \text{~by~triangle~inequality}\\
\geq & ~ \|y\|_1 \bigl( \Omega(1)- \epsilon  O(t\log t\log d) \bigr) \\
\geq & ~ \|y\|_1 \Omega(1).
\end{align*}

So $\|y\|_1$ should be $O(\log d)$. Otherwise, the total cost on this $1/5$ fraction of coordinates is at least $\Omega(d\log d)$.

Combining Property (\RN{1}) and (\RN{2}) completes the proof.
\end{proof}

Next, we need to show the following claim is true,
\begin{claim}\label{cla:d_cauchy_distributed_logd_levels}
For any $d$ independent Cauchy random variables $x_1, x_2,\cdots, x_d$ from $C(0,1)$, with probability $1-1/\poly(t\log d)$, for any $j\in [1,2,\cdots, \log d- \Theta( \log\log (t\log d)) ]$, there are $\Omega(d/2^j)$ variables belonging to $( 2^{j}, 2^{j+1} ]$.
\end{claim}
\begin{proof}
For each Cauchy random variable $x_i$, we have for any $j\in [1,2,\cdots, \Theta(\log d)]$,
\begin{align*}
\Pr \biggl[ |x_i| \in ( 2^{j}, 2^{j+1} ] \biggr]  = \Theta ( 1/2^{j} ).
\end{align*}

We define the indicator random variable $z_{i,j}$
\begin{align*}
z_{i,j} =
\begin{cases}
1 & \text{if~} |x_i| \in ( 2^{j}, 2^{j+1}], \\
0 & \text{otherwise.}
\end{cases}
\end{align*}
We define $z_j = \sum_{i=1}^d z_{i,j}$. It is clear that $\E[z_{i,j}] = \Theta(1/2^j)$.
We use a Chernoff bound,
\begin{align*}
\Pr[ X < (1-\delta) \mu ] < \left( \frac{e^{-\delta}}{ (1-\delta)^{1-\delta}} \right)^{\mu},
\end{align*}
and set $X=z_j$, $\delta=1/2$, $\mu = \E[ z_j]$. Then, this probability is at most $2^{-\Omega(\mu)} = 2^{-\Omega(\E[z_j])}$. For any $j\in [1, \log d- \Theta(\log\log (t\log d) )]$, we have $\E[z_j] = d \E[z_{ij}] =\Omega(d/2^j) = \Omega(\log (t\log d))$. Thus, this probability is at most $1/\poly(t\log d)$. Overall, we have
\begin{align*}
\Pr \biggl[ z_j \gtrsim d/2^j \biggr] \geq 1-1/\poly(t\log d).
\end{align*}
Taking a union bound over $\Theta(\log d)$ such $j$, we complete the proof.

\end{proof}

\begin{claim}\label{cla:constant_fraction_Rj_is_good}
Let $1>c_1>c_2>1/3>0$ be three arbitrary constants. We fix the first column of $(SI)^\top \in \mathbb{R}^{d\times t}$. All the rows are grouped together according to the value in the first column. Let $R_j$ be the set of rows for which the entry in the first column is $\in (2^j, 2^{j+1}]$. With probability at least $1-O(1/\poly(t\log(d)))$, for any $j\in [c_2\log d, c_1\log d]$, the following event holds. There exist $\Omega(d/2^j)$ rows such that the first coordinate $\in (2^j, 2^{j+1}]$ and all the other coordinates are at most $O(d^{1/3})$. %%%maybe $d^{1/3}$
\end{claim}
\begin{proof}
Let $R_j$ be a subset of rows of $S^\top$, such that for any row in $R_j$, the first coordinate is $\in (2^j, 2^{j+1}]$. By Claim \ref{cla:d_cauchy_distributed_logd_levels}, we have that with probability $1-1/\poly(t\log d)$, for any $j\in [c_2\log d, c_1\log d]$, $|R_j| \geq \Omega(d/2^j)$.  %and the all the remaining coordinates are at most $O(d^{1/3}})$.
We want to show that for any $j\in [c_2\log d, c_1\log d]$, there exists a constant fraction of rows in $R_j$ such that, the remaining coordinates are at most $O(d^{1/3})$.

For a fixed row in $R_j$ and a fixed coordinate in that row, the probability that the absolute value of this entry is at least $\Omega(d^{1/3})$ is at most $O(1/d^{1/3})$. By taking a union bound, with probability at least $1-O(t/d^{1/3})$, the row has every coordinate of absolute value at most $O(d^{1/3})$.

By applying a Chernoff bound, for a fixed subset $R_j$ of rows, the probability that there is a constant fraction of these rows for which every coordinate except the first one in absolute value is at most $O(d^{1/3})$, is at least $1-2^{-\Theta(|R_j|)}\geq 1-2^{-\Theta(d^{1-c_1})}\geq 1-O(1/\poly(t\log(d)))$.

%For a fixed set of rows $R_j$, we consider a fixed column $i$ (that is not the first column), applying the Chernoff bound,
%\begin{align*}
%\Pr\bigl[\exists \text{~a~constant~fraction~of~}R_j\text{~s.t.}~ | i ~\text{coordinate} | \leq O(1) \bigr] \geq 1-2^{-\Omega(|R_j|)} \geq 1-1/\poly(d)
%\end{align*}
After taking a union over all the $j$, we complete the proof.
\end{proof}

\begin{claim}\label{cla:SB_is_at_most_logd}
With probability at least $1-O(1/t)$, the absolute value of any coordinate in column vector $(SB)_1 \in \mathbb{R}^{t\times 1}$ is at most $O(t^2\log d)$.
\end{claim}
\begin{proof}
Because each entry is sampled from a Cauchy distribution $C(0,1)$ scaled by $O(\log d)$, then with probability at most $O(1/t^2)$, the absolute value of one coordinate is at least $\Omega(t^2\log d)$. Taking a union over all the coordinates, we complete the proof.
\end{proof}

Because $\|x\|_1 \leq O(\log d), \| x\|_2 \geq \Omega(1/t\log t)$, there exists one coordinate $i$ such that the absolute value of it is at least $\Omega(1/t\log t)$ and all the other coordinates are most $O(\log d)$. We can assume that $i=1$ for now. (To remove this assumption, we can take a union bound over all the possibilities for $i\in [t]$.) According to Claim \ref{cla:constant_fraction_Rj_is_good} and \ref{cla:SB_is_at_most_logd}, let $\wh{R}_j$ denote a subset of $R_j$, which is a ``good'' constant fraction of $R_j$.  Considering the $\ell_1$-norm of all coordinates $l \in \wh{R}_j \subset R_j \subset [d]$, we have
\begin{align*}
&\sum_{l \in \wh{R}_j } | ( (SA)^\top x)_l - O(\log d) |\\
\geq &  \sum_{l \in \wh{R}_j } |  \langle(SA)_l,x\rangle - O(\log d) |\\
\geq &  \sum_{l \in \wh{R}_j } \left(|(SI)_{1,l}\cdot x_1|-\sum_{j=2}^t|(SI)_{j,l}\cdot x_j|-|\langle (SB)_l,x\rangle|-O(\log d)\right)\\
\geq & \sum_{l \in \wh{R}_j } \left(\Omega(\frac{2^j}{t\log t})-O(td^{1/3}\log d)-|\langle (SB)_1,x\rangle|-O(\log d)\right)\\
\geq & \sum_{l \in \wh{R}_j } \left(\Omega(\frac{2^j}{t\log t})-O(td^{1/3}\log d)-\|(SB)_1\|_1\|x\|_1-O(\log d)\right)\\
\geq & \sum_{l \in \wh{R}_j } \left(\Omega(\frac{2^j}{t\log t})-O(td^{1/3}\log d)-O(t\log t\log^2 d)-O(\log d)\right)\\
\gtrsim & \sum_{l \in \wh{R}_j }2^j/(t\log t)\\
\gtrsim & d/2^j\cdot2^j/(t\log t)\\
\gtrsim & d/(t\log t).
\end{align*}
The second inequality follows by the triangle inequality. The third inequality follows by $(SB)_1=(SB)_l$, $|x_1|=\Omega(1/t\log t),(SI)_{1,l}\in[2^j,2^{j+1})$, and $\forall j\not=1, |x_j|<O(\log d),(SI)_{j,l}\leq O(d^{1/3})$. The fourth inequality follows by Cauchy-Schwarz and $\|\cdot\|_2\leq \|\cdot\|_1$. The fifth inequality follows by Equation~(\ref{eq:upper_on_SB}) and Claim~\ref{cla:upper_lower_x}. The sixth inequality follows by $t=o(\log d)$ where $c$ is a constant smaller than $1/3$ and $2^j\geq d^{c_2}>\poly(t)$. The seventh inequality follows from $|\wh{R}_j|\geq \Omega(d/2^j)$.

Since there are $c1-c2$ different $j$, the total cost is $\Omega(d\log d/(t\log t))$. The gap then is $\Omega(\log d/(t\log t))$.
This completes the proof of Theorem \ref{thm:hardinstance_cauchy}.
\end{proof}

\subsection{Hard instance for row subset selection}\label{sec:hardinstance_row}

%Regarding selecting a small subset $R$ of rows of an $n \times d$ matrix $A$ so that their span contains a 2-approximate rank-1 approximation with entrywise 1-norm error, I think one can show you need $|R| = \Omega(d)$. As discussed with Zhao today it seems we can achieve $\wt{O}(d) \poly(k)$ rows containing a $\rank$-$k$ $O(\log d) \poly(k)$-approximation in their span, so this is not that far off from the upper bound.

\begin{theorem}
Let $\epsilon \in (0,1)$. There exists a value $k\geq 1$ and matrix $A\in \mathbb{R}^{ (d-1) \times d}$ such that, for any subset $R$ of rows of $A$, letting $B$ be the $\ell_1$-projection of each row of $A$ onto $R$, we have  
\begin{align*}
\| A - B \|_1 > (2-\epsilon)\underset{\rank-k ~ A' }{\min} \| A' - A\|_1,
\end{align*}
unless $|R| \gtrsim \epsilon d$.
\end{theorem}
\begin{proof}

We construct the $(d-1) \times d$ matrix $A$ in the following way. The first column is all $1$s, and then the remaining $(d-1) \times (d-1)$ matrix is the identity matrix.

\begin{align*}
A = \begin{bmatrix}
1 & 1 & 0 & 0 & \cdots & 0 \\
1 & 0 & 1 & 0 & \cdots & 0 \\
1 & 0 & 0 & 1 & \cdots & 0 \\
\cdots & \cdots & \cdots & \cdots & \cdots & \cdots \\
1 & 0 & 0 & 0 & \cdots & 1 \\
\end{bmatrix}.
\end{align*}
Note that $\underset{\rank-1~A'}{\min}\|A- A'\|_1  = \OPT \leq d-1$, since one could choose $A'$ to have the first column all $1$s and remaining entries $0$. On the other hand, consider any subset of $r$ rows. We can permute the columns and preserve the entrywise $\ell_1$-norm so w.l.o.g., we can take the first $r$ rows for $R$.

%Consider any linear combination $v$ of these first $r$ rows. Note that the cost of using a rank-1 approximation with scalar multiples of v in each row is at least $d-(r+1)$, which is what we pay on the last $d-(r+1)$ columns, so let us restrict attention to the first $r+1$ columns. Further, let us restrict to the bottom $(d-1-r)$ rows of the first $r+1$ columns, that is, we will not charge the algorithm for its mistakes on the first $r$ rows.

Because we are taking the first $r$ rows, we do not pay any cost on the first rows. To minimize the cost on the $i$-th row, where $i\in \{r+1, r+2, \dotsc, d\}$, let $v_i$ denote the row vector we use for the $i$-th row. Then $v_i$ can be written as a linear combination of the first $r$ rows $\{A^1, A^2, \dotsc, A^r \}$,
\begin{align*}
v_i = \sum_{j=1}^r \alpha_{i,j} A^j.
\end{align*}
%Suppose v is $\alpha_1$ times the first row $A^1$, plus $\alpha_2$ times the second row $A^2$, etc.
%$v = \overset{r}{\underset{j=1}{\sum}} \alpha_{j} A^j$
 Then the cost of using $v_i$ to approximate the $i$-th row of A is:
\begin{align*}
\| v_i - A^i \|_1 & = ~(\text{cost~in~1st~col}) + (\text{~cost~in~2nd,3rd,}\dotsc,r+1\text{th~cols}) + (\text{~cost~in~}i+1\text{th~col})  \\
 & = ~ | \overset{r}{\underset{j=1}{\sum}} \alpha_{i,j} - 1| +  \overset{r}{\underset{j=1}{\sum}} |\alpha_{i,j} | + 1 \\
 & \geq ~ |  \overset{r}{\underset{j=1}{\sum}} \alpha_{i,j} - 1 -  \overset{r}{\underset{j=1}{\sum}} \alpha_{i,j} | +1 \text{~by~triangle~inequality} \\
 & = ~2.
\end{align*}
%  Suppose $|\sum_j \alpha_j-1| = \epsilon \leq 1$. Note, otherwise, the cost is already at least 1 on the $i$-th row. It follows that $\sum_j |\alpha_j| >= 1-\epsilon$.
Hence, the cost of using $v_i$ to approximate the $i$-th row of $A$ is at least $2$. So in total, across these $(d-1-r)$ rows, the algorithm pays at least $2(d-1-r)$ cost, which needs to be at most $C (d-1)$, and therefore $r \geq (d-1)(1-C/2) = \Omega(\epsilon d)$. Choosing $C = 2-\epsilon$ completes the proof.

\end{proof}

\begin{theorem}
There exists a value $k\geq 1$ and matrix $A\in \mathbb{R}^{(d-1) \times d}$ such that, there is no algorithm that is able to output a $\rank-k$ matrix $B$ in the row span of $A$ satisfying
\begin{align*}
%\| A - B \|_1 < 2(1-\frac{1}{d-1}) \underset{\rank-k ~ A' }{\min} \| A' - A\|_1. %% Not sure how to prove the tight version.
\| A - B \|_1 < 2(1- \Theta(\frac{1}{d}) ) \underset{\rank-k ~ A' }{\min} \| A' - A\|_1.
\end{align*}
\end{theorem}

\begin{proof}
We use the same matrix $A$ as in the previous theorem.
\begin{align*}
A = \begin{bmatrix}
1 & 1 & 0 & 0 & \cdots & 0 \\
1 & 0 & 1 & 0 & \cdots & 0 \\
1 & 0 & 0 & 1 & \cdots & 0 \\
\cdots & \cdots & \cdots & \cdots & \cdots & \cdots \\
1 & 0 & 0 & 0 & \cdots & 1 \\
\end{bmatrix}.
\end{align*}

Let vector $v$ be $(\sum_{i=1}^{d-1} \beta_i, \beta_1, \dotsc, \beta_{d-1})$. For the $i$-th row of $A$, we use $\alpha_j \cdot v$ to cancel the cost of that row, where $\alpha_j$ is a scalar. Then for any $\beta_1, \dotsc, \beta_{d-1}, \alpha_1, \dotsc, \alpha_{d-1}$, we can compute the entire residual cost,
\begin{align*}
f(\alpha,\beta)= & ~ \sum_{j=1}^{d-1} \| A^j - v  \alpha_j \|_1 = ~ \sum_{j=1}^{d-1} \bigl( |1-\alpha_j (\sum_{i=1}^{d-1} \beta_i) | + |1 - \alpha_j \beta_j| + \sum_{i\neq j} |\alpha_j \beta_i | \bigr).
\end{align*}
In the next few paragraphs, we will show that optimizing $f(\alpha,\beta)$ over some extra constraints for $\alpha,\beta$ does not change the optimality.
Without loss of generality, we can assume that $\sum_{i=1}^{d-1}\beta_i = 1$. If not we can just rescale all the $\alpha_i$. Consider a fixed index $j$. All the terms related to $\alpha_j$ and $\beta_j$ are,
\begin{align*}
|1-\alpha_j (\sum_{i=1}^{d-1} \beta_i)| + |1-\alpha_j \beta_j| + \sum_{i\neq j} |\alpha_j \beta_i| + |\alpha_i \beta_j|.
\end{align*}
%Because of $\sum_{i=1}^{d-1} \beta_i=1$, to minimize the entire cost, then choosing $\alpha_j,\beta_j \geq 0$ is always better than all the other choices.
We first show a simple proof by assuming $\beta_j \geq 0$. Later, we will prove the general version which does not have any assumptions.
\paragraph{Handling a special case}
We can optimize $f(\alpha,\beta)$ in the following sense,
\begin{align*}
\min ~& f(\alpha,\beta) \\
\text{~s.t.~} & \sum_{j=1}^{d-1} \beta_j = 1\\
& \alpha_j \geq 0, \beta_j \geq 0, \forall j\in [d-1].
\end{align*}

For each $j$, we consider three cases.
Case I, if $\alpha_j \geq \frac{1}{\beta_j}$, then
\begin{align*}
& ~ |1-\alpha_j| + |1-\alpha_j \beta_j| + \sum_{i\neq j} \alpha_j \beta_i = ~ \alpha_j -1 + \alpha_j \beta_j -1 +\sum_{i\neq j} \alpha_j \beta_i =  2\alpha_j -2 \geq  2(1/\beta_j -1) \geq 2(1-\beta_j),
\end{align*}
where the last step follows by $\beta_j + 1/\beta_j \geq 2$.
Case II, if $ 1 \leq \alpha_j < 1/\beta_j$,  then
\begin{align*}
& ~ |1-\alpha_j| + |1-\alpha_j \beta_j| + \sum_{i\neq j} \alpha_j \beta_i= ~ \alpha_j -1 + 1- \alpha_j \beta_j +\sum_{i\neq j} \alpha_j \beta_i = 2\alpha_j (1-\beta_j) \geq 2(1-\beta_j).
\end{align*}
Case III, if $\alpha_j < 1$, then
\begin{align*}
& ~ |1-\alpha_j| + |1-\alpha_j \beta_j| + \sum_{i\neq j} \alpha_j \beta_i= 1- \alpha_j + 1- \alpha_j \beta_j +\sum_{i\neq j} \alpha_j \beta_i =2(1-\alpha_j \beta_j) \geq 2(1-\beta_j).
\end{align*}
%For each $j\in\overline{S}$, we also consider the three cases. Case I, if $\alpha_j \geq 1$,
%\begin{align*}
%|1-\alpha_j| + |1-\alpha_j \beta_j| + \sum_{i\neq j} |\alpha_j \beta_i| = \alpha_j -1 + \alpha_j |\beta_j| +1 + \alpha_j \sum_{i\neq j} | \beta_i| \geq 1 + \sum_{i} |\beta_i|
%\end{align*}
% Case II, if $ 0 \leq \alpha_j < 1$,
%\begin{align*}
%|1-\alpha_j| + |1-\alpha_j \beta_j| + \sum_{i\neq j} |\alpha_j \beta_i| = \alpha_j -1 + 1+\alpha_j |\beta_j| + \sum_{i\neq j} \alpha_j |\beta_i|
%\end{align*}

Putting it all together, we have
\begin{align*}
f(\alpha,\beta) \geq \sum_{j=1}^{d-1} 2(1-\beta_j) = 2(d-1) - 2\sum_{j=1}^{d-1} \beta_j = 2(d-2).
\end{align*}

\paragraph{To handle the case where $\beta_i$ can be negative}
Without loss of generality, we can assume that $\sum_{i=1}^{d-1} \beta_i = 1$. Notice that we can also assume $\beta_i\neq 0$, otherwise it means we do not choose that row. We split all $\beta_i$ into two disjoint sets $S$ and $\overline{S}$. For any $i\in S$, $\beta_i > 0$ and for any $i\in \overline{S}$, $\beta_i <0$.

As a first step, we discuss the case when all the $j$ are in set $S$.
Case \RN{1}, if $1-\alpha_j \sum_{i=1}^{d-1} \beta_i <0$ and $1-\alpha_j \beta_j<0$, then it means $\alpha_j \geq \max(\frac{1}{\sum_{i=1}^{d-1} \beta_i} ,\frac{1}{\beta}_j )$. The cost of that row is,
\begin{align*}
 = & ~ \alpha_j \sum_{i=1}^{d-1}\beta_i -1 + \alpha_j \beta_j -1 + \sum_{i\neq j} |\alpha_j \beta_i| \\
 = & ~ 2 \alpha_j \beta_j + 2 \alpha_j \sum_{i\in S\backslash j} \beta_i - 2.
\end{align*}
If $\beta_j\geq 1$, then $\alpha_j \geq 1$. The cost is at least $2\sum_{i\in S \backslash j} \beta_i$. If $\beta_j < 1$, then $\alpha_j \geq 1/\beta_j$, and the cost is at least $2 \frac{1}{\beta_j} \sum_{i\in S \backslash j} \beta_i$. If there are $C$ such $j$ in Case \RN{1}, then the total cost of Case I is at least $0$ if $C=1$, and $2(C-1)$ if $C\geq 2$.

Case \RN{2}, if $1-\alpha_j \sum_{i=1}^{d-1} \beta_i < 0$ and $1-\alpha_j \beta_j >0$, then it means $1< \alpha_j <1/\beta_j$. The cost of that row is,
\begin{align*}
= & ~\alpha_j \sum_{i=1}^{d-1} \beta_i - 1 + 1 -\alpha_j \beta_j + |\alpha_j| \sum_{i\neq j} |\beta_j| \\
= & ~ 2 \alpha_j \sum_{i\in S \backslash j} \beta_i \\
\geq & ~ 2 \sum_{i\in S \backslash j} \beta_i.
\end{align*}
Similarly to before, if there are $C$ such $j$ in Case \RN{2}, then the total cost of Case \RN{2} is at least $0$ if $C=1$, and $2(C-1)$ if $C\geq 2$.

Case \RN{3}, if $1-\alpha_j \sum_{i=1}^{d-1} \beta_i > 0$ and $1-\alpha_j \beta_j < 0$, then it means $1/\beta_j < \alpha_j < \frac{1}{\sum_{i=1}^{d-1} \beta_i}$. The cost of the row is,
\begin{align*}
& ~ = 1-\alpha_j \sum_{i=1}^{d-1} \beta_i + \alpha_j \beta_j -1 +  |\alpha_j| \sum_{i\neq j} |\beta_j| \\
& ~ = 2 |\alpha_j| \sum_{i\in \overline{S} } |\beta_i| \\
& ~ = 2 |\alpha_j| ( \sum_{i\in S } |\beta_i| -1) \\
& ~ \geq 2 \frac{1}{\beta_j} \sum_{i\in S\backslash j} \beta_i.
\end{align*}
If there are $C$ such $j$ in Case \RN{3}, then the total cost of Case \RN{3} is at least $0$ if $C=1$, and $2(C\cdot(C-1))$ if $C\geq 2$.

Case \RN{4}, if $1-\alpha_j \sum_{i=1}^{d-1} \beta_i > 0$ and $1-\alpha_j \beta_j > 0$, then it means $ \alpha_j \leq \min (1/\beta_j , \frac{1}{\sum_{i=1}^{d-1} \beta_i} )$. The cost for the case $\alpha_j <0$ is larger than the cost for case $\alpha>0$. Thus we can ignore the case $\alpha_j<0$. The cost of the row is,
\begin{align*}
& ~ = 1-\alpha_j \sum_{i=1}^{d-1} \beta_i + 1- \alpha_j \beta_j +  |\alpha_j| \sum_{i\neq j} |\beta_j| \\
& ~ = 2 -2 \alpha_j \beta_j + 2\alpha_j \sum_{i\in \overline{S} } |\beta_i |.
\end{align*}
If $\beta_j < \sum_{i\in \overline{S}} |\beta_i|$, we know that the cost is at least $2$. Otherwise, using $\alpha_j \leq 1$, we have the cost is at least $ 2- 2\beta_j + \sum_{i\in \overline{S}} |\beta_i|$. Let $T$ denote the set of those $j$ with $\beta_j \geq \sum_{i\in \overline{S} }|\beta_i|$. If $|T|=1$, we know that cost is at least $0$. If $|T|\geq 2$, then the cost is at least
\begin{align*}
 &= ~\sum_{j\in T} (2 -2\beta_j +2\sum_{i\in \overline{S}} |\beta_i| ) \\
& \geq ~ 2 |T| - 2\sum_{j\in S} \beta_j + 2|T|\sum_{i\in \overline{S}} |\beta_i| \\
& \geq ~ 2 |T| - 2+ 2(|T|-1)\sum_{i\in \overline{S}} |\beta_i| \\
& \geq ~ 2(C-1).
\end{align*}
Now we discuss the case where $j\in \overline{S}$, which means $\beta_j <0$.

Case \RN{5} if $1-\alpha_j \sum_{i=1}^{d-1} \beta_i <0$ and $1-\alpha_j \beta_j<0$, then it means $\alpha_j > 1/\sum_{i=1}^{d-1} \beta_i$ and $\alpha_j < 1/\beta_j$. Notice that this case will never happen, because $\sum_{i=1}^{d-1} \beta_i = 1$ and $\alpha_j <0$.

Case \RN{6} if $1-\alpha_j \sum_{i=1}^{d-1} \beta_i <0$ and $1-\alpha_j \beta_j<0$, then it means $\alpha_j \geq \max(1/\sum_{i=1}^{d-1} \beta_i, 1/\beta_j)$. Because of $\beta_j<0$, then $\alpha_j \geq 1$. The cost of that is,
\begin{align*}
 =& ~\alpha_j \sum_{i=1}^{d-1} \beta_i - 1 + 1- \alpha_j \beta_j + |\alpha_j| \sum_{i\neq j} |\beta_i|\\
= & ~ 2\alpha_j \sum_{i\in S} |\beta_i|\\
\geq & ~ 2. &\text{~by~}\alpha_j \geq 1 \text{~and~} \sum_{i\in S} |\beta_i| \geq 1 .
\end{align*}

Case \RN{7} if $1-\alpha_j \sum_{i=1}^{d-1} \beta_i > 0$ and $1-\alpha_j \beta_j<0$, then it means $\alpha_j < \min(1/\beta_j, 1/\sum_{i=1}^{d-1} \beta_i)$. Because $\beta_j < 0$ and $\sum_{i=1}^{d-1} \beta_i=1$, thus $\alpha_j < 1/\beta_j$. The cost of that row is,
\begin{align*}
=& ~ 1- \alpha_j \sum_{i=1}^{d-1} \beta_i  +  \alpha_j \beta_j - 1 + |\alpha_j| \sum_{i\neq j} |\beta_i|\\
=& ~ 2 |\alpha_j| \sum_{i\in \overline{S}\backslash j} |\beta_i| \\
\geq & ~ 2 \frac{1}{ |\beta_j| } \sum_{i\in \overline{S}\backslash j} |\beta_i|.
\end{align*}
If there are $C$ such $j$ in Case \RN{7}, then the total cost of Case \RN{7} is at least $0$ if $C=1$, and $2(C\cdot(C-1))$ if $C\geq 2$.

Case \RN{8} if $1-\alpha_j \sum_{i=1}^{d-1} \beta_i > 0$ and $1-\alpha_j \beta_j>0$, then it means $1/\beta_j<\alpha_j < 1/\sum_{i=1}^{d-1} \beta_i$. The cost of that row,
\begin{align*}
=& ~ 1- \alpha_j \sum_{i=1}^{d-1} \beta_i  +  1- \alpha_j \beta_j + |\alpha_j| \sum_{i\neq j} |\beta_i|\\
= & ~ 2 - 2\alpha_j \beta_j + 2|\alpha_j| \sum_{i \in \overline{S}\backslash j} |\beta_i|.
\end{align*}
If $\alpha_j >0$, then the cost is always at least $2$. If $\alpha_j <0$, the cost is at least,
\begin{align*}
 2 - 2| \alpha_j \beta_j| + 2|\alpha_j| \sum_{i \in \overline{S}\backslash j} |\beta_i| .
\end{align*}
If $\sum_{i \in \overline{S}\backslash j} |\beta_i| > |\beta_j|$, we also have cost at least $2$. Otherwise, we have
\begin{align*}
 2 - 2 | \alpha_j| (  |\beta_j| -\sum_{i \in \overline{S}\backslash j} |\beta_i| )
 =
 \begin{cases}
    2 - 2|\beta_j| + 2 \sum_{i \in \overline{S}\backslash j} |\beta_i|      & \quad \text{if } 1/|\beta_j| \leq 1,\\
    2 \frac{1}{|\beta_j|}  \sum_{i \in \overline{S}\backslash j} |\beta_i|  & \quad \text{if }1/|\beta_j| > 1.\\
  \end{cases}
\end{align*}
Let $C$ denote the number of such $j$ with $1/|\beta_j| \leq 1$. If $C=1$, the cost is at least $0$. If $C\geq 2$, the cost is at least $2C$. Let $C'$ denote the number of such $j$ with $1/|\beta_j| > 1$. If $C'=1$, the cost is at least $0$, if $C'\geq 2$, the cost is at least $2(C'(C'-1))$.
Overall, putting all the eight cases together, we complete the proof.
\end{proof}

%%%%%%%%%%%%%%%%%%%%%%%%%%%%%%%%%%%%%%%%%%%%%%%%%%%%%%%%%%%%%%%%%%%%%%%%%%%%%%%%%%%%%%%
%%%%%%%%%%%%%%%%%%%%%%%%%%%%%%%%%%%%%%%%%%%%%%%%%%%%%%%%%%%%%%%%%%%%%%%%%%%%%%%%%%%%%%%

\subsection{Hard instance for oblivious subspace embedding and more row subset selection}\label{sec:hardinstance_ose}

The goal in this section is to prove Theorem~\ref{thm:use_yaominmax}. By applying Yao's minmax principle, it suffices to prove Theorem~\ref{thm:hard_distribution}.

\subsubsection{Definitions}
We first give the definition of total variation distance and Kullback-Leibler divergence.
\begin{definition}{\cite{lpw09,v14}}
The total variation distance between two probability measures $P$ and $Q$ on the measurable space (${\cal X}$, ${\cal F}$) is defined as,
\begin{align*}
 D_{\TV}(P,Q) = \underset{ {\cal A} \in {\cal F}  }{\mathrm{sup}} |  P({\cal A}) -  Q({\cal A}) |.
\end{align*}
The Kullback-Leibler(~$\KL$) divergence of $P$ and $Q$ is defined as,
\begin{align*}
D_{\KL} (P || Q) = \E_{P} \left( \log \frac{\mathrm{d} P}{\mathrm{d} Q}\right) = \int_{\cal X} \left( \log \frac{\mathrm{d} P}{\mathrm{d} Q} \right) \mathrm{d} P.
\end{align*}
\end{definition}

\begin{lemma}\label{lem:KLandTV}{\cite{p60,t09,ck11}}
Pinsker's inequality states that, if $P$ and $Q$ are two probability distributions on a measurable (${\cal X}$, ${\cal F}$), then
\begin{align*}
D_{\TV}(P,Q) \leq \sqrt{\frac{1}{2} D_{\KL} (P || Q) }.
\end{align*}
\end{lemma}

\begin{lemma}\label{lem:half_gaussian}
For any Gaussian random variable $x\sim N(0,\sigma^2)$, we have, for any $a>0$
\begin{align*}
\Pr[|x| <a] \lesssim \frac{a}{\sigma }.
\end{align*}
\end{lemma}
\begin{proof}
\begin{align*}
\Pr[|x|< a] & ~ = \text{erf}(\frac{a}{\sigma \sqrt{2} }) \\
& ~ \leq (1-e^{-\frac{-4a^2}{2\sigma^2 \pi} })^\frac{1}{2} & \text{~by~} 1-\exp(-4x^2/\pi) \geq \text{erf}(x)^2\\
& ~ \leq (\frac{4 a^2}{2\sigma^2 \pi })^\frac{1}{2} &\text{~by~} 1- e^{-x} \leq  x \\
& ~ \lesssim \frac{a}{\sigma}.
\end{align*}
\end{proof}

\begin{lemma}\label{lem:equivalent_distribution}
Let $V\in\mathbb{R}^{n\times m}$ be a matrix with orthonormal columns. Let $\mathcal{H}'$ be a distribution over $V^\top A$, where $A\in \mathbb{R}^{n\times k}$ is a random matrix with each entry i.i.d. Gaussian $N(0,1)$. Denote $\mathcal{H}$ as a distribution over $H\in\mathbb{R}^{m\times k}$, where each entry of $H$ is drawn from i.i.d. Gaussian $N(0,1)$. Then, $\mathcal{H}$ and $\mathcal{H}'$ are the same distribution.
\end{lemma}
\begin{proof}
It is clear that each entry of $V^\top A$ is a random Gaussian variable from $N(0,1)$, so our goal is to prove the entries of $V^\top A$ are fully independent. Since the the $j^{\text{th}}$ column of $V^\top A$ only depends on $A_j$, the variables from different columns are fully independent. Now, we look at one column of $x=V^\top A_j$. The density function is
\begin{align*}
f(x)=\frac{1}{\sqrt{(2\pi)^m |V^\top V|}}\exp\left(-\frac{1}{2} x^\top V^\top V x\right)=\frac{1}{\sqrt{(2\pi)^m }}\exp\left(-\frac{1}{2} x^\top  x\right),
\end{align*}
which is exactly the density function of $N(0,I_m)$. Thus $x\in N(0,I_m)$. Therefore, all the entries of $V^\top A$ are fully independent.
\end{proof}

\begin{lemma}[Matrix Determinant Lemma]\label{lem:Matrix_determinant_lemma}
Suppose $A\in\mathbb{R}^{n\times n}$ is an invertible matrix, and $u,v\in\mathbb{R}^n$. Then,
\begin{align*}
|A+uv^\top|=(1+v^\top A^{-1} u)|A|.
\end{align*}
\end{lemma}

\begin{lemma}[KL divergence between two multivariate Gaussians~\cite{pp08}\footnote{\url{http://stats.stackexchange.com/questions/60680/kl-divergence-between-two-multivariate-gaussians}}]\label{lem:KL_gaussian}
Given two $d$-dimensional multivariate Gaussian distribution $N(\mu_1,\Sigma_1)$ and $N(\mu_2,\Sigma_2)$, then
\begin{align*}
D_{\KL}(N(\mu_1||\Sigma_1),N(\mu_2,\Sigma_2))=\frac{1}{2}\left(\log\frac{|\Sigma_2|}{|\Sigma_1|}-d+\mathrm{tr}(\Sigma_2^{-1}\Sigma_1)+(\mu_2-\mu_1)^\top\Sigma_2^{-1}(\mu_2-\mu_1)\right).
\end{align*}
\end{lemma}

\begin{lemma}[Sherman-Morrison formula]\label{lem:sherman}
Suppose $A\in\mathbb{R}^{n\times n}$ is an invertible matrix and $u,v\in\mathbb{R}^n$. Suppose furthermore that $1+v^\top A^{-1}u\neq 0$. Then the Sherman-Morrison formula states that
\begin{align}
(A+uv^\top)^{-1}=A^{-1}-{A^{-1}uv^{\top}A^{-1} \over 1+v^{\top}A^{-1}u}.
\end{align}
\end{lemma}

\subsubsection{Main results}

\begin{lemma}\label{lem:tot_var_results}
Let $V\in \mathbb{R}^{n\times m}$ be a matrix with orthonormal columns, and let $A\in\mathbb{R}^{n\times k}$ be a random matrix with each entry drawn from i.i.d. Gaussian $N(0,1)$. We denote the distribution $\mathcal{D}_i$ over $D_i\in\mathbb{R}^{(m+1)\times k}$ where
\begin{align*}
 D_i=\begin{bmatrix}
 V^\top A \\ A^i
 \end{bmatrix}.
\end{align*}
 If $\|(V^\top)_i\|_2^2<\frac{1}{2}$, then
\begin{align*}
D_{\TV}(\mathcal{D}_i,\mathcal{G})\leq O(k\|V^i\|_2)+2^{-\Theta(k)},
\end{align*}
where $\mathcal{G}$ is a distribution over $G\in\mathbb{R}^{(m+1)\times k}$, where each entry of $G$ is drawn from the i.i.d. Gaussian $N(0,1)$.
\end{lemma}

\begin{proof}
Let $\mathcal{H}$ be a distribution over $H\in\mathbb{R}^{m\times k}$, where each entry of $H$ is an i.i.d. Gaussian $N(0,1)$. Let $\mathcal{H}'$ be a distribution over $V^\top A$. According to Lemma~\ref{lem:equivalent_distribution}, $\mathcal{H}$ and $\mathcal{H}'$ are the same distribution.

We define matrix $V^{\neg i} \in \mathbb{R}^{n \times m}$ as
\begin{align*}
V^{\neg i} =
\begin{bmatrix}
(V^1)^\top & (V^2)^\top & \cdots & (V^{i-1})^\top & 0 & (V^{i+1})^\top & \cdots & (V^n)^\top
\end{bmatrix}^\top,
\end{align*}
where $V^i$ denotes the $i^{\text{th}}$ row of $V\in \mathbb{R}^{n\times m}, \forall i\in[n]$.

Let $\mathcal{G}$ be a distribution over $G\in\mathbb{R}^{(m+1)\times k}$, where each entry of $G$ is an i.i.d. Gaussian $N(0,1)$.
Let $\mathcal{P}_i$ be a distribution over $P_i\in\mathbb{R}^{(m+1)\times k}$, where
\begin{align*}
P_i=\begin{bmatrix}(V^{\neg i})^\top A\\ A^i \end{bmatrix}.
\end{align*}
Let $\wh{\mathcal{P}}_i$ be a distribution over $\wh{P}_i\in\mathbb{R}^{m\times k}$, where $\wh{P}_i=(V^{\neg i})^\top A$.
Then we have:
\begin{align}\label{eq:tot_var_eq_chain}
D_{\TV}(\mathcal{P}_i,\mathcal{G}) = D_{\TV}(\wh{\mathcal{P}}_i,\mathcal{H'})=D_{\TV}(\wh{\mathcal{P}}_i,\mathcal{H}).
\end{align}
The first equality is because $(V^{\neg i})^\top A$ is independent from $A^i$. The second equality follows the Lemma~\ref{lem:equivalent_distribution}.

\begin{claim}\label{cla:small_KLdis}
If $\|V^i\|_2^2\leq \frac{1}{2}$,
\begin{align*}
D_{\KL}(\wh{\mathcal{P}}_i||\mathcal{H})= -\frac{k}{2} \left( \log  (1-\|V^i\|_2^2) +\|V^i\|_2^2  \right)\leq k\|V^i\|_2^2.
\end{align*}
\end{claim}
\begin{proof}
Let $\wh{P}_i\sim \wh{\mathcal{P}}_i, H\sim \mathcal{H}$. Notice that different columns of $\wh{P}_i$ are i.i.d, and all entries of $H$ are fully independent. We can look at one column of $\wh{P}_i$ and $H$. Since $\wh{P}_i=(V^{\neg i})^\top A$, it is easy to see that its column is drawn from $N(0,\Sigma_1)$ where $\Sigma_1=I_m-(V^i)^\top V^i$. Since $H$ is fully independent, each column of $H$ is drawn from $N(0,\Sigma_2)$ where $\Sigma_2=I_m$. Let $p(x)$ be the pdf of the column of $\wh{P}_i$, and let $q(x)$ be the pdf of the column of $H$. We have the following calculation~\cite{pp08}\footnote{\url{http://stats.stackexchange.com/questions/60680/kl-divergence-between-two-multivariate-gaussians}},
\begin{align*}
D_{\KL}(\wh{\mathcal{P}}_i||\mathcal{H})%&= k\int [ \log(p(x)) - \log(q(x))] p(x) \mathrm{d}x\\
%& = k\int \left( \frac{1}{2} \log \frac{|\Sigma_2|}{|\Sigma_1|} - \frac{1}{2}x^\top \Sigma_1^{-1} x+ \frac{1}{2} x^\top \Sigma_2^{-1} x \right) p(x) \mathrm{d} x \\
%& = \frac{k}{2} \log \frac{|\Sigma_2|}{|\Sigma_1|} - \frac{1}{2} \text{tr} (\E [ x x^\top ] \Sigma_1^{-1} ) + \frac{1}{2} \E[ x^\top \Sigma_2^{-1} x ] \\
%& = \frac{k}{2} \log \frac{|\Sigma_2|}{|\Sigma_1|} - \frac{1}{2} \text{tr} ( I_m )  + \frac{1}{2} \text{tr} ( \Sigma_2^{-1} \Sigma_1 ) \\
& = \frac{k}{2} \left( \log \frac{| \Sigma_2 |}{ |\Sigma_1| } - m + \text{tr} (\Sigma_2^{-1} \Sigma_1)  \right)\\
& = \frac{k}{2} \left( \log \frac{| I_m |}{ |I_m-(V^i)^\top V^i| } - m + \text{tr} (I_m-(V^i)^\top V^i)  \right)\\
& = \frac{k}{2} \left( -\log  |I_m-(V^i)^\top V^i|  - m + m -\|V^i\|_2^2  \right)\\
& = \frac{k}{2} \left( -\log  (1-\|V^i\|_2^2)  - m + m -\|V^i\|_2^2  \right)\\
& = -\frac{k}{2} \left( \log  (1-\|V^i\|_2^2) +\|V^i\|_2^2  \right)\\
& \leq -\frac{k}{2}\cdot 2\|V^i\|_2^2\\
& = k\|V^i\|_2^2.
\end{align*}
The first equality is due to Lemma~\ref{lem:KL_gaussian}. The sixth equality follows by $\Sigma_2=I_m$ and $\Sigma_1=I_m-(V^i)^\top V^i$. The eighth equality follows by Lemma~\ref{lem:Matrix_determinant_lemma}. The first inequality follows by $\log(1-x)+x\geq -2x$, when $0<x<1/2$.

\end{proof}

According to Lemma~\ref{lem:KLandTV}, we have
\begin{align}\label{eq:transdis}
D_{\TV}(\wh{\mathcal{P}}_i,\mathcal{H})&\leq \sqrt{\frac{1}{2}D_{\KL}(\wh{\mathcal{P}}_i||\mathcal{H})}\leq \sqrt{k}\|V^i\|_2.
\end{align}
Now, we want to argue that $D_{\TV}(\mathcal{D}_i,\mathcal{P}_i)$ is small, where $\mathcal{D}_i$ is a distribution over $D_i\in\mathbb{R}^{(m+1)\times k}$ that
\begin{align*}
 D_i=\begin{bmatrix}
 V^\top A \\ A^i
 \end{bmatrix}.
\end{align*}
For a fixed $x\in \mathbb{R}^{k}$, let $\wh{\mathcal{D}}_i(x)$ be a distribution over $\wh{D}_i(x)\in \mathbb{R}^{m\times k}$, where $\wh{D}_i(x)=(V^{\neg i})^\top A+(V^i)^\top x$. Let $p(x)$ be the pdf of $(A^i)^\top$, then
\begin{align}\label{eq:totdis}
D_{\TV}(\mathcal{D}_i,\mathcal{P}_i) &= \int D_{\TV}\left(\wh{\mathcal{D}}_i(x),\wh{\mathcal{P}}_i\right) p(x) \mathrm{d}x.
\end{align}

Now we look at the $j^{\text{th}}$ column of $\wh{D}_i(x)$ and the $j^{\text{th}}$ column of $\wh{P}_i$. The distribution of the previous one is over $N((V^i)^\top x_j,\Sigma_2)$, and the latter distribution as we said before is $N(0,\Sigma_2)$, where $\Sigma_2=I_m-(V^i)^\top V^i$. Now we can argue that the KL divergence between them is bounded:
\begin{align*}
&D_{\KL}(N((V^i)^\top x_j,\Sigma_2)||N(0,\Sigma_2))\\
=&\frac{1}{2}\left(\log\frac{|\Sigma_2|}{|\Sigma_2|}-m+\text{tr}(\Sigma_2^{-1}\Sigma_2)+x_j^2V^i\Sigma_2^{-1}(V^i)^{\top}\right)\\
=&\frac{1}{2}\left(-m+\text{tr}(I_m)+x_j^2V^i\Sigma_2^{-1}(V^i)^{\top}\right)\\
=&\frac{1}{2}x_j^2V^i\Sigma_2^{-1}(V^i)^{\top}\\
=&\frac{1}{2}x_j^2V^i(I_m-(V^i)^{\top} V^i)^{-1}(V^i)^{\top}\\
=&\frac{1}{2}x_j^2V^i\left(I_m+\frac{(V^i)^{\top}V^i}{1-V^i(V^i)^{\top}}\right)(V^i)^{\top}\\
=&\frac{1}{2}x_j^2\left(\|V^i\|_2^2+\frac{\|V^i\|_2^2}{1-\|V^i\|_2^2}\right)\\
=&\frac{1}{2}x_j^2\frac{\|V^i\|_2^2}{1-\|V^i\|_2^2}.
\end{align*}
The first equality is due to Lemma~\ref{lem:KL_gaussian}. The fourth equality follows by $\Sigma_2=I_m-(V^i)^\top V^i$. The fifth equality follows by Lemma~\ref{lem:sherman}.

By summing the KL divergence on all the columns up,
\begin{align*}
D_{\KL}(\wh{\mathcal{D}}_i(x)||\wh{\mathcal{P}}_i)=\frac{1}{2}\|x\|_2^2\frac{\|V^i\|_2^2}{1-\|V^i\|_2^2}.
\end{align*}

Applying Lemma~\ref{lem:KLandTV} again, we get
\begin{align*}
D_{\TV}(\wh{\mathcal{D}}_i(x),\wh{\mathcal{P}}_i)\leq \sqrt{\frac12D_{\KL}(\wh{\mathcal{D}}_i(x)||\wh{\mathcal{P}}_i)}=\frac{\|x\|_2\|V^i\|_2}{2\sqrt{1-\|V^i\|_2^2}}.
\end{align*}

Plugging it into Equation~(\ref{eq:totdis}), we get
\begin{align*}
D_{\TV}(\mathcal{D}_i,\mathcal{P}_i) &= \int D_{\TV}\left(\wh{\mathcal{D}}_i(x),\wh{\mathcal{P}}_i\right) p(x) \mathrm{d}x\\
&=\int_{\|x\|_2\leq 10k} D_{\TV}\left(\wh{\mathcal{D}}_i(x),\wh{\mathcal{P}}_i\right) p(x) \mathrm{d}x+\int_{\|x\|_2>10k}D_{\TV}\left(\wh{\mathcal{D}}_i(x),\wh{\mathcal{P}}_i\right) p(x) \mathrm{d}x\\
&\leq \int_{\|x\|_2\leq 10k} D_{\TV}\left(\wh{\mathcal{D}}_i(x),\wh{\mathcal{P}}_i\right) p(x) \mathrm{d}x+\int_{\|x\|_2>10k} p(x) \mathrm{d}x\\
&\leq \int_{\|x\|_2\leq 10k} \frac{\|x\|_2\|V^i\|_2}{2\sqrt{1-\|V^i\|_2^2}} p(x) \mathrm{d}x+\int_{\|x\|_2>10k} p(x) \mathrm{d}x\\
&\leq \frac{10k\|V^i\|_2}{2\sqrt{1-\|V^i\|_2^2}} +\int_{\|x\|_2>10k} p(x) \mathrm{d}x\\
&\leq \frac{10k\|V^i\|_2}{2\sqrt{1-\|V^i\|_2^2}} + 2^{-\Theta(k)}.
\end{align*}
The first inequality just follows from the fact that total variation distance is smaller than $1$. The second inequality is what we plugged in. The last inequality follows from the fact that $x\sim N(0,I_k)$ and from the tail bounds of a Gaussian.

Together with Equation~(\ref{eq:tot_var_eq_chain}) and Equation~(\ref{eq:transdis}), we can get
\begin{align*}
D_{\TV}(\mathcal{D}_i,\mathcal{G})& \leq D_{\TV}(\mathcal{D}_i,\mathcal{P}_i)+D_{\TV}(\mathcal{G},\mathcal{P}_i)\\
&=D_{\TV}(\mathcal{D}_i,\mathcal{P}_i)+D_{\TV}(\wh{\mathcal{P}}_i,\mathcal{H})\\
&\leq \frac{10k\|V^i\|_2}{2\sqrt{1-\|V^i\|_2^2}} + 2^{-\Theta(k)}+\sqrt{k}\|V^i\|_2\\
&\leq 10k\|V^i\|_2 + 2^{-\Theta(k)}+\sqrt{k}\|V^i\|_2.
\end{align*}
The last inequality follows by $\|V^i\|_2^2\leq \frac{1}{2}$. Then, we have completed the proof.

\end{proof}

\begin{lemma}\label{lem:spectral_of_A}
$A\in\mathbb{R}^{r\times k}$ ($r\geq k$) is a random matrix for which each entry is i.i.d. $N(0,1)$. With probability at least $1-e^{- \Theta(r) }$, the maximum singular value $\|A\|_2$ is at most $O(\sqrt{r})$.
\end{lemma}
\begin{proof}
Since $A\in\mathbb{R}^{k\times r}$ is a random matrix with each entry i.i.d. $N(0,1)$, this follows by standard arguments (Proposition 2.4 in \cite{RV10}). Since $r\geq k$, with probability at least $1-e^{-\Theta(r)}$, $\|A\|_2$ is at most $O(\sqrt{r})$.
\end{proof}

\begin{definition}
%For any $k\geq 1$, and any constants $c_2\geq c_1\geq 1$, let $k\leq r=O(k^{c_1}),r\leq n= O(k^{c_2})$,
Let $V\in\mathbb{R}^{n\times r}$ be a matrix with orthonormal columns, and let each entry of $A\in\mathbb{R}^{k\times r}$ be a random variable drawn from an i.i.d. Gaussian $N(0,1)$. Define event $\wh{\mathcal{E}}(A,V,\beta,\gamma)$ to be: $\forall y\in\mathbb{R}^n$, $\|y\|_1\leq O(k^\gamma)$ and each coordinate of $y$ has absolute value at most $1/k^\beta$, and also $AV^\top y$ has at most $O(k/\log k)$ coordinates with absolute value at least $\Omega(1/\log k)$, and $\|A\|_2\leq O(\sqrt{r})$.%where $y=y_0+y_1$ satisfies that every non-zero coordinate in $y_1$ has absolute value at most $1/k^\beta$, and every non-zero coordinate in $y_0$ has absolute value greater than $1/k^\beta$.
\end{definition}

\begin{lemma}\label{lem:useful_property}
For any $k\geq 1$, and any constants $c_2\geq c_1\geq 1$, let $k\leq r=O(k^{c_1}),r\leq n= O(k^{c_2})$, let $V\in\mathbb{R}^{n\times r}$ be a matrix with orthonormal columns, and let each entry of $A\in\mathbb{R}^{k\times r}$ be a random variable drawn from i.i.d. Gaussian $N(0,1)$. Furthermore, if $\beta$ and $\gamma$ are two constants which satisfy $\beta>\gamma>0$ and $\beta+\gamma<1$,
\begin{align*}
\Pr\biggl[\wh{\mathcal{E}}(A,V,\beta,\gamma)\biggr]\geq 1-2^{-\Theta(k)}.
\end{align*}
\end{lemma}
\begin{proof}
Due to Lemma~\ref{lem:spectral_of_A}, with probability at least $1-2^{-\Theta(r)}$, $\|A\|_2\leq O(\sqrt{r})$, so we can restrict attention to $\|A\|_2\leq O(\sqrt{r})$ in the following proof.

%For any $x\in\mathbb{R}^r$, we can find a $y\in\mathbb{R}^n$ such that $y=Vx$. Since $V\in\mathbb{R}^{n\times r}$ has orthonormal columns, $x=V^\top Vx=V^\top y$. Then our goal is to argue that with high probability if $\|AV^\top y-v\|_1\leq O(k^\gamma)$, then $\|y\|_1>\Omega(k^\gamma)$.
Take any $y \in \mathbb{R}^n$ which has each non-zero coordinate with absolute value at most $1/k^\beta$, %which can be expressed by $Vx$,
and write it as $y = \sum_{j=j_0}^{+\infty} y^j$, where the coordinates in $y^j$ have absolute value in the range $[2^{-j-1}, 2^{-j})$, and the supports of different $y^j$ are disjoint. Since each coordinate of $y$ is at most $1/k^\beta$, $2^{-{j_0}} = 1/k^{\beta}$. Since $\|y\|_1 \leq O(k^\gamma)$, the support size of $y^j \in \mathbb{R}^n$ is at most $s_j \leq O(2^{j+1} k^\gamma)$, so it follows that
\begin{align*}
\|y^j\|_2 \leq \sqrt{s_j  \cdot 2^{-2j}} \leq \sqrt{ 2^{-j+1} k^\gamma  }.
\end{align*}
We also know $\| y^j\|_2 \geq \sqrt{ s_j \cdot 2^{-2j-2} }$. Then we can conclude that $y^j$ has 2-norm $\Theta(2^{-j} \sqrt{s_j})$.
Now we state an $\varepsilon$-net for all the possible $y^j$: let $\varepsilon=O(1/(nrk^3))=O(1/(k^{c_1+c_2+3}))$. Let $\mathcal{N}_j\subset \mathbb{R}^n$ be the following:
\begin{align*}
\mathcal{N}_{j}=\left\{p\in\mathbb{R}^n \mid\ \exists q\in\mathbb{Z}^n,\text{~s.t.~}p=\varepsilon q, \|p\|_1\leq O(k^\gamma),\forall i\in [n],\text{~either~} |p_i|\in [2^{-j-1},2^{-j}) \text{~or~}p_i=0\right\}.
\end{align*}
Obviously, for any $y^j$, there exists $p\in\mathcal{N}_j$ such that $\|y^j-p\|_{\infty}\leq \varepsilon=O(1/(k^{c_1+c_2+3}))$, since $n\leq O(k^{c_2})$, $\|y^j-p\|_2\leq \|y^j-p\|_1\leq n\|y^j-p\|_{\infty}\leq O(1/k^{c_1+3})$. Now let us consider the size of $\mathcal{N}_j$. If $p\in \mathcal{N}_j$, the choice of one coordinate of $p$ is at most $2/\varepsilon+1$. And since $\|p\|_1\leq O(k^\gamma)$ and each coordinate of $p$ has absolute value at least $2^{-j-1}$, the number of supports of $p$ is at most $O(2^{j+1}k^\gamma)$. Therefore,
\begin{align*}
|\mathcal{N}_j|&\leq (n+1)^{O(2^{j+1}k^\gamma)}\cdot (2/\varepsilon+1)^{O(2^{j+1}k^\gamma)}\\
&\leq  2^{O(2^{j+1}k^\gamma(\log n+\log(2/\varepsilon+1))}\\
&\leq  2^{O(2^{j+1}k^\gamma\log k)}.
\end{align*}
The last inequality follows from $n\leq O(k^{c_2}),r\leq O(k^{c_1}),\varepsilon=O(1/(rnk^3))$.

%There are ${n \choose s_j}$ ways of choosing such a support. For each choice, we place a $\eps$-net  $\ell_2$-net on the support, the $\eps =1/k^3$. The size of the net for a given $j$ is
%\begin{align*}
%{n \choose s_j} (2^{-j-1}/\eps )^{s_j} & ={n \choose s_j} (2^{-j-1} k^3)^{s_j} \\
%& \leq {n \choose s_j} ( k^3)^{s_j} \\
%& \leq ( n k^3 )^{s_j} \\
%& \leq 2^{\Theta(s_j \log k)} &\text{~by~} n=\poly(k)\\
%& \leq 2^{\Theta(2^{j+1}k^\gamma \log k)} & \text{~by~} s_j \leq O(2^{j+1}k^\gamma)\\
%& \leq 2^{\Theta(k^{2-\alpha}\log k) } &\text{~by~} r = \Theta(k^{2-\alpha}) \\
%\end{align*}

% = \Theta(\frac{1}{\sqrt{k} \poly( \log k) } )$.

%\iffalse
%Using Claim \ref{cla:half_gaussian}, we know that for each single coordinate,
%\begin{align*}
%p:=\Pr\biggl[  | \text{~that~coordinate~} | \leq \frac{1}{ \poly(\log k)} \biggr] \lesssim \frac{1}{\poly(\log k) \sigma}.
%\end{align*}
% Then, the probability of existing $k/\poly(\log k)$ of the coordinates no more than $\frac{1}{\poly(\log k)}$ is at most

%\begin{align*}
%\sum_{i=k/\poly(\log k)}^{k} p^{i} (1-p)^{k-i} {k \choose i} \leq \sum_{i=k/\poly(\log k)}^{k} p^i {k \choose {k/2}} \lesssim p^{k/\poly(\log k) } 2^k %\leq  p^{k/2}  {k \choose {k/\poly(\log k)}} \leq 2^{-\Theta(k\log( \alpha \sqrt{k} \sigma )) }
%\end{align*}
%\fi
%\begin{definition}
We define an event $\mathcal{E}(y^j)$ to be: $A V^\top y^j$ has $k/\log^2 k$ coordinates which are each at least $2/\log^2 k$ in absolute value.
%\end{definition}
Now, we want to show,

\begin{claim}\label{cla:single_prob}
For any $j\geq j_0$, for a fixed $y^j\in\mathbb{R}^n$,
\begin{align*}
\Pr\biggl[\mathcal{E}(y^j)\mathrm{~happens~}\biggr]\leq 2^k e^{-\Theta(k/(\|y^j\|_2^2 \log^{6} k))}\leq e^{-\Theta(2^{j-1}k^{1-\gamma}/\log^6 k ) }.
\end{align*}
\end{claim}
%Now, we want to show

\begin{proof}
We let $p$ be the probability that the absolute value of a single coordinate of $A V^\top y^j$ is at least $1/\log^2 k$. Notice that each coordinate of $A V^\top y^j$ is i.i.d. Gaussian $N(0,\| V^\top y^j \|_2^2)$
%(since $V_S$ has unit orthogonal columns, $\|V_S^\top y^j\|_2=\|y^j\|_2$)
and because for any Gaussian random variable $g$, $\Pr[ |g| \geq t] \leq \exp(-\Theta(t^2/\sigma^2))$, then $p\leq \exp(- 1/(\| V^\top y^j\|_2^2 \log^4k))$, by plugging $\sigma^2 = \|V^\top y^j\|_2^2$ and $t =1/\log^2 k$. So the probability
$A V^\top y^j$  has $k/\log^2 k$ coordinates which are each at least $1/\log^2 k$ in absolute value is,
\begin{align*}
\sum_{i=k/\log^2 k}^{k} p^{i} (1-p)^{k-i} {k \choose i} &\leq ~ \sum_{i=k/\log^{2} k}^{k} p^i {k \choose {k/2}} \leq 2p^{k/\log^{2} k } 2^k \\
&\leq ~ e^{-\Theta(k t^2 /(\sigma^2 \log^{2} k) )} 2^k & \text{~by~} p\leq \exp(- \Theta( t^2/\sigma^2) )\\
&\leq ~ e^{-\Theta(k  /(\| V^\top y^j\|_2^2 \log^{6} k) )} 2^k \\
&\leq ~ e^{-\Theta(k  /(\|  y^j\|_2^2 \log^{6} k) )} 2^k &\text{~by~} \| V^\top y^j\|_2^2\leq \| V\|_2^2\| y^j\|_2^2\leq \|y^j\|_2^2 \\
&\leq ~ e^{-\Theta(k   /(2^{-2j}s_j  \log^{6} k) )} 2^k & \text{~by~} \| y^j\|_2^2 \leq 2^{-2j}s_j\\
&\leq ~ e^{-\Theta(k   /(2^{-j+1}k^\gamma  \log^{6} k) )} 2^k & \text{~by~} s_j\leq O(2^{j+1}k^\gamma)\\
&= ~ e^{-\Theta(k   /(2^{-j+1}k^\gamma  \log^{6} k) )} \\
&\leq ~ e^{-\Theta(2^{j-1}k^{1-\gamma}/\log^6 k ) } .\\
\end{align*}
The first equality follows from $2^{-j+1}k^\gamma  \log^{6} k\leq 2^{-j_0+1}k^\gamma  \log^{6} k\leq 2k^{\gamma-\beta}\log^{6} k=o(1)$.
\end{proof}

%\Zhao{David original email is $e^{- k \poly(\log k)}$}

%Because $s$ is at most $k \log^{c_0} k$, then the net size is at most $2^{ \Theta( k \log^{c_0+1} k)}$. In order to take the union over all the vectors in the net, we only need to set $c_3 \gtrsim c_0+1$. % with a smaller $\poly(\log k)$ in the exponent since we can choose the $\poly(\log k)$ for which $\|x^j\|_2 \leq 1/\poly(\log k)$ as small as we want by suffering a larger $\poly(\log k)$ in our approximation factor.

For $j_1 = \lceil100(c_1+c_2+1)\log k\rceil =\Theta(\log k)$, consider $j \in [j_1,\infty)$. We have $\|\sum_{j=j_1}^\infty y^j\|_2$ is at most $\Theta(2^{-j_1} \sqrt{n}) \leq 1/k^{100(1+c_1)}$, and so 
\begin{align*}
\|AV^\top \sum_{j=j_1}^{\infty}y^j\|_1 \leq \sqrt{k} \|AV^\top\sum_{j=j_1}^{\infty} y^j\|_2 \leq \sqrt{k} \|A\|_2\|V\|_2 \|\sum_{j=j_1}^{\infty}y^j\|_2 \leq \sqrt{k}\cdot \sqrt{r}\cdot 1/k^{100(1+c_1)}\leq 1/k^{50}.
\end{align*}
The last inequality follows from $r=O(k^{c_1})$.
%The last inequality follows by a sufficiently large $j_1$ and $n=\poly(k)$.
%where the last step follows by we can control $\| x^j \|_2$.
%can be at most $\sqrt{k}/\poly(\log k)$,
So the contribution of $y^j$ to $\|AV^\top y^j\|_1$ for all $j \geq j_1$ is at most $1/k^{50}$. Thus, if we only consider those $j$ which contribute, i.e., $j_0\leq j\leq j_1$, we have $O(\log k)$ values of $j$. Then we can only construct $O(\log k)$ nets $\mathcal{N}_{j_0},\mathcal{N}_{j_0+1},\cdots,\mathcal{N}_{j_1}$. Since the size of net $\mathcal{N}_j$ is $2^{\Theta(2^{j+1}k^\gamma \log k)}$, by combining Claim~\ref{cla:single_prob} and taking a union bound, we have
\begin{align*}
\Pr \left[ \exists y^j\in \overset{j_1}{\underset{j=j_0}{\bigcup}} \mathcal{N}_j,\ \mathcal{E}(y^j)\text{~happens~} \right] & \leq \sum_{j=j_0}^{j_1} 2^{\Theta(2^{j+1}k^\gamma \log k)} \cdot e^{-\Theta(2^{j-1}k^{1-\gamma}/\log^6 k ) } \\
&\leq e^{-\Theta(2^{j_0-1}k^{1-\gamma}/\log^6 k )}\\
&\leq e^{-\Theta(k^{1+\beta-\gamma}/\log^6 k )}\\
&\leq 2^{-\Theta(k)} .\\
\end{align*}
The second inequality follows since $k^{\gamma}=o(k^{1-\gamma})$. The third inequality follows since $2^{j_0}\geq k^\beta$. The fourth inequality follows since $1+\beta-\gamma>1$.

Then, $\forall j_0\leq j\leq j_1,\forall \wt{y}^j\not \in \mathcal{N}_j$, there exists a vector $\wh{y}^j$ in $\mathcal{N}_j$, such that $\| \wh{y}^j - \wt{y}^j \|_2 \leq 1/k^{3+c_1}$. We can upper bound the $\ell_{\infty}$ norm of $AV^\top \wh{y}^j - AV^\top \wt{y}^j$ in the following sense,
\begin{align*}
\|AV^\top \wh{y}^j - AV^\top \wt{y}^j\|_{\infty} & \leq \|AV^\top \wh{y}^j - AV^\top \wt{y}^j \|_2 & \text{~by~} \| \cdot \|_{\infty} \leq \| \cdot\|_2 \\
& \leq \| A \|_2 \cdot \| \wh{y}^j - \wt{y}^j\|_2 \\
& \leq \sqrt{r}/k^{3+c_1} &\text{~by~Claim~\ref{lem:spectral_of_A}~and~}\| \wh{y}^j - \wt{y}^j \|_2 \leq 1/k^{3+c_1}  \\
& =1/k^2 &\text{~by~} r\leq O(k^{c_1}) .
\end{align*}
We let $Y=\{y\in\mathbb{R}^n\ \mid\ \|y\|_1\leq O(k^\gamma)\text{~and~each~coordinate~of~}y\leq 1/k^\beta\}$. Since $1/k^2<1/\log^2 k$ we can conclude that,
\begin{align}\label{eq:tail_for_flat_vectors}
\Pr\biggl[\exists y\in Y,j\geq j_0,\mathcal{E}(y^j)\text{~happens~}\biggr]\leq 2^{-\Theta(k)} .
\end{align}

Recalling $y = \sum_j y^j$, by Equation~(\ref{eq:tail_for_flat_vectors}), for any $y$, with probability at most $2^{-\Theta(k)}$, there are at most $ O(\log k) \cdot  k/\log^2 k \leq O(k/\log k)$ coordinates for which $AV^\top \sum_{j=j_0}^{j_1} y^j$ (the same statement also holds for $AV^\top \sum_{j=j_0}^{\infty} y^j=AV^\top y$ since we argued that there is negligible ``contribution'' for those $j>j_1$) is at least $O(\log k) /\log^2 k =O(1/\log k)$ on that coordinate. Summarizing,
\begin{align*}
\Pr\biggl[\wh{\mathcal{E}}(A,V,\beta,\gamma)\biggr]\geq 1-2^{-\Theta(k)}.
\end{align*}
\end{proof}

\begin{lemma}\label{lem:Gaussian_can_not_fit_Gaussian}
For any $t,k\geq 1$, and any constants $c_2\geq c_1\geq 1$, let $k\leq r=O(k^{c_1}),r\leq n= O(k^{c_2})$, let $V\in\mathbb{R}^{n\times r}$ be a matrix with orthonormal columns, and let each entry of $A\in\mathbb{R}^{k\times r},v_1,v_2,\cdots,v_t\in\mathbb{R}^{k}$ be an i.i.d. Gaussian $N(0,1)$ random variable. For a constant $\alpha\in(0,0.5)$ which can be arbitrarily small, if $\wh{\mathcal{E}}(A,V,0.5+\alpha/2,0.5-\alpha)$ happens, then with probability at least $1-2^{-\Theta(tk)}$, there are at least $\lceil t/10\rceil$ such $j\in[t]$ that $\forall x\in\mathbb{R}^r$ either $\|Ax-v_j\|_1\geq\Omega(k^{0.5-\alpha})$ or $\|Vx\|_1\geq \Omega(k^{0.5-\alpha})$ holds.
\end{lemma}

\begin{proof}
For convenience, we define $\gamma=0.5-\alpha$ which can be an arbitrary constant in $(0,0.5)$. We let constant $\beta=0.5+\alpha/2$. Then we have $\beta+\gamma<1$ and $\beta>\gamma$. Let $v\in\mathbb{R}^k$ be a random vector with each entry drawn from i.i.d. Gaussian $N(0,1)$.
Suppose $\wh{\mathcal{E}}(A,V,\beta,\gamma)$ happens.
%Since $0<\gamma<0.5$, we can always find a constant $\beta$ which satisfies $\beta+\gamma<1$ and $\beta>\gamma$, e.g. let $\beta=0.75-\gamma/2$.
%\begin{claim}\label{cla:spectral_of_A}
%With probability at least $1-e^{- \Theta(r) }$, the maximum singular value $\|A\|_2$ is at most $O(\sqrt{r})$.
%\end{claim}
%\begin{proof}
%Recall that $A\in\mathbb{R}^{k\times r}$ is a random matrix with each entry from i.i.d. $N(0,1)$. By the standard arguments (Proposition 2.4 in \cite{RV10}). Since $r\geq k$, with probability at least $1-e^{-\Theta(r)}$, the $\|A\|_2$ is at most $O(\sqrt{r})$.
%\end{proof}

%Due to Lemma~\ref{lem:spectral_of_A},
%in the following proof, we can suppose $\|A\|_2\leq O(\sqrt{r})$.

For any $x\in\mathbb{R}^r$, we can find a $y\in\mathbb{R}^n$ such that $y=Vx$. Since $V\in\mathbb{R}^{n\times r}$ has orthonormal columns, $x=V^\top Vx=V^\top y$. Then our goal is to argue that with high probability if $\|AV^\top y-v\|_1\leq O(k^\gamma)$, then $\|y\|_1>\Omega(k^\gamma)$. Take any $y \in \mathbb{R}^n$ which can be expressed as $Vx$, and decompose it as $y = y_0 + y_1$, where $y_0\in\mathbb{R}^n$ and $y_1\in\mathbb{R}^n$ have disjoint supports, $y_0$ has each coordinate with absolute value greater than $1/k^\beta$, and $y_1$ has each coordinate with absolute value at most $1/k^\beta$. %where the coordinates in $y^j$ have absolute value in the range $[2^{-j-1}, 2^{-j})$, and where we start with $j=j_0$ with $2^{-{j_0}} = 1/k^{\beta}$. Since we can assume $\|y\|_1 \leq O(k^\gamma)$ (otherwise it already satisfies $\|y\|_1\geq \Omega(k^\gamma)$), the support size of $y^j \in \mathbb{R}^n$ is at most $s_j \leq O(2^{j+1} k^\gamma)$, so it follows that

Now we create an $\varepsilon$-net for $y_0$: let $\varepsilon=O(1/(rnk^3))=O(1/k^{c_1+c_2+3})$, we denote $\mathcal{N}\subset\mathbb{R}^n$ as follows:
\begin{align*}
\mathcal{N}=\left\{p\in\mathbb{R}^n \mid\ \exists q\in\mathbb{Z}^n,\text{~s.t.~}p=\varepsilon q, \|p\|_1\leq O(k^\gamma),\forall i\in [n],\text{~either~} |p_i|> 1/k^\beta \text{~or~}p_i=0\right\}.
\end{align*}
Obviously, for any $y_0$, there exists $p\in\mathcal{N}$ such that $\|y_0-p\|_{\infty}\leq \varepsilon=O(1/(k^{c_1+c_2+3}))$, since $n\leq O(k^{c_2})$, $\|y_0-p\|_2\leq \|y_0-p\|_1\leq n\|y_0-p\|_{\infty}\leq O(1/k^{c_1+3})$. Now let us consider the size of $\mathcal{N}$. If $p\in \mathcal{N}$, the number of choices of one coordinate of $p$ is at most $O(k^\gamma/\varepsilon)$. And since $\|p\|_1\leq O(k^\gamma)$ and each coordinate of $p$ has absolute value at least $1/k^\beta$, the number of supports of $p$ is at most $O(k^{\gamma+\beta})$. Therefore,
\begin{align*}
|\mathcal{N}|&\leq (n+1)^{O(k^{\gamma+\beta})}\cdot O(k^\gamma/\varepsilon)^{O(k^{\gamma+\beta})}\\
&\leq  2^{O(k^{\gamma+\beta}\log k)} .
\end{align*}
The last inequality follows from $n\leq O(k^{c_2}),r\leq O(k^{c_1}),\varepsilon=O(1/(rnk^3))$.

%Now we create a $\varepsilon'=1/k^3$ $\ell_2$-net $\mathcal{N}$ for $y_0$.
%We know that $\| y_0 \|_1 \leq \| y \|_1 \leq O(k^\gamma)$, and each coordinate is at least $ 1/k^{1-\frac{\alpha}2+\varepsilon_1}$.
%Thus the sparsity of $y_0$ is at most $s_0=O(k^\gamma\cdot k^{1-\frac{\alpha}2+\varepsilon_1})=O(k^{\gamma+1-\frac{\alpha}2+\varepsilon_1})$.

%, since each coordinate is at least $1/\sqrt{k \poly(\log k)}$ in magnitude and $|x|_1 <= \sqrt{k}/\poly(\log k)$.

%We can thus create a net with size
%\begin{align*}
%{r \choose s_0} \cdot (\frac{k^\gamma}{\varepsilon'})^{s_0} & \leq r^{s_0} \cdot (k^5)^{s_0}\\
% &\leq (k^{10})^{\Theta(k^{\gamma+1-\frac{\alpha}2+\varepsilon_1})} \\
 %& \leq { {k\log^{c_0} k} \choose k/\log^{b-a}k } \cdot (\frac{\sqrt{k}/\log^b k }{1/(\sqrt{k} \log^{c_5} k)})^{k/\log^{b-a}k} \\
 %& \leq ( \log^{c_0 + b-a} k )^{k/\log^{b-a} k} \cdot (k/\log^{b+c_5} k)^{k/\log^{b-a} k} \\
% & \leq 2^{\Theta(k^{\gamma+1-\frac{\alpha}2+\varepsilon_1}\log k)}\\
% & =2^{k^{o(1)}}
%\end{align*}
%The last equality follows by $\varepsilon_1<\frac{\alpha}2-\gamma$
%, which is $2^{k/\poly(\log k)}$.

%\newpage
%For any $x\in \mathbb{R}^s$, we have $x = x_0 + \sum_{j=j_0}^{\infty} x^j$. Then define $x' = \wt{x}_0 + \sum_{j=j_0}^{\infty} x^j$, where $\wt{x}_0$ is the closest vector $x_0$ in our net.

For $y_0\in\mathbb{R}^n$, we define event $\mathcal{E}_1(y_0)$ as: $\exists \text{~valid~} y_1\in\mathbb{R}^n$, $\|AV^\top y-v\|_1\leq O(k^\gamma)$, where $y=y_0+y_1$ is the decomposition of a possible $y\in\mathbb{R}^n$. Here $y_1$ is valid means that there exists $y\in\mathbb{R}^n$ such that $\|y\|_1\leq O(k^\gamma)$ and the decomposition of $y$ is $y=y_0+y_1$.
%\begin{definition}
We define event $\mathcal{E}_2(y_0)$ as: $\exists \text{~valid~}y_1$, the absolute value of $AV^\top y-v$ is at most $O(1/k^\gamma)$ for at least $k-O(k^{2\gamma})$ coordinates $i$ in $[k]$ .
%\end{definition}
%\begin{definition}
We define event $\mathcal{E}_3(y_0)$ as: at least $k- O(k^{2\gamma}) -O (k/\log k)$ coordinates of $Ay_0 - v$ have absolute value at most $O(1/\log k)$.
%\end{definition}

\begin{claim}\label{cla:goal_is_to_bound_E3}
For $y_0\in\mathbb{R}^n$,
\begin{align*}
\Pr[\mathcal{E}_3(y_0)\mathrm{~happens~}]\geq\Pr[\mathcal{E}_2(y_0)\mathrm{~happens~}]\geq \Pr[\mathcal{E}_1(y_0)\mathrm{~happens~}] .
\end{align*}
\end{claim}
\begin{proof}
If $\mathcal{E}_1(y_0)$ happens, then there exists a valid $y_1\in\mathbb{R}^n$ such that $y=y_0+y_1$ and $\|AV^\top y-v\|_1\leq O(k^\gamma)$. For this $y$, there are at least $k-O(1/k^{2\gamma})$ coordinates of $AV^\top y-v$ with absolute value at most $O(1/k^\gamma)$. Otherwise, $\|AV^\top y-v\|_1>\Omega(k^\gamma)$. Thus, $\mathcal{E}_1(y_0)$ implies $\mathcal{E}_2(y_0)$.

%Otherwise, there exists $y=y_0+\sum_j y^j$, and there at least $\Omega(k^{2\gamma})$ coordinates of $AV^\top y-v$ is at least $\Omega(1/k^\gamma)$. It contradicts to $\|AV^\top y-v\|\leq O(k^\gamma)$.

Now we want to show $\mathcal{E}_2(y_0)$ implies $\mathcal{E}_3(y_0)$. We suppose $\mathcal{E}_2(y_0)$ happens. Then there is a valid $y_1$ such that there are at least $k-O(1/k^{2\gamma})$ coordinates of $AV^\top y-v$ with absolute value at most $O(1/k^\gamma)$, where $y=y_0+ y_1$. Recall that the event $\wh{\mathcal{E}}(A,V,\beta,\gamma)$ happens, for any valid $y_1$ there are at most $O(k/\log k)$ coordinates of $AV^\top y_1$ are at least $\Omega(1/\log k)$. %Also notice that $AV^\top y-v$ has at most $O(k^\gamma)$ coordinates is at least $O(1/k^\gamma)$.
Therefore,
\begin{align*}
AV^\top y_0 - v =  \underbrace{AV^\top y - v}_{ \substack{ \geq k-O(k^{2\gamma}) \text{~coordinates~} \\ \text{each~}\leq O(1/k^\gamma) }} - \underbrace{ AV^\top y_1 }_{ \substack{\leq O(k/\log k) \text{~coordinates~} \\ \text{each} \geq \Omega(1/\log k)} } .
\end{align*}
 %Therefore, if we look at the set of coordinates $[k]\backslash S_1(y) \backslash S_2(y)$ of $AV^\top y-v-AV^\top\sum_{j=j_0}^{\infty} y^j =  AV^\top y_0 - v$ , the absolute value of each of those coordinates should be at most $O(1/\log k)+O(1/k^\gamma)=O(1/\log k)$.
 Therefore, at least $k- O(k^{2\gamma}) -O (k/\log k)$ coordinates of $Ay_0 - v$ in absolute value is at most $O(1/\log k)+O(1/k^\gamma)=O(1/\log k)$
 %Since $|S_1(y)|\geq O(k/\log k),|S_2(y)|\geq k^{2\gamma}$, then we finish the proof.
\end{proof}
\begin{claim}
Define event $\mathcal{F}_1$ to be the situation for which there exists $1/2$ of the coordinates of $v\in \mathbb{R}^k$ that are at least $1/100$. The probability this event $\mathcal{F}_1$ holds is at least $1-2^{-\Theta(k)}$.
\end{claim}
\begin{proof}
Note that $\overline{\mathcal{F}}_1$ means there exist $k/2$ of the coordinates of $v\in \mathbb{R}^k$ which are at most $1/100$. Using Lemma~\ref{lem:half_gaussian}, for each single coordinate the probability it is smaller than $1/100$ is $1/200$. Then the probability more than half the coordinates are no more than $1/100$ is

\begin{align*}
\sum_{i=k/2}^{k} p^{i} (1-p)^{k-i} {k \choose i} &\leq ~ \sum_{i=k/2 }^{k} p^i {k \choose {k/2}} \leq 2 p^{k/2 } 2^k  \leq (1/200)^{\Theta(k)} = 2^{-\Theta(k)} .
\end{align*}
\end{proof}

Conditioned on $\mathcal{F}_1$ happening, if $\mathcal{E}_3(y_0)$ happens, then due to the pigeonhole principle, there are at least $\frac{k}{2}-O(k^{2\gamma})-O(k/\log k)$ coordinates of $AV^\top y_0-v$ that are at most $O(1/\log k)$ and the corresponding coordinate of $v$ is larger than $1/100$. Now, let us look at the probability on a single coordinate of $AV^\top y_0-v$.

\begin{claim}\label{cla:single_success}
If the $i^{\text{th}}$ coordinate of $v\in\mathbb{R}^k$ is at least $1/100$. Then $\Pr[| (AV^\top y_0)_i - (v)_i | \leq O(1/\log k) ]\leq  O(1/\log k)$.
\end{claim}

\begin{proof}
 Since $|(v)_i|>1/100$ and $v$ is independent from $A\in\mathbb{R}^{k\times r}$, for $0<\eta<1/100$, $\Pr[ | (AV^\top y_0)_i - (v)_i | \leq \eta ]$ is always upper bounded by $\Pr[ (A x_0)_i \in [ 1/100- \eta, 1/100 + \eta ]$. Thus, it suffices to prove an upper bound for $\Pr[ (AV^\top x_0)_i \in [ 1/100- \eta, 1/100 + \eta ]$. Let $f(x)$ be the pdf of $N(0,\sigma^2)$, where $\sigma^2 = \| V^\top y_0\|_2^2,V\in\mathbb{R}^{n\times r},y_0\in\mathbb{R}^n $. Then
\begin{align*}
 & ~\Pr[ (AV^\top y_0)_i \in [ 1/100- \eta, 1/100 + \eta] ] \\
= & ~ \int_{1/100- \eta }^{ 1/100 + \eta } f(x) \mathrm{d} x \\
\leq & ~ f(1/200) \int_{1/100- \eta }^{ 1/100 + \eta }  \mathrm{d} x \\
\leq & O(\eta) .
\end{align*}
where the last step follows since $f(1/200)\leq 200$. We set $\eta=O(1/\log k)$, then we get the statement.
\end{proof}

\begin{claim}
Conditioned on $\mathcal{F}_1$, for a fixed $y_0\in \mathbb{R}^n$, with probability at most $2^{-\Theta(k)}$, there are at least $\frac{k}{10}$ coordinates of $AV^\top y_0-v\in\mathbb{R}^{k}$ which are at most $O(1/\log k)$ and the corresponding coordinate of $v$ is larger than $1/100$.

%for any $y$ which can be decomposed as $y=y_0+\sum y^j$, for all $i\in S_0(y)$, $| (AV^\top y_0)_i - v_i | \leq  O(1/\log k)$.% where $|S_{\mathcal{F}_1}|\geq k/2$, $|S_1(y)| = O(k/\log k) $ and $|S_2|= O(k^\gamma)$.
\end{claim}

\begin{proof}
We look at the coordinate $i\in[k]$ which has $|(v)_i|>1/100$. The probability $\|(AV^\top y_0)_i-(v)_i\|\leq O(1/\log k)$ is at most $O(1/\log k)$. Due to the independence between different coordinates of $v$, since there are at least $k/2$ coordinates of $v$ satisfying that they have absolute value greater than $1/100$, with probability at most $2^{-\Theta(k)}$, there are at least $\frac15\cdot\frac{k}2=k/10$ coordinates of $AV^\top y_0-v$ which are at most $O(1/\log k)$.

% and the corresponding coordinate of $v$ is larger than $1/100$.
%Fixed a $y_0$, for any $y$ which can be decomposed as $y=y_0+\sum y^j$, $|S_1(y)|\leq O(k/\log k)$, $|S_2(y)|\leq O(k^\gamma)$. Combining that $|S_{\mathcal{F}_1}|\geq k/2$ and Claim~\ref{cla:single_success}, the probability that $\forall i\in S_0(y)$, $| (AV^\top y_0)_i - v_i | \leq  O(1/\log k)$ is bounded by
%$$O(1/\log k)^{|S_0(y)|}\leq 2^{-\Theta(k)}$$
\end{proof}

Because $\mathcal{E}_3(y_0)$ implies the event described in the above claim when conditioning on $\mathcal{F}_1$, for a fixed $y_0$, the probability that $\mathcal{E}_3(y_0)$ holds is at most $2^{-\Theta(k)}$. Since $\gamma+\beta<1$, the $|\mathcal{N}|\leq 2^{k^{o(1)}}$, we can take a union bound over the $\mathcal{N}$:
\begin{align*}
\Pr[\exists y_0\in\mathcal{N},\mathcal{E}_1(y_0)\text{~happens~}]\leq 2^{-\Theta(k)} .
\end{align*}
It means that with probability at least $1-2^{-\Theta(k)}$, for any $y=y_0+y_1$ with $y_0\in\mathcal{N}$, $\|AV^\top y-v\|_1>\Omega(k^\gamma)$.
Let $\wt{y}=\wt{y}_0+\sum\wt{y}^j$, where $\wt{y}_0\not \in \mathcal{N}$. We can find $\wh{y}_0\in \mathcal{N}$ which is the closest to $\wt{y}_0$. Denote $\wh{y}=\wh{y}+\sum\wt{y}^j$. Then,
\begin{align*}
\|AV^\top \wh{y}-v\|_1 &\leq  \|AV^\top \wt{y}-v\|_1 + \|AV^\top \wt{y}-AV^\top \wh{y}\|_1 & \text{~by~triangle~inequality}\\
&\leq  \|AV^\top \wt{y}-v\|_1 + \sqrt{k} \|AV^\top \wt{y}-AV^\top \wh{y}\|_2 & \text{~by~} \|\cdot \|_1 \leq \sqrt{\text{dim}} \|\cdot \|_2 \\
&\leq  \|AV^\top \wt{y}-v\|_1 + \sqrt{k} \|A \|_2 \|\wt{y}-\wh{y}\|_2  \\
&\leq \|AV^\top \wt{y}-v\|_1 + \sqrt{k} \cdot \sqrt{r} \cdot \|\wt{y}_0-\wh{y}_0\|_2 & \text{~by~}\| A\|_2\leq \sqrt{r}\\
&\leq \|AV^\top \wt{y}-v\|_1 + \sqrt{k} \cdot \sqrt{r} \cdot \varepsilon & \text{~by~} \|\wt{y}_0-\wh{y}_0\|_2 \leq \varepsilon \\
&=\|AV^\top \wt{y}-v\|_1 +\sqrt{k} \cdot \sqrt{r}\cdot 1/k^{c_1+3} & \text{~by~} \varepsilon' =1/k^{c_1+3} \\
&= \|AV^\top \wt{y}-v\|_1 + 1/k . & \text{~by~} r=O(k^{c_1})
\end{align*}
and so if $\|AV^\top \wt{y}-v\|_1$ is at most $O(k^\gamma)$ then $\|A\wh{y}-v\|_1$ is at most $O(k^\gamma)$.

 For $j\in[t]$, we now use notation $\mathcal{E}_4(v_j)$ to denote the event: $\exists y_0\in \mathbb{R}^n$ with $\|y_0\|_1\leq O(k^\gamma)$ and each non-zero coordinate of $y_0$ is greater than $1/k^\beta$, at least $k- O(k^{2\gamma}) -O (k/\log k)$ coordinates of $AV^\top y_0 - v_j$ in absolute value are at most $O(1/\log k)$. Based on the previous argument, conditioned on $\wh{\mathcal{E}}(A,V,\beta,\gamma)$,
 \begin{align*}
 \Pr[\mathcal{E}_4(v_j)]\leq 2^{-\Theta(k)} .
 \end{align*}
Also notice that, conditioned on $\wh{\mathcal{E}}(A,V,\beta,\gamma)$, $\forall j\in[t],\ \mathcal{E}_4(v_j)$ are independent. Thus, conditioned on $\wh{\mathcal{E}}(A,V,\beta,\gamma)$, due to the Chernoff bound, the probability that there are $\lceil t/10\rceil$ such $j$ that $\mathcal{E}_4(v_j)$ happens is at most $2^{-\Theta(kt)}$. We define $\mathcal{E}_5(v_j)$ to be the event: $\exists y\in\mathbb{R}^n,\|AV^\top y - v_j\|\leq O(k^\gamma)$ with $\|y\|_1\leq O(k^\gamma)$. Similar to the proof of Claim~\ref{cla:goal_is_to_bound_E3}, conditioned on $\wh{\mathcal{E}}(A,V,\beta,\gamma)$, $\mathcal{E}_5(v_j)$ implies $\mathcal{E}_4(v_j)$. Thus, conditioned on $\wh{\mathcal{E}}(A,V,\beta,\gamma)$, the probability that there are $\lceil t/10\rceil$ such $j$ that $\mathcal{E}_5(v_j)$ happens is at most $2^{-\Theta(tk)}$. Then, we complete the proof.

\end{proof}

\begin{theorem}\label{thm:hard_distribution}
For any $k\geq 1$, and any constants $c_1,c_2$ which satisfy $c_2-2>c_1>1$, let $r=\Theta(k^{c_1}),n=\Theta(k^{c_2})$, and let ${\cal A}(k,n)$ denote a distribution over $n\times (k+n)$ matrices where each entry of the first $n\times k$ matrix is i.i.d. Gaussian $N(0,1)$ and the next $n\times n$ matrix is an identity matrix. For any fixed $r\times n$ matrix $S$ and a random matrix $\wh{A} \sim {\cal A}(k,n)$, with probability at least $1-O(k^{1+\frac{c_1-c_2}{2}})-2^{-\Theta(k)}$, there is no algorithm that is able to output a matrix $B\in \mathbb{R}^{n\times r}$ such that
\begin{align*}
\| BS\wh{A} - \wh{A} \|_1 \leq O(k^{0.5-\varepsilon}) \underset{\rank-k~A'}{\min} \| A' - \wh{A} \|_1 ,
\end{align*}
where $\varepsilon>0$ is a constant which can be arbitrarily small.
\end{theorem}
\begin{proof}
For convenience, we let $\gamma=0.5-\varepsilon$ be a constant which can be arbitrarily close to $0.5$.
Since the last $n$ columns of $\wh{A}$ is an identity matrix, we can fit the first $k$ columns of $\wh{A}$, so we have
\begin{align*}
\underset{\rank-k~A'}{\min} \| A' - \wh{A} \|_1\leq n .
\end{align*}
Now, we want to argue that, for a fixed $S$, with high probability, for any rank-$k$ $n\times r$ matrix $B$, the cost
\begin{align*}
\| BS\wh{A} - \wh{A} \|_1 \geq  \Omega(n \cdot k^{\gamma}) .
\end{align*}
Thus, the approximation gap will be at least $\Omega(k^{\gamma})$.

We denote the SVD of $S=U_S\Sigma_SV_S^\top$ where $U_S\in\mathbb{R}^{r\times r},\Sigma_S\in\mathbb{R}^{r\times r},V_S\in\mathbb{R}^{n\times r}$.
 %We denote $V_S^c\in \mathbb{R}^{r\times \frac{r}{2}}$ as a complement of $V_S$, i.e. $Q=[V_S\ V_S^c]$ is an $r\times r$ orthogonal matrix. Then we have
%\begin{align*}
%S\wh{A}=U_S\Sigma_SV_S^\top \wh{A}=U_S\Sigma_SV_S^\top Q Q^\top \wh{A}=[U_S\Sigma_S\ \text{\bf 0}] Q^\top \wh{A}.
%\end{align*}
%Due to Theorem~\ref{thm:close_to_fully_independent}, the first $\frac{r}2$ rows of $Q^\top\wh{A}$ XXX.
Then, we can rewrite
\begin{align}
\|BS\wh{A}-\wh{A}\|_1& = \|BU_S\Sigma_SV_S^\top\wh{A}-\wh{A}\|_1    \notag\\
&=\sum_{l=1}^n \|(BU_S\Sigma_S)^{l}(V_S^\top\wh{A})-\wh{A}^{l}\|_1 \notag\\
&\geq \sum_{l:\|V_S^l\|_2^2\leq 2r/n} \|(BU_S\Sigma_S)^{l}(V_S^\top\wh{A})-\wh{A}^{l}\|_1 . \label{eq:portion}
%&=\sum_{l=1}^n \|(BU_S\Sigma_S)^{l}(\begin{bmatrix}V_S^\top\\ {V_S^c}^\top\end{bmatrix}\wh{A})^{[1:\frac{r}{2}]}-\wh{A}^{l}\|_1\\
%&=\sum_{l=1}^n \|(BU_S\Sigma_S)^{l}(V_S^\top\wh{A})-\wh{A}^{l}\|_1\\
\end{align}
The first equality follows from the SVD of $S$. The second equality follows from the fact that the $\ell_1$-norm of a matrix is the sum of $\ell_1$-norms of rows. The third inequality follows since we just look at the cost on a part of the rows.

We use $\beta_l$ to denote $(BU_S\Sigma_S)^{l}$.
We look at a fixed row $l$, then the cost on this row is:
\begin{align}
&\|\beta_l(V_S^\top\wh{A})-\wh{A}^{l}\|_1\notag\\
=& \|\beta_l(V_S^\top\wh{A})_{[1:k]}-(\wh{A}^{l})_{[1:k]}\|_1+\|\beta_l(V_S^\top\wh{A})_{[k+1:k+n]}-(\wh{A}^{l})_{[k+1:n+k]}\|_1\notag\\
\geq & \|\beta_l(V_S^\top\wh{A})_{[1:k]}-(\wh{A}^{l})_{[1:k]}\|_1+\|\beta_l(V_S^\top\wh{A})_{[k+1:k+n]}\|_1-\|(\wh{A}^{l})_{[k+1:n+k]}\|_1\notag\\
\geq & \|\beta_l(V_S^\top\wh{A})_{[1:k]}-(\wh{A}^{l})_{[1:k]}\|_1+\|\beta_l(V_S^\top\wh{A})_{[k+1:k+n]}\|_1- 1\notag\\
\geq & \|\beta_l(V_S^\top\wh{A})_{[1:k]}-(\wh{A}^{l})_{[1:k]}\|_1+\|\beta_lV_S^\top\|_1- 1\notag . \\
\end{align}
where $(V_S^\top\wh{A})_{[1:k]}$ denotes the first $k$ columns of $(V_S^\top\wh{A})$, and similarly, $(V_S^\top\wh{A})_{[k+1:k+n]}$ denotes the last $n$ columns of $(V_S^\top\wh{A})$.
The first equality is because we can compute the sum of $\ell_1$ norms on the first $k$ coordinates and $\ell_1$ norm on the last $n$ coordinates. The first inequality follows from the triangle inequality. The second inequality follows since the last $n$ columns of $\hat{A}$ form an identity, so there is exactly one $1$ on the last $n$ columns in each row. The third inequality follows since the last $n$ columns of $\hat{A}$ form an identity. Let $\mathcal{D}_l$ be a distribution over $D_l\in\mathbb{R}^{(r+1)\times k}$, where
\begin{align*}
D_l=\begin{bmatrix}(V_S^\top\wh{A})_{[1:k]}\\(\wh{A}^{l})_{[1:k]}\end{bmatrix} .
\end{align*}
Let $\mathcal{G}$ be a distribution over $G\in\mathbb{R}^{(r+1)\times k}$ where each entry of $G$ is drawn from i.i.d. $N(0,1)$.
According to Lemma~\ref{lem:tot_var_results}, we have
\begin{align}\label{eq:needcombination}
D_{\TV}(\mathcal{D}_l,\mathcal{G})\leq O(k\|(V_S)^l\|_2)+2^{-\Theta(k)} .
\end{align}
Let $A=G^{[1:r]},v=G^{r+1}$. Due to Lemma~\ref{lem:useful_property}, with probability at least $1-2^{-\Theta(k)}$, $\wh{\mathcal{E}}(A^\top,V_S,0.75-\gamma/2,\gamma)$ happens. Then conditioned on $\wh{\mathcal{E}}(A^\top,V_S,0.75-\gamma/2,\gamma)$, due to Lemma~\ref{lem:Gaussian_can_not_fit_Gaussian}, with probability at most $2^{-\Theta(k)}$, there exists $\beta_l$ such that
\begin{align*}
\|\beta_lA-v\|_1+\|\beta_lV_S^\top\|_1=o(k^\gamma) .
\end{align*}
Combined with Equation~(\ref{eq:needcombination}), we can get that for a fixed $l$, with probability at most $O(k\|(V_S)^l\|_2)+2^{-\Theta(k)}$, there exists $\beta_l$ such that
\begin{align*}
\|\beta_l(V_S^\top\wh{A})_{[1:k]}-(\wh{A}^{l})_{[1:k]}\|_1+\|\beta_lV_S^\top\|_1=o(k^\gamma) .
\end{align*}
When $\|(V_S)^l\|_2^2\leq 2r/n=\Theta(k^{c_1-c_2})$, this probability is at most $\Theta(k^{1+(c_1-c_2)/2})+2^{-\Theta(k)}$.
Since $\sum_{l=1}^n \|(V_S)^l\|_2^2=r$, there are at most $n/2$ such $l$ that $\|(V_S)^l\|_2^2> 2r/n$ which means that there are at least $n/2$ such $l$ that $\|(V_S)^l\|_2^2\leq 2r/n$. Let $s$ be the number of $l$ such that $\|(V_S)^l\|_2^2\leq 2r/n$, then $s>n/2$. Let $t$ be a random variable which denotes that the number of $l$ which satisfies $\|(V_S)^l\|_2^2\leq 2r/n$ and achieve
\begin{align*}
\|\beta_l(V_S^\top\wh{A})_{[1:k]}-(\wh{A}^{l})_{[1:k]}\|_1+\|\beta_lV_S^\top\|_1=o(k^\gamma) ,
\end{align*}
at the same time.
Then,
\begin{align*}
\E[t]\leq (O(k\|(V_S)^l\|_2)+2^{-\Theta(k)})s=(O(k\sqrt{2r/n})+2^{-\Theta(k)})s .
\end{align*}
Due to a Markov inequality,
\begin{align*}
\Pr[t>s/2>n/4]\leq O(k\sqrt{2r/n})+2^{-\Theta(k)}=O(k^{1+(c_1-c_2)/2})+2^{-\Theta(k)} .
\end{align*}
The equality follows since $r=\Theta(k^{c_1}),n=\Theta(k^{c_2})$.
Plugging it into Equation~(\ref{eq:portion}), now we can conclude, with probability at least $1-O(k^{1+(c_1-c_2)/2})-2^{-\Theta(k)}$, $\forall B\in\mathbb{R}^{n\times r}$
\begin{align*}
\|BS\wh{A}-\wh{A}\|_1\geq \sum_{l:\|V_S^l\|_2^2\leq 2r/n} \|(BU_S\Sigma_S)^{l}(V_S^\top\wh{A})-\wh{A}^{l}\|_1\geq n/4\cdot \Omega(k^\gamma)=\Omega(n\cdot k^\gamma).
\end{align*}
\end{proof}

\begin{theorem}[Hardness for row subset selection]\label{thm:hard_for_row_subset}
For any $k\geq 1$, any constant $c\geq 1$, let $n = O(k^c)$, and let ${\cal A}(k,n)$ denote the same distribution stated in Theorem~\ref{thm:hard_distribution}. For matrix $\wh{A} \sim {\cal A}(k,n)$, with positive probability, there is no algorithm that is able to output $B\in\mathbb{R}^{n\times(n+k)}$ in the row span of any $r=n/2$ rows of $\wh{A}$ such that
\begin{align*}
\| \wh{A} - B \|_1 \leq O(k^{0.5-\alpha}) \underset{\rank-k~A'}{\min} \| A' - \wh{A} \|_1 ,
\end{align*}
where $\alpha\in(0,0.5)$ is a constant which can be arbitrarily small.
\end{theorem}
\begin{proof}
For convenience, we define $\gamma=0.5-\alpha$ which can be an arbitrary constant in $(0,0.5)$.
Since the last $n$ columns of $\wh{A}$ is an identity matrix, we can fit the first $k$ columns of $\wh{A}$, so we have
\begin{align*}
\underset{\rank-k~A'}{\min} \| A' - \wh{A} \|_1\leq n .
\end{align*}
We want to argue that $\forall B\in\mathbb{R}^{n\times (k+n)}$ in the row span of any $r=n/2$ rows of $\wh{A}$,
\begin{align*}
\| \wh{A} - B \|_1\geq \Omega(n\cdot k^\gamma) .
\end{align*}
Let $A^\top \in\mathbb{R}^{n\times k}$ be the first $k$ columns of $\wh{A}$, and let $S\subset [n]$ be a set of indices of chosen rows of $\wh{A}$ with $k\leq|S|\leq r$. Let $M^S\in\mathbb{R}^{k\times n}$ with the $i^{\text{th}}$ column $M^S_i=A_i$ if $i\in S$ and $M^S_i=0$ otherwise. We use $\wh{M}^S\in\mathbb{R}^{k\times r}$ to be $M^S$ without those columns of zeros, so it is a random matrix with each entry i.i.d. Gaussian $N(0,1)$. Then the minimum cost of using a matrix in the span of rows of $\wh{A}$ with index in $S$ to fit $\wh{A}$ is at least:
\begin{align*}
\sum_{l\not\in S} \min_{x_l\in\mathbb{R}^{n}}\left(\|M^Sx_l-A_l\|_1+\|x_l\|_1-1\right) .
\end{align*}
The part of $\|M^Sx_l-A_l\|_1$ is just the cost on the $l^{\text{th}}$ row of the first $k$ columns of $\wh{A}$, and the part of $\|x_l\|_1-1$ is just the lower bound of the cost on the $l^{\text{th}}$ row of the last $n$ columns of $\wh{A}$.
\begin{claim}\label{cla:global_imply_local}
$A,\wh{M}^S\in\mathbb{R}^{k\times n},\gamma\in(0,0.5)$,
\begin{align*}
\Pr\biggl[\wh{\mathcal{E}}(\wh{M}^S,I_r,0.75-\gamma/2,\gamma) ~\bigg|~ \wh{\mathcal{E}}(A,I_n,0.75-\gamma/2,\gamma)\biggr]=1.
\end{align*}
\end{claim}
\begin{proof}
Suppose $\wh{\mathcal{E}}(A,I_n,0.75-\gamma/2,\gamma)$ happens. Since $M^S$ has just a subset of columns of $A$, $\|M^S\|_2\leq \|A\|_2\leq \sqrt{n}$. Notice that $\forall x\in\mathbb{R}^n$ with $\|x\|_1\leq O(k^\gamma)$ and each non-zero coordinate of $x$ is at most $O(1/k^{0.75-\gamma/2})$, $M^Sx\equiv M^Sx^{S} \equiv Ax^S$, where $x^S\in\mathbb{R}^n$ has ${x^S}_i=x_i$ if $i\in S$ and ${x^S}_i=0$ otherwise. Because $\|x^S\|_1\leq O(k^\gamma)$ and $x^{S}$ has each coordinate in absolute value at most $O(1/k^{0.75-\gamma/2})$, $Ax^S$ has at most $O(k/\log k)$ coordinates in absolute value at least $\Omega(1/\log k)$. So, $M^Sx=Ax^S$ has at most $O(k/\log k)$ coordinates in absolute value at least $\Omega(1/\log k)$.
\end{proof}

We denote $\cost(S,l)=\min_{x_l\in\mathbb{R}^{n}}\left(\|M^Sx_l-A_l\|_1+\|x_l\|_1-1\right)$. Since $\forall l\not\in S,\ A_l$ are independent, and they are independent from $M^S$, due to Lemma~\ref{lem:Gaussian_can_not_fit_Gaussian},
\begin{align}\label{eq:local_prob}
\Pr\left[\sum_{l\not\in S}\cost(S,l)\leq O(n\cdot k^\gamma) ~ \bigg| ~ \wh{\mathcal{E}}(\wh{M}^S,I_r,0.75-\gamma/2,\gamma)\right]\leq 2^{-\Theta(rk)} .
\end{align}
Now we just want to upper bound the following:
\begin{align*}
& ~\Pr\left[\exists S\subset [n],|S|\leq r, \sum_{l\not\in S}\cost(S,l) \leq O(n\cdot k^\gamma)\right]\\
\leq & ~\Pr\left[\exists S\subset [n],|S|\leq r, \sum_{l\not\in S}\cost(S,l) \leq O(n\cdot k^\gamma) ~ \bigg| ~ \wh{\mathcal{E}}(A,I_n,0.75-\gamma/2,\gamma)\right] \\
+ & ~\Pr\left[\neg \wh{\mathcal{E}}(A,I_n,0.75-\gamma/2,\gamma)\right]\\
\leq & ~ \sum_{S\subset [n],|S|\leq r} \Pr\left[\sum_{l\not\in S}\cost(S,l)\leq O(n\cdot k^\gamma)~ \bigg| ~ \wh{\mathcal{E}}(A,I_n,0.75-\gamma/2,\gamma)\right] \\
+ & ~ \Pr\left[\neg \wh{\mathcal{E}}(A,I_n,0.75-\gamma/2,\gamma)\right]\\
\leq & ~ \sum_{S\subset [n],|S|\leq r} \Pr \left[\sum_{l\not\in S}\cost(S,l)\leq O(n\cdot k^\gamma) ~ \bigg|~ \wh{\mathcal{E}}(A,I_n,0.75-\gamma/2,\gamma)\right]+2^{-\Theta(k)}\\
\leq & ~ \sum_{S\subset [n],|S|\leq r} \Pr\left[\sum_{l\not\in S}\cost(S,l)\leq O(n\cdot k^\gamma) ~ \bigg|~ \wh{\mathcal{E}}(\wh{M}^S,I_r,0.75-\gamma/2,\gamma)\right]+2^{-\Theta(k)}\\
\leq & ~ (n+1)^r 2^{-\Theta(rk)}+2^{-\Theta(k)}\\
\leq & ~ 2^{-\Theta(rk)}+2^{-\Theta(k)}\\
\leq & ~ 2^{-\Theta(k)} .
\end{align*}
The second inequality follows by a union bound. The third inequality follows by Lemma~\ref{lem:useful_property}. The fourth inequality follows by Claim~\ref{cla:global_imply_local}. The fifth inequality is due to Equation~(\ref{eq:local_prob}). The sixth inequality follows by $n\leq O(k^c),r=n/2$.
Thus, with probability at least $1-2^{-\Theta(k)}$, $\forall B\in\mathbb{R}^{n\times (n+k)}$ which is in the span of any $r\leq n/2$ rows of $\wh{A}$,
\begin{align*}
\|B-\wh{A}\|_1\geq \Omega(n\cdot k^\gamma).
\end{align*}
Then, we have completed the proof.
\end{proof}

\begin{definition}
Given a matrix $A\in \mathbb{R}^{n\times d}$, a matrix $S\in \mathbb{R}^{r\times n}$, $k\geq 1$ and $\gamma\in (0,\frac12)$, we say that an algorithm ${\cal M}(A,S,k,\gamma)$ which outputs a matrix $B\in\mathbb{R}^{n\times r}$ ``succeeds'', if
\begin{align*}
\|BSA-A\|_1\leq k^\gamma \cdot \underset{\rank-k~A'}{\min} \| A' - A \|_1,
\end{align*}
holds.
\end{definition}

\begin{theorem}[Hardness for oblivious embedding]\label{thm:use_yaominmax}
Let $\Pi$ denote a distribution over matrices $S\in \mathbb{R}^{r\times n}$. For any $k\geq 1$, any constant $\gamma\in (0,\frac{1}{2})$, arbitrary constants $c_1,c_2>0$ and $\min(n,d) \geq \Omega(k^{c_2})$, if for all $ A\in\mathbb{R}^{n\times d}$, it holds that
\begin{align*}
\underset{S\sim \Pi}{\Pr}[{\cal M}(A,S,k,\gamma)\mathrm{~succeeds~}]\geq \Omega(1/k^{c_1}).
\end{align*}
Then $r$ must be at least $\Omega(k^{c_2-2c_1-2}).$
\end{theorem}

\begin{proof}
We borrow the idea from~\cite{nn14,psw16}.
We use Yao's minimax principle~\cite{y77} here. Let $\mathcal{D}$ be an arbitrary distribution over $\mathbb{R}^{n\times d}$, then
$$\underset{A\sim \mathcal{D},S\sim \Pi}{\Pr}[\mathcal{M}(A,S,k,\gamma)]\geq 1-\delta.$$
It means that there is a fixed $S_0$ such that
$$\underset{A\sim \mathcal{D}}{\Pr}[\mathcal{M}(A,S_0,k,\gamma)]\geq 1-\delta.$$
Therefore, we want to find a hard distribution $\mathcal{D}_{\text{hard}}$ that if
$$\underset{A\sim \mathcal{D}_{\text{hard}}}{\Pr}[\mathcal{M}(A,S_0,k,\gamma)]\geq 1-\delta,$$
$S_0$ must have at least some larger $\poly(k)$ rows.

Here, we just use the distribution ${\cal A}(k,\Omega(k^{c_2}))$ described in Theorem~\ref{thm:hard_distribution} as our hard distribution. We can just fill zeros to expand the size of matrix to $n\times d$. We can complete the proof by using Theorem~\ref{thm:hard_distribution}.

\end{proof}

\begin{remark}
Actually, in Lemma~\ref{lem:useful_property} and Lemma~\ref{lem:Gaussian_can_not_fit_Gaussian}, the reason we need $\beta>\gamma>0$ is that we want $k^{\beta-\gamma}=\omega(\poly(\log k))$, and the reason we need $\beta+\gamma<1$ is that we want $k^{\beta+\gamma}\poly(\log k)=o(k)$. Thus we can replace all the $k^\gamma$ by $\sqrt{k}/\poly(\log k)$, e.g., $\sqrt{k}/\log^{20} k$, and replace all the $k^\beta$ by $\sqrt{k}\poly(\log k)$ with a smaller $\poly(\log k)$, e.g., $\sqrt{k}\log^{10} k$. Our proofs still work. Therefore, if we replace the approximation ratio in Theorem~\ref{thm:hard_distribution}, Theorem~\ref{thm:hard_for_row_subset}, and Theorem~\ref{thm:use_yaominmax} to be $\sqrt{k}/\log^{c} k$ where $c$ is a sufficiently large constant, the statements are still correct.
\end{remark}

%%%%%%%%%%%%%%%%%%%%%%%%%%%%%%%%%%%%%%%%%%%%%%%%%%%%%%%%%%%%%%%%%%%%%%%%%%%%%%%%%%%%%%%%%%%%%%%%
%%%%%%%%%%%%%%%%%%%%%%%%%%%%%%%%%%%%%%%%%%%%%%%%%%%%%%%%%%%%%%%%%%%%%%%%%%%%%%%%%%%%%%%%%%%%%%%%

\section{Hardness}\label{sec:hardness}
This section presents our hardness results. Section \ref{sec:previous_hardness} states several useful tools from literature. Section \ref{sec:np_result} shows that, it is {\bf NP-hard} to get some multiplicative error. Assuming \ETH~is true, we provide a stronger hardness result in Section \ref{sec:eth_result}. Section \ref{sec:eth_rankk_result} extends the result from the $\rank$-$1$ case to the $\rank$-$k$ case. 

\subsection{Previous results}\label{sec:previous_hardness}

\begin{definition}[$\|A\|_{\infty\rightarrow 1}$,\cite{gv15}]
Given matrix $A\in\mathbb{R}^{n\times d}$,
\begin{align*}
\|A\|_{\infty\rightarrow 1}=\min_{x\in\{-1,+1\}^n,y\in\{-1,+1\}^d} x^\top A y.
\end{align*}
\end{definition}

The following lemma says that computing $\|A\|_{\infty\rightarrow 1}$ for matrix $A$ with entries in $\{-1,+1\}$ is equivalent to computing a best $\{-1,+1\}$ matrix which is an $\ell_1$ norm rank-1 approximation to $A$.
\begin{lemma}[Lemma 3 of \cite{gv15}]\label{lem:sum}
Given matrix $A\in\{-1,+1\}^{n\times d}$,
\begin{align*}
\|A\|_{\infty\rightarrow 1}+\min_{x\in\{-1,+1\}^n,y\in\{-1,+1\}^d} \|A-xy^\top\|_1=nd.
\end{align*}
\end{lemma}

\begin{lemma}[Theorem 2 of \cite{gv15}]\label{lem:equiv}
Given $A\in\{-1,+1\}^{n\times d}$, we have
\begin{align*}
\min_{x\in\{-1,+1\}^n,y\in\{-1,+1\}^d} \|A-xy^\top\|_1=\min_{x\in\mathbb{R}^n,y\in\mathbb{R}^d}\|A-xy^\top\|_1 .
\end{align*}
\end{lemma}

Combining with Lemma~\ref{lem:sum} and Lemma~\ref{lem:equiv}, it implies that computing $\|A\|_{\infty\rightarrow 1}$ for $A\in\{-1,+1\}^{n\times d}$ is equivalent to computing the best $\ell_1$ norm rank-$1$ approximation to the matrix $A$:
\begin{align*}
\|A\|_{\infty\rightarrow 1}+\min_{x\in\mathbb{R}^n,y\in\mathbb{R}^d} \|A-xy^\top\|_1=nd.
\end{align*}

\begin{theorem}[{\bf NP-hard} result, Theorem 1 of \cite{gv15}]
Computing $\|A\|_{\infty\rightarrow 1}$ for matrix $A\in\{-1,+1\}^{n\times d}$ is {\bf NP-hard}.
\end{theorem}

The proof of the above theorem in \cite{gv15} is based on the reduction from \MAXCUT problem. The above theorem implies that computing $\min_{x\in\mathbb{R}^n,y\in\mathbb{R}^d} \|A-xy^\top\|_1$ is also {\bf NP-hard}.

\subsection{Extension to multiplicative error $\ell_1$-low rank approximation}\label{sec:np_result}
%Firstly, we recall the reduction from MAX-CUT to computing $\|\cdot\|_{\infty\rightarrow 1}$ for $\{-1,+1\}$ matrices.
The previous result only shows that solving the exact problem is {\bf NP-hard}. This section presents a stronger hardness result, which says that, it is still {\bf NP-hard} even if the goal is to find a solution that is able to achieve some multiplicative error. The proof in this section and the next section are based on the reduction from the \MAXCUT problem. For recent progress on \MAXCUT problem, we refer the readers to \cite{gw95,bgs98,tssw00,h01,kkmo07,flp15}.
\begin{theorem}\label{thm:main}
Given $A\in\{-1,+1\}^{n\times d}$, computing an $\wh{x}\in\mathbb{R}^n,\wh{y}\in\mathbb{R}^d$ s.t.
\begin{align*}
\|A-\wh{x}^\top\wh{y}\|_1 \leq (1+\frac{1}{nd})\min_{x\in\mathbb{R}^n,y\in\mathbb{R}^d} \|A-x^\top y\|_1
\end{align*}
is {\bf NP-hard}.
\end{theorem}
%\cite{psw16,nn14}
%\Peilin{Actually, we can prove the hardness when approximation ratio is $(1+\frac{1}{(nd)^c})$ for some constant $0<c<1$}

%Gap MAX-CUT decision problem: Given a positive integer $c^*$ and an unweighted graph $G=(V,E)$ where $V$ is the set of vertices of $G$ and $E$ is the set of edges of $G$, and there is a promise on the input $c^*$ that either $c^*$ is greater than the number of edges of max cut of $G$ or $c^*$ is smaller than the $\frac{16}{17}$ of number of edges of max cut of $G$, the goal is to determine whether there is a cut of $G$ has at least $c^*$ edges.

%\begin{lemma}[\cite{TSSW00}]
%Gap MAX-CUT decision problem is NP-hard.
%\end{lemma}

\MAXCUT decision problem: Given a positive integer $c^*$ and an unweighted graph $G=(V,E)$ where $V$ is the set of vertices of $G$ and $E$ is the set of edges of $G$, the goal is to determine whether there is a cut of $G$ has at least $c^*$ edges.

\begin{lemma}
\MAXCUT decision problem is {\bf NP-hard}.
\end{lemma}

We give the definition for the Hadamard matrix,
\begin{definition}
%We use $H$ to denote a Hadamard matrix. 
The Hadamard matrix $H_p$ of size $p\times p$ is defined recursively : $ \begin{bmatrix} H_{p/2} & H_{p/2} \\ H_{p/2} & - H_{p/2} \end{bmatrix} $ with $H_2 = \begin{bmatrix} +1 & +1 \\ +1 & -1 \end{bmatrix}$.
\end{definition}
For simplicity, we use $H$ to denote $H_p$ in the rest of the proof.

%Recall that a 1-by-1 Hadamard matrix is just $(1)$, and $p$-by-$p$ Hadamard matrix $H$ is $$\left(\begin{array}{cc}H_0 & H_0\\ H_0 & -H_0\end{array}\right)$$ where $p$ is a power of $2$ and $H_0$ is a $p/2$-by-$p/2$ Hadamard matrix.

Recall the reduction shown in \cite{gv15} which is from \MAXCUT to computing $\|\cdot\|_{\infty\rightarrow 1}$ for $\{-1,+1\}$ matrices. We do the same thing: for a given graph $G=(V,E)$, we construct a matrix $A\in\{-1,+1\}^{n\times d}$ where $n=p|E|$ and $d=p|V|$. Notice that $p=\poly(|E|,|V|)$ is a parameter which will be determined later, and also $p$ is a power of $2$.

We divide the matrix $A$ into $|E|\times |V|$ blocks, and each block has size $p\times p$. For $e\in[|E|]$, if the $e^{th}$ edge has endpoints $i\in [|V|],j\in [|V|]$ and $i<j$, let all the $p\times p$ elements of $(e,i)$ block of $A$ be $1$, all the $p\times p$ elements of $(e,j)$ block of $A$ be $-1$, and all the $(e,l)$ block of $A$ be $p\times p$ Hadamard matrix $H$ for $l\neq i,j$.

\begin{claim}[Lower bound of $\|A\|_{\infty\rightarrow 1}$, proof of Theorem 1 of \cite{gv15}] \label{cla:lowerA2}
If there is a cut of $G$ with cut size at least $c$,
\begin{align*}
\|A\|_{\infty\rightarrow 1}\geq 2p^2c-|E||V|p^{3/2}.
\end{align*}
\end{claim}

\begin{claim}[Upper bound of $\|A\|_{\infty\rightarrow 1}$, proof of Theorem 1 of \cite{gv15}] \label{cla:upperA2}
If the max cut of $G$ has fewer than $c$ edges,
\begin{align*}
\|A\|_{\infty\rightarrow 1}\leq 2p^2(c-1)+|E||V|p^{3/2}.
\end{align*}
\end{claim}

\begin{remark}
In \cite{gv15}, they set $p$ as a power of $2$ and $p>|E|^2|V|^2$. This implies
\begin{align*}
\forall c\in [|E|],2p^2(c-1)+|E||V|p^{3/2}<2p^2c-|E||V|p^{3/2}.
\end{align*}
Therefore, according to Claim~\ref{cla:lowerA2} and Claim~\ref{cla:upperA2}, if we can know the precise value of $\|A\|_{\infty\rightarrow 1}$, we can decide whether $G$ has a cut with cut size at least $c^*$.
\end{remark}

%\begin{claim}[Lower bound of $\|A\|_{\infty\rightarrow 1}$ in Gap MAX-CUT decision problem] \label{cla:lowerA2}
%If parameter $c^*$ satisfies the promise of input of Gap MAX-CUT decision problem, and there is a cut of $G$ with cut size at least $c^*$, then
%$$\|A\|_{\infty\rightarrow 1}\geq 2p^2\cdot \frac{17}{16}c^*-|E||V|p^{3/2}$$
%\end{claim}
%\begin{proof}

%\end{proof}

For convenience, we use $T^*$ to denote $\|A\|_{\infty\rightarrow1}$ and use $L^*$ to denote
\begin{align*}
\min_{x\in\mathbb{R}^n,y\in\mathbb{R}^d} \|A-x^\top y\|_1.
\end{align*}
Also, we use $L$ to denote a $(1+\frac{1}{nd})$ relative error approximation to $L^*$, which means:
\begin{align*}
L^*\leq L \leq (1+\frac1{nd})L^*.
\end{align*}
We denote $T$ as $nd-L$.

\begin{proof}[Proof of Theorem~\ref{thm:main}]

Because $L^*\leq L \leq (1+\frac1{nd})L^*$, we have:
\begin{align*}
nd-L^* \geq nd-L \geq nd-(1+\frac1{nd})L^*.
\end{align*}
Due to Lemma~\ref{lem:sum} and the definition of $T$, it has:
\begin{align*}
T^*\geq T \geq T^*-\frac1{nd}L^*.
\end{align*}
Notice that $A$ is a $\{-1,+1\}$ matrix, we have
\begin{align*}
L^*\leq \|A\|_1\leq 2nd.
\end{align*}
Thus,
\begin{align*}
T^*\geq T \geq T^*-2.
\end{align*}
It means
\begin{align*}
T+2\geq T^* \geq T.
\end{align*}
According to Claim~\ref{cla:lowerA}, if $G$ has a cut with cut size at least $c$, we have:
\begin{align*}
T+2\geq T^*\geq 2p^2c-|E||V|p^{3/2}.
\end{align*}
That is
\begin{align*}
T\geq 2p^2c-|E||V|p^{3/2}-2.
\end{align*}
According to Claim~\ref{cla:upperA}, if the max cut of $G$ has fewer than $c$ edges,
\begin{align*}
T\leq T^* \leq 2p^2(c-1)+|E||V|p^{3/2}.
\end{align*}

Let $p$ be a power of $2$ and $p>|E|^3|V|^3$, we have
\begin{align*}
2p^2(c-1)+|E||V|p^{3/2}<2p^2c-|E||V|p^{3/2}-2.
\end{align*}
Therefore, we can decide whether $G$ has a cut with size at least $c$ based on the value of $T$.

Thus, if we can compute $\wh{x}\in\mathbb{R}^n,\wh{y}\in\mathbb{R}^d$ s.t.
\begin{align*}
\|A-\wh{x}^\top\wh{y}\|_1 \leq (1+\frac{1}{nd})\min_{x\in\mathbb{R}^n,y\in\mathbb{R}^d} \|A-x^\top y\|_1,
\end{align*}
in polynomial time, it means we can compute $L$ and $T$ in polynomial time, and we can solve \MAXCUT decision problem via the value of $T$, which leads to a contradiction.
\end{proof}

\subsection{Using the \ETH~assumption}\label{sec:eth_result}
The goal of this section is to prove Theorem \ref{thm:hard_main2}. We first introduce the definition of \SAT and Exponential Time Hypothesis(\ETH). For the details and background of \SAT problem, we refer the readers to \cite{ab09}.

\begin{definition}[\SAT problem]\label{def:SAT_problem}
Given an $r$ variables and $m$ clauses conjunctive normal form \CNF formula with size of each clause at most $3$, the goal is to decide whether there exists an assignment for the $r$ boolean variables to make the \CNF formula be satisfied.
\end{definition}

\begin{hypothesis}[Exponential Time Hypothesis (\ETH)~\cite{IPZ98}]\label{hyp:ETH}
There is a $\delta>0$ such that \SAT problem defined in Definition~\ref{def:SAT_problem} cannot be solved in $O(2^{\delta r})$ running time.
\end{hypothesis}

The main lower bound is stated as follows:

\begin{theorem}\label{thm:hard_main2}
Unless \ETH (see Hypothesis~\ref{hyp:ETH}) fails, for arbitrarily small constant $\gamma >0$, given some matrix $A\in\{-1,+1\}^{n\times d}$, there is no algorithm can compute $\wh{x}\in\mathbb{R}^n,\wh{y}\in\mathbb{R}^d$ s.t.
\begin{align*}
\|A-\wh{x}^\top\wh{y}\|_1 \leq (1+\frac{1}{\log^{1+\gamma} nd})\min_{x\in\mathbb{R}^n,y\in\mathbb{R}^d} \|A-x^\top y\|_1,
\end{align*}
in $(nd)^{O(1)}$ running time. %Here $\gamma>0$ is an arbitrarily small constant.
\end{theorem}

Before we prove our lower bound, we introduce the following theorem which is used in our proof.

\begin{definition}[\MAXCUT decision problem]\label{def:max_cut}
Given a positive integer $c^*$ and an unweighted graph $G=(V,E)$ where $V$ is the set of vertices of $G$ and $E$ is the set of edges of $G$, the goal is to determine whether there is a cut of $G$ has at least $c^*$ edges.
\end{definition}

\begin{theorem}[Theorem 6.1 in \cite{flp15}]\label{thm:maxcut_sparse2}
There exist constants $a,b\in(0,1)$ and $a>b$, such that, for a given \MAXCUT (see Definition~\ref{def:max_cut}) instance graph $G=(E,V)$ which is an $n$-vertices 5-regular graph, if there is an algorithm in time $2^{o(n)}$ which can distinguish the following two cases:
\begin{enumerate}
\item At least one cut of the instance has at least $a|E|$ edges,
\item All cuts of the instance have at most $b|E|$ edges,
\end{enumerate}
then \ETH (see Hypothesis~\ref{hyp:ETH}) fails.
\end{theorem}

%\begin{corollary}[Hardness for sparse graphs]\label{cor:maxcut_sparse}
%Give a positive integer $c^*$ and a MAX-CUT (see Definition~\ref{def:max_cut}) instance $G$ which has $n$ vertices and $O(n)$ edges. If ETH (see Hypothesis~\ref{hyp:ETH}) is true, there is no algorithm can decide whether $G$ has a cut at least $c^*$ in time $2^{n^{1-\varepsilon}}$ for any $\varepsilon>0$.
%\end{corollary}

\begin{proof}[Proof of Theorem~\ref{thm:hard_main2}]
We prove it by contradiction. We assume, for any given $A\in\{-1,+1\}^{n\times d}$, there is an algorithm can compute $\wh{x}\in\mathbb{R}^n,\wh{y}\in\mathbb{R}^d$ s.t.
\begin{align*}
\|A-\wh{x}^\top\wh{y}\|_1 \leq (1+\frac{1}{W})\min_{x\in\mathbb{R}^n,y\in\mathbb{R}^d} \|A-x^\top y\|_1,
\end{align*}
in time $\poly(nd)$, where $W=\log^{1+\gamma} d$ for arbitrarily small constant $\gamma>0$. Then, we show the following. There exist constants $a,b\in [0,1],a>b$, for a given \MAXCUT instance $G=(V,E)$ with $|E|=O(|V|)$, such that we can distinguish whether $G$ has a cut with size at least $a|E|$ or all the cuts of $G$ have size at most $b|E|$ in $2^{o(|V|)}$ time, which leads to a contradiction to Theorem~\ref{thm:maxcut_sparse2}.

%In the following, we use $[N]$ to denote the set $\{1,2,\cdots,N\}$ and $H$ to denote a Hadamard matrix. Recall that a 1-by-1 Hadamard matrix is just $(1)$, and $p$-by-$p$ Hadamard matrix $H$ is
%$$\left(\begin{array}{cc}H_0 & H_0\\ H_0 & -H_0\end{array}\right)$$
%where $p$ is a power of $2$ and $H_0$ is a $p/2$-by-$p/2$ Hadamard matrix.

Recall the reduction shown in \cite{gv15} which is from \MAXCUT to computing $\|\cdot\|_{\infty\rightarrow 1}$ for $\{-1,+1\}$ matrices. We do similar things here: for a given graph $G=(V,E)$ where $|E|=O(|V|)$, we construct a matrix $A\in\{-1,+1\}^{n\times d}$ where $n=p|E|$ and $d=p|V|$. Notice that $p$ is a parameter which will be determined later, and also $p$ is a power of $2$.

We divide the matrix $A$ into $|E|\times |V|$ blocks, and each block has size $p\times p$. For $e\in[|E|]$, if the $e^{th}$ edge has endpoints $i\in [|V|],j\in [|V|]$ and $i<j$, let all the $p\times p$ elements of $(e,i)$ block of $A$ be $1$, all the $p\times p$ elements of $(e,j)$ block of $A$ be $-1$, and all the $(e,l)$ block of $A$ be $p\times p$ Hadamard matrix $H$ for $l\neq i,j$.

We can construct the matrix in $nd$ time, which is $p^2|E||V|$. We choose $p$ to be the smallest number of power of $2$ which is larger than $2^{\frac{2}{a-b}|V|^{1-\frac\gamma{10}}}$ for some $\gamma>0$. Thus, the time for construction of the matrix $A$ is $O(nd)=2^{O(|V|^{1-\frac{\gamma}{10}})}$.

We will show, if we can compute a $(1+1/W)$-approximation to $A$, we can decide whether $G$ has a cut with size at least $a|E|$ or has no cut with size larger than $b|E|$. For convenience, we use $T^*$ to denote $\|A\|_{\infty\rightarrow1}$ and use $L^*$ to denote
\begin{align*}
\min_{x\in\mathbb{R}^n,y\in\mathbb{R}^d} \|A-x^\top y\|_1.
\end{align*}
Also, we use $L$ to denote a $(1+\frac{1}{W})$ relative error approximation to $L^*$, which means:
\begin{align*}
L^*\leq L \leq (1+\frac1{W})L^*.
\end{align*}
We denote $T$ as $nd-L$.

Because $L^*\leq L \leq (1+\frac1{W})L^*$, we have:
\begin{align*}
nd-L^* \geq nd-L \geq nd-(1+\frac1{W})L^*.
\end{align*}
Due to Lemma~\ref{lem:sum} and the definition of $T$, it has:
\begin{align*}
T^*\geq T \geq T^*-\frac1{W}L^*.
\end{align*}
Notice that $A$ is a $\{-1,+1\}$ matrix, we have
\begin{align*}
L^*\leq \|A\|_1\leq 2nd.
\end{align*}
Thus,
\begin{align*}
T^*\geq T \geq T^*-2nd/W.
\end{align*}
It means
\begin{align*} 
T+2nd/W\geq T^* \geq T.
\end{align*}

\begin{claim}[Lower bound of $\|A\|_{\infty\rightarrow 1}$, proof of Theorem 1 of \cite{gv15}] \label{cla:lowerA}
If there is a cut of $G$ with cut size at least $c$,
\begin{align*}
\|A\|_{\infty\rightarrow 1}\geq 2p^2c-|E||V|p^{3/2}.
\end{align*}
\end{claim}

\begin{claim}[Upper bound of $\|A\|_{\infty\rightarrow 1}$, proof of Theorem 1 of \cite{gv15}] \label{cla:upperA}
If the max cut of $G$ has fewer than $c$ edges,
\begin{align*}
\|A\|_{\infty\rightarrow 1}\leq 2p^2(c-1)+|E||V|p^{3/2}.
\end{align*}
\end{claim}

According to Claim~\ref{cla:lowerA}, if $G$ has a cut with cut size at least $a|E|$, we have:
\begin{align*}
T+2nd/W\geq T^*\geq 2p^2a|E|-|E||V|p^{3/2}.
\end{align*}
That is
\begin{align}\label{eq:right_hand}
T\geq 2p^2a|E|-|E||V|p^{3/2}-2nd/W.
\end{align}
According to Claim~\ref{cla:upperA}, if the max cut of $G$ has fewer than $b|E|$ edges,
\begin{align}\label{eq:left_hand}
T\leq T^* \leq 2p^2b|E|+|E||V|p^{3/2}.
\end{align}

Using these conditions $p\geq 2^{\frac2{a-b}|V|^{1-\frac{\gamma}{10}}},d=p|V|,W\geq\log^{1+\gamma} d$, we can lower bound $|W|$ by $|V|$ up to some constant,
\begin{align}\label{eq:hardness_lower_bound_W}
W &  \geq ~  \log^{1+\gamma} d & \text{~by~} W\geq\log^{1+\gamma} d \notag \\
  &= ~ \log^{1+\gamma} (p|V|) &\text{~by~} d=p|V| \notag \\
  &= ~ (\log |V|+\log p)^{1+\gamma} \notag \\
  &\geq ~ (\log |V|+\frac2{a-b}|V|^{1-\frac{\gamma}{10}})^{1+\gamma} & \text{~by~} p\geq 2^{\frac2{a-b}|V|^{1-\frac{\gamma}{10}}} \notag \\
  & \geq~   \frac{2}{a-b}|V|. & \text{~by~} (1-\gamma/10)(1+\gamma)>1 \text{~for~} \gamma \text{~small enough~}
\end{align}
Thus, we can upper bound $1/W$ in the following sense,
\begin{align}\label{eq:hardness_upper_bound_frac1W}
\frac1W  \leq \frac{a-b}{2|V|} \leq \frac{a-b}{|V|}-p^{-\frac12},
\end{align}
where the first inequality follows by Equation (\ref{eq:hardness_lower_bound_W}) and the second inequality follows by $p\geq 2^{\frac2{a-b}|V|^{1-\frac{\gamma}{10}}}$, $\gamma$ is sufficient small, and $|V|$ is large enough.

Now, we can conclude,
\begin{align}
& ~\frac1W  \leq \frac{a-b}{|V|}-p^{-\frac12} \notag \\
\iff & ~p^2|E||V|/W \leq  (a-b)p^2|E|-|E||V|p^{\frac32} \notag\\
&\text{~by~multiplying~$p^2|E||V|$~on~both~sides}\notag \\
\iff & ~2nd/W\leq2(a-b)p^2|E|-2|E||V|p^{\frac32} \notag\\
 & \text{~by~multiplying~$2$~on~both~sides~and~$p^2|E||V|=nd$} \notag \\
\iff & ~2p^2b|E|+|E||V|p^{3/2}<2p^2a|E|-|E||V|p^{3/2}-2nd/W, \label{eq:final} \\
&\text{~by~adding~$2p^2b|E|+|E||V|p^{3/2}-2nd/W$~on~both~sides} \notag
\end{align}
which implies that Equation (\ref{eq:final}) is equivalent to Equation (\ref{eq:hardness_upper_bound_frac1W}).
Notice that the LHS of (\ref{eq:final}) is exactly the RHS of (\ref{eq:left_hand}) and the RHS of (\ref{eq:final}) is exactly the RHS of (\ref{eq:right_hand}).
Therefore, we can decide whether $G$ has a cut with size larger than $a|E|$ or has no cut with size larger than $b|E|$.

Thus, if we can compute $\wh{x}\in\mathbb{R}^n,\wh{y}\in\mathbb{R}^d$ s.t.
\begin{align*}
\|A-\wh{x}^\top\wh{y}\|_1 \leq (1+\frac{1}{\log^{1+\gamma} d})\min_{x\in\mathbb{R}^n,y\in\mathbb{R}^d} \|A-x^\top y\|_1,
\end{align*}
in $\poly(nd)$ time, which means we can compute $L$ and $T$ in $\poly(nd)$ time. Notice that $nd=p^2|E||V|,|E|=O(|V|),p\geq 2^{\frac2{a-b}|V|^{1-\frac{\gamma}{10}}}$, it means $\poly(nd)=2^{O(|V|^{1-\frac{\gamma}{10}})}$.
 Because we decide whether $G$ has a cut with size larger than $a|E|$ or has no cut with size larger than $b|E|$ via the value of $T$, we can solve it in $2^{O(|V|^{1-\frac{\gamma}{10}})}$ time which leads to a contradiction to Theorem~\ref{thm:maxcut_sparse2}.
\end{proof}

\subsection{Extension to the rank-$k$ case}\label{sec:eth_rankk_result}

This section presents a way of reducing the rank-$k$ case to the rank-$1$ case. Thus, we can obtain a lower bound for general $k\geq 1$ under \ETH.

\begin{theorem}
For any constants $c_1>0,c_2>0$ and $c_3>0$, and any constant $c_4 \geq 10(c_1+c_2+c_3+1)$, given any matrix $A\in \mathbb{R}^{n\times n}$ with absolute value of each entry bounded by $n^{c_1}$, we define a block diagonal matrix $\wt{A}\in \mathbb{R}^{(n+k-1)\times (n+k-1)}$ as
\begin{align*}
\wt{A}=\begin{bmatrix} A & 0 & 0 & \cdots & 0 \\ 0 & B & 0 & \cdots & 0\\ 0 & 0 & B & \cdots & 0 \\ \cdots & \cdots & \cdots & \cdots & \cdots\\ 0 & 0 & 0 & \cdots & B \end{bmatrix},
\end{align*}
where $B=n^{c_4}$.
If $\wh{A}$ is an $\ell_1$-norm rank-$k$ $C$-approximation solution to $\wt{A}$, i.e.,
\begin{align*}
\|\wh{A}-\wt{A}\|_1\leq C\cdot \min_{\rank-k\ \wh{A}'} \|\wh{A}'-\wt{A}\|_1,
\end{align*}
where $ C \in [1, n^{c_3}]$, then there must exist $j^*\in[n]$ such that
\begin{align*}
\min_{v\in\mathbb{R}^{n}} \|\wh{A}_{j^*}^{[1:n]}v^\top-A\|_1\leq C\cdot \min_{u,v\in\mathbb{R}^{n}}\|uv^\top-A\|_1+1/n^{c_2},
\end{align*}
i.e., the first $n$ coordinates of the column $j^*$ of $\wh{A}$ can give an $\ell_1$-norm rank-$1$ $C$-approximation to $A$.
\end{theorem}

\begin{proof}
The first observation is that because we can use a rank-$1$ matrix to fit $A$ and use a rank-$(k-1)$ matrix to fit other $B$s, we have
\begin{equation}\label{eq:relations_between_rank1_and_rankk}
\min_{\rank-k\ \wh{A}'} \|\wh{A}'-\wt{A}\|_1\leq \min_{u,v\in\mathbb{R}^{n}}\|uv^\top-A\|_1 \leq \|A\|_1.
\end{equation}
%Now, we define 
\begin{claim}\label{cla:rank_kminus1}
Let $\wh{A}$ denote the $\rank$-$k$ $C$-approximate solution to $\wt{A}$.  Let $Z \in \mathbb{R}^{(n+k-1)\times (k-1)}$ denote the rightmost $k-1$ columns of $\wh{A}$, then, $\rank(Z)=k-1.$ %where $\wh{A}$ satisfies
%$$\|\wh{A}-\wt{A}\|_1\leq C\cdot \min_{\rank-k\ \wh{A}'} \|\wh{A}'-\wt{A}\|_1.$$ 
\end{claim}
\begin{proof}
Consider the $(k-1)\times(k-1)$ submatrix $Z^{[n+1:n+k-1]}$ of $Z$. Each element on the diagonal of this submatrix should be at least $B-C\|A\|_1$, and each element not on the diagonal of the submatrix should be at most $C\|A\|_1$. Otherwise $\|\wh{A}-\wt{A}\|_1>C\|A\|_1$ which will lead to a contradiction. Since $B=n^{c_4}$ is sufficiently large, $Z^{[n+1:n+k-1]}$ is diagonally dominant. Thus $\rank(Z)\geq \rank(Z^{[n+1:n+k-1]})=k-1$. Because $Z$ only has $k-1$ columns, $\rank(Z)=k-1$.
\end{proof}

\begin{claim}\label{cla:large_gap}
$\forall x\in\mathbb{R}^{k-1},i\in [n],\exists j\in \{n+1,n+2,\cdots,n+k-1\}$ such that,
 \begin{align*} 
\frac{|(Zx)_j|}{|(Zx)_i| }\geq \frac{B}{2(k-1)C\|A\|_1}.
\end{align*}
\end{claim}
\begin{proof}
Without loss of generality, we can let $\|x\|_1=1$. Thus, there exists $j$ such that $|x_j|\geq \frac1{k-1}$. So we have
\begin{align*}
|(Zx)_{n+j}|&=|\sum_{i=1}^{k-1}Z_{n+j,i}x_i|\\
&\geq |Z_{n+j,j}x_j|-\sum_{i\not= j} |Z_{n+j,i}x_i|\\
&\geq (B-C\|A\|_1)|x_j|-\sum_{i\not=j}|x_i|C\|A\|_1\\
&\geq (B-C\|A\|_1)/(k-1)-C\|A\|_1\\
&\geq \frac{B}{2(k-1)}.
\end{align*}
The second inequality follows because $|Z_{n+j,j}|\geq B-C\|A\|_1$ and $\forall i\not=j,|Z_{n+j,i}|\leq C\|A\|_1$ (otherwise $\|\wh{A}-\wt{A}\|_1>C\|A\|_1$ which leads to a contradiction.) The third inequality follows from $|x_j|>1/(k-1)$ and $\|x\|_1=1$. The fourth inequality follows since $B$ is large enough such that $\frac{B}{2(k-1)}\geq \frac{C\|A\|_1}{k-1}+C\|A\|_1$.

Now we consider any $q\in[n]$. We have
\begin{align*}
|(Zx)_q|=\sum_{i=1}^{k-1}|Z_{q,i}x_i|\leq \max_{i\in [k-1]} |Z_{q,i}|\cdot\sum_{i=1}^{k-1}|x_i|\leq C\|A\|_1.
\end{align*}
The last inequality follows that $\|x\|_1=1$ and $\forall i\in[k-1],Z_{q,i}\leq C\|A\|_1$. Otherwise, $\|\wh{A}-\wt{A}\|_1>C\|A\|_1$ which will lead to a contradiction.

Look at $|(Zx)_{n+j}|/|(Zx)_q|$, it is greater than $\frac{B}{2(k-1)C\|A\|_1}$.

\end{proof}

We now look at the submatrix $\wh{A}_{[1:n]}^{[1:n]}$, we choose $i^*,j^*\in[n]$ such that
\begin{align*}
|\wh{A}_{i^*,j^*}|\geq 1/n^{c_2-2}.
\end{align*}
If there is no such $(i^*,j^*)$, it means that we already found a good rank-$1$ approximation to $A$
\begin{align*}
\|\text{\bf 0}-A\|_1 &\leq \|\text{\bf 0}-\wh{A}_{[1:n]}^{[1:n]}\|_1+\|\wh{A}_{[1:n]}^{[1:n]}-A\|_1\\
& \leq \|\text{\bf 0}-\wh{A}_{[1:n]}^{[1:n]}\|_1 + \|\wh{A}-\wt{A}\|_1\\
& \leq 1/n^{c_2}+\|\wh{A}-\wt{A}\|_1\\
& \leq 1/n^{c_2}+C\min_{\rank-k\ \wh{A}'} \|\wh{A}'-\wt{A}\|_1\\
& \leq 1/n^{c_2}+C\min_{u,v\in\mathbb{R}^{n}}\|uv^\top-A\|_1,
\end{align*}
where $\text{\bf 0}$ is an $n\times n$ all zeros matrix. The second inequality follows since $\wh{A}_{[1:n]}^{[1:n]}-A$ is a submatrix of $\wh{A}-\wt{A}$. The third inequality follows since each entry of $\wh{A}_{[1:n]}^{[1:n]}$ should be no greater than $1/n^{c_2-2}$ (otherwise, we can find $(i^*,j^*)$). The last inequality follows from equation~\ref{eq:relations_between_rank1_and_rankk}.

\begin{claim}\label{cla:not_in_the_span}
$\wh{A}_{j^*}$ is not in the column span of $Z$, i.e.,
\begin{align*}
\forall x\in\mathbb{R}^{k-1}, \wh{A}_{j^*}\not= Zx .
\end{align*}
\end{claim}

\begin{proof}
If there is an $x$ such that $\wh{A}_{j^*}=Zx$, it means $(Zx)_{i^*}=\wh{A}_{i^*,j^*}\geq 1/n^{c_2-2}$. Due to Claim~\ref{cla:large_gap}, there must exist $i'\in\{n+1,n+2,\cdots,n+k-1\}$ such that $(Zx)_{i'}\geq \frac{B}{2(k-1)C\|A\|_1 n^{c_2-2}}$. Since $B$ is sufficiently large, $\wh{A}_{i',j^*}=(Zx)_{i'}>C\|A\|_1$ which implies that $\|\wh{A}-\wt{A}\|_1>C\|A\|_1$, and so leads to a contradiction.
\end{proof}

Due to Claim~\ref{cla:rank_kminus1} and Claim~\ref{cla:not_in_the_span}, the dimension of the subspace spanned by $\wh{A}_{j^*}$ and the column space of $Z$ is $k$. Since $\wh{A}$ has rank at most $k$, it means that each column of $\wh{A}$ can be written as a linear combination of $\wh{A}_{j^*}$ and the columns of $Z$.

Now consider the $j^{th}$ column $\wh{A}_j$ of $\wh{A}$ for $j\in[n]$. We write it as
\begin{align*}
\wh{A}_j=\alpha_j\cdot \wh{A}_{j^*}+Zx^j.
\end{align*}

\begin{claim}\label{cla:no_large_coef}
$\forall j\in[n],\alpha_j\leq 2C\|A\|_1 n^{c_2+2}$.
\end{claim}

\begin{proof}
Otherwise, suppose $\alpha_j>2C\|A\|_1 n^{c_2+2}$. We have
\begin{align*}
|(Zx^j)_{i^*}|&\geq \alpha_j\cdot |\wh{A}_{i^*,j^*}|-|\wh{A}_{i^*,j}|\\
& \geq \alpha_j\cdot \frac1{n^{c_2-2}}- |\wh{A}_{i^*,j}|\\
& \geq \frac12\alpha_j\cdot \frac1{n^{c_2-2}} .
\end{align*}
The second inequality follows from $|\wh{A}_{i^*,j^*}|\geq 1/n^{c2-2}$. The third inequality follows from $|\wh{A}_{i^*,j}|\leq \|A\|_1$ and $\frac12\alpha_j\cdot \frac1{n^{c_2-2}}\geq C\|A\|_1\geq \|A\|_1$.

Due to Claim~\ref{cla:large_gap}, there exists $i\in\{n+1,n+2,\cdots,n+k-1\}$ such that $|(Zx^j)_i|\geq \frac{B}{2(k-1)C\|A\|_1}\cdot \frac12\alpha_j\cdot \frac1{n^{c_2-2}}$. For sufficiently large $B$, we can have $|(Zx^j)_i|\geq \alpha_jB^{1/2}$. Then we look at
\begin{align*}
|\wh{A}_{i,j}|&\geq |(Zx^j)_i|-\alpha_j |\wh{A}_{i,j^*}|\\
&\geq \alpha_j (B^{1/2}-C\|A\|_1)\\
&\geq \alpha_j\frac12 B^{1/2} .
\end{align*}
The second inequality follows by $|\wh{A}_{i,j^*}|\leq C\|A\|_1$, otherwise $\|\wh{A}-\wt{A}\|_1>C\|A\|_1$ which will lead to a contradiction. The third inequality follows that $B$ is sufficient large that $\frac12 B^{1/2}>C\|A\|_1$.

Since $|\wh{A}_{i,j}|\geq\alpha_j\frac12 B^{1/2}>C\|A\|_1$, it contradicts to the fact $\|\wh{A}-\wt{A}\|_1\leq C\|A\|_1$.

Therefore, $\forall j\in[n],\alpha_j\leq 2C\|A\|_1 n^{c_2+2}$.
\end{proof}

\begin{claim}\label{cla:kind_of_tail_bound}
$\forall j\in[n],i\in \{n+1,n+2,\cdots,n+k-1\},|(Zx^j)_i|\leq 4C^2\|A\|_1^2 n^{c_2+2}$
\end{claim}

\begin{proof}
Consider $j\in[n],i\in \{n+1,n+2,\cdots,n+k-1\}$, we have
\begin{align*}
|(Zx^j)_i|&\leq |\wh{A}_{i,j}|+\alpha_j |\wh{A}_{i,j^*}|\\
&\leq C\|A\|_1+\alpha_j\cdot C\|A\|_1\\
&\leq 4C^2\|A\|_1^2 n^{c_2+2} .
\end{align*}
The second inequality follows by $|\wh{A}_{i,j}|\leq C\|A\|_1$ and $|\wh{A}_{i,j^*}|\leq C\|A\|_1$, otherwise the $\|\wh{A}-\wt{A}\|_1$ will be too large and leads to a contradiction. The third inequality is due to $\alpha_j+ 1 \leq 4C\|A\|_1 n^{c_2+2}$ via Claim~\ref{cla:no_large_coef}.
\end{proof}

\begin{claim}\label{cla:closeness}
$\forall j\in [n], \|\wh{A}_j^{[1:n]}-\alpha_j\cdot\wh{A}_{j^*}^{1:n}\|_{\infty}\leq 1/n^{c_2-2}$
\end{claim}

\begin{proof}

Due to Claim~\ref{cla:kind_of_tail_bound} and Claim~\ref{cla:large_gap}, $\forall i,j\in[n],$ we have
\begin{align*}
|(Zx^j)_i|\leq \frac{4C^2\|A\|_1^2 n^{c_2+2}}{B/(2(k-1)C\|A\|_1)}\leq 1/B^{1/2}.
\end{align*}
The second inequality follows for a large enough $B$.

Therefore, $\forall i,j\in[n],$
\begin{align*}
|\wh{A}_{i,j}-\alpha_j\cdot\wh{A}_{i,j^*}|\leq |(Zx^j)_i|\leq 1/B^{1/2}\leq 1/n^{c_2-2} .
\end{align*}
The last inequality follows since $B$ is large enough.

\end{proof}

Now, let us show that $\wh{A}_{j^*}^{1:n}$ can provide a good rank-$1$ approximation to $A$:
\begin{align*}
\|\wh{A}_{j^*}^{1:n}\alpha^\top-A\|_1 &\leq \|\wh{A}_{j^*}^{1:n}\alpha^\top-\wh{A}_{[1:n]}^{[1:n]}\|_1+\|\wh{A}_{[1:n]}^{[1:n]}-A\|_1\\
&= \sum_{j=1}^n \|\alpha_j\wh{A}_{j^*}-\wh{A}_{j}^{[1:n]}\|_1+\|\wh{A}_{[1:n]}^{[1:n]}-A\|_1\\
&\leq n^2\cdot 1/n^{c_2-2}+\|\wh{A}_{[1:n]}^{[1:n]}-A\|_1\\
&\leq 1/n^{c_2}+\|\wh{A}-\wt{A}\|_1\\
&\leq 1/n^{c_2}+C\min_{u,v\in\mathbb{R}^{n}}\|uv^\top-A\|_1.
\end{align*}
The first inequality follows by triangle inequality. The first equality is due to the linearity of $\ell_1$ norm. The second inequality is due to Claim~\ref{cla:closeness}. The third inequality follows since $\wh{A}_{[1:n]}^{[1:n]}-A$ is a submatrix of $\wh{A}-\wt{A}$. The fourth inequality is due to the equation~\ref{eq:relations_between_rank1_and_rankk}. 

\end{proof}

\section{Limited Independent Cauchy Random Variables}\label{sec:limind}
This section presents the fundamental lemmas with limited independent Cauchy variables, which will be used in Section~\ref{sec:streaming} and \ref{sec:distributed}. In Section~\ref{sec:lim_notation}, we provide some notation, definitions and tools from previous work. Section~\ref{sec:lim_l1_cauchy} includes the main result.
%\Zhao{Below is David's original email, need to be polished more}
%Hey, I think limited independent should work for the Cauchy random variables and therefore we should be able to do this in the distributed setting.

\subsection{Notations and tools}\label{sec:lim_notation}

For a function $f: \mathbb{R} \rightarrow \mathbb{R}$ and nonnegative integer $\ell$, $f^{(\ell)}$ denotes the $\ell$th derivative of $f$, with $f^{(0)} = f$. We also often use $x\approx_{\eps} y$ to state that $|x-y| = O(\eps)$. We use $I_{[a,b]}$ to denote the indicator function of the interval $[a,b]$.

To optimize the communication complexity of our distributed algorithm, we show that instead of using fully independent Cauchy variables, $\poly(k,d)$-wise independent Cauchy variables suffice.

We start by stating two useful Lemmas from previous work \cite{knw10}.

\begin{lemma}[Lemma 2.2 in \cite{knw10}]\label{lem:lemma2_2_in_KNW10}
There exists an $\epsilon_0 > 0$ such that the following holds. Let $n$ be a positive integer
and $0 < \epsilon < \epsilon_0$, $0< p< 2$ be given. Let $f: \mathbb{R} \rightarrow \mathbb{R}$
satisfy $\| f^{(\ell)} \|_{\infty} = O(\alpha^{\ell})$ for all $\ell \geq 0$, for some $\alpha$
satisfying $\alpha^p \geq \log(1/\epsilon)$. Let $k = \alpha^p$. Let $a \in \mathbb{R}^n$ satisfy
$\| a \|_p = O(1)$. Let $X_i$ be a $3Ck$-independent family of $p$-stable random variables.
Let $X= \sum_{i} a_i X_i$ and $Y = \sum_{i} a_i Y_i$. Then $\E[f(x)] = \E[f(Y)]+O(\epsilon)$.
\end{lemma}

\begin{lemma}[Lemma 2.5 in \cite{knw10}]
There exist constants $c',\epsilon_0 >0$ such that for all $c > 0$ and $0 < \epsilon < \epsilon_0$,
and for all $[a,b] \subseteq \mathbb{R}$, there exists a function $J_{[a,b]}^c : \mathbb{R} \rightarrow \mathbb{R}$
satisfying:

i. $\| (J^c_{[a,b]})^{(\ell)} \|_{\infty} = O(c^\ell)$ for all $\ell \geq 0$.

ii. For all $x$ such that $a,b\notin [x-\epsilon, x+\epsilon]$, and as long as $c> c' \epsilon^{-1} \log^3(1/\epsilon)$,
$| J_{[a,b]}^c(x) - I_{[a,b]} (x) | < \epsilon$.
\end{lemma}

%Consider a $t \times n$ matrix of i.i.d. Cauchy's $S$ where $t=O(k\log k)$. We need the property that
%\begin{equation*}
%\| Sy \|_1 \geq \frac{t}{100} \| y \|_1,
%\end{equation*}
% which is implied by $| (Sy)_i| \geq \frac{1}{50} \| y \|_1$ for at least half the values of $i$ in $[t]$.
\subsection{Analysis of limited independent random Cauchy variables}\label{sec:lim_l1_cauchy}
\begin{lemma}\label{lem:limited_no_contraction}
Given a vector $y \in \mathbb{R}^n$, choose $Z$ to be the $t\times n$ random Cauchy matrices with $1/t$ rescaling and $t=O(k\log k)$. The variables from different rows are fully independent, and the variables from the same rows are $O(1)$-wise independent. Then, we have 
\begin{equation*}
 \| Z y \|_1 \gtrsim  \|y \|_1
\end{equation*}
holds with probability at least $1- 2^{-\Omega(t)}$.
\end{lemma}
\begin{proof}
Let $S$ denote the original fully independent matrix and $Z$ denote the matrix for which the entries in the same row are $w$-wise independent, and the entries from different rows are fully independent. Notice we define the random matrices without rescaling by $1/t$ and it will be added back at the end. (We will decide $w$ later)

 %We have $t$ rows, so we can afford to make the rows of $S$ independent. In each row of $S$, we use random variables being $\poly(k)$-wise independent. %Renormalizing by $\| y \|_1$,
%Consider the indicator variable $I_i$ which is 1 if and only if $\|(Sy)_i\|_1 \leq \frac{1}{50}$ Then applying Lemma 2.2 and Lemma 2.5 here:

%http://researcher.watson.ibm.com/researcher/files/us-dpwoodru/knw10.pdf

%(see also the proof of Theorem 2.1 after Lemma 2.5)
We define random variable $X$ such that $X = 1$ if $| (Zy)_i | \leq \frac{1}{50}$ and $X=0$ otherwise. We also define random variable $Y$ such that $Y=1$ if $|(Sy)_i | \leq \frac{1}{50}$ and $Y=0$ otherwise. Then, we have
\begin{align*}
\E[X] & = \Pr \biggl[ | (Zy)_i | \leq \frac{1}{50} \biggr] = \E \biggl[ I_{ [ -\frac{1}{50}, \frac{1}{50}] }((Zy)_i) \biggr] \\
\E[Y] & = \Pr \biggl[ | (Sy)_i | \leq \frac{1}{50} \biggr] = \E \biggl[ I_{ [ -\frac{1}{50}, \frac{1}{50}] }((Sy)_i) \biggr]
\end{align*}
The goal is to show that $\E[X] \approx_{\eps} \E[Y]$. Following the same idea from \cite{knw10}, we need to argue this chain of inequalities,
\begin{align*}
\E[I_{[a,b]} (X)] \approx_{\eps} \E[J_{[a,b]}^c (X)] \approx_{\eps} \E[J_{[a,b]}^c (Y)] \approx_{\eps} \E[I_{[a,b]} (Y)]
\end{align*}

Using Lemma 2.2 and Lemma 2.5 from \cite{knw10}, choosing sufficiently small constant $\eps$(which implies $w=O(1)$), it follows that for each $i\in [t]$, we still have
\begin{align*}
\Pr\biggl[ | (Zy)_i| > \frac{1}{50} \| y\|_1 \biggr] \gtrsim \Pr \biggl[ | (Sy)_i| > \frac{1}{50} \|y\|_1 \biggr]  \geq 0.9
\end{align*}
Because all rows of $Z$ are fully independent, using the Chernoff bound we can get that
\begin{equation*}
\Pr \biggl[ \| Z y \|_1 \lesssim t \| y\|_1 \biggr] \leq \exp(-\Omega(t) )
\end{equation*}
%$\|Sy\|_1 = \Omega(k \log k)$ with probability $1-\exp(-k \log k)$
as we needed for the ``no contraction'' part of the net argument.
\end{proof}

For the  no dilation, we need to argue that
\begin{lemma}\label{lem:limited_no_dilation}
Given a set of vectors $\{ y_1, y_2, \dotsc, y_d\}$ where $y_i\in \mathbb{R}^n, \forall i\in [d]$, choose $Z$ to be the $t\times n$ random Cauchy matrices with $1/t$ rescaling and $t=O(k\log k)$, where the variables from different rows are fully independent, and the variables from the same rows are $w$-wise independent.

$\mathrm{\RN{1}}$. If $w=\wt{O}(dk)$, we have
\begin{equation*}
\sum_{i=1}^d \| Z y_i \|_1 \leq O(\log d) \sum_{i=1}^d \| y_i\|_1
\end{equation*}
holds with probability at least $.999$.

$\mathrm{\RN{2}}$. If If $w=\wt{O}(d)$, we have
\begin{equation*}
\sum_{i=1}^d \| Z y_i \|_1 \leq O(k\log d) \sum_{i=1}^d \| y_i\|_1
\end{equation*}
holds with probability at least $.999$.
\end{lemma}

\begin{proof}
Let $m=t$. Let $S\in\mathbb{R}^{m\times n}$ denote the original fully independent matrix and $Z$ denote the matrix that for each entry in the same row are $w$-wise independent, where the entries from different rows are fully independent. (We will decide on$w$ later)

Applying matrix $S$ to those fixed set of vectors, we have
\begin{align*}
\sum_{i=1}^d \|S y_i \|_1 = \sum_{i=1}^d \sum_{j=1}^m | \sum_{l=1}^n \frac{1}{m} S_{j,l} \cdot (y_i)_l | =  \frac{1}{m}  \sum_{i=1}^d \sum_{j=1}^m | \sum_{l=1}^n S_{j,l} \cdot (y_i)_l |
\end{align*}
Applying matrix $Z$ to those fixed set of vectors, we have a similar thing,
\begin{align*}
\sum_{i=1}^d \|Z y_i \|_1 = \sum_{i=1}^d \sum_{j=1}^m | \sum_{l=1}^n \frac{1}{m} Z_{j,l} \cdot (y_i)_l | =  \frac{1}{m}  \sum_{i=1}^d \sum_{j=1}^m | \sum_{l=1}^n Z_{j,l} \cdot (y_i)_l |
\end{align*}
The goal is to argue that, for any $i\in [d], j\in [m]$,
\begin{align*}
\E\biggl[ | \sum_{l=1}^n Z_{j,l} \cdot (y_i)_l |  \bigg| \xi \biggr] \lesssim \E\biggl[ | \sum_{l=1}^n S_{j,l} \cdot (y_i)_l |  \bigg| \xi\biggr] +\delta
\end{align*}
As long as $\delta$ is small enough, we are in a good shape.

Let $X= \frac{1}{\| y_i\|_1} \sum_{l=1}^n Z_{j,l} (y_i)_l$, and $Y=  \frac{1}{\| y_i\|_1}  \sum_{l=1}^m S_{j,l} (y_i)_l$. Let $D$ be the truncating threshold of each Cauchy random variable. Define $T=O(\log D)$. On one hand, we have
\begin{align}\label{eq:upper_bound_on_X_given_xi}
\E[ |X| | \xi] & \leq 2 \left( \Pr[X\in [0,1] | \xi] \cdot 1+  \sum_{j=0}^{T} \Pr[ X \in (2^j, 2^{j+1} ] | \xi ] \cdot 2^{j+1} \right) \notag \\
& \leq 2 \left( 1+  \sum_{j=0}^{T} \Pr[I_{(2^j, 2^{j+1})}(X) =1 | \xi ] \cdot 2^{j+1} \right)\notag \\
& = 2 \left( 1+ \sum_{j=0}^{T} \E[I_{(2^j, 2^{j+1})}(X) | \xi ] \cdot 2^{j+1} \right)
\end{align}
On the other hand, we can show
\begin{align}\label{eq:lower_bound_on_Y_given_xi}
\E[ |Y| | \xi] & \geq 2 \left( \Pr[X\in [0,1] | \xi] \cdot 0 + \sum_{j=0}^{T} \Pr[ Y\in (2^j,2^{j+1} ] | \xi ] \cdot 2^j  \right) \notag \\
& \geq  2\sum_{j=0}^{T} \Pr[ I_{ (2^j,2^{j+1} ]}(Y) = 1 | \xi ] \cdot 2^j \notag \\
& \geq 2\sum_{j=0}^{T} \E [I_{ (2^j,2^{j+1} ] }(Y) | \xi ] \cdot 2^j
\end{align}
Thus, we need to show that, for each $j$,
\begin{align}\label{eq:additive_error_for_each_interval}
 \E[I_{(2^j, 2^{j+1})}(X) | \xi ]\approx_{\eps} \E [I_{ (2^j,2^{j+1} ] }(Y) | \xi ]
\end{align}
Following the same idea from \cite{knw10}, we need to argue this chain of inequalities,
\begin{align*}
\E[I_{[a,b]} (X)] \approx_{\eps} \E[J_{[a,b]}^c (X)] \approx_{\eps} \E[J_{[a,b]}^c (Y)] \approx_{\eps} \E[I_{[a,b]} (Y)]
\end{align*}
We first show $\E[I_{[a,b]} (X)] \approx_{\eps} \E[J_{[a,b]}^c (X)] $. Notice that $I_{[a,b]}$ and $J_{[a,b]}$ are within $\eps$ everywhere except for two intervals of length $O(\eps)$. Also the Cauchy distribution is anticoncentrated (any length-$O(\eps)$ interval contains $O(\eps)$ probability mass) and $\| I_{[a,b]} \|_{\infty}, \| J_{[a,b]}^c \|_{\infty} =O(1)$, these intervals contribute $O(\eps)$ to the difference.

Second, we show $\E[J_{[a,b]}^c (X)] \approx_{\eps} \E[J_{[a,b]}^c (Y)]$. This directly follows by Lemma \ref{lem:lemma2_2_in_KNW10} by choosing $\alpha = O(\eps^{-1} \log^3 (1/\eps))$.

Third, we show $\E[J_{[a,b]}^c (Y)] \approx_{\eps} \E[I_{[a,b]} (Y)] $. The argument is similar as the first step, but we need to show anticoncentration of $Y$. Suppose for any $t\in \mathbb{R}$ we had a nonnegative function $f_{\eps,t}: \mathbb{R} \rightarrow \mathbb{R}$ symmetric about $t$ satisfying:
\begin{align*}
\text{\RN{1}}. ~&~ \| f^{(\ell)}_{t,\eps} \|_{\infty} = O(\alpha^{\ell}) \text{~for~all~} \ell \geq 0, \text{~with~} \alpha = O(1/\eps) \\
\text{\RN{2}}. ~&~ \E[ f_{t,\eps}(z) ] = O(\eps) \text{~for~} z \sim {\cal D}_1 \\
\text{\RN{3}}. ~&~ f_{t,\eps}(t+\eps) = \Omega(1) \\
\text{\RN{4}}. ~&~ f_{t,\eps}(x) \text{~is~strictly~decreasing~as~} |x-t| \rightarrow \infty
\end{align*}
By \RN{1}, \RN{2} and Lemma \ref{lem:lemma2_2_in_KNW10} we could have $\E[ f_{\eps,t} (Y) ] \approx_{\eps} \E[f_{t,\eps}(z)] = O(\eps)$. Then, $\E[f_{t,\eps}(Y) ] \geq f_{t,\eps}(t+\eps) \cdot \Pr[Y \in [t-\eps, t+\eps] ]= \Omega( \Pr[Y \in [t-\eps,t+\eps]] )$ by \RN{3} and \RN{4}, implying anticoncentration in $[t-\eps, t+\eps]$ as desired. For the details of function $f_{t,\eps}$, we refer the readers to Section A.4 in \cite{knw10}.

Now, combining Equation (\ref{eq:upper_bound_on_X_given_xi}), (\ref{eq:lower_bound_on_Y_given_xi}) and (\ref{eq:additive_error_for_each_interval}) gives
\begin{align*}
\E\biggl[ | \sum_{l=1}^n Z_{j,l} \cdot (y_i)_l |  \bigg| \xi \biggr] & \lesssim \E\biggl[ | \sum_{l=1}^n S_{j,l} \cdot (y_i)_l |  \bigg| \xi\biggr] + \sum_{j=0}^T 2^j \cdot \eps \cdot \| y_i \|_1 \\
& \lesssim \E\biggl[ | \sum_{l=1}^n S_{j,l} \cdot (y_i)_l |  \bigg| \xi\biggr] +D \cdot \eps \cdot \| y_i\|_1 \\
\end{align*}

Overall, for the fixed $j$, $S_{j,l}$ is $\wt{O}(1/\eps)$-independent family of Cauchy random variable. Choosing $D= O(dk)$ and $\eps = O(1/D)$, we can show
\begin{align*}
\frac{1}{m}\sum_{i=1}^d \sum_{j=1}^m \E\biggl[ | \sum_{l=1}^n S_{j,l} \cdot (y_i)_l |  \bigg| \xi\biggr] \leq O(\log d) \sum_{i=1}^d \| y_i \|_1
\end{align*}
as before. Notice that $D \eps \| y_i\|_1 = O(\| y_i\|_1)$. Thus, we complete the proof of the first result.

Choosing $D= O(dk)$ and $\eps = O(1/d)$, the dominant term becomes $D \eps \| y_i\|_1 = O(k \| y_i\|_1)$. Thus, we complete the proof of second result.
\end{proof}

\begin{corollary}\label{cor:limited_no_contraction}
Given $U\in\mathbb{R}^{n\times k}$, let $Z\in\mathbb{R}^{t\times n}$ be the same as the matrix stated in the Lemma~\ref{lem:limited_no_contraction}, then with probability at least $.95$,
\begin{align*}
\forall x\in\mathbb{R}^k,\|ZUx\|_1\gtrsim\|Ux\|_1.
\end{align*}
\end{corollary}
The proof is very similar to the proof of Lemma~\ref{lem:dense_cauchy_l1_k_subspace}.
Without loss of generality, we can suppose $U$ is a well-conditioned basis.
Due to Lemma~\ref{lem:limited_no_dilation}, with arbitrarily high constant probability $\|ZU\|_1$ is bounded by $\poly(t,k)$. By simply applying the net argument and using Lemma~\ref{lem:limited_no_contraction} to take a union bound over net points, we can get the above corollary.

\section{Streaming Setting}\label{sec:streaming}
Section~\ref{sec:streaming_definition} provides some notation and definitions about row-update streaming model and the turnstile streaming model. For some recent developments of row-update streaming and turnstile streaming models, we refer the readers to \cite{ cw09, kl11, gp13, lib13,klmms14,bwz16} and the references therein. Section~\ref{sec:streaming_turnstile} presents our turnstile streaming algorithm. Section~\ref{sec:streaming_rowupdate} presents our row-update streaming algorithm.

\subsection{Definitions}\label{sec:streaming_definition}

\begin{definition}[Row-update model]\label{def:model_rowupdate}
Let matrix $A\in\mathbb{R}^{n\times d}$ be a set of rows $A_1,\cdots,A_n$. In the row-update streaming model, each row of $A$ will occur in the stream exactly once. But the rows can be in arbitrary order. An algorithm in this model is only allowed a single pass over these rows. At the end of the stream, the algorithm stores some information of $A$. The space of the algorithm is the total number of words required to store this information during the stream. Here, each word is $O(\log(nd))$ bits.
\end{definition}

\begin{definition}[Turnstile model]\label{def:model_turnstile}
At the beginning, let matrix $A\in\mathbb{R}^{n\times d}$ be a zero matrix. In the turnstile streaming model, there is a stream of update operations, and the $i^{th}$ operation has the form $(x_i,y_i,c_i)$ which means that $A_{x_i,y_i}$ should be incremented by $c_i$. An algorithm in this model is only allowed a single pass over the stream. At the end of the stream, the algorithm stores some information of $A$. The space complexity of the algorithm is the total number of words required to store this information during the stream. Here, each word is $O(\log(nd))$ bits.
\end{definition}

\subsection{Turnstile model, $\poly(k,\log(d),\log(n))$ approximation}\label{sec:streaming_turnstile}

\begin{definition}[Turnstile model $\ell_1$-low rank approximation - $\rank$-$k$ subspace version]\label{def:streaming_sub}
Given matrix $A\in\mathbb{R}^{n\times d}$ and $k \in \mathbb{N}_+$, the goal is to propose an algorithm in the streaming model of Definition~\ref{def:model_turnstile} such that
\begin{enumerate}
\item Upon termination, the algorithm outputs a matrix $V^*\in \mathbb{R}^{k\times d}$.
\item $V^*$ satisfies that
$$\min_{U\in\mathbb{R}^{n\times k}} \|A-UV^*\|_1\leq \poly(k,\log(d),\log(n))\cdot \min_{U\in\mathbb{R}^{n\times k},V\in\mathbb{R}^{k\times d}} \|A-UV\|_1.$$
\item The space complexity is as small as possible
\end{enumerate}
\end{definition}

%\Peilin{Target space cost is $\poly(k)+kd$}

\begin{theorem}
Suppose $A\in\mathbb{R}^{n\times d}$ is given in the turnstile streaming model (See Definition~\ref{def:model_turnstile}), there is an algorithm(in Algorithm~\ref{alg:turnstile_dec} without decomposition) which solves the problem in Definition~\ref{def:streaming_sub} with constant probability. Further, the space complexity of the algorithm is $\poly(k)+\wt{O}(kd)$ words.
\end{theorem}

%Here we state our algorithm in Algorithm~\ref{alg:turnstile_dec} (without decomposition):

\begin{proof}

\textbf{Correctness.} The correctness is implied by \rm{(\RN{4})} of Lemma~\ref{lem:con_dil_summary}, and the proof of Theorem~\ref{thm:input_sparsity_algorithm}. %the proof of Theorem~\ref{thm:l1_rank_k_approx_polyklogdlogn} and the limited independence results shown in Section~\ref{sec:limind}.
Notice that $L=T_1AR,N=SAT_2,M=T_1AT_2$, so $\wh{X} \in \mathbb{R}^{O(k\log k)\times O(k\log k)}$ minimizes
$$\min_{\rank-k~X} \| T_1 AR X  SA T_2 - T_1 A T_2\|_F.$$
According to the proof of Theorem~\ref{thm:input_sparsity_algorithm}, $ AR \wh{X}  SA $ gives a $\ell_1$ rank-$k$ $\poly(k,\log(d),\log(n))$-approximation to $A$. Because $\wh{X}=\wh{U}\wh{\Sigma}\wh{V}^\top$, $V^*=\wh{\Sigma}\wh{V}^\top SA$ satisfies:
$$\min_{U\in\mathbb{R}^{n\times k}} \|A-UV^*\|_1\leq \poly(k,\log(d),\log(n))\cdot \min_{U\in\mathbb{R}^{n\times k},V\in\mathbb{R}^{k\times d}} \|A-UV\|_1.$$

\textbf{Space complexity.} Generating $\wt{O}(kd)$-wise independent random Cauchy variables needs $\wt{O}(kd)$ bits. The size of $L,N$ and $M$ are $k^2\log^2 k, k^2\log^3 k$ and $k^2\log^3 k$ words separately. So the space of maintaining them is $O(k^2\log^3 k)$ words. The size of $D$ is $O(k\log k)\times d$, so maintaining it needs $O(kd\log k)$ words. Therefore, the total space complexity of the algorithm is $\poly(k)+\wt{O}(kd)$ words.

\end{proof}

It is easy to extend our algorithm to output a decomposition. The formal definition of the decomposition problem is as the following:

\begin{definition}[Turnstile model $\ell_1$-low rank approximation - $\rank$-$k$ decomposition version]\label{def:streaming_dec}
Given matrix $A\in\mathbb{R}^{n\times d}$ and $k\in \mathbb{N}_+$, the goal is to propose an algorithm in the streaming model of Definition~\ref{def:model_turnstile} such that
\begin{enumerate}
\item Upon termination, the algorithm outputs a matrix $U^*\in\mathbb{R}^{n\times k},V^*\in \mathbb{R}^{k\times d}$.
\item $U^*,V^*$ satisfies 
$$\|A-U^*V^*\|_1\leq \poly(k,\log(d),\log(n))\cdot \min_{U\in\mathbb{R}^{n\times k},V\in\mathbb{R}^{k\times d}} \|A-UV\|_1.$$
\item The space complexity is as small as possible
\end{enumerate}
\end{definition}

%\Peilin{Target space cost is $\poly(k)+k(n+d)$}

\begin{theorem}
Suppose $A\in\mathbb{R}^{n\times d}$ is given by the turnstile streaming model (See Definition~\ref{def:model_turnstile}). There is an algorithm( in Algorithm~\ref{alg:turnstile_dec} with decomposition) which solves the problem in Definition~\ref{def:streaming_dec} with constant probability. Further, the space complexity of the algorithm is $\poly(k)+\wt{O}(k(d+n))$ words.
\end{theorem}

%The algorithm is stated in Algorithm~\ref{alg:turnstile_dec} (with decomposition):

\begin{algorithm}[h!]\caption{Turnstile Streaming Algorithm}\label{alg:turnstile_dec}
\begin{algorithmic}[1]
%\State \textbf{Input:}
%\State A stream of update operation $\{(x_i,y_i,c_i)\ |\ i\in[l]\}$, rank parameter $k<\rank(A)$.
%\State \textbf{Output:}
%\State $V^*\in\mathbb{R}^{k\times d}$
%\State
\Procedure{TurnstileStreaming}{$k$,${\cal S}$}
\State Construct sketching matrices $S\in\mathbb{R}^{O(k\log k )\times n},R\in\mathbb{R}^{d\times O(k\log k)},T_1\in \mathbb{R}^{O(k\log k)\times n},T_2\in \mathbb{R}^{d\times O(k\log^2 k)}$ where $R,\ T_2$ are fully independent random Cauchy matrices, and $S,\ T_1$ are random Cauchy matrices with fully independent variables across different rows and $\wt{O}(d)$-wise independent variables from the same row.
\State Initialize matrices:
\State $L\leftarrow \{0\}^{O(k\log k)\times O(k\log k)},N\leftarrow \{0\}^{O(k \log k)\times O(k\log^2 k)}$.
\State $M\leftarrow \{0\}^{O(k\log k)\times O(k \log^2 k)},D\leftarrow\{0\}^{O(k\log k)\times d}$.
\If{ need decomposition}
	\State $C\leftarrow\{0\}^{n\times O(k\log k)}$.
\EndIf
\For{$i\in [l]$}
    \State Receive update operation $(x_i,y_i,c_i)$ from the data stream ${\cal S}$.
    \For{$r=1\to O(k\log k),s=1\to O(k\log k)$}
        \State $L_{r,s}\leftarrow L_{r,s}+{T_1}_{r,x_i}\cdot c_i \cdot R_{y_i,s}$.
    \EndFor
    \For{$r=1\to O(k\log k),s=1\to O(k\log^2 k)$}
        \State $N_{r,s}\leftarrow N_{r,s}+S_{r,x_i}\cdot c_i \cdot {T_2}_{y_i,s}$.
    \EndFor
    \For{$r=1\to O(k\log k),s=1\to O(k\log^2 k)$}
        \State $M_{r,s}\leftarrow M_{r,s}+{T_1}_{r,x_i}\cdot c_i\cdot {T_2}_{y_i,s}$.
    \EndFor
    \For{$r=1\to O(k\log k)$}
        \State $D_{r,y_i}\leftarrow D_{r,s}+S_{r,x_i}\cdot c_i$.
    \EndFor
    \If{need decomposition}
    	   \For{$s=1\to O(k\log k)$}
        		\State $C_{x_i,s} \leftarrow C_{x_i,s}+c_i\cdot R_{y_i,s}$.
	       \EndFor
	\EndIf
\EndFor
\State Compute the SVD of $L=U_L\Sigma_LV_L^\top$.
\State Compute the SVD of $N=U_N\Sigma_NV_N^\top$.
\State Compute $\wh{X}=L^\dagger(U_LU_L^\top MV_NV_N^\top)_kN^\dagger$.
\State Compute the SVD of $\wh{X}=\wh{U}\wh{\Sigma}\wh{V}^\top$.
\If {need decomposition}
	\State \Return $V^*=\wh{\Sigma}\wh{V}^\top D,U^*=C\wh{U}$.
\Else
	\State \Return $V^*=\wh{\Sigma}\wh{V}^\top D$.
\EndIf
\EndProcedure
\end{algorithmic}
\end{algorithm}

\begin{proof}

\textbf{Correctness.} The only difference from the Algorithm~\ref{alg:turnstile_dec} (without decomposition) is that the algorithm maintains $C$. Thus, finally it can compute $U^*=AR\wh{U}$. Notice that $U^*V^*=AR\wh{X}SA$, according to the proof of Theorem~\ref{thm:input_sparsity_algorithm}, $U^*V^*$ gives a  $\ell_1$ rank-$k$ $\poly(k,\log(d),\log(n))$-approximation to $A$.

\textbf{Space complexity.} Since the size of $C$ is $O(nk\log k)$ words, the total space is $\poly(k)+\wt{O}(k(d+n))$ words.

\end{proof}

\subsection{Row-update model, $\poly(k)\log d$ approximation}\label{sec:streaming_rowupdate}

\begin{definition}[Row-update model $\ell_1$-low rank approximation - $\rank$-$k$ subspace version]\label{def:streaming_sub2}
Given matrix $A\in\mathbb{R}^{n\times d}$ and $k\in \mathbb{N}_+$, the goal is to propose an algorithm in the streaming model of Definition~\ref{def:model_rowupdate} such that
\begin{enumerate}
\item Upon termination, the algorithm outputs a matrix $V^*\in \mathbb{R}^{k\times d}$.
\item $V^*$ satisfies that
$$\min_{U\in\mathbb{R}^{n\times k}} \|A-UV^*\|_1\leq \poly(k)\log(d)\cdot \min_{U\in\mathbb{R}^{n\times k},V\in\mathbb{R}^{k\times d}} \|A-UV\|_1.$$
\item The space complexity is as small as possible
\end{enumerate}
\end{definition}

%\Peilin{Target space cost is $\poly(k)+kd$}

\begin{theorem}
Suppose $A\in\mathbb{R}^{n\times d}$ is given by the row-update streaming model (See Definition~\ref{def:model_rowupdate}), there is an algorithm( in Algorithm~\ref{alg:rowupdate_dec} without decomposition ) which solves the problem in Definition~\ref{def:streaming_sub2} with constant probability. Further, the space complexity of the algorithm is $\poly(k)+\wt{O}(kd)$ words.
\end{theorem}

%Here we state our algorithm in Algorithm~\ref{alg:rowupdate_dec} (without decomposition):

\begin{proof}
\textbf{Correctness.} Notice that $L=T_1BR,N=SBT_2,M=T_1BT_2$. Thus, $\wh{X}\in\mathbb{R}^{O(k\log k)\times O(k\log k)}$ actually minimizes
$$\min_{\rank-k~X} \| T_1 BR X  SB T_2 - T_1 B T_2\|_F.$$
Also notice that $B$ is just taking each row of $A$ and replacing it with its nearest point in the row span of $S'A$. According to the proof of Theorem~\ref{thm:polyklogd_approx_algorithm} and \rm{(\RN{4})} of Lemma~\ref{lem:con_dil_summary}
%the limited independence results shown in Section~\ref{sec:limind},
$BR\wh{X}SB$ gives a $\poly(k)\log d$ $\ell_1$ norm $\rank$-$k$ approximation to $A$. Since $\wh{X}=\wh{U}\wh{\Sigma}\wh{V}^\top$, $V^*=\wh{\Sigma}\wh{V}^\top SB$ satisfies:
$$\min_{U\in\mathbb{R}^{n\times k}} \|A-UV^*\|_1\leq \poly(k,\log(d),\log(n))\cdot \min_{U\in\mathbb{R}^{n\times k},V\in\mathbb{R}^{k\times d}} \|A-UV\|_1.$$

\textbf{Space complexity.} Constructing sketching matrices needs $\wt{O}(kd)$ bits to store random seeds. Maintaining $L,N,M$ needs $O(k^2\log^3 k)$ words. The cost of storing $\wh{X}$ is also $O(k^2\log^2 k)$ words. Maintaining $D$ needs $O(kd\log k)$ words. Therefore, the total space cost of the algorithm is $\poly(k)+\wt{O}(kd)$ words.
\end{proof}

It is easy to extend our algorithm to output a decomposition. The formal definition of the decomposition problem is as the following:

\begin{definition}[Row-update model $\ell_1$-low rank approximation - $\rank$-$k$ decomposition version]\label{def:streaming_dec2}
Given matrix $A\in\mathbb{R}^{n\times d}$ and $k\in \mathbb{N}_+$, the goal is to propose an algorithm in the streaming model of Definition~\ref{def:model_rowupdate} such that
\begin{enumerate}
\item Upon termination, the algorithm outputs matrices $U^*\in\mathbb{R}^{n\times k},V^*\in \mathbb{R}^{k\times d}$.
\item $U^*,V^*$ satisfies that
$$\|A-U^*V^*\|_1\leq \poly(k)\log d\cdot \min_{U\in\mathbb{R}^{n\times k},V\in\mathbb{R}^{k\times d}} \|A-UV\|_1.$$
\item The space complexity is as small as possible.
\end{enumerate}
\end{definition}

%\Peilin{Target space cost is $\poly(k)+(n+d)k$}

\begin{theorem}
Suppose $A\in\mathbb{R}^{n\times d}$ is given by the row-update streaming model (See Definition~\ref{def:model_rowupdate}), there is an algorithm(in Algorithm~\ref{alg:rowupdate_dec} with decomposition ) which solves the problem in Definition~\ref{def:streaming_dec2} with constant probability. Further, the space complexity of the algorithm is $\poly(k)+\wt{O}(k(n+d))$ words.
\end{theorem}

%Here we state our algorithm in Algorithm~\ref{alg:rowupdate_dec} (with decomposition):

\begin{algorithm}[h!]\caption{Row Update Streaming Algorithm}\label{alg:rowupdate_dec}
\begin{algorithmic}[1]
%\State \textbf{Input:}
%\State A stream of rows of $A$: $\{A_i |\ i\in[n]\}$ in arbitrary order, rank parameter $k<\rank(A)$.
%\State \textbf{Output:}
%\State $V^*\in\mathbb{R}^{k\times d}$
%\State
\Procedure{RowUpdateStreaming}{$k,{\cal S}$}
\State Construct sketching matrices $S'\in\mathbb{R}^{O(k\log k)\times n},S\in\mathbb{R}^{O(k\log k )\times n},R\in\mathbb{R}^{d\times O(k\log k)},T_1\in \mathbb{R}^{O(k\log k)\times n},T_2\in \mathbb{R}^{d\times O(k\log^2 k)}$ where $R,\ T_2$ are fully independent random Cauchy variables, and $S\ ,S'\ ,T_1$ are random Cauchy matrices with fully independent random variables from different rows and $\wt{O}(d)$-wise independent in the same row.
\State Initialize matrices:
\State $L\leftarrow\{0\}^{O(k\log k)\times O(k\log k)},N\leftarrow \{0\}^{O(k \log k)\times O(k\log^2 k)}$.
\State $M\leftarrow\{0\}^{O(k\log k)\times O(k \log^2 k)},D\leftarrow\{0\}^{O(k\log k)\times d}$.
\If{need decomposition}
	\State $C\leftarrow \{0\}^{n\times O(k\log k)}$.
\EndIf
\For{$i\in [n]$}
    \State Receive a row update $(i,A_i)$ from the data stream ${\cal S}$.
    \State Compute $Y^*_i\in\mathbb{R}^{1\times O(k\log k)}$ which minimizes $\min_{Y\in\mathbb{R}^{1\times O(k\log k)}}\|YS'_{:,i}A_i-A_i\|_1 $.
    \State Compute $B_i=Y^*_iS'_{:,i}A_i$.
    \For{$r=1\to O(k\log k),s=1\to O(k\log k),j=1\to d$}
        \State $L_{r,s}\leftarrow L_{r,s}+{T_1}_{r,i}\cdot B_{i,j} \cdot R_{j,s}$.
    \EndFor
    \For{$r=1\to O(k\log k),s=1\to O(k\log^2 k),j=1\to d$}
        \State $N_{r,s}\leftarrow N_{r,s}+S_{r,i}\cdot B_{i,j} \cdot {T_2}_{j,s}$.
    \EndFor
    \For{$r=1\to O(k\log k),s=1\to O(k\log^2 k),j=1\to d$}
        \State $M_{r,s}\leftarrow M_{r,s}+{T_1}_{r,i}\cdot B_{i,j}\cdot {T_2}_{j,s}$.
    \EndFor
    \For{$r=1\to O(k\log k),j=1\to d$}
        \State $D_{r,j}\leftarrow D_{r,j}+S_{r,i}\cdot B_{i,j}$.
    \EndFor
	\If {need decomposition}
    		\For{$s=1\to O(k\log k),j=1\to d$}
        		\State $C_{i,s}:=C_{i,s}+B_{i,j}\cdot R_{j,s}$.
    		\EndFor
	\EndIf
\EndFor
\State Compute the SVD of $L=U_L\Sigma_LV_L^\top$.
\State Compute the SVD of $N=U_N\Sigma_NV_N^\top$.
\State Compute $\wh{X}=L^\dagger(U_LU_L^\top MV_NV_N^\top)_kN^\dagger$.
\State Compute the SVD of $\wh{X}=\wh{U}\wh{\Sigma}\wh{V}^\top$.
\If {need decomposition}
	\State \Return $V^*=\wh{\Sigma}\wh{V}^\top D,U^*=C\wh{U}$.
\Else
	\State \Return $V^*=\wh{\Sigma}\wh{V}^\top D$.
\EndIf
\EndProcedure
\end{algorithmic}
\end{algorithm}

\begin{proof}
\textbf{Correctness.} The only difference is that the above algorithm maintains $C$. Thus, We can compute $U^*=C\wh{U}$ in the end. Notice that $U^*V^*=BR\wh{X}SB$, according to the proof of Theorem~\ref{thm:polyklogd_approx_algorithm}, $U^*V^*$ gives a $\poly(k)\log d$ $\ell_1$ norm $\rank$-$k$ approximation to $A$.

\textbf{Space complexity.} Since the size of $C$ is $nk\log k$ words, the total space is $\poly(k)+\wt{O}(k(n+d))$ words.
\end{proof}

\section{Distributed Setting}\label{sec:distributed}
Section~\ref{sec:distributed_definition} provides some notation and definitions for the Row-partition distributed model and the Arbitrary-partition model. These two models were recently studied in a line of works such as~\cite{tisseur1999parallel,qu2002principal,bai2005principal,sensors2008,macua2010consensus,fegk13,poulson2013elemental,kvw14,bklw14,kdd16,bwz16,wz16}.
Section~\ref{sec:arb_sub} and \ref{sec:arb_dec} presents our distributed protocols for the Arbitrary-partition distributed model. Section~\ref{sec:row_sub} and \ref{sec:row_dec} presents our distributed protocols for the Row-partition distributed model.

\subsection{Definitions}\label{sec:distributed_definition}

\begin{definition}[Row-partition model~\cite{bwz16}]\label{def:model_row}
There are $s$ machines, and the $i^{\mathrm{th}}$ machine has a matrix $A_i\in\mathbb{R}^{n_i\times d}$ as input. Suppose $n=\sum_{i=1}^s n_i$, and the global data matrix $A\in\mathbb{R}^{n\times d}$ is denoted as
$$\left(\begin{array}{c}A_1\\A_2\\\cdots\\A_s\end{array}\right),$$
we say $A$ is row-partitioned into these $s$ matrices distributed in $s$ machines respectively. Furthermore, there is a machine which is a coordinator. The model only allows communication between the machines and the coordinator. The communication cost in this model is the total number of words transferred between machines and the coordinator. Each word is $O(\log(snd))$ bits.
\end{definition}

\begin{definition}[Arbitrary-partition model~\cite{bwz16}]\label{def:model_arb}
There are $s$ machines, and the $i^{\mathrm{th}}$ machine has a matrix $A_i\in\mathbb{R}^{n\times d}$ as input. Suppose the global data matrix $A\in\mathbb{R}^{n\times d}$ is denoted as $A=\sum_{i=1}^s A_i$. We say $A$ is arbitrarily partitioned into these $s$ matrices distributed in $s$ machines respectively. Furthermore, there is a machine which is a coordinator. The model only allows communication between the machines and the coordinator. The communication cost in this model is the total number of words transferred between machines and the coordinator. Each word is $O(\log(snd))$ bits.
\end{definition}

\subsection{Arbitrary-partition model, subspace, $\poly(k,\log(d),\log(n))$ approximation}\label{sec:arb_sub}

\begin{definition}[Arbitrary-partition model $\ell_1$-low rank approximation - $\rank$-$k$ subspace version]\label{def:distrisub}
Given matrix $A\in\mathbb{R}^{n\times d}$ arbitrarily partitioned into $s$ matrices $A_1,A_2,\cdots,A_s$ distributed in $s$ machines respectively, and $k\in \mathbb{N}_+$, the goal is to propose a protocol in the model of Definition~\ref{def:model_arb} such that
\begin{enumerate}
\item Upon termination, the protocol leaves a matrix $V^*\in \mathbb{R}^{k\times d}$ on the coordinator.
\item $V^*$ satisfies that
$$\min_{U\in\mathbb{R}^{n\times k}} \|A-UV^*\|_1\leq \poly(k,\log(d),\log(n))\cdot \min_{U\in\mathbb{R}^{n\times k},V\in\mathbb{R}^{k\times d}} \|A-UV\|_1.$$
\item The communication cost is as small as possible
\end{enumerate}
%Constant failure probability of the protocol can be tolerated.
\end{definition}

%\Peilin{Target communication cost is $s(\poly(k)+kd)$}

\begin{theorem}
Suppose $A\in\mathbb{R}^{n\times d}$ is partitioned in the arbitrary partition model (See Definition~\ref{def:model_arb}). There is a protocol(in Algorithm~\ref{alg:arb_sub}) which solves the problem in Definition~\ref{def:distrisub} with constant probability. Further, the communication complexity of the protocol is $s(\poly(k)+\wt{O}(kd))$ words.
\end{theorem}

%Here, we state our protocol in Algorithm~\ref{alg:arb_sub}:

\begin{algorithm}[h!]\caption{Arbitrary Partition Distributed Protocol}\label{alg:arb_sub}
\begin{algorithmic}[1]
\Procedure{ArbitraryPartitionDistributedProtocol}{$k$,$s$,$A$}
%\State \textbf{Input:}
\State $A\in\mathbb{R}^{n\times d}$ was arbitrarily partitioned into $s$ matrices $A_1,\cdots,A_s\in \mathbb{R}^{n\times d}$ distributed in $s$ machines.
\State \hspace{2cm} {\bf Coordinator} \hspace{4.5cm} {\bf Machines} $i$
\State Chooses a random seed.
\State Sends it to all machines. \label{sta:send_seed}
\State \hspace{4.5cm} $--------->$
%The coordinator chooses a random seed for sketching matrices $S\in\mathbb{R}^{O(k\log k)\times n},R\in\mathbb{R}^{d\times O(k\log k)},T_1\in \mathbb{R}^{O(k\log k)\times n},T_2\in \mathbb{R}^{d\times O(k\log^2 k )}$. The coordinator sends the random seed to all machines.
%\State According to the random seed, machines agree upon sketching matrices $S\in\mathbb{R}^{O(k\log k )\times n},R\in\mathbb{R}^{d\times O(k\log k)},T_1\in \mathbb{R}^{O(k\log k)\times n},T_2\in \mathbb{R}^{d\times O(k\log^2 k)}$ with $W$-wise independent random Cauchy variables.  %%%comment by Zhao
\State \hspace{8cm} Agrees on $R,\ T_2$ which are fully
\State \hspace{8cm} independent random Cauchy matrices.
\State \hspace{8cm} Agrees on $S,\ T_1$ which are random Cauchy
\State \hspace{8cm} matrices with fully independent entries
\State \hspace{8cm} from different rows, and $\wt{O}(d)$-wise indepen-
\State \hspace{8cm} dent variables from the same row.
\State \hspace{8cm} Computes $L_i=T_1A_iR,N_i=SA_iT_2$.
\State \hspace{8cm} Computes $M_i=T_1A_iT_2$.
\State \hspace{8cm} Sends $L_i,N_i,M_i$ to the coordinator.\label{sta:send_LNM}
\State \hspace{4.5cm} $<---------$
\State Computes $L=\overset{s}{\underset{i=1}{\sum}}L_i,N=\overset{s}{\underset{i=1}{\sum}}N_i$.
\State Computes $M=\overset{s}{\underset{i=1}{\sum}} M_i$. %%%(Thus, $L=T_1AR,N=SAT_2,M=T_1AT_2$) %%% comment by Zhao
\State Computes the SVD of $L=U_L\Sigma_LV_L^\top$.
\State Computes the SVD of $N=U_N\Sigma_NV_N^\top$.
\State Computes $\wh{X}=L^\dagger(U_LU_L^\top MV_NV_N^\top)_kN^\dagger$.
\State Sends $\wh{X}$ to machines. \label{sta:send_X}
\State \hspace{4.5cm} $--------->$
\State \hspace{8cm}Computes the SVD of $\wh{X}=\wh{U}\wh{\Sigma}\wh{V}^\top$.
\State \hspace{8cm}Computes $V_i^*=\wh{\Sigma}\wh{V}^\top SA_i$.
\State \hspace{8cm}{\bf If} need decomposition
\State \hspace{9cm}$U_i^*=A_i R \wh{U}$.
\State \hspace{9cm}Sends $U_i^*$, $V_i^*$ to the coordinator.
\State \hspace{8cm}{\bf Else}
\State \hspace{9cm}Sends $V_i^*$ to the coordinator. \label{sta:send_V}
\State \hspace{8cm}{\bf Endif}
\State \hspace{4.5cm} $<---------$
\State {\bf If} need decomposition,
\State \hspace{1cm} \Return $V^*=\sum_{i=1}^s V_i^*$, $U^*=\sum_{i=1}^s U_i^*$.
\State {\bf Else}
\State \hspace{1cm} \Return $V^*=\sum_{i=1}^s V_i^*$.
\State {\bf Endif}
\EndProcedure
\end{algorithmic}
\end{algorithm}

%\begin{small}
%{\bf Input:}
%\begin{enumerate}
%\item $A\in\mathbb{R}^{n\times d}$ arbitrarily partitioned into $s$ matrices $A_1,\cdots,A_s$ distributed in $s$ machines.
%\item rank parameter $k<\rank(A)$
%\end{enumerate}
%
%{\bf Algorithm}
%
%\begin{enumerate}
%
%\item The coordinator chooses a random seed for sketching matrices $S\in\mathbb{R}^{O(k\log k)\times n},R\in\mathbb{R}^{d\times O(k\log k)},T_1\in \mathbb{R}^{O(k\log k)\times n},T_2\in \mathbb{R}^{d\times O(k\log^2 k )}$. The coordinator sends the random seed to all machines.
%
%\item According to the random seed, machines agree upon sketching matrices $S\in\mathbb{R}^{O(k\log k )\times n},R\in\mathbb{R}^{d\times O(k\log k)},T_1\in \mathbb{R}^{O(k\log k)\times n},T_2\in \mathbb{R}^{d\times O(k\log^2 k)}$ with $W$-wise independent random Cauchy variables.
%
%\item Machine $i$ computes $L_i=T_1A_iR,N_i=SA_iT_2,M_i=T_1A_iT_2$, and sends $L_i,N_i,M_i$ to the coordinator.
%
%\item The coordinator computes $L=\sum_{i=1}^s L_i,N=\sum_{i=1}^s N_i, M=\sum_{i=1}^s M_i$. (Thus, $L=T_1AR,N=SAT_2,M=T_1AT_2$)
%
%\item The coordinator computes the SVD of $L=U_L\Sigma_LV_L^\top$ and the SVD of $N=U_N\Sigma_NV_N^\top$.
%
%\item The coordinator computes $\wh{X}=L^\dagger(U_LU_L^\top MV_NV_N^\top)_kN^\dagger$
%
%\item The coordinator sends $\wh{X}$ to machines.
%
%\item Machine $i$ computes the SVD of $\wh{X}=U_{\wh{X}}\Sigma_{\wh{X}}V_{\wh{X}}^\top$.
%
%\item Machine $i$ computes $V_i^*=\Sigma_{\wh{X}}V_{\wh{X}}^\top SA_i$, and sends them to the coordinator.
%
%\item The coordinator computes $V^*=\sum_{i=1}^s V_i^*$.
%
%\end{enumerate}
%
%\end{small}

\begin{proof}

\textbf{Correctness.} The correctness is shown by the proof of Theorem~\ref{thm:input_sparsity_algorithm} and \rm{(\RN{4})} of Lemma~\ref{lem:con_dil_summary}. %the limited independence results shown in the Section~\ref{sec:limind}.
Notice that $\wh{X}\in\mathbb{R}^{O(k\log k)\times O(k\log k)}$ minimizes
$$\min_{\rank-k~X}\|LXN-M\|_F.$$
which is
$$\min_{\rank-k~X} \| T_1 AR X  SA T_2 - T_1 A T_2\|_F.$$
According to the proof of Theorem~\ref{thm:input_sparsity_algorithm}, $ AR \wh{X}  SA $ gives an $\ell_1$ rank-$k$ $\poly(k,\log(d),\log(n))$-approximation to $A$. Because $\wh{X}=\wh{U}\wh{\Sigma}\wh{V}^\top$, $V^*=\wh{\Sigma}\wh{V}^\top SA$ satisfies:
$$\min_{U\in\mathbb{R}^{n\times k}} \|A-UV^*\|_1\leq \poly(k,\log(d),\log(n))\cdot \min_{U\in\mathbb{R}^{n\times k},V\in\mathbb{R}^{k\times d}} \|A-UV\|_1.$$

\textbf{Communication complexity.} Since the random seed generates $\wt{O}(kd)$-wise independent random Cauchy variables, the cost of line~\ref{sta:send_seed} is $\wt{O}(skd)$ bits. The size of $L_i,N_i$ and $M_i$ are $k^2\log^2 k, k^2\log^3 k$ and $k^2\log^3 k$ words separately. So the cost of line~\ref{sta:send_LNM} is $O(sk^2\log^3 k)$ words. Because the size of $\wh{X}$ is $O(k^2\log^2 k)$, the cost of line~\ref{sta:send_X} is $O(sk^2\log^2 k)$ words. line~\ref{sta:send_V} needs $skd$ words of communication. Therefore, the total communication of the protocol is $s(\poly(k)+\wt{O}(kd))$ words.

\end{proof}

\subsection{Arbitrary-partition model, decomposition, $\poly(k,\log(d),\log(n))$ approximation}\label{sec:arb_dec}

\begin{definition}[Arbitrary-partition model $\ell_1$-low rank approximation - $\rank$-$k$ decomposition version]\label{def:distridec}
Given matrix $A\in\mathbb{R}^{n\times d}$ arbitrarily partitioned into $s$ matrices $A_1,A_2,\cdots,A_s$ distributed in $s$ machines respectively, and $k\in \mathbb{N}_+$, the goal is to propose a protocol in the model of Definition~\ref{def:model_arb} such that
\begin{enumerate}
\item Upon termination, the protocol leave matrices $U^*\in\mathbb{R}^{n\times k},V^*\in \mathbb{R}^{k\times d}$ on the coordinator.
\item $U^*,V^*$ satisfies that
$$\|A-U^*V^*\|_1\leq \poly(k,\log(d),\log(n))\cdot \min_{U\in\mathbb{R}^{n\times k},V\in\mathbb{R}^{k\times d}} \|A-UV\|_1.$$
\item The communication cost is as small as possible.
\end{enumerate}
%Constant failure probability of the protocol can be tolerated.
\end{definition}

%\Peilin{Target communication cost is $s(\poly(k)+kd+kn)$}

\begin{theorem}
Suppose $A\in\mathbb{R}^{n\times d}$ is partitioned in the arbitrary partition model (See Definition~\ref{def:model_arb}). There is a protocol(in Algorithm~\ref{alg:arb_sub} with decomposition) which solves the problem in Definition~\ref{def:distridec} with constant probability. Further, the communication complexity of the protocol is $s(\poly(k)+\wt{O}(k(d+n)))$ words.
\end{theorem}

\begin{proof}

\textbf{Correctness.} The only difference from the protocol (without decomposition) in Section~\ref{sec:arb_sub} is that the protocol sends $U_i$. Thus, the coordinator can compute $U^*=AR \wh{U}$. Notice that $U^*V^*=AR\wh{X}SA$. According to the proof of Theorem~\ref{thm:input_sparsity_algorithm}, $U^*V^*$ gives a  $\ell_1$ rank-$k$ $\poly(k,\log(d),\log(n))$-approximation to $A$.

\textbf{Communication complexity.} Since the size of $U_i$ is $kn$ words, the total communication is $s(\poly(k)+\wt{O}(k(d+n)))$ words.

\end{proof}

\subsection{Row-partition model, subspace, $\poly(k)\log d$ approximation}\label{sec:row_sub}

\begin{definition}[Row-partition model $\ell_1$-low rank approximation - $\rank$-$k$ subspace version]\label{def:rowpsub}
Given matrix $A\in\mathbb{R}^{n\times d}$ row-partitioned into $s$ matrices $A_1,A_2,\cdots,A_s$ distributed in $s$ machines respectively, and $k\in \mathbb{N}_+$, the goal is to propose a protocol in the model of Definition~\ref{def:model_row} such that
\begin{enumerate}
\item Upon termination, the protocol leaves a matrix $V^*\in \mathbb{R}^{k\times d}$ on the coordinator.
\item $V^*$ satisfies that
$$\min_{U\in\mathbb{R}^{n\times k}} \|A-UV^*\|_1\leq \poly(k)\log d\cdot \min_{U\in\mathbb{R}^{n\times k},V\in\mathbb{R}^{k\times d}} \|A-UV\|_1.$$
\item The communication cost is as small as possible
\end{enumerate}
%Constant failure probability of the protocol can be tolerated.
\end{definition}

\begin{theorem}
Suppose $A\in\mathbb{R}^{n\times d}$ is partitioned in the row partition model (See Definition~\ref{def:model_row}). There is a protocol(in Algorithm~\ref{alg:row_sub} without decomposition) which solves the problem in Definition~\ref{def:rowpsub} with constant probability. Further, the communication complexity of the protocol is $s(\poly(k)+\wt{O}(kd))$ words.
\end{theorem}

%The protocol is shown in Algorithm~\ref{alg:row_sub} (without decomposition).

\begin{algorithm}[h!]\caption{Row Partition Distributed Protocol}\label{alg:row_sub}
\begin{algorithmic}[1]
\Procedure{RowPartitionDistributedProtocol}{$k$,$s$,$A$}
\State $A\in\mathbb{R}^{n\times d}$ was row partitioned into $s$ matrices $A_1\in\mathbb{R}^{n_1\times d},\cdots,A_s\in\mathbb{R}^{n_s\times d}$ distributed in $s$ machines.
\State \hspace{2cm} {\bf Coordinator} \hspace{4.5cm} {\bf Machines} $i$
\State Chooses a random seed.
\State Sends it to all machines.
\State \hspace{4.5cm} $--------->$ \label{sta:send_seed2}
\State \hspace{8cm} Agrees on $R,\ T_2$ which are fully
\State \hspace{8cm} independent random Cauchy variables.
\State \hspace{8cm} Generates random Cauchy matrices
\State \hspace{8cm} $S'_i\in\mathbb{R}^{O(k\log k)\times n_i}$, $S_i,{T_1}_i\in\mathbb{R}^{O(k\log^2 k)\times n_i}$.
\State \hspace{8cm} Computes ${\tiny Y^*_i = \underset{Y\in\mathbb{R}^{n_i\times O(k\log k)}}{\arg\min} \|YS'_iA_i-A_i\|_1}$.
\State \hspace{8cm} Computes $B_i=Y^*_iS'_iA_i$.
\State \hspace{8cm} Computes $L_i={T_1}_iB_iR,N_i=S_iB_iT_2$.
\State \hspace{8cm} Computes $M_i={T_1}_iB_iT_2$.
\State \hspace{8cm} Sends $L_i,N_i,M_i$ to the coordinator.
\State \hspace{4.5cm} $<---------$ \label{sta:send_LNM2}
\State Computes $L=\sum_{i=1}^s L_i,N=\sum_{i=1}^s N_i$.
\State Computes $M=\sum_{i=1}^s M_i$.
\State Computes the SVD of $L=U_L\Sigma_LV_L^\top$.
\State Computes the SVD of $N=U_N\Sigma_NV_N^\top$.
\State Computes $\wh{X}=L^\dagger(U_LU_L^\top MV_NV_N^\top)_kN^\dagger$.
\State Sends $\wh{X}$ to machines.
\State \hspace{4.5cm} $--------->$ \label{sta:send_X2}
\State \hspace{8cm} Computes the SVD of $\wh{X}=\wh{U}\wh{\Sigma}\wh{V}^\top$.
\State \hspace{8cm} Computes $V_i^*=\wh{\Sigma} \wh{V}^\top S_iB_i$.
\State \hspace{8cm} {\bf If} need decomposition
\State \hspace{9cm} Computes $U_i^*=B_iR \wh{U}$.
\State \hspace{9cm} Sends $U_i^*, V_i^*$ to the coordinator.
\State \hspace{8cm} {\bf Else}
\State \hspace{9cm} Sends $V_i^*$ to the coordinator.
\State \hspace{8cm} {\bf Endif}
\State \hspace{4.5cm} $<---------$  \label{sta:send_V2}
%\State Computes $V^*=\sum_{i=1}^s V_i^*$.
\State {\bf If} need decomposition
\State \hspace{1cm} \Return $V^*=\sum_{i=1}^s V_i^*$, $U^*=\sum_{i=1}^s U_i^*$.
\State {\bf Else}
\State \hspace{1cm} \Return $V^*=\sum_{i=1}^s V_i^*$.
\State {\bf Endif}
\EndProcedure
\end{algorithmic}
\end{algorithm}

\begin{proof}

\textbf{Correctness.} For convenience, we denote matrices $B\in\mathbb{R}^{n\times d},S\in\mathbb{R}^{O(k \log^2 k)\times n}$ and $T_1\in\mathbb{R}^{O(k\log^2 k)\times n}$ as
$$\begin{array}{ccc}
B=\left(\begin{array}{c}B_1\\B_2\\\cdots\\B_s\end{array}\right)
&
S=\left(\begin{array}{cccc}S_1 & S_2 & \cdots & S_s\end{array}\right)
&
T_1=\left(\begin{array}{cccc}{T_1}_1 & {T_1}_2 & \cdots & {T_1}_s\end{array}\right)
\end{array}.$$
%Separately.
Notice that $L=T_1BR,N=SBT_2,M=T_1BT_2$. Thus, $\wh{X}\in\mathbb{R}^{O(k\log k)\times O(k\log k)}$ actually minimizes
$$\min_{ \rank-k~X} \| T_1 BR X  SB T_2 - T_1 B T_2\|_F.$$
Also notice that $B$ is just taking each row of $A$ and replacing it with its nearest point in the row span of $S'A$. According to the proof of Theorem~\ref{thm:polyklogd_approx_algorithm}, $BR\wh{X}SB$ gives a $\poly(k)\log d$ $\ell_1$ norm $\rank$-$k$ approximation to $A$. Since $\wh{X}=\wh{U}\wh{\Sigma}\wh{V}^\top$, $V^*=\wh{\Sigma}\wh{V}^\top SB$ satisfies:
$$\min_{U\in\mathbb{R}^{n\times k}} \|A-UV^*\|_1\leq \poly(k,\log(d),\log(n))\cdot \min_{U\in\mathbb{R}^{n\times k},V\in\mathbb{R}^{k\times d}} \|A-UV\|_1.$$

\textbf{Communication complexity.} Since the $R$ and $T_2$ are $\wt{O}(kd)$-wise independent, line~\ref{sta:send_seed2} needs $O(sW)$ bits of communication. Line~\ref{sta:send_LNM2} needs $O(sk^2\log^3 k)$ words. The cost of line~\ref{sta:send_X2} is $O(sk^2\log^2 k)$ words. Line~\ref{sta:send_V2} needs $skd$ words of communication. Therefore, the total communication of the protocol is $s(\poly(k)+\wt{O}(kd))$ words.

\end{proof}

\subsection{Row-partition model, decomposition, $\poly(k)\log d$ approximation}\label{sec:row_dec}

\begin{definition}[Row-partition model $\ell_1$-low rank approximation - rank-k decomposition version]\label{def:rowpdec}
Given matrix $A\in\mathbb{R}^{n\times d}$ row partitioned into $s$ matrices $A_1,A_2,\cdots,A_s$ distributed in $s$ machines respectively, and a positive integer $k<\rank(A)$, the goal is to propose a protocol in the model of Definition~\ref{def:model_row} such that
\begin{enumerate}
\item Upon termination, the protocol leaves matrices $U^*\in\mathbb{R}^{n\times k},V^*\in \mathbb{R}^{k\times d}$ on the coordinator.
\item $U^*,V^*$ satisfies that
$$\|A-U^*V^*\|_1\leq \poly(k)\log d\cdot \min_{U\in\mathbb{R}^{n\times k},V\in\mathbb{R}^{k\times d}} \|A-UV\|_1.$$
\item The communication cost is as small as possible.
\end{enumerate}
%Constant failure probability of the protocol can be tolerated.
\end{definition}

%\Peilin{Target communication cost is $s(\poly(k)+kd+kn)$}

\begin{theorem}
Suppose $A\in\mathbb{R}^{n\times d}$ is partitioned in the row partition model (See Definition~\ref{def:model_row}). There is a protocol(in Algorithm~\ref{alg:row_sub} with decomposition) which solves the problem in Definition~\ref{def:rowpdec} with constant probability. Further, the communication complexity of the protocol is $s(\poly(k)+\wt{O}(k(n+d)))$ words.
\end{theorem}

%The protocol is as shown in Algorithm~\ref{alg:row_sub} (with decomposition).

\begin{proof}
\textbf{Correctness.} The only difference is that the above protocol sends $U_i$. Thus, the coordinator can compute $U^*=BR\wh{U}$. Notice that $U^*V^*=BR\wh{X}SB$, according to the proof of Theorem~\ref{thm:polyklogd_approx_algorithm}, $U^*V^*$ gives a $\poly(k)\log d$ $\ell_1$ norm $\rank$-$k$ approximation to $A$.

\textbf{Communication complexity.} Since the size of $U_i$ is $kn$ words, the total communication is $s(\poly(k)+\wt{O}(k(n+d)))$ words.

\end{proof}

\section{Experiments and Discussions}\label{sec:exp}
In this section, we provide some counterexamples for the other heuristic algorithms such that, for those
examples, the heuristic algorithms can output a solution with a very ``bad'' approximation ratio, i.e., $n^c$, where $c>0$ and the input matrix has size $n\times n$. We not only observe that heuristic algorithms sometimes have very bad performance in practice, but also give a proof in theory.

\subsection{Setup}
We provide some details of our experimental setup. We obtained the R package of \cite{kk05,kwak08,bdb13} from \url{https://cran.r-project.org/web/packages/pcaL1/index.html}.
%\url{http://www.optimization-online.org/DB_HTML/2012/04/3436.html}.
 We also implemented our algorithm and the r1-pca algorithm \cite{dzhz06} using the R language. The version of the R language is 3.0.2. We ran experiments on a machine with Intel X5550$@2.67$GHz CPU and $24$G memory. The operating system of that machine is Linux Ubuntu 14.04.5 LTS.  All the experiments were done in single-threaded mode.

\subsection{Counterexample for \cite{dzhz06}}

\begin{figure}[t]
  \centering
    \includegraphics[width=0.7\textwidth]{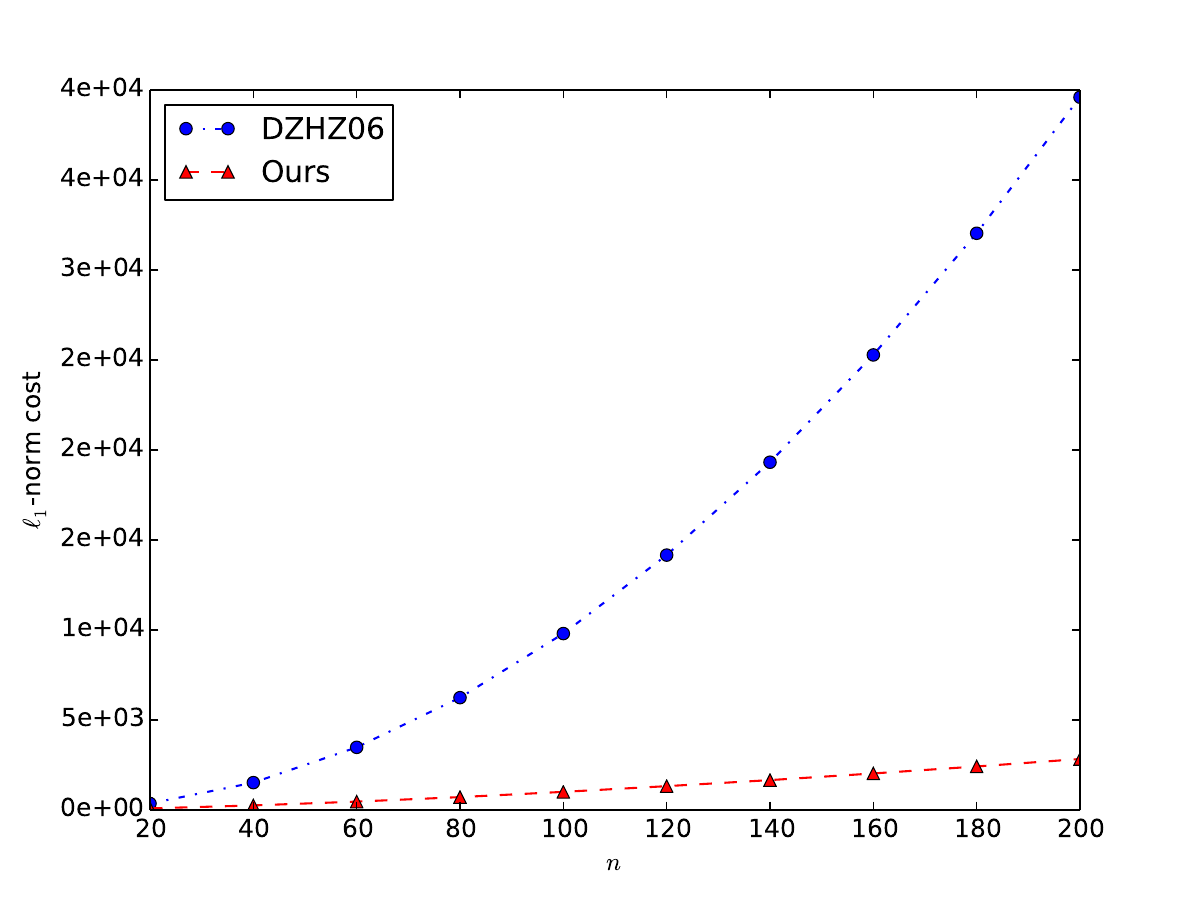} %{./exp_figs/dzhz06/counterexample_dzhz06_ours}
    \caption{The $x$-axis is $n$ where $A\in \mathbb{R}^{n\times n}$, and the $y$-axis is $\| A' - A\|_1$ where $\rank(A')=k$. This figure shows the performance of both our algorithm and \cite{dzhz06} on input matrix $A$ defined as Equation (\ref{eq:def_ce_dzhz06}). }
\end{figure}

The goal is to find a rank $k=1$ approximation for matrix $A$. For any $\eps\in [0,0.5)$, we define $A\in \mathbb{R}^{n\times n}$ as
\begin{align}\label{eq:def_ce_dzhz06}
A=\begin{bmatrix} n^{1.5+\eps} & 0 \\ 0 & 0 \end{bmatrix} + \begin{bmatrix} 0 & 0 \\ 0 & B \end{bmatrix},
\end{align}
where $B\in \mathbb{R}^{(n-1)\times(n-1)}$ is all $1$s matrix. It is immediate that the optimal cost is at most $n^{1.5+\eps}$. However, using the algorithm in \cite{dzhz06}, the cost is at least $\Omega(n^2)$. Thus, we can conclude,
using algorithm \cite{dzhz06} to solve $\ell_1$ low rank approximation problem on $A$ cannot achieve an approximation ratio better than $n^{0.5-\eps}$.

\subsection{Counterexample for \cite{bdb13}}

\begin{figure}[t]
  \centering
    \includegraphics[width=0.7\textwidth]{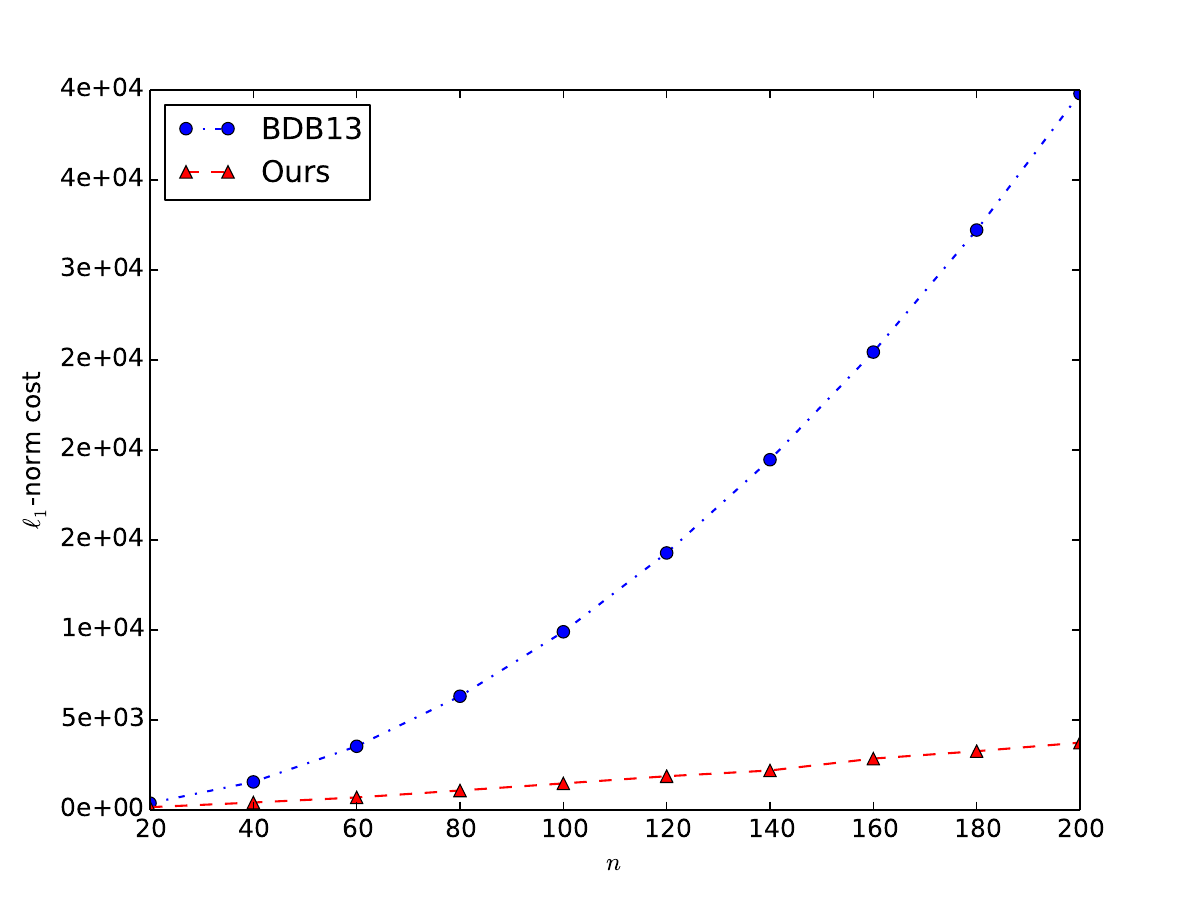}%{./exp_figs/bdb13/counterexample_bdb13_ours} %%% change for arxiv
    \caption{The $x$-axis is $n$ where $A\in \mathbb{R}^{n\times n}$, and the $y$-axis is $\| A' - A\|_1$ where $\rank(A')=k$. This figure shows the performance of both our algorithm and \cite{bdb13} on input matrix $A$ defined as Equation (\ref{eq:def_ce_bdb13}).}
\end{figure}

%The goal is to find a rank $k=1$ approximation for matrix $A$.
%The input matrix $A \in \mathbb{R}^{d \times d}$ for algorithm \cite{bdb13} is defined as,
%\begin{align}\label{eq:def_ce_bdb13}
%A = B + C
%\end{align}
%where $B$ is the matrix that contains $1$ everywhere and $C$ is a diagonal matrix where the first $s$ coordinates are $a$ and the remaining $n-s$ entries are $0$. Consider a rank-1 matrix $A'$ which contain all $1$s, it means the optimal is cost is at most $ \| A - A'\|_1 = sa$. Because their algorthim will rotate the matrix based on its svd, they will try to fit those outliers on the diagonal. But if they pay more weight on the outliers, they cannot fit other flat elements very well. Thus they will lose $\omega(n)$ cost on those flat elements. XXX

The goal is to find a rank $k=1$ approximation for matrix $A$.
The input matrix $A \in \mathbb{R}^{d \times d}$ for algorithm \cite{bdb13} is defined to be,
\begin{align}\label{eq:def_ce_bdb13}
A=\begin{bmatrix} n^{1.5} & 0 \\ 0 & 0 \end{bmatrix} + \begin{bmatrix} 0 & 0 \\ 0 & B \end{bmatrix},
\end{align}
where $B\in \mathbb{R}^{(n-1)\times(n-1)}$ is an all $1$s matrix. It is immediate that the optimal cost is at most $n^{1.5}$.
Now, let us look at the procedure of \cite{bdb13}. Basically, the algorithm described in~\cite{bdb13} is that they first find a rank $n-1$ approximation via a best $\ell_1$-fit hyperplane algorithm, then they rotate it based on the right singular vectors of the rank $n-1$ approximation matrix, and next they recursively do the same thing for the rotated matrix which has only $n-1$ columns.

When running their algorithm on $A$, they will fit an arbitrary column except the first column of $A$. Without loss of generality, it just fits the last column of $A$. After the rotation, the matrix will be an $n\times (n-1)$ matrix:
\begin{align*}
\begin{bmatrix}
n^{1.5} & 0 & 0 & \cdots & 0 \\
0 & \sqrt{n-2} & 0 &\cdots & 0 \\
0 & \sqrt{n-2} & 0 &  \cdots & 0\\
\cdots & \cdots & \cdots & \cdots & \cdots\\
0 & \sqrt{n-2} & 0 &  \cdots & 0
\end{bmatrix}.
\end{align*}
Then, after the $t^{\text{th}}$ iteration for $t<(n-1)$, they will get an $n\times(n-t)$ matrix:
\begin{align*}
\begin{bmatrix}
n^{1.5} & 0 & 0 & \cdots & 0 \\
0 & \sqrt{n-2} & 0 & \cdots & 0 \\
0 & \sqrt{n-2} & 0 & \cdots & 0\\
\cdots & \cdots & \cdots & \cdots & \cdots\\
 0 & \sqrt{n-2}& 0 & \cdots & 0
\end{bmatrix}.
\end{align*}
This means that their algorithm will run on an $n \times 2$ matrix:
\begin{align*}
\begin{bmatrix} n^{1.5} & 0  \\ 0 & \sqrt{n-2} \\ 0 & \sqrt{n-2}\\ \cdots & \cdots \\ 0 & \sqrt{n-2} \end{bmatrix},
\end{align*}
in the last iteration. Notice that $n\times \sqrt{n-2}<n^{1.5}$. This means that their algorithm will fit the first column which implies that their algorithm will output a rank-$1$ approximation to $A$ by just fitting the first column of $A$. But the cost of this rank-$1$ solution is $(n-1)^2$. Since the optimal cost is at most $n^{1.5}$, their algorithm cannot achieve an approximation ratio better than $n^{0.5}$.

\subsection{Counterexample for \cite{kwak08}}

\begin{figure}[t]
  \centering
    \includegraphics[width=0.7\textwidth]{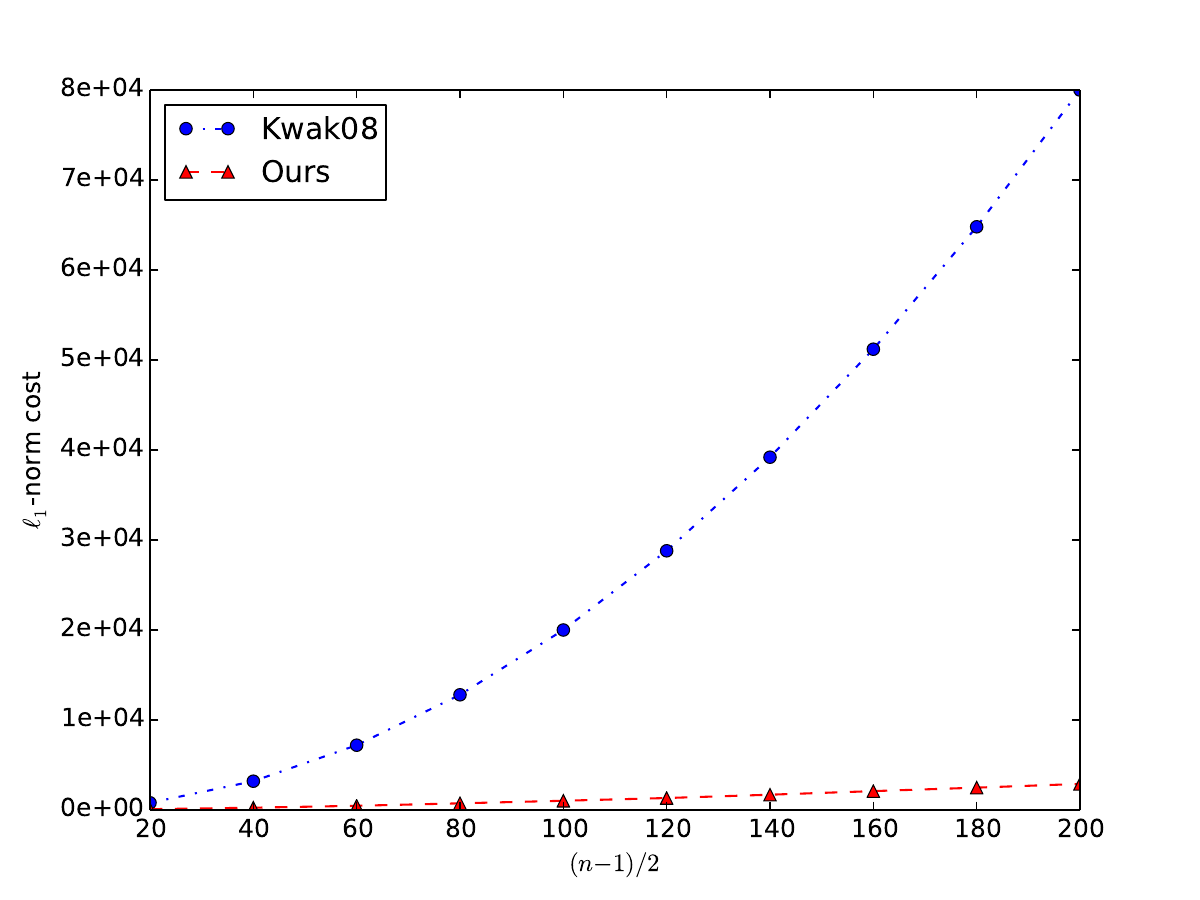}%{./exp_figs/kwak08/counterexample_kwak08_ours} %%%change for arxiv
    \caption{The $x$-axis is $n$ where $A\in \mathbb{R}^{n\times n}$, and the $y$-axis is $\| A' - A\|_1$ where $\rank(A')=k$. This figure shows the performance of both our algorithm and \cite{kwak08} on input matrix $A$ defined as Equation (\ref{eq:def_ce_kwak08}). Algorithm \cite{kwak08} has two ways of initialization, which have similar performance on matrix $A$.} %%% both random initialization and top singular vector
\end{figure}

We show that the algorithm \cite{kwak08} cannot achieve an approximation ratio better than $\Theta(n)$ on the matrix $A\in \mathbb{R}^{ (2n+1) \times (2n+1) }$ defined as,
\begin{align}\label{eq:def_ce_kwak08}
A = \begin{bmatrix} n^{1.5} & 0 & 0 \\ 0 & B & 0 \\ 0 & 0 & B \end{bmatrix},
\end{align}
where $B$ is an $n\times n$ matrix that contains all $1$s.
 We consider the rank-2 approximation problem for this input matrix $A$. The optimal cost is at most $n^{1.5}$.

 We run their algorithm. Let $x_0$ denote the initial random unit vector. Consider the sign vector $s\in \{\pm 1\}^{2n+1}$ where each entry is the sign of the inner product between $x_0$ and each column of $A$. There are only three possibilities,
\begin{align*}
s=
  \begin{cases}
    (1,\{1\}^n,\{1\}^n)  \\
    (1,\{1\}^n,\{-1\}^n) \\
    (1,\{-1\}^n,\{1\}^n) \\
  \end{cases}.
\end{align*}
Case \RN{1}, $s=(1, \{1\}^n, \{1\}^n )$. Define $\wh{u} = \sum_{i=1}^{2n+1} s_i \cdot A_i = (n^{1.5},n,\cdots,n)$. Let $u = \wh{u} / \| u\|_2 = \wh{u} / (\sqrt{3} n^{1.5} ) = (1/\sqrt{3}, 1/\sqrt{3n}, \cdots, 1/\sqrt{3n})$. Define matrix $D$ to be $A- u u^\top A$. Then, we can compute $D$,
\begin{align*}
D=& ~ A - u u^\top A \\
=& ~ A -
\begin{bmatrix}
1/3 & \frac{1}{3\sqrt{n}} {\bf 1}^\top & \frac{1}{3\sqrt{n}} {\bf 1}^\top \\
\frac{1}{3\sqrt{n}} {\bf 1} & \frac{1}{3n}B & \frac{1}{3n}B  \\
\frac{1}{3\sqrt{n}} {\bf 1} & \frac{1}{3n}B & \frac{1}{3n}B
\end{bmatrix} A \\
=&
\begin{bmatrix}
n^{1.5} & 0 & 0 \\
0 & B & 0 \\
0 & 0 & B
\end{bmatrix}
-
\begin{bmatrix}
\frac{n^{1.5}}{3} & \frac{\sqrt{n}}{3} {\bf 1}^\top & \frac{\sqrt{n}}{3} {\bf 1}^\top  \\
\frac{n}{3} {\bf 1} & \frac{1}{3}B & \frac{1}{3}B \\
\frac{n}{3} {\bf 1} & \frac{1}{3}B & \frac{1}{3}B \\
\end{bmatrix}\\
= &
\begin{bmatrix}
2n^{1.5}/3 & - \frac{\sqrt{n}}{3} {\bf 1}^\top & - \frac{\sqrt{n}}{3} {\bf 1}^\top \\
-\frac{n}{3} {\bf 1} & \frac{2}{3} B & -\frac{1}{3} B \\
-\frac{n}{3} {\bf 1} & -\frac{1}{3} B & \frac{2}{3} B \\
\end{bmatrix}.
\end{align*}
Now we need to take the linear combination of columns $ D = A - u u^\top A$.
Let $w$ denote another sign vector $\{-1,+1\}^{2n+1}$. Then let $v$ denote the basis vector, $v=\sum_{i=1}^{2n+1}D_i$. There are three possibilities,  Case I(a), if $w = (1, \{1\}^n,\{1\}^n )$, then $v$ is the all $0$ vector. Using vector $v$ to interpolate each column of $D$, the cost we obtain is at least $ 2n^2$. Case I(b), if $w = (1, \{1\}^n,\{-1\}^n )$, then $v=(2n^{1.5}/3, \{2/3\}^n, \{-4/3\}^n)$. We also obtain at least $2n^2$ cost if we use that $v$ to interpolate each column of $D$. Case I(c), if $w = ( 1, \{-1\}^n, \{-1\}^n )$, then $v=(0, \{-2/3\}^n, \{-2/3\}^n)$. The cost is also at least $2n^2$.

Case \RN{2}, $s=(1, \{1\}^n, \{-1\}^n )$. Define ${\bf 1}$ to be a length $n$ all $1$s column vector. Define $\wh{u} = \sum_{i=1}^{2n+1} s_i \cdot A_i = (n^{1.5}, \{n\}^n, \{-n\}^n)$. Let $u = \wh{u} / \| u\|_2 = \wh{u} / (\sqrt{3} n^{1.5} ) = (1/\sqrt{3}, \{1/\sqrt{3n} \}^n, \{-1/\sqrt{3n}\}^n)$. Define $(2n+1) \times (2n+1)$ matrix $D$ to be $A- u u^\top A$. Then, we can compute $D$,

\begin{align*}
D=& ~ A - u u^\top A \\
=& ~ A -
\begin{bmatrix}
1/3 & \frac{1}{3\sqrt{n}} {\bf 1}^\top & -\frac{1}{3\sqrt{n}} {\bf 1}^\top \\
\frac{1}{3\sqrt{n}} {\bf 1} & \frac{1}{3n} B & -\frac{1}{3n} B \\
-\frac{1}{3\sqrt{n}} {\bf 1} & -\frac{1}{3n}B & \frac{1}{3n} B \\
\end{bmatrix} A \\
=&
\begin{bmatrix}
n^{1.5} & 0 & 0 \\
0 & B & 0 \\
0 & 0 & B
\end{bmatrix}
-
\begin{bmatrix}
n^{1.5}/3 & \frac{\sqrt{n}}{3} {\bf 1}^\top & - \frac{\sqrt{n}}{3} {\bf 1}^\top \\
\frac{n}{3} {\bf 1} & \frac{1}{3} B & -\frac{1}{3} B \\
-\frac{n}{3} {\bf 1} & -\frac{1}{3} B & \frac{1}{3} B \\
\end{bmatrix}\\
= &
\begin{bmatrix}
2n^{1.5}/3 & - \frac{\sqrt{n}}{3} {\bf 1}^\top & \frac{\sqrt{n}}{3} {\bf 1}^\top \\
-\frac{n}{3} {\bf 1} & \frac{2}{3}B & \frac{1}{3} B \\
\frac{n}{3} {\bf 1} & \frac{1}{3} B & \frac{2}{3} B \\
\end{bmatrix}.
\end{align*}
Similarly to the previous case, we can also discuss three cases.

Case \RN{3}, $s=(1, \{1\}^n, \{-1\}^n )$. Define ${\bf 1}$ to be a length $n$ all $1$s column vector. Define $\wh{u} = \sum_{i=1}^{2n+1} s_i \cdot A_i = (n^{1.5}, \{-n\}^n, \{-n\}^n)$. Let $u = \wh{u} / \| u\|_2 = \wh{u} / (\sqrt{3} n^{1.5} ) = (1/\sqrt{3}, \{-1/\sqrt{3n} \}^n, \{-1/\sqrt{3n}\}^n)$. Define matrix $D$ to be $A- u u^\top A$. Then, we can compute $D$,

\begin{align*}
D=& ~ A - u u^\top A \\
=& ~ A -
\begin{bmatrix}
1/3 & -\frac{1}{ 3\sqrt{n} }) {\bf 1}^\top & - \frac{1}{ 3\sqrt{n} } {\bf 1}^\top \\
-\frac{1}{ 3\sqrt{n} } {\bf 1} & \frac{1}{3n}B & \frac{1}{3n} B \\
-\frac{1}{ 3\sqrt{n} } {\bf 1} & \frac{1}{3n}B & \frac{1}{3n} B \\
\end{bmatrix} A \\
=&
\begin{bmatrix}
n^{1.5} & 0 & 0 \\
0 & B & 0 \\
0 & 0 & B
\end{bmatrix}
-
\begin{bmatrix}
n^{1.5}/3 & - \frac{\sqrt{n}}{3} {\bf 1}^\top & -\frac{\sqrt{n}}{3} {\bf 1}^\top \\
-\frac{n}{3} {\bf 1} & \frac{1}{3} B & \frac{1}{3} B \\
-\frac{n}{3} {\bf 1} & \frac{1}{3} B & \frac{1}{3} B \\
\end{bmatrix}\\
= &
\begin{bmatrix}
2n^{1.5}/3 & \frac{\sqrt{n}}{3} {\bf 1}^\top & \frac{\sqrt{n}}{3} {\bf 1}^\top \\
\frac{n}{3} {\bf 1} & \frac{2}{3} B & \frac{1}{3} B \\
\frac{n}{3} {\bf 1} & \frac{1}{3} B & \frac{2}{3} B \\
\end{bmatrix}.
\end{align*}
Similarly to the previous case, we can also discuss three cases.

\subsection{Counterexample for \cite{kk05}}

We show that there exist matrices such that the algorithm of \cite{kk05} cannot achieve an approximation ratio better than $\Theta(n)$. Their algorithm has two different ways of initialization. We provide counterexamples for each of the initialization separately.

\begin{figure}[t]
  \centering
    \includegraphics[width=0.47\textwidth]{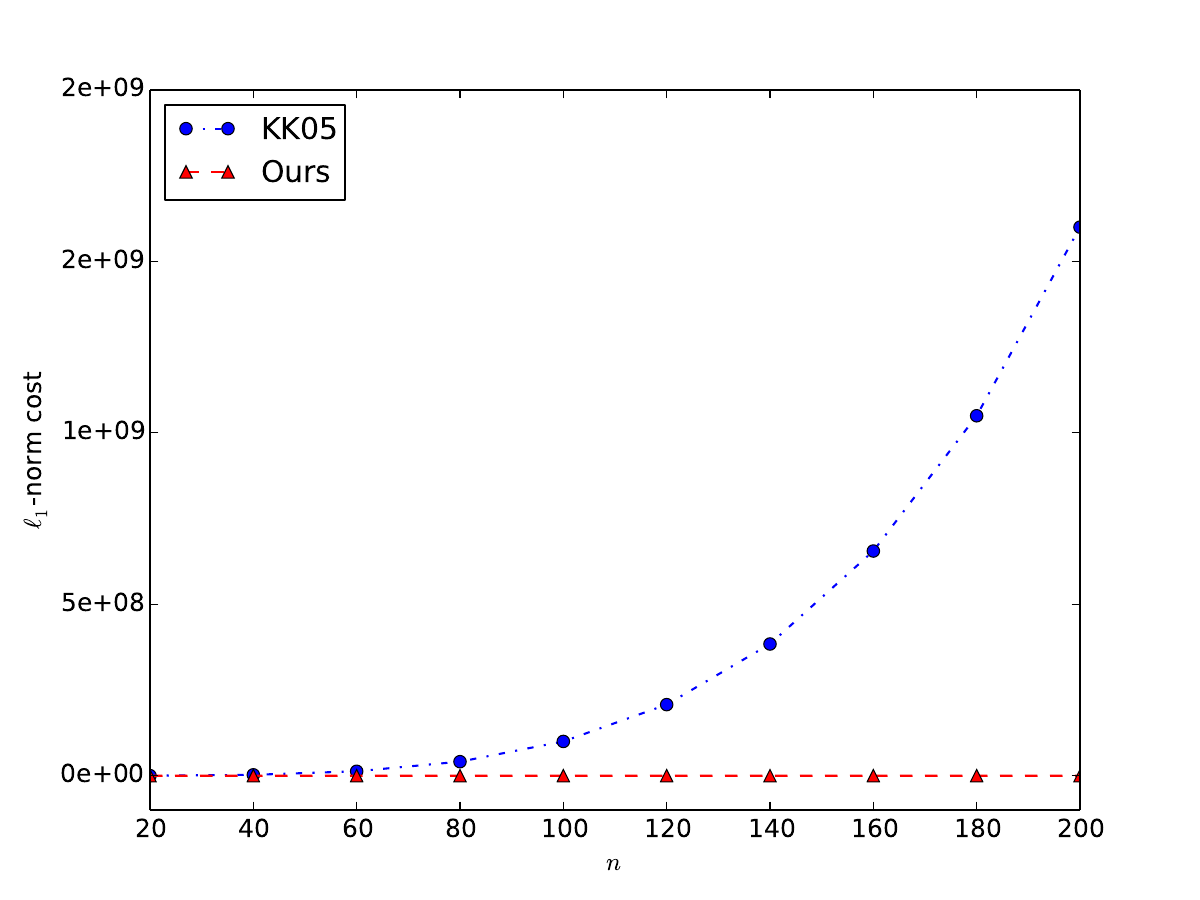}%{./exp_figs/kk05/counterexample_kk05_ours1}  %%% change for arxiv
    \hspace{-2mm} %%% random
    \includegraphics[width=0.47\textwidth]{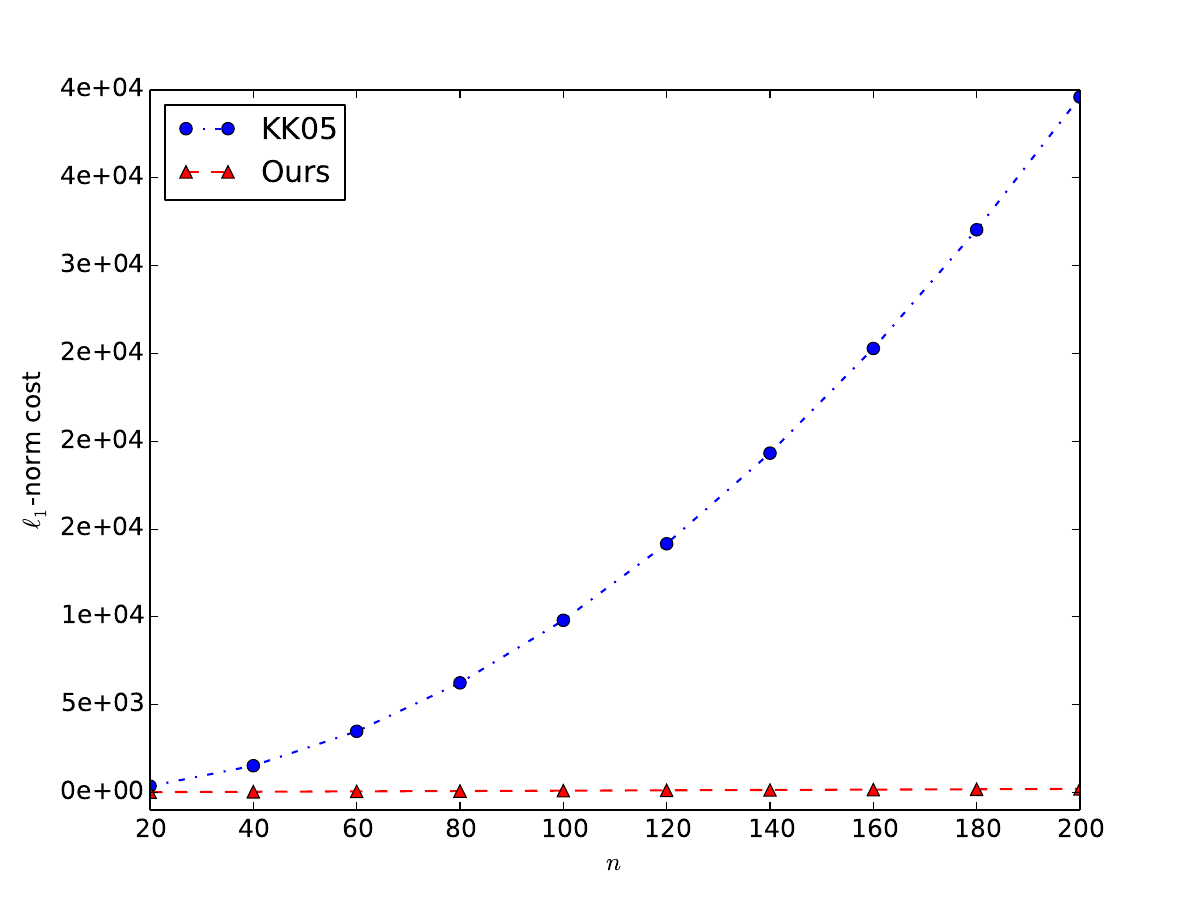}%{./exp_figs/kk05-l2pca/counterexample_kk05_ours2} %%% top singular vector %%% change for arxiv
    \caption{The $x$-axis is $n$ where $A\in \mathbb{R}^{n\times n}$, and the $y$-axis is $\| A' - A\|_1$ where $\rank(A')=k$. Algorithm \cite{kk05} has two ways of initialization. The left figure shows the performance of both our algorithm and \cite{kk05} with random vector initialization on input matrix $A$ defined as Equation (\ref{eq:def_ce_kk05_a}). The right figure shows the performance of both our algorithm and \cite{kk05} with top singular vector initialization on input matrix $A$ defined as Equation (\ref{eq:def_ce_kk05_b}).}
\end{figure}
\paragraph{Random vector initialization}

We provide a counterexample matrix $A\in \mathbb{R}^{ n \times n }$ defined as,
\begin{align}\label{eq:def_ce_kk05_a}
A =
\begin{bmatrix}
n^c & 0 & \cdots & 0 \\
0 & 0 & \cdots & 0 \\
\cdots & \cdots & \cdots & \cdots \\
0 & 0 & \cdots & 0
\end{bmatrix} + I,
\end{align}
where $c\geq 2$. Consider the rank-1 approximation problem for this matrix $A$. The optimal cost is at most $n-1$.

 Run their algorithm. The starting vectors are $u(0) \sim N(0,I)$ and $v(0) \sim N(0,1)$.
We define two properties for a given vector $y\in \mathbb{R}^n$. Property \RN{1} is for all $i\in [n]$, $ |y_i| \leq n/8$, and Property \RN{2} is there exist half of the $i$ such that $|y_i| \geq 1/2$.
We can show that with probability $1-2^{-\Omega(n)}$, both $u(0)$ and $v(0)$ satisfy Property \RN{1} and \RN{2}. After $1$ iteration, we can show that $u(1)_1 = v(1)_1 = 0$

Now let us use column vector $u(0) \in \mathbb{R}^n$ to interpolate the first column of matrix $A$. We simplify $u(0)$ to be $u$. Let $u_i$ denote the $i$-th coordinate of vector $u$, $\forall i\in [n]$. Let $A_1$ denote the first column of matrix $A$. We define $\alpha = v(1)_1 = \arg\min_{\alpha} \| \alpha u(0) - A_1 \|_1$. For any scalar $\alpha$, the cost we pay on the first column of matrix $A$ is,
\begin{align}
 & ~ | \alpha \cdot u_1 -  n^c | + \sum_{i=2}^n | \alpha u_i|
\geq  ~  | \alpha \cdot u_1 -  n^c | + \frac{n}{2} | \alpha \cdot  \frac{1}{2} |
\geq ~ | \alpha \cdot u_1 -  n^c | + \frac{n}{4} | \alpha|.
\end{align}
Notice that $c\geq 2$. Then $| \alpha \cdot u_1 -  n^c | + \frac{n}{4} | \alpha|  \geq | \alpha \cdot \frac{n}{8} -  n^c | + \frac{n}{4} | \alpha| $. If $\alpha\leq 0$, then the cost is minimized when $\alpha=0$. If $\alpha \in [0, 8n^{c-1}]$, the cost is minimized when $\alpha=0$. If $\alpha \geq 8n^{c-1}$, the cost is minimized when $\alpha = 8n^{c-1}$. Putting it all together, to achieve the minimum cost, there is only one choice for $\alpha$, which is $\alpha=0$. The optimal cost is at least $n^c$.

Then after $T$ iterations(for any $T\geq 1$), $u(T)_1 = v(T)_1=0$. Thus, we always pay at least $n^c$ cost on the first entry.

Therefore, their algorithm cannot achieve any approximation ratio better than $n^{c-1}$. Because $c\geq 2$, we complete the proof.

\paragraph{Top singular vector initialization}
The counterexample input matrix $A\in \mathbb{R}^{n\times n}$ is defined as,
\begin{align}\label{eq:def_ce_kk05_b}
A =
\begin{bmatrix}
n & 0 \\
0 & B \\
\end{bmatrix},
\end{align}
where matrix $B\in \mathbb{R}^{ (n-1) \times (n-1)}$ contains all $1$s. Consider the rank-1 approximation problem for this matrix $A$. The optimal cost is at most $n$. Run their algorithm. The starting vectors $u(0)$ and $v(0)$ will be set to $(1,0,\cdots,0) \in \mathbb{R}^n$. After $T$ iterations(for any $T>0$), the support of $u(T)$(resp. $v(T)$) is the same as $u(0)$(resp. $v(T)$). Thus, the cost is at least $\| B\|_1 = (n-1)^2$. Therefore, we can conclude that their algorithm cannot achieve any approximation ratio  better than $(n-1)^2 / n = \Theta(n)$.

\subsection{Counterexample for all}

\begin{figure}[!ht]
  \centering
\subfloat[]{
    \includegraphics[width=0.7\textwidth]{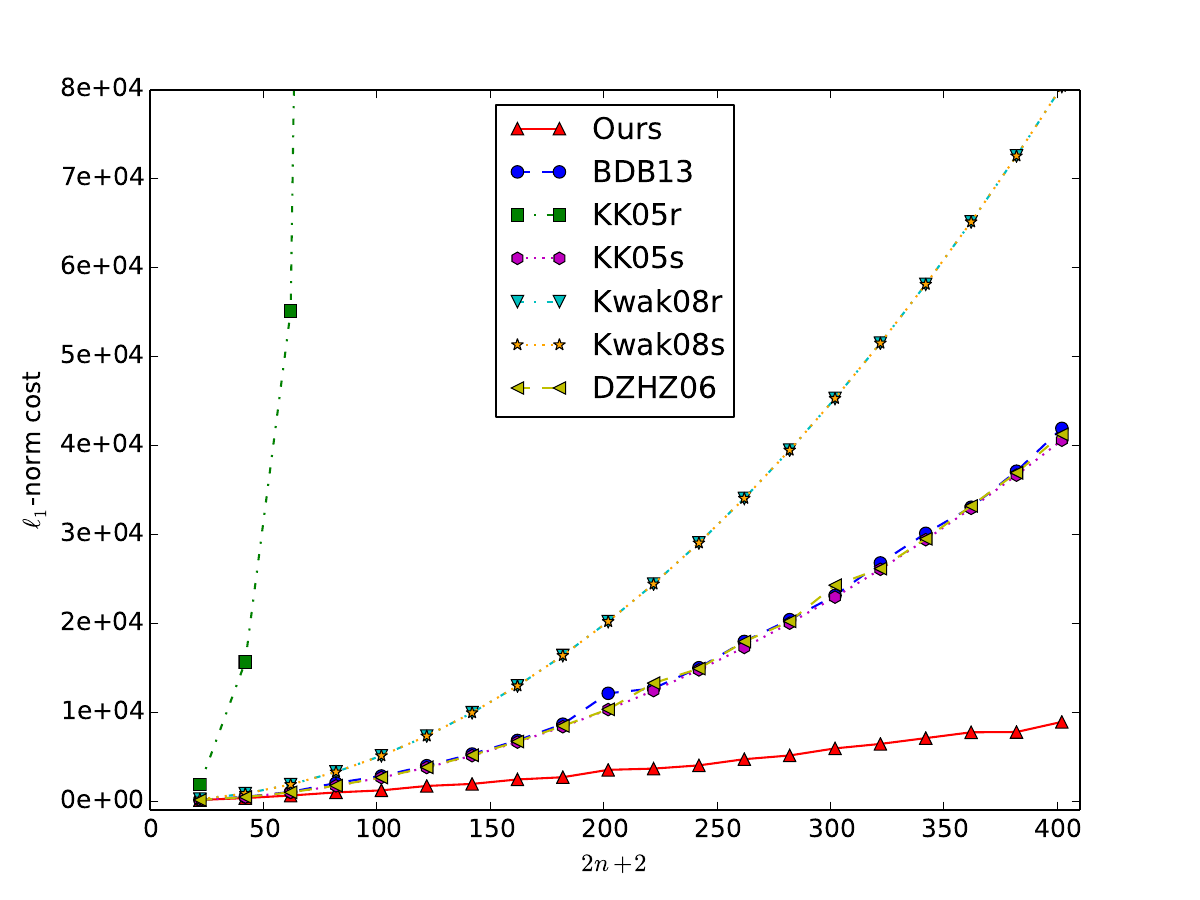} %{./exp_figs/all/counterexample_all_ours2} %%% change for arxiv
}
\vspace{-3mm}
\subfloat[]{
    \includegraphics[width=0.7\textwidth]{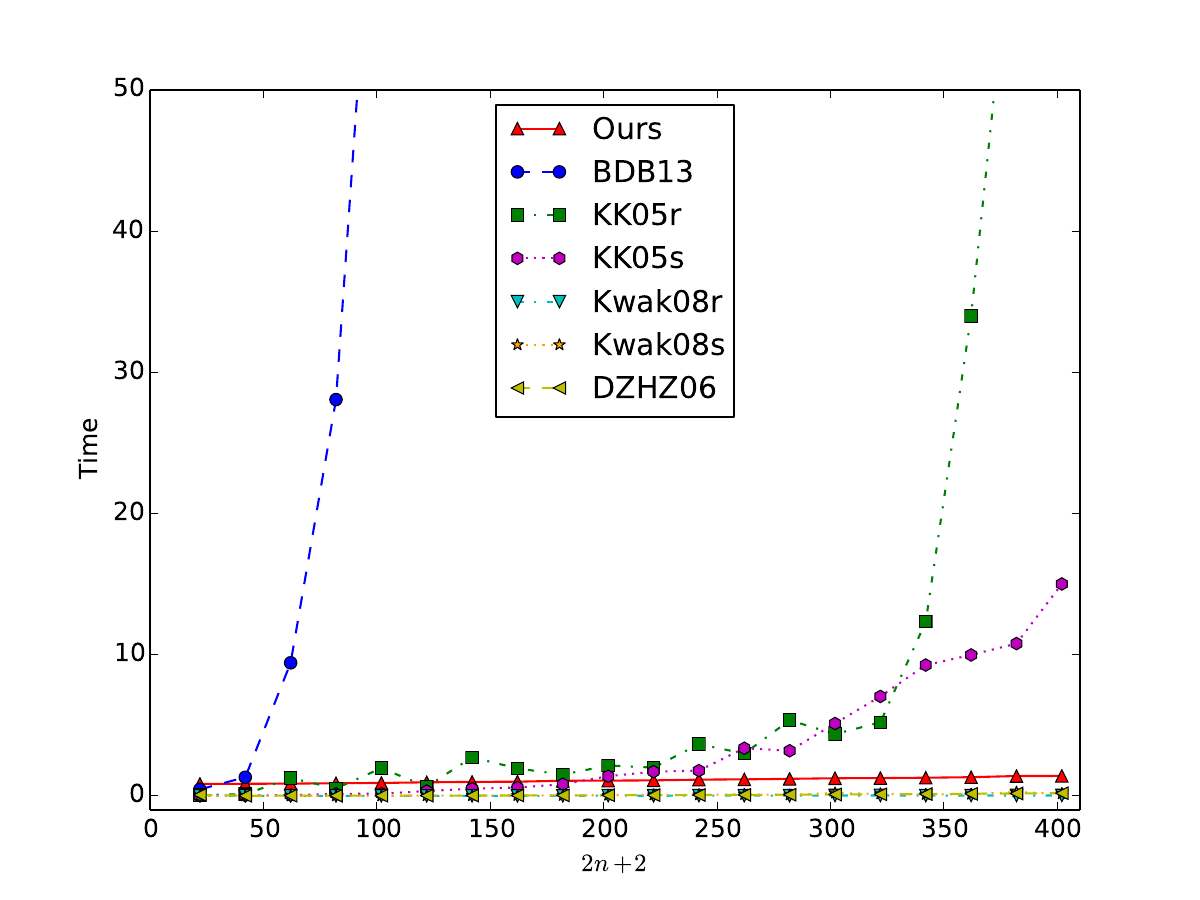}%{./exp_figs/all/counterexample_all_ours2_time} %%% change for arxiv
}
    \caption{Let $A$ be $(2n+2)\times (2n+2)$ input matirx. (a) shows the performance of all the algorithms when the matrix dimension is growing. The $x$-axis is $n$, and the $y$-axis is $\| A' - A\|_1$ where $A'$ is the rank-$3$ solution output by all the heuristic algorithms and also ours. The $\ell_1$ residual cost of all the other algorithms is growing much faster than ours, which is consistent with our theoretical results. (b) shows the running time (in seconds) of all the algorithms when the matrix dimension $n$ is growing. The $x$-axis is $n$ and the $y$-axis is time (seconds). The running time of  some of the algorithms is longer than $3$ seconds. For most of the algorithms (including ours), the running time is always less than $3$ seconds.}\label{fig:all_vs_ours}
\end{figure}

For any $\eps \in (0,0.5)$ and $\gamma >0$, we construct the input matrix $A\in \mathbb{R}^{(2n+2)\times (2n+2)}$ as follows
\begin{align*}
A=
\begin{bmatrix}
n^{2+\gamma} & 0 & 0 & 0 \\
0 & n^{1.5+\epsilon} & 0 & 0 \\
0 & 0 & B & 0 \\
0 & 0 & 0 & B
\end{bmatrix},
\end{align*}
where $B$ is $n \times n$ all $1$s matrix.
We want to find a rank $k=3$ solution for $A$. Then any of those four heuristic algorithms \cite{kk05,dzhz06,kwak08,bdb13} is not able to achieve better than $n^{\min ( \gamma, 0.5-\eps ) } $ approximation ratio. We present our main experimental results in Figure~\ref{fig:all_vs_ours}. Both \cite{kk05} and \cite{kwak08} have two different ways of initialization. In Figure~\ref{fig:all_vs_ours} we use KK05r(resp. Kwak08r) to denote the way that uses random vector as initialization, and use KK05s(resp. Kwak08s) to denote the way that uses top singular vector as initialization. Figure~\ref{fig:all_vs_ours}(a) shows the performance of all the algorithms and Figure~\ref{fig:all_vs_ours}(b) presents the running time. The $\ell_1$ residual cost of all the other algorithm is growing much faster than our algorithm. Most of the algorithms (including ours) are pretty efficient, i.e., the running time is always below $3$ seconds. The running time of \cite{bdb13,kk05} is increasing very fast when the matrix dimension $n$ is growing.

In Figure~\ref{fig:all_vs_ours}(a), the cost of KK05r at $\{ 82,\cdots, 142\}$ is in $[10^5, 10^6 ]$, at $\{ 162,\cdots, 302\}$ is in $[10^6, 10^7]$, and at $\{322,\cdots,402\}$ is in $[10^7, 10^8]$. In Figure~\ref{fig:all_vs_ours}(b), the time of KK05r at $\{382,482\}$ is $64$s and $160$s. The running time of BDB13 at $\{ 82,\cdots,222\}$ is between $1$ minute and $1$ hour. The running time of BDB13 at $\{ 242,\cdots,322\}$ is between $1$ hour and $20$ hours. The running time of BDB13 at $\{ 342, \cdots, 402\}$ is more than $20$ hours.

%%% We won't show any real-life datasets any more
%\subsection{Real-life datasets}

%This \url{http://www.optimization-online.org/DB_FILE/2012/04/3436.pdf} package  \cite{bj12} implemented three methods, PCA-L1 (\cite{kwak08}),
%L1-PCA (\cite{kk03} \cite{kk05}), and L1-PCA$^*$ (\cite{bdb13}).
%Not sure if we still need to compare with \cite{nhclw11}

\subsection{Discussion for Robust PCA~\cite{clmw11}}

A popular method is robust PCA \cite{clmw11}, which given a matrix $A$, tries to find a matrix $L$ for which $\lambda \|A-L\|_1 + \|L\|_*$ is minimized, where $\lambda > 0$ is a tuning parameter and $\|L\|_*$ is the nuclear norm of $L$. This is a convex program, but it need not return a low rank matrix $L$ with relative error. As a simple example, suppose $\lambda = 1$, and the $n\times n$ matrix $A$ is a block-diagonal matrix of rank $k$ and $n=\frac{k}{2}(b+1)$. Further, the first $k/2$ blocks are $b \times b$ matrices of all $1$s, while the next $k/2$ blocks are just a single value $b$ on the diagonal.

Then the solution to the above problem may return $L$ to be the first $k/2$ blocks of $A$. The total cost of $\lambda\| A-L\|_1+\| L\|_*$is $ (k/2)b +  (k/2)b = kb$.

Also, the solution to the above problem may return $L$ to be $A$, which has cost $0+ \| A\|_* = kb$.
Because this solution has the same cost, it means that their algorithm might output a rank-$k$ solution, and also might output a rank-$k/2$ solution.

Therefore, the relative error of the output matrix may be arbitrarily bad for $\ell_1$-low rank approximation.

We also consider the following example. Suppose $\lambda=1/\sqrt{n}$, and let $n\times n$ matrix $A$ denote the Hadamard matrix $H_n$.
Recall that the Hadamard matrix $H_p$ of size $p\times p$ is defined recursively : $ \begin{bmatrix} H_{p/2} & H_{p/2} \\ H_{p/2} & - H_{p/2} \end{bmatrix} $ with $H_2 = \begin{bmatrix} +1 & +1 \\ +1 & -1 \end{bmatrix}$. Notice that every singular values of $A$ is $\sqrt{n}$. We consider the objective function $\lambda \| A-L\|_1 + \| L\|_*$.

Then the solution to the above problem may return $L$ to be the first $n/2$ rows of $A$. The total cost of $\lambda\| A-L\|_1+\| L\|_*$ is $ (1/\sqrt{n}) n^2/2 +  (n/2) \sqrt{n} = n^{1.5}$.
Also, the solution to the above problem may return $L$ to be $A$, which has cost $0+ \| A\|_* = n \sqrt{n} = n^{1.5}$.
For any $i$, if the solution takes $i$ rows of $A$, the cost is $(1/\sqrt{n}) (n-i) n + i \sqrt{n} = n^{1.5}$.
Because this solution has the same cost, it means that their algorithm might output a rank-$n$ solution, and also might output a rank-$n/2$ solution. Therefore, the relative error of the output matrix may be arbitrarily bad for $\ell_1$-low rank approximation.

%Then the solution to the above problem may return $L$ to be the last $k/2$ blocks of $A$. The total cost of $\lambda\| A-L\|_1+\| L\|_*$is $ (k/2)b^2 +  (k/2)b$.

%which is the same as the cost of $A$ itself, which is rank-$k$.

%Notice that $\| A\|_1 = \frac{k}{2} b^2 + \frac{k}{2}b=kb(b+1)/2$.

%The total cost is $b$, which is the same as the cost of $A$ itself, which is rank-$k$. Therefore, the relative error of the output matrix may be arbitrarily bad for $\ell_1$-low rank approximation.

%We consider $A\in\mathbb{R}^{n\times n}$ as the following block diagonal matrix:
%\begin{align*}
%A=\begin{bmatrix}
%B & 0 & 0 & \cdots & 0\\
%0 & C & 0 & \cdots & 0\\
%0 & 0 & C & \cdots & 0\\
%\cdots  & \cdots  & \cdots  & \cdots & \cdots\\
%0 & 0 & 0 & \cdots & 0\\
%\end{bmatrix},
%\end{align*}
%where $B\in\mathbb{R}^{b\times b}$ is an all $1$ matrix, and
%\begin{align*}
%C=\begin{bmatrix}
%b/2 & b/2\\
%b/2 & b/2\\
%\end{bmatrix}.
%\end{align*}

%We just set $\lambda=1/\sqrt{n}$ which is chosen in Theorem

\section{Acknowledgments}
%abcdefghijklmnopqrstuvwxyz
The authors would like to thank Alexandr Andoni, Saugata Basu, Cho-Jui Hsieh, Daniel Hsu, Chi Jin, Fu Li, Ankur Moitra, Cameron Musco, Richard Peng, Eric Price, Govind Ramnarayan, James Renegar, and Clifford Stein for useful discussions. The authors also thank Jiyan Yang, Yinlam Chow, Christopher R\'{e}, and Michael Mahoney for sharing the code.

\newpage
\addcontentsline{toc}{section}{References}
\bibliographystyle{alpha}
\bibliography{ref}
\newpage
%%% some writing rules

%% Writing rule for creating tags.
%% Tags :
%% Theorem    \ref{thm:bla_bla}
%% Lemma      \ref{lem:bla_bla}
%% Claim      \ref{cla:bla_bla}
%% Corollary  \ref{cor:bla_bla}
%% Fact       \ref{fac:bla_bla}
%% Definition \ref{def:bla_bla}
%% Section    \ref{sec:bla_bla}
%% Subsection \ref{sub:bla_bla}
%% Equation   \ref{eq:bla_bla}

\end{document}